\pdfoutput=1
\documentclass[11pt,a4paper]{article}

\usepackage[margin=1in]{geometry}

\usepackage{mathpazo}

\usepackage{mathrsfs}

\usepackage{microtype}

\usepackage{bm}

\usepackage{enumerate}

\usepackage[
  breaklinks = true, 
  pdfusetitle = true, 
  colorlinks,
  linkcolor = blue,
  citecolor = blue,
  urlcolor = blue]{hyperref}

\usepackage{tikz}
\usetikzlibrary{calc,decorations.markings,backgrounds}
\usetikzlibrary {arrows.meta}

\usepackage{amsmath,amssymb,amsthm}
\numberwithin{equation}{section}

\usepackage{mathtools}
\mathtoolsset{centercolon}

\usepackage{thmtools}

\usepackage{cleveref}
\crefname{lemma}{Lemma}{Lemmas}
\crefname{definition}{Definition}{Definitions}
\crefname{theorem}{Theorem}{Theorems}
\crefname{conjecture}{Conjecture}{Conjectures}
\crefname{section}{Section}{Sections}
\crefname{claim}{Claim}{Claims}
\crefname{appendix}{Appendix}{Appendices}
\crefname{figure}{Fig.}{Figs.}
\crefname{table}{Table}{Tables}
\crefname{prop}{Proposition}{Propositions}
\crefname{cor}{Corollary}{Corollaries}

\usepackage{tcolorbox}

\newtcolorbox[auto counter, crefname = {Algorithm}{Algorithms}]%
{alg}[2][]{%
  colbacktitle = gray!50,
  coltitle = black,
  title = \textbf{Algorithm~\thetcbcounter:} #2,%
  #1
}

\usepackage{epigraph}

\newcommand{\R}{\mathbb{R}}

\newcommand{\C}{\mathbb{C}}
\newcommand{\Z}{\mathbb{Z}}
\newcommand{\Sphere}{\mathbb{S}}

\DeclarePairedDelimiter{\set}{\lbrace}{\rbrace}
\DeclarePairedDelimiter{\abs}{\lvert}{\rvert}
\DeclarePairedDelimiter{\norm}{\lVert}{\rVert}
\DeclarePairedDelimiter{\floor}{\lfloor}{\rfloor}

\DeclarePairedDelimiter{\of}{\lparen}{\rparen}
\DeclarePairedDelimiter{\sof}{\lbrack}{\rbrack}

\newcommand{\ket}[1]{|{#1}\rangle}
\newcommand{\bra}[1]{\langle{#1}|}
\newcommand{\braket}[2]{\langle{#1}|{#2}\rangle}
\newcommand{\ketbra}[2]{\ket{#1}\bra{#2}}
\newcommand{\proj}[1]{\ketbra{#1}{#1}}

\newcommand{\ct}{^{\dagger}}
\newcommand{\tp}{^{\mathsf{T}}}

\newcommand{\+}{\oplus}
\newcommand{\x}{\otimes}
\newcommand{\xp}[1]{^{\otimes #1}}
\newcommand{\ctxp}[1]{^{\dagger\otimes #1}}

\newcommand{\U}[1]{\mathrm{U}(#1)}
\newcommand{\SU}[1]{\mathrm{SU}(#1)}
\newcommand{\GL}[1]{\mathrm{GL}(#1)}
\newcommand{\SL}[1]{\mathrm{SL}(#1)}
\renewcommand{\S}{\mathrm{S}} 

\newcommand{\mc}[1]{\mathcal{#1}}
\newcommand{\id}{I}

\newcommand{\pt}{\mathbin{\vdash}} 
\newcommand{\Usch}{U_{\mathrm{Sch}}} 
\newcommand{\J}[1]{J(#1)} 
\newcommand{\JJ}{J'} 
\newcommand{\Sym}[1]{\mathrm{Sym}(#1)} 

\newcommand{\tin}{\mathrm{in}}
\newcommand{\tout}{\mathrm{out}}
\newcommand{\din}{{d_\tin}}   
\newcommand{\dout}{{d_\tout}} 

\newcommand{\ie}{\textit{i.e.}}
\newcommand{\eg}{\textit{e.g.}}
\newcommand{\cf}{\textit{cf.}}

\newcommand{\mx}[1]{\begin{pmatrix}#1\end{pmatrix}}
\newcommand{\smx}[1]{\begin{psmallmatrix}#1\end{psmallmatrix}} 
\newcommand{\rmx}[1]{\begin{pmatrix*}[r]#1\end{pmatrix*}} 
\newcommand{\srmx}[1]{
  \setcounter{MaxMatrixCols}{20}
  \scalebox{0.8}{\mbox{\ensuremath{\displaystyle\rmx{#1}}}}
}

\newcommand{\ID}{\mathrm{ID}}
\newcommand{\OR}{\mathrm{OR}}
\newcommand{\AND}{\mathrm{AND}}
\newcommand{\NOT}{\mathrm{NOT}}
\newcommand{\MAJ}{\mathrm{MAJ}}
\newcommand{\PAR}{\mathrm{PAR}}
\newcommand{\NMAJ}{\mathrm{NMAJ}}
\newcommand{\NPAR}{\mathrm{NPAR}}

\newcommand{\UNOT}{\mathrm{UNOT}}

\newcommand{\CG}{C}
\newcommand{\DCG}{D}
\newcommand{\DCGct}{\DCG\ct}

\DeclareMathOperator{\SWAP}{SWAP}

\DeclareMathOperator{\Tr}{Tr}
\DeclareMathOperator{\spn}{span}
\DeclareMathOperator{\diag}{diag}

\DeclareMathOperator{\rank}{rank}
\DeclareMathOperator{\Comm}{Comm}

\newcommand{\CPTP}[2]{\mathrm{C}(#1,#2)}
\newcommand{\UCPTP}[2]{\mathrm{UC}(#1,#2)} 
\renewcommand{\L}{\mathrm{L}} 

\newcommand{\0}{0}

\newcommand{\vt}{\mathbf{t}}

\newtheorem{theorem}{Theorem}[section] 
\newtheorem{lemma}[theorem]{Lemma} 
\newtheorem{cor}[theorem]{Corollary}

\newtheorem{prop}[theorem]{Proposition}
\newtheorem{example}[theorem]{Example}
\newtheorem{definition}[theorem]{Definition}

\newcommand{\email}[1]{\href{mailto:#1}{\texttt{#1}}}

\title{Quantum majority vote}

\author{
Harry Buhrman\thanks{QuSoft, CWI, and University of Amsterdam, \email{harry.buhrman@cwi.nl}.}
\and
Noah Linden\thanks{University of Bristol, \email{n.linden@bristol.ac.uk}.}
\and
Laura Man\v{c}inska\thanks{University of Copenhagen, \email{mancinska@math.ku.dk}.}
\and
Ashley Montanaro\thanks{Phasecraft Ltd and University of Bristol, \email{ashley.montanaro@bristol.ac.uk}.}
\and
Maris Ozols\thanks{QuSoft and University of Amsterdam, \email{marozols@gmail.com}.}
}

\begin{document}
\maketitle

\begin{abstract}
Majority vote is a basic method for amplifying correct outcomes that is widely used in computer science and beyond. While it can amplify the correctness of a quantum device with classical output, the analogous procedure for quantum output is not known. We introduce \emph{quantum majority vote} as the following task: given a product state $|\psi_1\rangle \otimes \dots \otimes |\psi_n\rangle$ where each qubit is in one of two orthogonal states $|\psi\rangle$ or $|\psi^\perp\rangle$, output the majority state. We show that an optimal algorithm for this problem achieves worst-case fidelity of $1/2 + \Theta(1/\sqrt{n})$. Under the promise that at least $2/3$ of the input qubits are in the majority state, the fidelity increases to $1 - \Theta(1/n)$ and approaches $1$ as $n$ increases.

We also consider the more general problem of computing any symmetric and equivariant Boolean function $f: \{0,1\}^n \to \{0,1\}$ in an unknown quantum basis, and show that a generalization of our quantum majority vote algorithm is optimal for this task. The optimal parameters for the generalized algorithm and its worst-case fidelity can be determined by a simple linear program of size $O(n)$. The time complexity of the algorithm is $O(n^4 \log n)$ where $n$ is the number of input qubits.
\end{abstract}

\setcounter{tocdepth}{2}
\tableofcontents

\setlength{\epigraphwidth}{10cm}
\epigraph{Majority rule is a precious, sacred thing worth dying for. But—like other precious, sacred things—it's not only worth dying for; it can make you wish you were dead. Imagine if all of life were determined by majority rule. Every meal would be a pizza.
}{\textit{P.~J.~O'Rourke}}

\section{Introduction}\label{sec:Intro}

For most practical matters, it is safe to assume that the material world around us obeys the laws of classical physics.
However, explaining certain phenomena, particularly at small scales, requires a more fundamental theory -- quantum mechanics.
Similarly, our digital world at present is fully classical, \ie, it consists of zeroes and ones, while the ultimate description of information is quantum mechanical \cite{Watrous}.
Quantum information, at least from a mathematical perspective, is not only more fundamental but also more abundant -- classical states form a measure zero subset of quantum states.
From this perspective, a (reversible) classical algorithm is nothing but a quantum algorithm whose state remains in the standard basis throughout computation.
In particular, classical algorithms constitute only a measure zero subset of quantum algorithms.\footnote{If we think of an algorithm as a CPTP map (see \cref{sec:Math preliminaries}), probabilistic algorithms are described by diagonal Choi matrices that constitute a measure zero subset of ones describing general quantum algorithms.}

While typically quantum algorithms are designed for solving classical problems, \ie, their input and output consists of standard basis states, a quantum algorithm should generally be thought of as a quantum channel that processes quantum information -- it has quantum input and quantum output.
Far from being just a philosophical stance, this is also a pragmatic view motivated by quantum programming.

New quantum (and classical) algorithms are often obtained by combining existing building blocks, \ie, subroutines that solve smaller sub-problems.
Even if the algorithm itself solves a classical problem, its constituent subroutines are often fully quantum.
For example, quantum Fourier transform \cite{Coppersmith}, amplitude amplification \cite{AA1,AA2}, and SWAP test \cite{SWAPtest} are subroutines for intrinsically quantum tasks.
This motivates the development of new fully quantum subroutines that can be used as building blocks of new quantum algorithms.

\emph{Majority vote} is a simple but ubiquitous method that ensures robustness to errors in various contexts, ranging from elections to error-correction.
For example, majority vote is used in essentially all randomized (and in particular also quantum) algorithms to amplify their success probability by repeating the computation several times and selecting the most common outcome.
Since majority vote only deals with classical data, this begs the question of whether there exists a natural analogue to majority vote that is fully quantum.
In fact, one can ask the same question for any Boolean function.

\subsection{Problem}

To state the problem more formally, consider two parties: Alice and Bob.
Alice has an $n$-bit string $x \in \set{0,1}^n$ represented by the $n$-qubit standard basis state $\ket{x} \in \C^{2^n}$. She would like Bob to evaluate some particular Boolean function $f: \set{0,1}^n \to \set{0,1}$ (\eg, the majority function) on her state and return the answer as a single-qubit state $\ket{f(x)} \in \C^2$. However, assume that instead of $\ket{x}$, Bob receives $U\xp{n} \ket{x}$ for some unknown single-qubit unitary $U \in \U{2}$, and his goal is to produce $U \ket{f(x)}$ for the same unknown $U$ (see \cref{fig:Problem}).

\begin{figure}
\centering


\begin{tikzpicture}[thick,
    arr/.style = {arrows = {|[scale = 0.8]->}}
  ]
  \def\W{3.0cm}
  \def\H{1.0cm}
  \node at (-\W,1.8*\H) {Alice};
  \node at ( \W,1.8*\H) {Bob};
  \path (-\W, \H) node {$\ket{x}$};
  \path (-\W,-\H) node {$\ket{f(x)}$};
  \path ( \W, \H) node {$U\xp{n} \ket{x}$};
  \path ( \W,-\H) node {$U \ket{f(x)}$};
  \draw[arr] (-0.6*\W, \H) to (0.6*\W, \H);
  \draw[arr] (0.6*\W,-\H) to (-0.6*\W,-\H);
  \draw[arr] (\W,0.5*\H) to (\W,-0.5*\H);
  \path (0, \H)+(0,0.3) node {$U\xp{n}$};
  \path (0,-\H)+(0,0.3) node {$U\ct$};
  \path (\W,0)+(0.3,0) node {?};
\end{tikzpicture}
\caption{\label{fig:Problem}The ideal functionality of a unitary-equivariant Boolean function. For a given function $f: \set{0,1}^n \to \set{0,1}$, Bob would like to implement the ideal $n$-to-$1$ qubit map $U\xp{n} \ket{x} \mapsto U \ket{f(x)}$ for all $U \in \U{2}$ and $x \in \set{0,1}^n$. Since in general the ideal map cannot be implemented by a completely positive operation (see \cref{sec:Impossiblity}), the problem is to find the best quantum channel that approximates it.}
\end{figure}
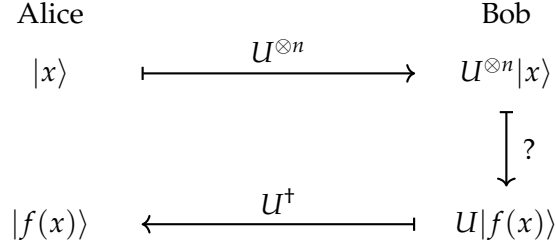

This problem appears naturally in several scenarios. For instance, if the communication channel between Alice and Bob is imperfect, it may cause an unknown unitary drift $U$ to each qubit transmitted in one direction, and $U\ct$ when transmitted in the reverse direction (this is a simple way to model an idealized lossless optical fiber). Alternatively, the laboratories of Alice and Bob might not share a common reference frame, such as in the case of satellite-based communication \cite{Satellite}. Finally, another, cryptographically-motivated setting is when Alice scrambles (and later unscrambles) the data on purpose to hide it from Bob who has been delegated with computing $f$.

In either case, Bob is faced with the task of computing a Boolean function in an unknown quantum basis determined by a unitary $U$ (both his input and output qubits are in the same unknown basis). We call such basis-independent computation \emph{unitary-equivariant} (see \cref{sec:Defs} for more details). Note that in each of the above scenarios, despite Bob being unaware of the basis, Alice can still correctly interpret his output because it either matches with her basis or she can undo the scrambling unitary $U$ if she did it herself.

This problem is particularly well-motivated when $f$ is the $n$-bit \emph{majority} function. In this case, Bob receives an $n$-qubit product state $\ket{\psi_1} \x \dotsb \x \ket{\psi_n}$ where each qubit $\ket{\psi_i} \in \C^2$ is one of two orthogonal states $\ket{\psi}$ or $\ket{\psi^\perp}$. He is asked to return either $\ket{\psi}$ or $\ket{\psi^\perp}$, depending on which state occurs more often. The difficulty lies in the fact that the basis $\set{\ket{\psi}, \ket{\psi^\perp}}$ is unknown to Bob, so the procedure has to work in any basis. We refer to this problem as \emph{quantum majority vote} since Bob's task is to compute the majority function in an unknown quantum basis.

Just like the classical majority vote, quantum majority vote can also be seen as a technique for mitigating errors. Assume that a given quantum device is supposed to produce some unknown state $\ket{\psi}$, but it occasionally fails and instead produces $\ket{\psi^\perp}$. Quantum majority vote can then be used to distill the correct state from several samples (see \cref{fig:Majority}). In particular, this method can be used to amplify the output fidelity of any quantum-output black-box subroutine within a quantum algorithm.

\begin{figure}
\centering


\begin{tikzpicture}[thick,
    box/.style = {fill = lightgray},
    btn/.style = {fill = red},
    arr/.style = {arrows = {|[scale = 0.8]->}}
  ]
  \def\H{0.6}
  \def\L{0.8}
  \def\W{0.5}
  \draw[box]
     (-\L,-\H) -- ++( 2*\L,0) -- ++( \W, \W) --
    ++(0,2*\H) -- ++(-2*\L,0) -- ++(-\W,-\W) -- cycle;
  \draw (-\L,\H) -- ++(2*\L,0) -- ++(0,-2*\H);
  \draw ( \L,\H) -- ++(\W,\W);
  \path (-0.5*\L,\H)+(\W/2,\W/2) coordinate (B);
  \def\h{0.1}
  \def\l{0.1}
  \def\w{0.05}
  \draw[btn]
    (B)++(\w,\w) -- ++(0,\h) -- ++(-2*\l,0) --
    ++(-2*\w,-2*\w) -- ++(0,-\h) -- ++(2*\l,0) -- cycle;
  \draw (B)++(0,\h)++(\w,\w) -- ++(-2*\w,-2*\w);
  \draw (B)++(-\w,-\w) -- ++(0,\h) -- ++(-2*\l,0);
  \def\l{0.17}
  \def\Rx{0.1}
  \def\Ry{0.3}
  \path (\L,0)+(\W/2,\W/2) coordinate (O);
  \draw (O)++(0,+\Ry) -- ++(\l,0);
  \draw (O)++(0,-\Ry) -- ++(\l,0);
  \draw (O)++(0, \Ry) arc [x radius = \Rx, y radius = \Ry, start angle = 90, end angle = 270];
  \draw[fill = gray] (O)++(\l,0) ellipse [x radius = \Rx, y radius = \Ry];
  \path (O)+(0.6,0) node [anchor = west] {$
    \ket{\psi}
    \ket{\psi}
    \ket{\psi^\perp}
    \ket{\psi}
    \ket{\psi}
    \ket{\psi^\perp}
    \ket{\psi}
    \ket{\psi}
    \ket{\psi}$};
  \path (O)+(10,0) node {$\ket{\psi}$};
  \draw[arr] (O)++(6.6,0) to +(2.8,0);
  \path (O)++(8,0.3) node {?};
\end{tikzpicture}
\caption{\label{fig:Majority}Quantum majority vote: given a sequence of $n$ quantum states from an unknown basis $\set{\ket{\psi}, \ket{\psi^\perp}}$, output the majority state.}
\end{figure}
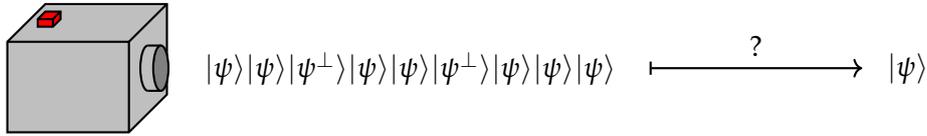

For a general function $f$ it is not possible to perfectly implement the map $U\xp{n}\ket{x} \mapsto U \ket{f(x)}$ for all input strings $x \in \set{0,1}^n$ and unitaries $U \in \U{2}$ since it is not completely positive (see \cref{sec:Impossiblity} for a simple example). Hence, we will instead look for a quantum channel that optimally approximates this map, in the sense that it maximizes the fidelity with the ideal output state for the worst-case input.

The simplest non-trivial instance of this problem is when $n = 1$ and $f = \NOT$ is the Boolean negation. The corresponding ideal map $U\ket{x} \mapsto U\ket{\NOT(x)}$ is known as the \emph{universal $\NOT$} since it is a unitary-equivariant extension of the $\NOT$ function. This anti-unitary map sends an arbitrary qubit state $\ket{\psi} = \smx{\alpha \\ \beta}$ to its unique perpendicular state $\ket{\psi^\perp} = \smx{\overline{\beta} \\ -\overline{\alpha}}$. The best quantum channel that approximates this map achieves the worst-case output fidelity of $2/3$ \cite{BHW99} and has been implemented experimentally \cite{ExperimentalUNOT1,ExperimentalUNOT2,ExperimentalUNOT3,ExperimentalUNOT4}.

\subsection{Our results}\label{sec:Our results}

Our main result is a generic algorithmic template $\mc{A}_\vt$ for computing any unitary-equivariant Boolean function. This template has a tuneable set of interpolation parameters $\vt \in [0,1]^{\floor{n/2}+1}$, for which we derive an efficient linear program that finds the optimal setting of $\vt$,
and we show that the algorithm $\mc{A}_\vt$ can be implemented efficiently.
Below is an informal statement of our main result (see \cref{thm:template,thm:main,thm:implementation} in \cref{sec:Optimal algorithm} for more details).

\begin{theorem}[informal]
Let $f: \set{0,1}^n \to \set{0,1}$ be a symmetric and equivariant Boolean function.
For appropriately chosen interpolation parameters $\vt$, the algorithm $\mc{A}_\vt$ is optimal for computing the unitary-equivariant extension of $f$.
The optimal choice of $\vt$ and the corresponding worst-case fidelity for computing $f$ can be determined by a simple linear program of size $O(n)$.
The resulting algorithm $\mc{A}_\vt$ can be implemented using $O(n^4 \log n)$ quantum gates.
\end{theorem}

We apply our linear programming characterization to determine the optimal fidelity $F_{\MAJ}(n)$ that can be achieved for the unitary-equivariant $n$-bit majority function $\MAJ$.
We distinguish two cases depending on whether the input string $x \in \set{0,1}^n$ is arbitrary or its Hamming weight $|x|$ is at least a constant fraction of $n$ away from $n/2$ (see \cref{sec:Majority} for more details).

\begin{figure}
\centering


\newcommand{\interval}[2]{
  \path (#1) coordinate (#2);
  \begin{scope}[on background layer]
    \draw [line width = 2pt, green!80!black] (#1)+(-\L,0) -- +(\L,0);
  \end{scope}
  \draw (#2)+(-\L,0) -- +(\L,0);
  \draw (#2)++(-\L,0)+(0,\t) node {$0$}+(0,\d) -- +(0,-\d);
  \draw (#2)++( \L,0)+(0,\t) node {$n$}+(0,\d) -- +(0,-\d);
  \draw (#2)         +(0,\t) node {$n\over2$};
  \draw[densely dotted] (#2)+(0,1.6*\d) -- +(0,-3.2*\d);
  \foreach \i in {1,...,14} {
    \draw (#2)++(-\L,0)++(2*\L*\i/15,\d) -- +(0,-2*\d);
  }
}

\newcommand{\curlybrace}[3]{
  \def\r{0.15}
  \def\w{\L/3}
  \path (#1)++(#2\w,-2*\r) coordinate (M)
            ++(0,-0.25) node {#3};
  \draw[radius = \r] (M)
    arc [start angle = 180, end angle = 90] -- ++(\w-2*\r,0)
    arc [start angle = -90, end angle = 0];
  \draw[radius = \r] (M)
    arc [start angle =   0, delta angle = 90] -- ++(-\w+2*\r,0)
    arc [start angle = 270, delta angle = -90];
}


\begin{tikzpicture}[semithick]
  \def\W{4}
  \def\H{1.7}
  \def\h{1.3}
  \path (0,\H+\h) node {No promise}
          +(\W,0) node {Promise};
  \draw (-\W,\h/2) -- (1.5*\W,\h/2);
  \draw (-\W/2,\H+1.4*\h) -- (-\W/2,-1.2*\H);
  \draw ( \W/2,\H+1.4*\h) -- ( \W/2,-1.2*\H);
  \path (0,0) node {$\displaystyle\frac{1}{2} + \frac{1}{2n}$}
      +(\W,0) node {$\displaystyle\frac{5}{6}$}
      +(-0.8*\W,0) node {Trivial};
  \path (0,-\h) node {$\displaystyle\frac{1}{2} + \Theta\of*{\frac{1}{\sqrt{n}}}$}
        +(\W,0) node {$\displaystyle 1 - \Theta\of*{\frac{1}{n}}$}
        +(-0.8*\W,0) node {Optimal};
  \def\d{0.08}
  \def\t{0.4}
  \def\L{0.4*\W}
  \interval{ 0,\H}{L}
  \interval{\W,\H}{R}
  \curlybrace{R}{+}{\scriptsize$n/3$}
  \curlybrace{R}{-}{\scriptsize$n/3$}
  \begin{scope}[on background layer]
    \draw[line width = 2pt, red!80!black] (R)+(-2*\L/3,0) -- +(2*\L/3,0);
  \end{scope}
\end{tikzpicture}
\caption{\label{fig:Comparison}Optimal fidelities for the trivial strategy that picks one of the $n$ input qubits at random, and the optimal strategy. The first column corresponds to the case when there is no promise on the input string $x \in \set{0,1}^n$, and the second when the Hamming weight $|x|$ is at least $n/3$ away from the central value of $n/2$.}
\end{figure}

\begin{restatable}{theorem}{majority}\label{thm:majority}
Let $n \geq 1$ be odd. If $|x|$ is arbitrary, then $F_{\MAJ}(n) = 1/2 + \Theta(1/\sqrt{n})$. If we are promised that $\abs[\big]{|x|-\frac{n}{2}} \geq \frac{n}{3}$, then $F_{\MAJ}(n) = 1 - \Theta(1/n)$.
\end{restatable}

For comparison, the trivial strategy that returns one of the input qubits at random achieves worst-case fidelities $(\frac{n}{2} + \frac{1}{2})/n = \frac{1}{2} + \frac{1}{2n}$ and $\frac{5}{6}$, respectively (see \cref{fig:Comparison}).
Note that in our setting the input state is of the form $\ket{\psi} \x \ket{\psi^\perp} \x \dotsb \x \ket{\psi}$ and in particular is pure.
Hence we cannot directly compare our results to the classical majority vote on the $n$-bit probability distribution $(p,1-p)\xp{n}$ where $p \in [0,1]$, since this would require considering mixed input states of the form $\rho\xp{n}$ where $\rho$ is an arbitrary single-qubit density matrix.
This version of the problem is known as \emph{quantum purification}~\cite{CEM99,KW01}.

\subsection{Related work}

While the idea of processing information encoded by parallel and antiparallel spins has been suggested as an interesting direction already in \cite{Spins}, we are not aware of any work that explores this problem in the setting of computing Boolean functions in an unknown basis.
Thus, here we survey work on problems of similar flavor.

Montanaro has considered a problem that is closely related to ours, namely computing a symmetric Boolean function on a quantum input but producing a \emph{classical} output~\cite{Mon09}.
Since the value of a symmetric function depends only on the input's Hamming weight, this problem effectively reduces to determining the ratio of $\ket{\psi}$ versus $\ket{\psi^\perp}$ in the tensor product input state.
Similar to Montanaro's approach, our algorithm is also based on weak Schur sampling.
However, producing a quantum output is more challenging and requires further processing of the state that remains after weak Schur sampling.
This involves applying certain extremal unitary-covariant channels with single-qubit output, which we characterize by employing tools from representation theory of $\U{2}$ such as the dual Clebsch--Gordan transform.

Another closely related problem is \emph{quantum purification}: given $n$ copies of a mixed state
\begin{equation}
  \rho = (1-\epsilon) \proj{\psi} + \epsilon \proj{\psi^\perp}
\end{equation}
for some $\epsilon \in [0,1/2]$, output the state $\proj{\psi}$.
This problem is closely related to quantum majority vote since the binomial expansion of $\rho\xp{n}$ as a linear combination of rank-$1$ product states $\proj{\psi} \x \proj{\psi^\perp} \x \dotsb \x \proj{\psi}$ is concentrated on states where the fraction of $\proj{\psi}$ is roughly $1-\epsilon$.
Since on such states quantum majority vote would output a good approximation of $\proj{\psi}$, we highly suspect that our algorithm is also optimal for purification.
Note that previous work on purification has focused only on the information-theoretical aspects of the problem, namely, at what rate and with what fidelity the task can be achieved~\cite{CEM99,KW01}.
However, as far as we are aware, algorithmic implementation of purification has not been studied before.
Therefore, if our algorithm coincides with an optimal purifier, we are first to provide an efficient implementation of purification -- our algorithm uses $O(n^4 \log n)$ gates.
Other related quantum error suppression techniques have recently been studied in \cite{KoczorEigenvector,KoczorSuppression,VirtualDistillation}.

Quantum cloning \cite{Cloning,OptimalCloning} can be considered as a reverse of purification since it dilutes information by mapping a smaller number of perfect states to a larger number of imperfect ones.
Purification and cloning are examples of so-called \emph{quantum machines} \cite{Machines,Keyl}, \ie, algorithms with quantum input or quantum output.
Recent examples of these include quantum subtracting machine \cite{SubtractingMachine}, unsupervised classification of quantum data \cite{Clustering}, overlap estimation \cite{OverlapEstimation}, and others \cite{Wright}.
Quantum majority vote and its generalization to equivariant Boolean functions also fit into this framework.
In fact, our template algorithm $\mc{A}_\vt$ (\cref{alg:Template}) is an example of a \emph{programmable quantum processor} \cite{Processors,PortbasedTeleportation}.
Its interpolation parameters $\vt$ can be considered as a program that is supplied either as a probability distribution or a quantum state.
For appropriately chosen $\vt$ the algorithm can optimally evaluate any symmetric equivariant Boolean function.

In the special case when $n = 1$, there are only two relevant Boolean functions -- the identity function and NOT.
While computing the identity function in a basis-independent way is trivial, computing the NOT function corresponds to applying the universal $\NOT$ operation mentioned earlier, which can only be done with fidelity $2/3$ \cite{BHW,BHW99,Spins}.
The post-processing step after weak Schur sampling in our algorithm can be viewed as a convex combination of quantum channels for computing the identity function and the universal $\NOT$.
In particular, we provide a novel and efficient implementation of a variant of the universal $\NOT$, where the input state has been isometrically scaled down from the symmetric subspace to a single system of the corresponding dimension.

The recent optimization framework \cite{GO22} for problems with permutational and unitary symmetries is highly relevant to our problem, however it does not easily apply in our setting due to a certain technical limitation.
Our problem corresponds to the parameter choice $(p,q,d) = (n,1,2)$ in \cite{GO22}.
Unfortunately, the partial trace constraint in their semidefinite program requires that $d \geq p + q - 1 = n$ (see eq.~(117) in \cite{GO22}), which is not satisfied when $d = 2$ and $n \geq 3$.
We expect a future refinement of \cite{GO22} to overcome this technical limitation.

One of our technical results is \cref{lem:1param} which characterizes a certain set of unitary-covariant quantum channels.
Upon completing the manuscript we realized that a similar characterization was already obtained by Al Nuwairan \cite{Nuwairan}; she derives extremal points of the set of $\SU{2}$-covariant quantum channels and determines their Stinespring, Kraus, and Choi representations.
Various further properties of these channels have been investigated also in \cite{SU2CovariantLowRank}, which has some overlap\footnote{Despite this, we decided to keep our writing as it is for the sake of completeness and coherence.} with our results in \cref{apx:Kraus,apx:Action}.
Note that the same class of channels is also discussed in Section~V of \cite{NoethersPrinciple}, where they play a central role in spin amplification and inversion.
The optimal channels for these two tasks are reminiscent of our extremal channels
$\Phi^\ell_1 = \Phi^\ell_{\Tr}$ and
$\Phi^\ell_2 = \Phi^\ell_{\UNOT}$, respectively, see \cref{sec:Extremal,sec:Understanding}.
More generally, tensors and tensor networks with $\SU{2}$ symmetry have been considered in \cite{SU2InvariantTensors} and are known as \emph{spin networks} \cite{SpinNetworks} in the context of loop quantum gravity.

Our results originally appeared as a talk in QIP~2021 \cite{Majority}.

\subsection{Open problems}

While our work settles the problem of computing all symmetric equivariant Boolean functions in a basis-independent way, it also opens up many interesting new directions.
\begin{enumerate}
  \setlength\itemsep{0pt}
  \item The patterns exhibited by the exact numerical values of the optimal fidelity in \cref{sec:Optimal fidelities} resemble the behavior of quantum query complexity $\mathsf{Q}(f)$ for symmetric Boolean functions $f$ \cite{Beals}. Is there a closer relationship between the optimal fidelity $F_f$ and the quantum query complexity $\mathsf{Q}(f)$ of equivariant Boolean functions $f$?
  \item It should be possible to find a closed-form solution for the linear program \eqref{eq:LP} in \cref{thm:main} that characterizes the optimal fidelity.
  \item Although we focus on Boolean functions, the problem can be considered for any alphabet with $d \geq 2$ symbols (see \cref{sec:Defs}).
  For example, a \emph{universal inverter} that generalizes the qubit universal $\NOT$ operation can be considered in arbitrary dimension \cite{QuditInverter}.
  It is also interesting to consider generalizations to functions with several outputs
  For example, quantum cloning \cite{Distributor,Cloning} and other unitary-equivariant primitives \cite{Keyl} naturally fit into this framework.
  We expect that the recent optimization framework of \cite{GO22} could be very useful for these types of problems.
  \item Investigate in more depth the notion of unitary-equivariant Boolean functions introduced in \cref{sec:U-equivariant functions}. Does this concept also extend to non-symmetric Boolean functions? How about functions over the alphabet $[d]$, for any $d \geq 2$? Does it also extend to relations?
  \item On a high level, our algorithm applies Schur transform and measures the partition label $\lambda$ (\ie, performs \emph{weak Schur sampling}), and then post-processes the leftover state on the unitary register. However, the post-processing assumes that the unitary register is in the so-called \emph{Gelfand--Tsetlin basis} \cite{GelfandII}, see \cref{apx:Construction}. While most implementations of Schur transform (in particular, \cite{KS17}) guarantee this, some do not \cite{Krovi18}. Can our algorithm be adapted so that the type of Schur transform used does not matter?

  One possible route towards this could be to explicitly apply the inverse Schur transform at the end of the algorithm. Currently, we discard the maximally mixed state from the permutation register and obtain the output by processing the unitary register. Could one instead replace the permutation register by some pure state and apply the universal NOT directly onto the unitary register? Then one could hope to obtain the desired output by performing the inverse Schur transform and discarding all qubits but one.

  Alternatively, could Schur transform be eliminated altogether? For example, weak Schur sampling can instead be implemented via generalized phase estimation \cite[Section~8.1]{Harrow05}. Unfortunately, without going to Schur basis it is not clear how to discard or replace the permutation register. Moreover, the post-processing step also needs to be modified. While implementing the partial trace channel in the standard basis is trivial, one also needs to implement the universal $\NOT$ operation with a multi-qubit input. Can this be done without Schur transform?
  \item Our algorithm appears to extract some classical side-information about the input state while still achieving optimal output fidelity. What does this information contain and could it be used for anything? For example, the side-information produced by an optimal algorithm for the universal $\NOT$ operation contains a full classical description of the output state, implying that any number of copies of the output state can be produced with the same fidelity. Understanding, for other Boolean functions, the trade-off between the optimal fidelity and how much classical side-information or how many additional copies of the output state can be produced could have cryptographic applications. Indeed, our setup in \cref{fig:Problem} suggests a cryptographic interpretation: Alice is delegating the computation of $f$ to Bob while trying to hide her input $x$ and the output $f(x)$ from him.
  \item While our quantum majority vote problem can be interpreted as amplifying the correct output of a quantum device, this device has a somewhat contrived failure mode: whenever it fails to produce the ideal state $\ket{\psi}$, it outputs a perfectly orthogonal state $\ket{\psi^\perp}$ (see \cref{fig:Majority}). A more natural setting instead would be to have the input states randomly distributed around some unknown ideal state $\ket{\psi}$. \emph{Quantum averaging} is then the problem of producing an approximate copy of the ideal state $\ket{\psi}$. Presumably our quantum majority vote algorithm is optimal also for quantum averaging. However, we do not know how its output fidelity scales with the number of input states and their variance around the ideal state $\ket{\psi}$.
  \item Various unitary-equivariant operations, such as approximations of universal $\NOT$ and quantum cloning, have been implemented experimentally \cite{Experiments}. Quantum majority vote could be a useful primitive for reducing experimental errors and cleaning up noisy quantum data. More generally, an experimental implementation of our template algorithm $\mc{A}_\vt$ (see \cref{alg:Template}) would allow to optimally evaluate any symmetric and equivariant Boolean function on quantum data.
\end{enumerate}

\subsection{Outline}

This paper is structured as follows.
\Cref{sec:Intro} motivates the problem and summarizes our main results.
\Cref{sec:Preliminaries} formally states the problem and analyzes its symmetries.
\Cref{sec:Generic pre-processing} discusses Schur--Weyl duality and constructs a particularly nice orthonormal Schur basis for the $n$-qubit space.
These ingredients are then used to derive a generic pre-processing procedure that can be used without loss of generality as the first step for computing any symmetric and equivariant Boolean function.
\Cref{sec:Completing} describes the second step of the algorithm and derives our main results, \cref{thm:template,thm:main,thm:implementation}, which describe the optimal algorithm, its fidelity, and implementation.
\Cref{sec:Applications} considers various applications of our general framework, such as exactly computing the optimal fidelities for all Boolean functions with $n \leq 7$ arguments and analyzing the asymptotic fidelity of the quantum majority vote.

This paper also contains a number of appendices with detailed calculations and extended background material.
\Cref{apx:Graphical} introduces graphical notation and
\cref{apx:Nice} analyzes our preferred Schur basis $\ket{(\lambda,w,i)}$.
\cref{apx:2x2} is a pedagogical introduction to the representation theory of $2 \times 2$ matrices, with an emphasis on explicit constructions.
\Cref{apx:Sym} discusses the symmetric subspace and its representation-theoretic properties, and
\cref{apx:Extremal} derives quantum circuits for implementing certain extremal unitary-covariant quantum channels needed in our algorithm.
Finally,
\cref{apx:Monotonicity} proves that the optimal fidelity is monotonic in the Hamming weight, and
\cref{apx:Choi} explicitly lists optimal Choi matrices for computing all symmetric equivariant Boolean functions with $n \leq 3$ arguments.

\section{Preliminaries}\label{sec:Preliminaries}

This section covers some preliminary material and sets up the problem.
First, in \cref{sec:Math preliminaries} we introduce notation, basic concepts, and terminology.
Next, in \cref{sec:Defs}, we formally state the problem in a general setting
and then, in \cref{sec:Boolean functions}, specialize to Boolean functions.
In \cref{sec:Impossiblity}, we demonstrate the impossibility of perfectly computing quantum majority on three inputs.
Finally, we describe symmetries of optimal quantum algorithms in \cref{sec:Symmetries}.

\subsection{Mathematical preliminaries}\label{sec:Math preliminaries}

\paragraph{Linear algebra.}
Fix an integer $d \geq 2$ and let $[d] := \set{0, \dotsc, d-1}$.
We denote the \emph{standard basis} of $\C^d$ by $\ket{i}$ where $i \in [d]$.
Then $\C^d = \spn \set{\ket{i} : i \in [d]}$ where ``$\spn$'' denotes the \emph{complex linear span}.
We write $\L(\C^d)$ to denote the set of all \emph{linear operators} on $\C^d$, \ie, the set of all $d \times d$ matrices over $\C$, and we let $\id_d := \sum_{i \in [d]} \proj{i}$ denote the $d \times d$ \emph{identity matrix}.
We denote the \emph{commutator} of $A,B \in \L(\C^d)$ by $[A,B] := AB - BA$.
The \emph{partial trace} over the first register, denoted $\Tr_1: \L(\C^{d_1} \x \C^{d_2}) \to \L(\C^{d_2})$, is defined as $\Tr_1(A \x B) := \Tr(A) B$, for any $A \in \L(\C^{d_1})$ and $B \in \L(\C^{d_2})$, and extended by linearity to non-product operators.
The partial trace over the second register, $\Tr_2$, is defined similarly.
The \emph{unitary group} $\U{d} := \set{U \in \L(\C^d): UU\ct = \id_d}$ consists of all $d \times d$ \emph{unitary matrices}, where $U\ct := \overline{U}\tp$ denotes the \emph{conjugate transpose} of $U$.
The \emph{special unitary group} $\SU{d} := \set{U \in \U{d} : \det(U) = 1}$ consists of all unitary matrices with determinant one.
We write $\L(\C^\din,\C^\dout)$ to denote all \emph{linear operators} from $\C^\din$ to $\C^\dout$.
When $\din \leq \dout$, the set of all \emph{isometries} is $\U{\C^\din,\C^\dout} := \set{U \in \L(\C^\din,\C^\dout) : U\ct U = \id_\din}$.

\paragraph{Representation theory.}
A $d$-dimensional \emph{representation} of a group $G$ is a map $R: G \to \U{d}$ such that $R(g h) = R(g) R(h)$ for all $g,h \in G$ (\ie, $R$ is a \emph{homomorphism} from $G$ to $\U{d}$).
It is \emph{reducible} if $U R(g) U\ct = R_1(g) \+ R_2(g)$ for some $U \in \U{d}$ and representations $R_1$ and $R_2$, and \emph{irreducible} otherwise.
The \emph{dual representation} of $R$ is defined as $R^*(g) := R(g^{-1})\tp$, for all $g \in G$.
Since $R$ is a homomorphism, $R(g^{-1}) = R(g)^{-1}$ for any $g \in G$.
Since $R$ is unitary, \ie, $R(g) \in \U{d}$, we have
$R^*(g) = R(g^{-1})\tp = \of[\big]{R(g)^{-1}}\tp = \of[\big]{R(g)\ct}\tp = \overline{R(g)}$.\footnote{While $R^*(g) = \overline{R(g)}$ for all $g \in G$, we make a distinction between $R^*(g)$ and $\overline{R(g)}$ since they might differ when $R$ is extended beyond $G$. For example, if $G = \U{1} \subset \C$ and $R(z) := z$ for all $z \in \U{1}$, we have $R^*(z) = 1/z$ while $\overline{R(z)} = \overline{z}$ for $z \in \C$, so the two maps differ when $R$ is extended from the unit circle to the whole complex plane.}
The following lemma by Schur is an important tool in representation theory.

\begin{lemma}[Schur's Lemma]\label{lem:Schur}
Let $R_1$ and $R_2$ be irreducible representations of dimension $d_1$ and $d_2$ of some group $G$. Assume that $M$ is a $d_1 \times d_2$ complex matrix such that $R_1(g) M = M R_2(g)$, for all $g \in G$. Then
\begin{itemize}
  \item if $R_1$ and $R_2$ are isomorphic then $d_1 = d_2 =: d$ and $M = c \id_d$ for some $c \in \C$ where $\id_d$ is the $d \times d$ identity matrix;
  \item if $R_1$ and $R_2$ are not isomorphic then $M = 0$.
\end{itemize}
\end{lemma}

\noindent
Since the representation theory of $\SU{2}$ and $\U{2}$ is particularly important to us, we provide a self-contained introduction in \cref{apx:2x2} and a brief refresher in \cref{sec:Refresher}.

\paragraph{Action of the symmetric group.}
The \emph{symmetric group} $\S_n$ on $n \geq 1$ elements consists of all $n$-element permutations. For any permutation $\pi \in \S_n$, let $P(\pi) \in \U{d^n}$ denote the permutation matrix whose action on the standard basis of $(\C^d)\xp{n}$ is given by
\begin{equation}
  P(\pi) \; \ket{x_1} \x \dotsb \x \ket{x_n}
  := \ket{x_{\pi^{-1}(1)}} \x \dotsb \x \ket{x_{\pi^{-1}(n)}},
  \qquad
  \forall x \in [d]^n,
  \label{eq:P}
\end{equation}
meaning that $P(\pi)$ permutes the $n$ qudits (or terms in the tensor product) according to the permutation $\pi$. The map $P: \S_n \to \U{d^n}$ is a representation of the symmetric group $\S_n$.

\paragraph{Bloch sphere representation of qubit states.}
Any \emph{pure} single-qubit state $\ket{\psi} \in \C^2$ can be written (up to a global phase) as
\begin{equation}
  \ket{\psi} = \mx{\cos \frac{\theta}{2} \\ e^{i\varphi} \sin \frac{\theta}{2}},
  \label{eq:psi angles}
\end{equation}
for some $\theta \in [0,\pi]$ and $\varphi \in [0, 2\pi)$.
To also capture \emph{mixed} states, we instead use a \emph{density matrix} $\rho \in \L(\C^2)$ such that $\rho \geq 0$ (\ie, $\rho$ is positive semidefinite) and $\Tr \rho = 1$.
Any such $\rho$ is of the form
\begin{equation}
  \rho(\vec{r}) := \frac{1}{2} \mx{1+z & x-iy \\ x+iy & 1-z},
  \label{eq:rho xyz}
\end{equation}
for some \emph{Bloch vector} $\vec{r} := (x,y,z) \in \R^3$ with $x^2 + y^2 + z^2 \leq 1$.
By solving $\rho = \proj{\psi}$, we see that the angles $\theta,\varphi$ of a pure state $\ket{\psi}$ are related to the spherical coordinates
$x = \sin \theta \cos \varphi$,
$y = \sin \theta \sin \varphi$,
$z = \cos \theta$
of the point $(x,y,z)$.
Hence, pure states $\rho(\vec{r})$ correspond to unit vectors $\vec{r} \in \R^3$ and form the so-called \emph{Bloch sphere}.
The \emph{uniform measure} on the Bloch sphere is given by
\begin{equation}
  d\psi := \frac{1}{4\pi} \sin \theta \, d\theta \, d\varphi,
  \label{eq:dpsi}
\end{equation}
which induces via \cref{eq:psi angles} a uniform measure on pure single-qubit states $\ket{\psi} \in \C^2$.
Note that $d\psi$ is normalized so that $\int d\psi = 1$.
If $\rho(\vec{r})$ is pure, its unique orthogonal state is $\rho(-\vec{r})$ since $\Tr \sof{\rho(\vec{r}_1) \rho(\vec{r}_2)} = (1 + x_1 x_2 + y_1 y_2 + z_1 z_2) / 2 = 0$ precisely when $\vec{r}_1 = - \vec{r}_2$.

\paragraph{Choi matrix and CPTP maps.}
Recall the following standard terminology and basic facts from quantum information theory~\cite{Watrous}.
We refer to linear maps of the form $\Phi: \L(\C^\din) \to \L(\C^\dout)$ as \emph{superoperators}.
The \emph{Choi matrix} $\J{\Phi} \in \L(\C^\dout \x \C^\din)$ of a superoperator $\Phi$ is defined as
\begin{equation}
  \J{\Phi} := \sum_{i,j \in [\din]} \Phi\of[\big]{\ketbra{i}{j}} \x \ketbra{i}{j}.
  \label{eq:Choi}
\end{equation}
The action of $\Phi$ on any $\rho \in \L(\C^\din)$ can be recovered from its Choi matrix $\J{\Phi}$ as follows:
\begin{equation}
  \Phi(\rho) = \Tr_2 \sof[\big]{ \J{\Phi} \cdot \of{\id_\dout \x \rho\tp}},
  \label{eq:ChoiAction}
\end{equation}
where we discard the second register that corresponds to the input system.
If a Choi matrix $\J{\Phi}$ is subject to
\begin{align}
  \J{\Phi} &\geq 0, &
  \Tr_1 \J{\Phi} = \id_\din,
  \label{eq:ChoiCPTP}
\end{align}
the corresponding map $\Phi$ is called \emph{completely positive} (CP) and \emph{trace-preserving} (TP), respectively. If both properties are satisfied simultaneously, the map is called CPTP or a \emph{quantum channel}. Such maps are precisely the operations that can in principle be implemented on a quantum computer. We denote the set of all CPTP maps with given input and output dimensions by $\CPTP{\C^\din}{\C^\dout}$.
In addition to the Choi matrix, quantum channels admit other representations, such as \emph{Kraus} and \emph{Stinespring}, which are described in \cref{apx:Kraus and Stinespring}.

\paragraph{Quantum circuits and algorithms.}
A \emph{quantum circuit} is a sequence of elementary quantum operations, where each operation either applies a unitary gate on at most two quantum systems, introduces an additional ancillary system initialized in $\ket{0}$, or discards a system that is no longer needed. Any such circuit implements a CPTP map. A \emph{quantum algorithm} is a uniform sequence of quantum circuits, meaning that a single Turing machine can produce any of the circuits upon receiving the corresponding problem size $n$ as its input. A quantum algorithm is \emph{efficient} if the corresponding Turing machine runs in $\mathrm{poly}(n)$ time. While every quantum algorithm corresponds to a family of CPTP maps, not every such family can be implemented by an efficient quantum algorithm. We will use the term ``quantum algorithm'' as opposed to ``quantum channel'' whenever we want to stress that our goal is to explicitly implement the underlying family of CPTP maps as a sequence of elementary quantum operations.

\subsection{Computing functions in a basis-independent way}\label{sec:Defs}

In this section, we formalize the problem of computing a function whose input is provided in an unknown quantum basis. We also specify the figure of merit used to measure the performance of a quantum algorithm for this task.

Fix a unitary $U \in \U{d}$ of dimension $d \geq 2$ and a string $x \in [d]^n$ of length $n \geq 1$ over the alphabet $[d] = \set{0, \dotsc, d-1}$. One can think of the state $U\xp{n} \ket{x}$ or its density matrix
\begin{equation}
  \rho^U_x := U\xp{n} \proj{x} U\ctxp{n}
\end{equation}
as a quantum encoding of the string $x$ in the basis given by $U$. Given this state and a function $f: [d]^n \to [d]$, we would like to find a quantum algorithm that computes $f$ and outputs the answer $f(x)$ in the same basis as the input. In other words, we want a quantum algorithm that approximates the map
\begin{equation}
  U\xp{n} \ket{x} \mapsto U \ket{f(x)},
  \label{eq:ideal}
\end{equation}
for any unknown unitary $U \in \U{d}$ and string $x \in [d]^n$. For each $n$, such algorithm is described by a CPTP map $\Phi$ from $n$ qudits to $1$ qudit.

Let $F(\ket{\psi}, \rho) := \bra{\psi} \rho \ket{\psi}$ denote the \emph{fidelity} between a mixed state $\rho$ and a pure state $\ket{\psi}$ of the same dimension.
The following definition formalizes the worst-case fidelity for computing a function in a basis-independent way on a quantum computer.
Determining this fidelity and deriving the corresponding optimal quantum algorithm is the main goal of this paper.

\newcommand{\minxu}[1][{[d]}]{\min_{\substack{x \in #1^n \\ U \in \U{d}}}}

\begin{definition}[Fidelity for computing $f$]\label{def:fid}
Let $f: [d]^n \to [d]$ be a function and $\Phi \in \CPTP{\C^{d^n}}{\C^d}$ a CPTP map for computing $f$. For a particular input string $x \in [d]^n$ and a single-qudit unitary $U \in \U{d}$, the \emph{fidelity of $\Phi$ for computing $f$ on input $x$ in basis $U$} is
\begin{equation}
  F\of[\big]{ U \ket{f(x)}, \Phi(\rho^U_x) }
  := \bra{f(x)} U\ct \Phi(\rho^U_x) U \ket{f(x)}.
\end{equation}
The \emph{(worst-case) fidelity of $\Phi$ for computing $f$} is
\begin{equation}
  F_f(\Phi) := \minxu F\of[\big]{ U \ket{f(x)}, \Phi(\rho^U_x) },
  \label{eq:fff}
\end{equation}
and the \emph{optimal fidelity for computing $f$ in a basis-independent way} is
\begin{equation}
  F_f := \max_{\Phi \in \CPTP{\C^{d^n}}{\C^d}} F_f(\Phi).
  \label{eq:Ff}
\end{equation}
A CPTP map $\Phi$ is \emph{optimal for computing $f$} if $F_f(\Phi) = F_f$.
\end{definition}

Note that the ideal functionality in \cref{eq:ideal} is not well-defined for all functions $f: [d]^n \to [d]$. To see this, consider any non-trivial permutation matrix $M \in \U{d}$ that relabels the symbols in $[d]$. Then $M\xp{n} \ket{x} = \ket{x'}$ for some other string $x' \in [d]^n$. Since the unitary basis change $M$ is not explicitly included as part of the input, it is not possible to tell whether the actual argument of $f$ should be $x$ or $x'$, so on one hand the correct output should be $M \ket{f(x)}$ while on the other hand it should be $\ket{f(x')}$. Hence, the ideal functionality of the algorithm for a given function $f$ is well-defined only when these two states agree for any permutation $M$ of $[d]$ (this is equivalent to the function $f$ being \emph{equivariant}, see \cref{def:symcov} below).

Although formally \cref{def:fid} makes sense for any function $f: [d]^n \to [d]$, we focus specifically on equivariant functions, since for this class of functions the ideal functionality in \cref{eq:ideal} is well-defined. Moreover, we further restrict our attention to \emph{symmetric} functions (see \cref{def:symcov} below). While this restriction is not necessary and the problem can also be studied for non-symmetric functions, assuming symmetry significantly simplifies the analysis by enabling the use of the so-called \emph{Schur--Weyl duality} (see \cref{sec:SchurWeyl}).

\begin{definition}[Symmetric and equivariant functions]\label{def:symcov}
A function $f: [d]^n \to [d]$ is \emph{symmetric} if
\begin{equation}
  f\of[\big]{x_{\pi^{-1}(1)}, \dotsc, x_{\pi^{-1}(n)}} = f(x), \qquad
  \forall x \in [d]^n, \quad
  \forall \pi \in \S_{n},
  \label{eq:fsym}
\end{equation}
where $\pi$ permutes the characters in the string $x$, and \emph{equivariant} if
\begin{equation}
  f\of[\big]{\sigma(x_1), \dotsc, \sigma(x_n)} = \sigma\of[\big]{f(x)}, \qquad
  \forall x \in [d]^n, \quad
  \forall \sigma \in \S_{d},
  \label{eq:fcov}
\end{equation}
where $\sigma$ relabels the symbols in the alphabet $[d]$.
\end{definition}

\subsection{Symmetric and equivariant Boolean functions}\label{sec:Boolean functions}

The main focus of this paper is the $d = 2$ case of Boolean functions.
Since the only non-trivial permutation of $\set{0,1}$ is the negation $\NOT(a) := a \oplus 1$, where $a \in \set{0,1}$ and ``$\oplus$'' denotes the bit-wise addition modulo $2$, equivariant Boolean functions have a particularly simple characterization -- they are invariant under simultaneously negating all input and output bits.
Alluding to the duality between $\AND$ and $\OR$ established by De~Morgan's laws, such functions are also known as \emph{self-dual}.

\begin{definition}[Equivariant or self-dual Boolean functions]
A Boolean function $f: \set{0,1}^n \to \set{0,1}$ is \emph{equivariant} or \emph{self-dual} if
\begin{equation}
  f(x \oplus 1^n) = f(x) \oplus 1,
  \qquad \forall x \in \set{0,1}^n,
  \label{eq:NOTcov}
\end{equation}
where $1^n$ denotes the string of $n$ ones and ``$\oplus$'' denotes the bit-wise addition modulo $2$.
\end{definition}

\newcommand{\symf}{\bar{f}}
\newcommand{\ham}{h}

To fully determine a symmetric Boolean function, it is enough to specify its value only on one input string for each possible Hamming weight.
Thus one can replace a symmetric $n$-argument Boolean function $f: \set{0,1}^n \to \set{0,1}$ by a function $\symf: \set{0,\dotsc,n} \to \set{0,1}$ whose argument is the Hamming weight of the original input string: $\symf(|x|) := f(x)$, for all $x \in \set{0,1}$.

If a symmetric function $f$ is also equivariant, the following constraint is imposed on $\symf$:
\begin{equation}
  \symf(\ham) = \symf(n-\ham) \oplus 1,
  \qquad \forall \ham \in \set{0,\dotsc,n}.
\end{equation}
This constraint cannot be satisfied when $n$ is even since $\symf(n/2) = \symf(n/2) \oplus 1$ is not possible.
Hence, it is only meaningful to consider the case of \emph{odd} $n$.

While we will not do this here, an alternative option is to consider relations instead of functions when $n$ is even. For example, the majority relation can output either $0$ or $1$ when the Hamming weight of the input is exactly $n/2$. We expect that our results can be easily extended to this variation of the problem.

\begin{example}[Majority and parity]
We are particularly interested in the $n$-bit \emph{majority} function $\MAJ_n$ that returns the value that appears the most often in the input string.
Another interesting function is the $n$-bit \emph{parity}:
\begin{align}
  \PAR_n(x_1,\dotsc,x_n) := x_1 \oplus \dotsb \oplus x_n,
  \qquad \forall x \in \set{0,1}^n.
\end{align}
It is evident that $\MAJ_n$ and $\PAR_n$ are symmetric and (when $n$ is odd) equivariant.
All symmetric and equivariant Boolean functions on $n=1$ and $n=3$ bits are summarized in \cref{tab:1-3bit}.
\end{example}

\begin{table}
  \centering
  \begin{tabular}{c|cccc}
    $x$ & $\ID$ & $\NOT$ \\ \hline
    0 & 0 & 1 \\
  \end{tabular}
  \qquad\qquad
  \begin{tabular}{c|cccc}
    $x$ & $\MAJ_3$ & $\PAR_3$ & $\NPAR_3$ & $\NMAJ_3$ \\ \hline
    000 & 0 & 0 & 1 & 1 \\
    001 & 0 & 1 & 0 & 1
  \end{tabular}
  \caption{\label{tab:1-3bit}All $1$-bit and $3$-bit symmetric equivariant Boolean functions. For $1$-bit functions, it is sufficient to list their values only on the input $0$, while for $3$-bit functions only the strings $000$ and $001$ are sufficient (on all other strings the value can be inferred from symmetry and equivariance). Here we use the notation
  $\NMAJ_n(x) := \NOT\of{\MAJ_n(x)}$ and
  $\NPAR_n(x) := \NOT\of{\PAR_n(x)}$.}
\end{table}

\subsection{Impossibility of perfect quantum majority}\label{sec:Impossiblity}

\newcommand{\junk}{\mathrm{j}}
\newcommand{\Junk}{\mathrm{J}}

Let us consider an example of computing quantum majority on three qubits, and show that no quantum algorithm can achieve this task perfectly.

A perfect algorithm would have to work correctly on all inputs. In particular, for the standard basis states we should have
\begin{equation}
  W \ket{x} \ket{\Omega} = \ket{\MAJ_3(x)} \ket{\junk_{x}}, \qquad \forall x \in \set{0,1}^3,
\end{equation}
where $W$ is an isometry (Stinespring dilation) that fully describes the algorithm.
The algorithm may use an additional ancillary system, which is initialized in some fixed state $\ket{\Omega}$, and produce some junk states $\ket{\junk_x}$ that are discarded once the computation is finished. More explicitly:
\begin{align}
  \begin{aligned}
    W \ket{000} \ket{\Omega} &= \ket{0} \ket{\junk_{000}}, \\
    W \ket{001} \ket{\Omega} &= \ket{0} \ket{\junk_{001}}, \\
    W \ket{010} \ket{\Omega} &= \ket{0} \ket{\junk_{010}}, \\
    W \ket{100} \ket{\Omega} &= \ket{0} \ket{\junk_{100}},
  \end{aligned} &&
  \begin{aligned}
    W \ket{111} \ket{\Omega} &= \ket{1} \ket{\junk_{111}}, \\
    W \ket{110} \ket{\Omega} &= \ket{1} \ket{\junk_{110}}, \\
    W \ket{101} \ket{\Omega} &= \ket{1} \ket{\junk_{101}}, \\
    W \ket{011} \ket{\Omega} &= \ket{1} \ket{\junk_{011}}.
  \end{aligned}
  \label{eq:junk}
\end{align}
Since $W$ is an isometry, all the $\ket{\junk_x}$ are normalized. Moreover, the $\ket{\junk_x}$ with a majority of $0$'s amongst $x_1, x_2, x_3 \in \set{0,1}$ are all orthogonal to each other, as are the $\ket{\junk_x}$ with a majority of $1$'s.

Now consider the ``all qubits equal'' state in an arbitrary basis:
\begin{equation}
  \ket{\Psi}
  := \of[\big]{\alpha \ket{0} + \beta \ket{1}} \x
     \of[\big]{\alpha \ket{0} + \beta \ket{1}} \x
     \of[\big]{\alpha \ket{0} + \beta \ket{1}}
\end{equation}
where $\alpha \ket{0} + \beta \ket{1} \in \C^2$ is an arbitrary single-qubit state. Then by linearity and \cref{eq:junk},
\begin{align}
  W \ket{\Psi} \ket{\Omega}
  = \alpha \ket{0} \of[\big]{
      \alpha^2     \ket{\junk_{000}}
   &+ \alpha \beta \ket{\junk_{001}}
    + \alpha \beta \ket{\junk_{010}}
    + \alpha \beta \ket{\junk_{100}}
    } \\
{}+ \beta \ket{1} \of[\big]{
      \beta^2      \ket{\junk_{111}}
   &+ \alpha \beta \ket{\junk_{110}}
    + \alpha \beta \ket{\junk_{101}}
    + \alpha \beta \ket{\junk_{011}}
    }. \nonumber
\end{align}
For this to be the correct outcome of the majority, we need the two junk states
\begin{align}
\begin{aligned}
  \ket{\Junk_0}
   &:= \of[\big]{
          \alpha^2     \ket{\junk_{000}}
        + \alpha \beta \ket{\junk_{001}}
        + \alpha \beta \ket{\junk_{010}}
        + \alpha \beta \ket{\junk_{100}}
       }, \\
  \ket{\Junk_1}
   &:= \of[\big]{
          \beta^2      \ket{\junk_{111}}
        + \alpha \beta \ket{\junk_{110}}
        + \alpha \beta \ket{\junk_{101}}
        + \alpha \beta \ket{\junk_{011}}
       }
\end{aligned}
\end{align}
to be equal. But for general $\alpha$ and $\beta$ they are not since their norms are different. Thus there is no perfect algorithm for quantum majority of $n = 3$ states. A similar argument can show that there is also no perfect algorithm for quantum majority for any $n \geq 3$ states.

According to the above argument, a linear map that implements the ideal operation
\begin{equation}
  U\xp{n} \proj{x} U\ctxp{n} \mapsto U \proj{\MAJ_n(x)} U\ct
\end{equation}
for any $n \geq 3$ cannot be completely positive since it has no Stinespring dilation $W$.
The same conclusion can be easily reached also for the ideal universal $\NOT$ operation.
Let $\vec{r} \in \R^3$ be an arbitrary unit vector on the Bloch sphere and $\rho(\vec{r})$ the corresponding density matrix given in \cref{eq:rho xyz}.
Since the unique pure state orthogonal to $\rho(\vec{r})$ has density matrix $\rho(-\vec{r})$, the Choi matrix $\J{\UNOT} \in \L(\C^4)$ for the ideal universal $\NOT$ operation must satisfy
\begin{equation}
  \Tr_2 \sof[\big]{J \cdot (\id_2 \x \rho(\vec{r})\tp)} = \rho(-\vec{r})
\end{equation}
for all $\vec{r} = (x,y,z)$. Comparing the coefficients at $x,y,z$, we find the unique solution to be
\begin{equation}
  \J{\UNOT} = \mx{
    0 & 0 & 0 & -1 \\
    0 & 1 & 0 & 0 \\
    0 & 0 & 1 & 0 \\
   -1 & 0 & 0 & 0
  }.
  \label{eq:J UNOT}
\end{equation}
Since its eigenvalues are $(-1,1,1,1)$, the ideal universal $\NOT$ operation is not completely positive and hence cannot be implemented on a quantum computer.
Nevertheless, as an abstract superoperator it is still meaningful.

We discuss this perspective more in \cref{sec:U-equivariant functions}, where we also point out that the ideal superoperator appears to be unique for any symmetric and equivariant Boolean function.
Moreover, in \cref{apx:Ideal} we provide explicit Choi matrices for the ideal superoperators of all symmetric and equivariant Boolean functions with $n=1$ and $n=3$ arguments.

\subsection{Symmetries of optimal algorithms}\label{sec:Symmetries}

Let us show that an optimal quantum algorithm for computing any symmetric function $f:[d]^n \to [d]$ in a basis-independent way can always be chosen so that it is equivariant with respect to the unitary group $\U{d}$ and symmetric with respect to the permutation group $\S_n$.

\newcommand{\PhiQ}{\Phi_Q}
\newcommand{\PhiP}{\Phi_P}

To formalize this, let $\Phi \in \CPTP{\C^{d^n}}{\C^d}$ be a quantum channel and consider two different symmetrized versions of $\Phi$, one over the symmetric group and the other over the unitary group:
\begin{align}
  \PhiP(\rho) &:=
    \sum_{\pi \in \S_n}
    \frac{1}{n!}
    \Phi\of[\big]{ P(\pi) \rho P(\pi)\ct }, &
  \PhiQ(\rho) &:=
    \int_{V \in \U{d}}
    V\ct \Phi\of[\big]{ V\xp{n} \rho V\ctxp{n} } V dV,
  \label{eq:PhiPQ}
\end{align}
where $P(\pi)$ permutes the $n$ tensor factors in $\C^d \x \dotsb \x \C^d$ according to the permutation $\pi$, see \cref{eq:P}, and $dV$ denotes the Haar measure on $\U{d}$ normalized so that $\int_{V \in \U{d}} dV = 1$. Notice that $\PhiP$ and $\PhiQ$ are quantum channels since $\CPTP{\C^{d^n}}{\C^d}$ is a convex set.

\begin{lemma}\label{lem:Opt}
For any symmetric function $f:[d]^n \to [d]$ and any $\Phi \in \CPTP{\C^{d^n}}{\C^d}$, the symmetrized channels $\PhiP$ and $\PhiQ$ compute $f$ with fidelity no worse than the original channel $\Phi$:
\begin{align}
  F_f(\PhiP) &\geq F_f(\Phi), &
  F_f(\PhiQ) &\geq F_f(\Phi).
\end{align}
\end{lemma}

\begin{proof}
Recall from \cref{def:fid} that $\PhiQ$ computes $f$ with worst-case fidelity
\begin{equation}
  F_f(\PhiQ) = \minxu F\of[\big]{ U \ket{f(x)}, \PhiQ(\rho^U_x) }.
  \label{eq:FQ}
\end{equation}
We can bound the fidelity function within the minimization as follows:
\begin{align}
  \bra{f(x)} U\ct \PhiQ(\rho^U_x) U \ket{f(x)}
  &= \int_{V \in \U{d}}
     \bra{f(x)} U\ct \cdot V\ct
     \Phi \of[\big]{ V\xp{n} \rho^U_x V\ctxp{n} }
     V \cdot U \ket{f(x)} \, dV \\
  &= \int_{V \in \U{d}}
     \bra{f(x)} (VU)\ct
     \Phi \of[\big]{ (VU)\xp{n} \proj{x} (VU)\ctxp{n} }
     (VU) \ket{f(x)} \, dV \nonumber \\
  &= \int_{V \in \U{d}}
     F\of[\big]{ VU \ket{f(x)}, \Phi(\rho^{VU}_x) } \, dV \nonumber \\
  &\geq \int_{V \in \U{d}} F_f(\Phi) \, dV \nonumber \\
  &= F_f(\Phi), \nonumber
\end{align}
where the inequality holds because of the minimization over $x \in [d]^n$ and $U \in \U{d}$ within the definition of $F_f(\Phi)$. Since the above calculation works for any $x$ and $U$, we get $F_f(\PhiQ) \geq F_f(\Phi)$. The calculation for $\PhiP$ is similar, except we symmetrize over $\S_n$:
\begin{align}
  \bra{f(x)} U\ct \PhiP(\rho^U_x) U \ket{f(x)}
  &= \sum_{\pi \in \S_n} \frac{1}{n!}
     \bra{f(x)} U\ct
     \Phi \of[\big]{ P(\pi) \rho^U_x P(\pi)\ct }
     U \ket{f(x)} \\
  &= \sum_{\pi \in \S_n} \frac{1}{n!}
     \bra{f(x_\pi)} U\ct
     \Phi \of{ \rho^U_{x_\pi} }
     U \ket{f(x_\pi)} \nonumber \\
  &\geq \sum_{\pi \in \S_n} \frac{1}{n!}
     F_f(\Phi) \nonumber \\
  &= F_f(\Phi), \nonumber
\end{align}
where $x_\pi$ is the string obtained by permuting the characters of $x$ according to $\pi$, and we used $f(x) = f(x_\pi)$ in the second equality since $f$ is symmetric.
\end{proof}

Due to \cref{lem:Opt}, we are particularly interested in maps $\Phi$ that are invariant under the symmetric group $\S_n$ and equivariant with respect to the unitary group $\U{d}$.
The following definition generalizes $\S_n$-invariance from $n$-argument functions (see \cref{def:symcov}) to superoperators with $n$ input qudits.

\begin{definition}[Permutation invariance]
Let $\mc{X} := (\C^d)\xp{n}$ be a tensor product input space and $\mc{Y}$ an arbitrary output space. A linear map $\Phi: \L(\mc{X}) \to \L(\mc{Y})$ is \emph{permutation-invariant} if
\begin{equation}
  \Phi(P(\pi) \rho P(\pi)\ct) = \Phi(\rho),
  \qquad \forall \rho \in \L(\mc{X}), \quad \forall \pi \in \S_n,
\end{equation}
where $P(\pi)$ permutes $n$ qudits according to $\pi$, see \cref{eq:P}.
\end{definition}

Similarly, we can also extend the notion of equivariance from functions (see \cref{def:symcov}) to superoperators by replacing functions $f: [d]^n \to [d]$ with superoperators from $n$ qudits to one qudit.
Moreover, we extend the discrete $\S_d$-equivariance to the continuous $\U{d}$-equivariance.

\begin{definition}[Unitary equivariance]
A linear map $\Phi: \L(\C^{d^n}) \to \L(\C^d)$ is \emph{unitary-equivariant} if
\begin{equation}
  \Phi(U\xp{n} \rho U\ctxp{n}) = U \Phi(\rho) U\ct,
  \qquad \forall \rho \in \L(\C^{d^n}), \quad \forall U \in \U{d}.
\end{equation}
\end{definition}

We will also need a slightly more general notion than the invariance under applying $U$ to all input qudits and $U\ct$ to all output qudits.
The following definition accounts for arbitrary $\U{d}$-representations on the input and output spaces.

\begin{definition}[Unitary covariance]\label{def:unitary covariance}
Let $\mc{X}$ and $\mc{Y}$ be arbitrary linear spaces admitting some $\U{d}$-representations
$R_\tin : \U{d} \to \U{\mc{X}}$ and
$R_\tout: \U{d} \to \U{\mc{Y}}$, respectively.
A linear map $\Phi: \L(\mc{X}) \to \L(\mc{Y})$ is \emph{unitary-covariant} if
\begin{equation}
  \Phi\of[\big]{ R_\tin(U) \rho R_\tin(U)\ct }
  = R_\tout(U) \Phi(\rho) R_\tout(U)\ct,
  \qquad \forall \rho \in \L(\mc{X}), \quad \forall U \in \U{d}.
  \label{eq:Covariance}
\end{equation}
We denote the set of all \emph{unitary-covariant channels} by $\UCPTP{\mc{X}}{\mc{Y}}$.
\end{definition}

Unitary equivariance corresponds to the special case when $R_\tin(U) = U\xp{n}$ and $R_\tout(U) = U$.
In all other cases the $\U{d}$-representations $R_\tin$ and $R_\tout$ will be clear from the context.

It can be easily verified that the CPTP maps $\PhiP$ and $\PhiQ$ in \cref{eq:PhiPQ} are permutation-invariant and unitary-equivariant, respectively. The following corollary is a direct consequence of \cref{lem:Opt}.

\begin{cor}\label{cor:PICU}
For any symmetric function $f:[d]^n \to [d]$, with $d \geq 2$ and $n \geq 1$, there exists an optimal quantum algorithm $\Phi \in \CPTP{\C^{d^n}}{\C^d}$ for computing $f$ in a basis-independent way, such that $\Phi$ is permutation-invariant and unitary-equivariant.
\label{cor:Opt}
\end{cor}

Our next goal is to characterize Choi matrices of unitary-equivariant maps.
The following lemma describes how the Choi matrix transforms under an arbitrary input and output basis change of $\Phi$.

\begin{lemma}\label{lem:Commuting}
Let $\mc{X}$ and $\mc{Y}$ be linear spaces, $\Phi: \L(\mc{X}) \to \L(\mc{Y})$ a linear map, and let $V \in \L(\mc{X})$ and $W \in \L(\mc{Y})$ be arbitrary matrices. Then the Choi matrix for the map $\Phi'(\rho) := W\ct \Phi(V \rho V\ct) W$ is
\begin{equation}
  \J{\Phi'} =
  (W \x \overline{V})\ct \J{\Phi} (W \x \overline{V}).
\end{equation}
\end{lemma}

\begin{proof}
Recall from \cref{eq:ChoiAction} that the action of $\Phi$ on any $\rho \in \L(\mc{X})$ can be expressed via its Choi matrix:
$\Phi(\rho) = \Tr_2 \sof[\big]{ \J{\Phi} \cdot \of{\id_\dout \x \rho\tp}}$
where $\dout := \dim \mc{Y}$.
If we modify the input of $\Phi$ by conjugating it with $V$, the new output looks as follows:
\begin{align}
   \Phi\of{V \rho V\ct}
&= \Tr_2 \sof[\Big]{
     \J{\Phi} \cdot \of[\big]{
       \id_\dout \x \of{V \rho V\ct}\tp
     }
   } \\
&= \Tr_2 \sof[\Big]{
     \J{\Phi} \cdot \of[\big]{
       \id_\dout \x \of{\overline{V} \rho\tp \overline{V}\ct}
     }
   } \\
&= \Tr_2 \sof[\Big]{
     \J{\Phi}
     \cdot (\id_\dout \x \overline{V})
     \cdot (\id_\dout \x \rho\tp)
     \cdot (\id_\dout \x \overline{V}\ct)
   } \\
&= \Tr_2 \sof[\Big]{
     (\id_\dout \x \overline{V}\ct)
     \cdot \J{\Phi}
     \cdot (\id_\dout \x \overline{V})
     \cdot (\id_\dout \x \rho\tp)
   },
\end{align}
where the last equality follows from the cyclic property of the partial trace (it is crucial here that the matrix $\id_\dout \x \overline{V}\ct$ acts trivially on the first register which is not traced out).
If we also modify the output of $\Phi$ by conjugating it with $W\ct$, we get
\begin{equation}
  \Phi'(\rho)
= W\ct \Phi\of[\big]{V \rho V\ct} W
= \Tr_2 \sof[\Big]{
    (W\ct \x \overline{V}\ct)
    \cdot \J{\Phi}
    \cdot (W \x \overline{V})
    \cdot (\id_\dout \x \rho\tp)
  }
  \label{eq:UJU}
\end{equation}
after moving the basis change inside the partial trace.
Since
$\Phi'(\rho) = \Tr_2 \sof[\big]{ \J{\Phi'} \cdot \of{\id_\dout \x \rho\tp}}$
and both formulas hold for all $\rho \in \L(\mc{X})$,
we conclude that
$\J{\Phi'} = (W \x \overline{V})\ct \J{\Phi} (W \x \overline{V})$.
\end{proof}

\begin{cor}\label{cor:Commuting}
A linear map $\Phi: \L(\mc{X}) \to \L(\mc{Y})$ is unitary-covariant with respect to representations
$R_\tin : \U{d} \to \U{\mc{X}}$ and
$R_\tout: \U{d} \to \U{\mc{Y}}$
if and only if
\begin{equation}
  \sof[\big]{ \J{\Phi}, R_\tout(\overline{U}) \x R_\tin(U) } = 0,
  \qquad \forall U \in \U{d}.
\end{equation}
\end{cor}

\begin{proof}
According to \cref{lem:Commuting}, the unitary covariance condition from \cref{def:unitary covariance} translates into
\begin{equation}
  \J{\Phi}
= \of[\big]{R_\tout(U) \x R_\tin(\overline{U})}\ct
  \J{\Phi}
  \of[\big]{R_\tout(U) \x R_\tin(\overline{U})},
  \qquad \forall U \in \U{d},
\end{equation}
since $\overline{R(U)} = R(\overline{U})$ for any (polynomial) representation $R$ of $\U{d}$.
The result follows by substituting $U$ with $\overline{U}$.
\end{proof}

It seems from \cref{eq:fff} that infinitely many different input states need to be considered to determine the fidelity $F_f(\Phi)$ with which a given quantum algorithm $\Phi$ computes a function $f$ in a basis-independent manner. It turns out that in the cases relevant to us it suffices to consider a finite list of inputs to determine this fidelity.
Since our primary interest is in Boolean functions, we restrict to $d = 2$.

\begin{lemma}\label{lem:Fidelity}
Let $f: \set{0,1}^n \to \set{0,1}$ be a symmetric and equivariant Boolean function.
A permutation-invariant and unitary-equivariant channel $\Phi \in \CPTP{\C^{2^n}}{\C^2}$ computes $f$ with fidelity
\begin{equation}
  F_f(\Phi)
= \min_{0 \leq \ham \leq \floor{n/2}}
  F\of[\big]{
    \ket{f(0^{n-\ham}1^{\ham})},
    \Phi(\proj{0^{n-\ham}1^{\ham}})
  },
\end{equation}
where $0^{n-\ham}1^{\ham} := \underbrace{0\dots0}_{n-\ham} \underbrace{1\dots1}_{\ham}$.
\end{lemma}

\begin{proof}
Let $\rho_x^U := U\xp{n} \proj{x} U\ctxp{n}$ denote the input state corresponding to an arbitrary $x \in \set{0,1}^n$ and $U \in \U{2}$.
Recall from \cref{eq:fff} that $\Phi$ computes $f$ with fidelity
\begin{equation}
  F_f(\Phi) = \minxu[\set{0,1}] F\of[\big]{ U \ket{f(x)}, \Phi(\rho^U_x) }.
  \label{eq:fff2}
\end{equation}
Since $\Phi$ is permutation-invariant and unitary-equivariant,
\begin{align}
  F\of[\big]{U \ket{f(x)}, \Phi(\rho_x^U)}
  &= \bra{f(x)} U\ct \Phi(\rho_x^U) U \ket{f(x)} \\
  &= \bra{f(x)} U\ct U \Phi(\proj{x}) U\ct U \ket{f(x)} \\
  &= F\of[\big]{\ket{f(x)}, \Phi(\proj{x})}.
\end{align}
Therefore, when $\Phi$ is unitary-equivariant, the minimization over unitaries $U$ in \cref{eq:fff2} can be omitted.
Moreover, we can also simplify the minimization over $x$.

Let $X := \smx{0&1\\1&0}$ and note that $X\xp{n} \ket{x} = \ket{x \oplus 1^n}$ and $X \ket{f(x)} = \ket{f(x \oplus 1^n)}$ since $f$ is equivariant, see \cref{eq:NOTcov}.
Setting $U := X$ in the above calculation results in
\begin{equation}
  F\of[\big]{\ket{f(x \oplus 1^n)}, \Phi(\proj{x \oplus 1^n})}
  = F\of[\big]{\ket{f(x)}, \Phi(\proj{x})},
\end{equation}
so it suffices to minimize only over strings $x \in \set{0,1}^n$ with $\abs{x} \leq \frac{n}{2}$.
Finally, if $\abs{x} = \ham$ then $f(x) = f(0^{n-\ham}1^{\ham})$ and $\Phi(\proj{x}) = \Phi(\proj{0^{n-\ham}1^{\ham}})$ since $f$ is symmetric and $\Phi$ is permutation-invariant. Therefore
\begin{equation}
  F\of[\big]{\ket{f(x)}, \Phi(\proj{x})}
  = F\of[\big]{\ket{f(0^{n-\ham}1^{\ham})}, \Phi(\proj{0^{n-\ham}1^{\ham}})}
\end{equation}
and it suffices to optimize only over Hamming weights $\ham \in \set{0, \dotsc, \floor{n/2}}$.
\end{proof}

\section{Generic pre-processing}\label{sec:Generic pre-processing}

While the problem of computing $f: [d]^n \to [d]$ in a basis-independent way can be studied for any $d \geq 2$, for the rest of this paper we focus specifically on the $d = 2$ or Boolean case. In this case we can take advantage of the relatively simple structure of the qubit Schur basis discussed in \cref{sec:BasisQubits}.
While some of our discussion applies also to any $d \geq 2$, we expect the general case to be much more complicated and leave it for future work.

As an additional simplifying assumption, we consider only symmetric functions $f$.
This helps us to narrow down the structure of possible quantum channels for computing $f$, and reduces the problem of finding an optimal channel from a semidefinite optimization problem to a linear one (see \cite{GO22} for an in-depth study of this phenomenon).

In this section, we outline a generic procedure for pre-processing the quantum input $\rho_x^U$, which can be used to simplify the task of computing any symmetric equivariant $n$-bit Boolean function $f: \set{0,1}^n \to \set{0,1}$ in a basis-independent manner. This pre-processing does not lose any relevant information about the input, so it is always possible to afterwards complete the algorithm and still compute $f$ optimally. The pre-processing is based on Schur--Weyl duality and uses a particularly nice basis of the $n$-qubit space which we explain in the following subsections.

\subsection{Schur--Weyl duality}\label{sec:SchurWeyl}

This section provides a brief summary of Schur--Weyl duality, which is one of our main tools.
While it holds for any $d \geq 2$, here we focus specifically on $d = 2$.
We list some standard material and recommend \cite[Section~5.3]{Harrow05} and \cite[Section~9.1]{GW09} for more background.

Let $U \in \U{2}$ be a single-qubit unitary and $\pi \in \S_n$ a permutation of $n$ elements. We can extend both operations to unitaries that act on the $n$-qubit space $(\C^2)\xp{n}$ by considering $Q(U) := U\xp{n}$, which simultaneously applies $U$ to each qubit, and $P(\pi)$, see \cref{eq:P}, which permutes the qubits according to $\pi$. Note that $Q(U)$ and $P(\pi)$ are unitary representations of the groups $\U{2}$ and $\S_n$, respectively.
Moreover, their actions commute: $[Q(U), P(\pi)] = 0$.

The representations $Q$ and $P$ not only commute with each other but in fact generate algebras that are mutual commutants of one another.
Recall that an \emph{algebra} $\mathscr{A} \subseteq \L(\C^{2^n})$ is a linear subspace that is closed under matrix multiplication, and the \emph{commutant} of $\mathscr{A}$ is defined as
\begin{equation}
  \Comm(\mathscr{A}) := \set{M \in \L(\C^{2^n}) : MA = AM, \forall A \in \mathscr{A}}.
\end{equation}
If we denote the algebras generated by the representations $Q$ and $P$ as follows
\begin{align}
  \mathscr{Q} &:= \spn \set{Q(U) : U \in \U{2}}, &
  \mathscr{P} &:= \spn \set{P(\pi) : \pi \in \S_n},
\end{align}
then \emph{Schur--Weyl duality} says that
$\Comm(\mathscr{Q}) = \mathscr{P}$ and
$\Comm(\mathscr{P}) = \mathscr{Q}$.

The above statement of Schur--Weyl duality is very concise but not necessarily very insightful.
We can make it much more concrete as follows.
First, we decompose the $n$-qubit state space $\C^{2^n}$ as follows:
\begin{equation}
  \underbrace{\C^2 \x \dotsb \x \C^2}_{n}
  \cong
  \bigoplus_{\lambda \pt n} \mc{Q}_\lambda \x \mc{P}_\lambda,
  \label{eq:Schur}
\end{equation}
where $\lambda \pt n$ means that $\lambda$ runs over all \emph{partitions} of $n$ with \emph{exactly two parts}\footnote{This somewhat differs from the conventional meaning of $\lambda \pt n$, which is that $\lambda$ can have any number of integer parts and each part has strictly positive size.}:
\begin{align}
  \lambda \pt n \quad \Longleftrightarrow \quad
  \lambda \in \set{
    (\lambda_1, \lambda_2) \in \Z^2:
    \lambda_1 + \lambda_2 = n,
    \lambda_1 \geq \lambda_2 \geq 0
  }.
  \label{eq:pt}
\end{align}
Equivalently, $\lambda = (n-k,k)$ for some $k \in \set{0, \dotsc, \floor{n/2}}$.
The subspace $\mc{Q}_\lambda \x \mc{P}_\lambda$ corresponding to a particular $\lambda$ in \cref{eq:Schur} is called the \emph{$\lambda$-isotypic subspace}. The dimensions of the \emph{unitary register} $\mc{Q}_\lambda$ and the \emph{permutation register} $\mc{P}_\lambda$ within this subspace are
\begin{align}
  \dim \mc{Q}_\lambda
  &= \lambda_1 - \lambda_2 + 1
  =: m_\lambda, &
  \dim \mc{P}_\lambda
  &= \binom{\lambda_1 + \lambda_2}{\lambda_1} \frac{\lambda_1 - \lambda_2 + 1}{\lambda_1 + 1}
  =: d_\lambda,
  \label{eq:dims}
\end{align}
which obey the non-trivial combinatorial identity
$2^n = \sum_{\lambda \pt n} m_\lambda d_\lambda$.

The defining property of the decomposition \eqref{eq:Schur} is that the operators $Q(U)$ and $P(\pi)$ take a particularly simple form in the right-hand side basis. Namely, if we write $\hat{M}$ to denote a matrix $M \in \L(\C^{2^n})$ expressed in this basis then
\begin{align}
  \hat{Q}(U)
  &= \bigoplus_{\lambda \pt n} Q_\lambda(U) \x \id_{d_\lambda}, &
  \hat{P}(\pi)
  &= \bigoplus_{\lambda \pt n} \id_{m_\lambda} \x P_\lambda(\pi),
  \label{eq:QP}
\end{align}
where $Q_\lambda$ and $P_\lambda$, acting on registers $\mc{Q}_\lambda$ and $\mc{P}_\lambda$, are irreducible representations of the unitary group $\U{2}$ and the permutation group $\S_n$, respectively, and $\id_d$ denotes the identity matrix of dimension $d$.
The dual nature of $\hat{Q}(U)$ and $\hat{P}(\pi)$ in \cref{eq:QP} constitute a concrete embodiment of the \emph{Schur--Weyl duality}.
It is evident from \cref{eq:QP} that $[\hat{Q}(U),\hat{P}(\pi)] = 0$.

\subsection{Schur transform}

Any orthonormal basis of $(\C^2)\xp{n}$ that decomposes the space as in \cref{eq:Schur} and that block-diagonalizes the action of $\U{2}$ and $\S_n$ as in \cref{eq:QP} is called a \emph{Schur basis}, and the corresponding unitary basis change
\begin{equation}
  \Usch: (\C^2)\xp{n} \to \bigoplus_{\lambda \pt n} \of[\big]{\C^{m_\lambda} \x \C^{d_\lambda}}
\end{equation}
is called a \emph{Schur transform}.

If $R: G \to \U{d}$ is a representation of a group $G$ then so is $R'(g) = V R(g) V\ct$ for any $V \in \U{d}$.
Since each of the $\mc{Q}_\lambda$ and $\mc{P}_\lambda$ registers in \cref{eq:Schur} has a similar unitary degree of freedom, Schur basis is not unique and thus neither is Schur transform.
Any two Schur transforms $\Usch, \Usch' \in \U{\C^{2^n}}$ are related as
\begin{equation}
  \Usch = \of*{ \bigoplus_{\lambda\pt n} \of[\big]{V_\lambda \otimes W_\lambda} } \Usch',
\end{equation}
for some $V_\lambda \in \U{m_\lambda}$ and $W_\lambda \in \U{d_\lambda}$.

Quantum circuits for implementing Schur transform, efficient in the local dimension $d \geq 2$ and the number of systems $n \geq 1$, are known \cite{Harrow05,BCH06,BCH07,KS17,Krovi18}.
However, as far as we can tell, our algorithm requires not just any old Schur transform but one that produces a particular basis on the unitary registers $\mc{Q}_\lambda$.
Namely, the action of $U\xp{n}$ in Schur basis should be given by
\begin{equation}
  \bigoplus_{\lambda \pt n}
  Q_\lambda(U) \x \id_{d_\lambda},
  \label{eq:QI}
\end{equation}
where the irreducible $\U{2}$-representations $Q_\lambda(U)$ are in the so-called \emph{Gelfand--Tsetlin basis}.
We describe in detail the construction and properties of such representations in \cref{apx:2x2}, and in \cref{sec:BasisQubits} we provide a family of Schur bases that produce them (see \cref{lem:Unitary blocks} for proof).

Generally, a Schur basis that is suitable for us can be produced by employing the representation theory of the unitary group and the associated Clebsch--Gordan transform.\footnote{This is in contrast to the approach of \cite{Krovi18} which exploits the representation theory of the symmetric group.}
For example, the $n$-qubit Schur transform implemented by \cite{KS17} is suitable for use in our algorithm.

Analogous to the Gelfand--Tsetlin basis for the unitary registers $\mc{Q}_\lambda$, a similarly elegant choice, the so-called \emph{Young--Yamanouchi basis}, exists also for the permutation registers $\mc{P}_\lambda$, \ie, the second register in \cref{eq:QI}.
However, for us the actual basis used for the permutation registers will not matter since our algorithm will discard them.
Indeed, as we consider only symmetric functions, this register will be in the maximally mixed state that looks the same in any basis.

\subsection{Nice Schur basis of the \texorpdfstring{$n$}{n}-qubit space}\label{sec:BasisQubits}

Since Schur basis is not unique, we would like to choose one with a particularly nice structure. In this section, we describe a construction borrowed from~\cite{CEM99} (see \cref{apx:Nice} for a comprehensive treatment).
While this construction does not determine a unique basis, any member of the resulting family has enough structure for our purposes.

From now on we focus exclusively on the $d = 2$ case because of the following particularly nice Schur basis for $(\C^2)\xp{n}$. It consists of vectors
\begin{equation}
  \ket{(\lambda, w, i)},
  \qquad \text{where} \qquad
  \lambda \pt n, \;
  w \in [m_\lambda], \;
  i \in [ d_\lambda].
\end{equation}
Here $\lambda$ labels the blocks in \cref{eq:Schur} while $w$ and $i$ label bases for the unitary register $\mc{Q}_\lambda$ and the permutation register $\mc{P}_\lambda$, respectively.
Note that the range of $w$ and $i$ depends on $\lambda$.

To construct this basis, we first define the $i = 0$ vectors as
\begin{equation}
  \ket{(\lambda, w, 0)}
  := \ket{s_{\lambda_1-\lambda_2}(w)} \x
     \ket{\Psi^-}\xp{\lambda_2},
  \label{eq:lw0}
\end{equation}
where $\ket{\Psi^-} := (\ket{01} - \ket{10})/\sqrt{2}$ is the \emph{singlet state} and $\ket{s_\ell(w)} \in (\C^2)\xp{\ell}$ is the $\ell$-qubit \emph{symmetric} or \emph{Dicke state}~\cite{Dicke} of Hamming weight $w$:
\begin{equation}
  \ket{s_\ell(w)}
  := \binom{\ell}{w}^{-1/2}
     \sum_{\substack{x\in\set{0,1}^\ell\\\abs{x}=w}} \ket{x}.
  \label{eq:slw}
\end{equation}
Note that $\braket{s_\ell(w)}{s_\ell(w')} = \delta_{ww'}$, so these states form an orthonormal basis of the $\ell$-qubit symmetric subspace.

Next, we (semi-explicitly) describe the states $\ket{(\lambda,w,i)}$ with $i \in \set{1,\dotsc,d_\lambda-1}$. Let
\begin{equation}
  \ket{(\lambda,w,i)} := \sum_{\pi \in \S_n} \alpha_\pi^{\lambda,i} P(\pi) \ket{(\lambda,w,0)},
  \label{eq:lwi}
\end{equation}
where the coefficients $\alpha_\pi^{\lambda,i} \in \R$ do not depend on $w$ and are chosen so that for $w=0$ the states $\set{\ket{(\lambda,0,i)} : i \in [d_\lambda]}$ form an orthonormal basis of $\mc{P}_{\lambda,0}$ where
\begin{equation}
  \mc{P}_{\lambda,w} := \spn \set[\big]{P(\pi) \ket{(\lambda,w,0)} : \pi \in \S_n}.
  \label{eq:Plw}
\end{equation}
In other words, we fix $\lambda$ and choose $\alpha_\pi^{\lambda,i}$ so that
\begin{equation}
  \braket{(\lambda,0,i)}{(\lambda,0,j)} = \delta_{ij},
  \label{eq:loi}
\end{equation}
and then reuse them for all other values of $w$ and the same $\lambda$.
Surprisingly, this ensures that the states $\set{\ket{(\lambda,w,i)} : i \in [d_\lambda]}$ with $w \neq 0$ also form an orthonormal basis of $\mc{P}_{\lambda,w}$ (see \cref{prop:Orthonormal} in \cref{apx:Nice}).
Although the above construction does not fully determine the vectors $\ket{(\lambda,w,i)}$, this ambiguity will not make any difference in the arguments that follow.
In particular, the full set of vectors $\ket{(\lambda,w,i)}$ is always orthonormal (see \cref{lem:FullOrthonormal} in \cref{apx:Nice}) and hence forms a basis for $(\C^2)\xp{n}$.

Let us now argue that these vectors are indeed compatible with the action of $\S_n$ and $\U{d}$ in \cref{eq:QP}. First, note from \cref{eq:Plw} that each subspace $\mc{P}_{\lambda,w}$ is invariant under qubit permutations $P(\pi)$. Moreover, \cref{lem:QInvar} in \cref{apx:Nice} shows that the subspaces
\begin{equation}
  \mc{Q}_{\lambda,i} := \spn \set[\big]{\ket{(\lambda,w,i)} : w \in [m_\lambda]}
\end{equation}
are invariant under $U\xp{n}$, for any $U \in \U{2}$.
The actions of $\pi \in \S_n$ and $U \in \U{2}$ on subspaces $\mc{P}_{\lambda,w}$ and $\mc{Q}_{\lambda,i}$ are given by the corresponding $\lambda$-irreps $P_\lambda(\pi)$ and $Q_\lambda(U)$, respectively, see \cref{lem:P blocks,lem:Unitary blocks} in \cref{apx:Nice}.
This establishes the explicit Schur--Weyl duality in \cref{eq:QP} and implies that the vectors $\ket{(\lambda,w,i)}$ indeed form a Schur basis.
While the exact form of $P_\lambda(\pi)$ will not be important to us, we will need more information about the $\U{2}$-irreps $Q_\lambda(U)$, so we provide an extensive discussion on their construction and properties in \cref{apx:2x2}.

With a basis at hand we can express the projection onto the $\lambda$-isotypic subspace as
\begin{equation}
  \Pi_\lambda
 := \sum_{w \in [m_\lambda]}
    \sum_{i \in [d_\lambda]}
    \proj{(\lambda,w,i)},
\end{equation}
which takes the following very simple form in the Schur basis:
\begin{equation}
  \hat{\Pi}_\lambda
  = \bigoplus_{\kappa \pt n}
    \delta_{\lambda\kappa}
    \id_{m_\kappa d_\kappa}.
  \label{eq:ProjSch}
\end{equation}
Moreover, we can write the associated Schur transform as follows:
\begin{equation}
  \Usch
  := \bigoplus_{\lambda \pt n}
     \of*{
       \sum_{w \in [m_\lambda]}
       \sum_{i \in [d_\lambda]}
       \ket{w} \x \ket{i} \bra{(\lambda,w,i)}
     }.
  \label{eq:NiceSchur}
\end{equation}

A key property of the above construction is that $\ket{(\lambda,w,0)}$ and $\ket{(\lambda,w,i)}$ satisfy \cref{eq:lw0,eq:lwi}, respectively, and the actual values of the coefficients $\alpha_\pi^{\lambda,i}$ in \cref{eq:lwi} do not matter to us.
Any basis of this form can be used in our analysis, and the associated Schur transform~\eqref{eq:NiceSchur} can be used as the first step of our algorithm to compute any symmetric equivariant Boolean function $f: \set{0,1}^n \to \set{0,1}$ in a basis-independent manner.

From now on by \emph{Schur basis} we will mean any basis of the above form, and by \emph{Schur transform} we will mean the corresponding transform \eqref{eq:NiceSchur}.

\subsection{Generic pre-processing}\label{sec:Preproc}

We now describe a generic pre-processing procedure $\Gamma$ that can be used as the first step when computing any symmetric and equivariant Boolean function $f: \set{0,1}^n \to \set{0,1}$.
We will conclude in \cref{cor:Generic} that applying $\Gamma$ does not incur any loss in fidelity.

\newcommand{\wlx}{w_{\lambda,x}}
\newcommand{\clix}{c_{\lambda,i}(x)}

\begin{alg}[label={alg:Gamma}]{Generic pre-processing $\Gamma$}
\textbf{Input:}
Quantum state $U\xp{n} \ket{x}$ with an unknown $x \in \set{0,1}^n$ and $U \in \U{2}$.\medskip\\
\textbf{Output:}
Partition $\lambda = (\lambda_1,\lambda_2) \pt n$ and quantum state $Q_\lambda(U) \ket{\wlx}$ of dimension $m_\lambda := \lambda_1 - \lambda_2 + 1$ where $\wlx := \abs{x} - \lambda_2$ (see \cref{lem:Gamma output}).
\begin{enumerate}[\bf Step 1:]
  \item Apply Schur transform $\Usch$.
  \item Measure $\lambda \pt n$ (\emph{weak Schur sampling}).
  \item Discard the permutation register $\mc{P}_\lambda$.
\end{enumerate}
\end{alg}

Since the output dimension $m_\lambda$ of this procedure is a random variable, we can either think of $\Gamma$ as a collection of CP maps $(\Gamma_\lambda : \lambda \pt n)$ where $\Gamma_\lambda: \L(\C^{2^n}) \to \L(\C^{m_\lambda})$, or we can identify the output space of $\Gamma$ with the direct sum $\bigoplus_{\lambda \pt n} \C^{m_\lambda}$.
The following lemma shows that $\Gamma$ is permutation-invariant, \ie, invariant under applying $P(\pi)$ for any $\pi \in \S_n$ on the input.
In addition, $\Gamma$ is also unitary-covariant, meaning that each $\Gamma_\lambda$ is unitary-covariant, \ie, instead of applying $U\xp{n}$ on the input we can apply $\bigoplus_{\lambda \pt n} Q_\lambda(U)$ on the output of $\Gamma$.

\begin{lemma}\label{lem:Gamma}
The generic pre-processing $\Gamma \in \CPTP{\C^{2^n}}{\bigoplus_{\lambda \pt n} \C^{m_\lambda}}$ implemented by \cref{alg:Gamma} is permutation-invariant and unitary-covariant.
In other words, for any $\rho \in \L(\C^{2^n})$,
\begin{align}
  \Gamma \of[\big]{P(\pi) \, \rho \, P(\pi)\ct}
 &= \Gamma \of[\big]{\rho},
 &  \forall \pi &\in \S_n,
    \label{eq:GP} \\
  \Gamma \of[\big]{U\xp{n} \, \rho \, U\ctxp{n}}
 &= \bigoplus_{\lambda \pt n}
    Q_\lambda(U)
    \, \Gamma_\lambda\of{\rho} \,
    Q_\lambda(U)\ct,
 &  \forall U &\in \U{2}.
    \label{eq:GU}
\end{align}
\end{lemma}

\begin{proof}
To see that $\Gamma$ is permutation-invariant, recall from \cref{eq:QP} that the $n$-qubit permutation $P(\pi)$ acts as
$\hat{P}(\pi) = \bigoplus_{\lambda \pt n} \id_{m_\lambda} \x P_\lambda(\pi)$
in Schur basis.
Hence, it does not interfere with the measurement performed in Step~2 since the measurement operators
$\hat{\Pi}_\lambda = \bigoplus_{\kappa \pt n} \delta_{\lambda\kappa} \id_{m_\kappa} \x \id_{d_\kappa}$
from \cref{eq:ProjSch} are diagonal in Schur basis.
The effect of the permutation is negated when the permutation register is discarded in Step~3, thus $\Gamma$ is permutation-invariant.

To argue unitary covariance, recall from \cref{eq:QP} that $U\xp{n}$ acts as
$\hat{Q}(U) = \bigoplus_{\lambda \pt n} Q_\lambda(U) \x \id_{d_\lambda}$
in Schur basis, which also does not interfere with the measurement performed in Step~2.
Since $\hat{Q}(U)$ acts trivially on the permutation register, we can instead apply $Q_\lambda(U)$ on the output once the permutation register has been discarded in Step~3.
This shows that $\Gamma$ is unitary-covariant.
\end{proof}

\begin{lemma}\label{lem:Gamma output}
When the generic pre-processing $\Gamma$ given by \cref{alg:Gamma} is applied to $U\xp{n} \ket{x}$ for some $U \in \U{2}$ and $x \in \set{0,1}^n$,
the probability of measurement outcome $\lambda = (\lambda_1, \lambda_2)$ in Step~2 is
\begin{equation}
	p_\lambda(x) :=
  \begin{cases}
    \frac
      {\tbinom{n}{\lambda_2} - \tbinom{n}{\lambda_2-1}}
      {\tbinom{n}{\abs{x}}}
      & \text{if $\lambda_2 \leq \abs{x}$}, \\
	  0 & \text{otherwise},
	\end{cases}
  \label{eq:plx}
\end{equation}
where $\abs{x}$ denotes the number of ones in the string $x$.
The corresponding output state after Step~3 is
\begin{equation}
  Q_\lambda(U) \ket{\wlx} \in \C^{m_\lambda},
\end{equation}
where $Q_\lambda(U)$ is a representation of $U$ of dimension
$m_\lambda := \lambda_1 - \lambda_2 + 1$,
see \cref{eq:Q-lambda} in \cref{apx:U2}, and
$\wlx := \abs{x} - \lambda_2$.
The full output of $\Gamma$ is
\begin{equation}
  \Gamma \of[\big]{U\xp{n} \proj{x} U\ctxp{n}}
= \bigoplus_{\lambda \pt n}
    p_\lambda(x) \,
    Q_\lambda(U)
    \proj{\wlx}
    Q_\lambda(U)\ct.
\end{equation}
\end{lemma}

\begin{proof}
Recall from \cref{eq:GU} in \cref{lem:Gamma} that $\Gamma$ is unitary-covariant:
\begin{equation}
	\Gamma \of[\big]{U\xp{n} \proj{x} U\ctxp{n}}
= \bigoplus_{\lambda \pt n}
  Q_\lambda(U)
  \Gamma_\lambda \of[\big]{\proj{x}}
  Q_\lambda(U)\ct.
\end{equation}
Hence, it suffices to argue that
$\Gamma_\lambda \of[\big]{\proj{x}} = p_\lambda(x) \, \proj{\wlx}$,
\ie, applying $\Gamma$ to $\ket{x}$ produces $\ket{\wlx}$ with probability $p_\lambda(x)$
(note that $\Gamma$ is trace-preserving while $\Gamma_\lambda$ are not).

Let us analyze the effect of Steps~1--3 on $\ket{x}$.
Since Step~1 corresponds to applying the Schur transform $\Usch$ defined in \cref{eq:NiceSchur}, let us first express $\ket{x}$ in the Schur basis $\ket{(\lambda,w,i)}$.
Note from \cref{eq:lw0,eq:lwi} that the standard basis expansion of $\ket{(\lambda,w,i)}$ involves only states of Hamming weight $w + \lambda_2$, therefore $\ket{x}$ can be expressed as a linear combination of $\ket{(\lambda,\wlx,i)}$ where $\wlx = \abs{x} - \lambda_2$:
\begin{equation}
  \ket{x}
= \sum_{\lambda \pt n}
  \sum_{i \in [d_\lambda]}
  \clix
  \ket{(\lambda, \wlx, i)},
\end{equation}
for some coefficients $\clix \in \R$.
Using \cref{eq:NiceSchur}, we find that $\ket{x}$ in the Schur basis looks as follows:
\begin{equation}
	\Usch \ket{x}
= \bigoplus_{\lambda \pt n}
  \sof*{
    \sum_{i \in [d_\lambda]}
    \clix
    \of[\big]{\ket{\wlx} \x \ket{i}}
  }
= \bigoplus_{\lambda \pt n}
  \sof[\Bigg]{
    \ket{\wlx} \x
    \underbrace{
      \sum_{i \in [d_\lambda]}
      \clix
      \ket{i}
    }_{\ket{\psi_\lambda(x)}}
  }.
  \label{eq:USchx}
\end{equation}
The probability of obtaining measurement outcome $\lambda$ in Step~2 is
\begin{equation}
  p_\lambda(x)
  = \norm{\ket{\psi_\lambda(x)}}^2
  = \sum_{i \in [d_\lambda]} \abs{\clix}^2,
  \label{eq:p and c}
\end{equation}
and the corresponding post-measurement state is
\begin{equation}
  \ket{\wlx} \x
  \frac{\ket{\psi_\lambda(x)}}{\sqrt{p_\lambda(x)}}.
\end{equation}
Discarding the second register results in the desired state $\ket{\wlx} \in \C^{m_\lambda}$.
The formula for $p_\lambda(x)$ stated in \cref{eq:plx} has been derived in \cite[Lemma~4]{Mon09}.
This is a special case of a more general formula derived in \cite{GeneralizedProbability}, \ie, $p_\lambda(x) = p(\lambda_2 | n,\ham,n,\ham)$ where $\ham = |x|$.
\end{proof}

The following lemma shows that the generic pre-processing channel $\Gamma$ implemented by \cref{alg:Gamma} is reversible.
More specifically, we can find a channel $\Gamma'$ that undoes the action of $\Gamma$ on all input states that are symmetric over $\S_n$.

\begin{lemma}\label{lem:Inverse}
There exists a unitary-covariant quantum channel $\Gamma' \in \CPTP{\bigoplus_{\lambda \pt n} \C^{m_\lambda}}{\C^{2^n}}$ such that
\begin{equation}
  (\Gamma' \circ \Gamma \circ \Theta)
  \of[\big]{U\xp{n} \proj{x} U\ctxp{n}}
= \Theta
  \of[\big]{U\xp{n} \proj{x} U\ctxp{n}},
  \label{eq:Invert}
\end{equation}
for all $x \in \set{0,1}^n$ and $U \in \U{2}$,
where $\Gamma$ is the generic pre-processing implemented by \cref{alg:Gamma}
and $\Theta: \L(\C^{2^n}) \to \L(\C^{2^n})$ denotes symmetrization over $\S_n$:
\begin{equation}
	\Theta(\rho)
 := \frac{1}{n!}
    \sum_{\pi \in \S_n}
    P(\pi) \rho P(\pi)\ct, \qquad
    \forall \rho \in \L(\C^{2^n}).
  \label{eq:Theta}
\end{equation}
\end{lemma}

\begin{proof}
Recall from \cref{lem:Gamma} that $\Gamma$ is unitary-covariant and note from \cref{eq:Theta} that so is $\Theta$. In addition, once we define the map $\Gamma'$, we will see that it is also unitary-covariant. Therefore, we can assume without loss of generality that $U = \id_2$ in \cref{eq:Invert}, so it suffices to show that, for all $x \in \set{0,1}^n$,
\begin{equation}
  (\Gamma' \circ \Gamma \circ \Theta) \of[\big]{\proj{x}}
  = \Theta \of[\big]{\proj{x}}.
  \label{eq:inverse on x}
\end{equation}
It will be easier to argue this in the Schur basis, so let us show that, for all $x \in \set{0,1}^n$,
\begin{equation}
  \Usch (\Gamma' \circ \Gamma \circ \Theta) \of[\big]{\proj{x}} \Usch\ct
  = \Usch \Theta \of[\big]{\proj{x}} \Usch\ct.
  \label{eq:inverse on x in Schur basis}
\end{equation}

Recall from \cref{eq:USchx} that
\begin{equation}
  \Usch \ket{x}
= \bigoplus_{\lambda \pt n}
  \sof[\Big]{
    \ket{\wlx} \x
    \ket{\psi_\lambda(x)}
  }.
\end{equation}
Using \cref{eq:QP} we can extend this to any input of the form $P(\pi) \ket{x}$ where $\pi \in \S_n$:
\begin{equation}
  \Usch P(\pi) \ket{x}
= \of*{
    \bigoplus_{\lambda \pt n}
    \id_{m_\lambda} \x
    P_\lambda(\pi)
  }
  \Usch \ket{x}
= \bigoplus_{\lambda \pt n}
  \sof[\Big]{
    \ket{\wlx} \x
    P_\lambda(\pi)
    \ket{\psi_\lambda(x)}
  }.
\end{equation}
Substituting $\Theta$ from \cref{eq:Theta}, the right-hand side of \cref{eq:inverse on x in Schur basis} becomes
\begin{align}
  \Usch \Theta \of[\big]{\proj{x}} \Usch\ct
&=
  \Usch
  \of*{
    \frac{1}{n!}
    \sum_{\pi \in \S_n}
    P(\pi) \proj{x} P(\pi)\ct
  }
  \Usch\ct \\
&= \bigoplus_{\lambda \pt n}
  \sof*{
    \proj{\wlx} \x
    \frac{1}{n!}
    \sum_{\pi \in \S_n}
    P_\lambda(\pi)
    \proj{\psi_\lambda(x)}
    P_\lambda(\pi)\ct
  }.
\end{align}
Since the matrix on the second register commutes with $P_\lambda(\sigma)$, for any $\sigma \in \S_n$, by Schur's lemma (\cref{lem:Schur}) we get
\begin{equation}
  \frac{1}{n!}
  \sum_{\pi \in \S_n}
  P_\lambda(\pi)
  \proj{\psi_\lambda(x)}
  P_\lambda(\pi)\ct
= c_\lambda(x)
  \frac{\id_{d_\lambda}}{d_\lambda}.
\end{equation}
By comparing the trace of both sides, we find from \cref{eq:p and c} that
$c_\lambda(x) = \norm{\ket{\psi_\lambda(x)}}^2 = p_\lambda(x)$.
In summary, the right-hand side of \cref{eq:inverse on x in Schur basis} can be written as
\begin{equation}
  \Usch \Theta \of[\big]{\proj{x}} \Usch\ct
= \bigoplus_{\lambda \pt n}
  \sof*{
    p_\lambda(x)
    \proj{\wlx} \x
    \frac{\id_{d_\lambda}}{d_\lambda}
  }.
  \label{eq:RHS}
\end{equation}

To analyze the left-hand side of \cref{eq:inverse on x in Schur basis}, recall that Step~1 of $\Gamma$ corresponds to applying the Schur transform $\Usch$.
Since $\Gamma$ is acting on $\Theta \of[\big]{\proj{x}}$, the result of Step~1 coincides with \cref{eq:RHS}.
Hence, we only need to analyze Steps~2 and~3 that correspond to measuring $\lambda$ and discarding the permutation register.
It is immediate from \cref{eq:RHS} that these steps produce
\begin{equation}
  (\Gamma \circ \Theta) \of[\big]{\proj{x}}
= \bigoplus_{\lambda \pt n}
  p_\lambda(x)
  \proj{\wlx}.
\end{equation}
It is now easy to see that the following steps implement a map $\Gamma'$ that undoes the action of $\Gamma$:
\begin{enumerate}[\bf Step 1':]
\setlength\itemindent{20pt}
\setlength\itemsep{0pt}
  \item Measure $\lambda$.
  \item Attach a new (permutation) register containing the maximally mixed state $\id_{d_\lambda} / d_\lambda$.
  \item Apply the inverse Schur transform $\Usch\ct$.
\end{enumerate}
Note that Steps~1' and~2' recover the state in \cref{eq:RHS} and Step~3' recovers the original symmetrized state $\Theta \of[\big]{\proj{x}}$.
This establishes \cref{eq:inverse on x}.

It remains to show that $\Gamma'$ is unitary-covariant.
Similar to the argument in \cref{lem:Gamma}, we recall from \cref{eq:QP} that $U\xp{n}$ acts as
$\hat{Q}(U) = \bigoplus_{\lambda \pt n} Q_\lambda(U) \x \id_{d_\lambda}$
in the Schur basis.
Since this does not interfere with the measurement in Step~1' and acts trivially on the permutation register that is attached in Step~2', we conclude that $\Gamma'$ is unitary-covariant.
\end{proof}

While the above proof is useful for understanding the evolution of any input state $\rho_x^U$ during the generic pre-processing $\Gamma$, \cref{lem:Inverse} can actually be generalized to all inputs.
The following result shows that $\Gamma'$ is the inverse of $\Gamma$ on symmetric inputs.\footnote{On general inputs $\Gamma$ is not invertible.}

\begin{lemma}\label{lem:InverseGeneral}
Let $\Gamma \in \CPTP{\C^{2^n}}{\bigoplus_{\lambda \pt n} \C^{m_\lambda}}$ be the generic pre-processing implemented by \cref{alg:Gamma} and $\Gamma' \in \CPTP{\bigoplus_{\lambda \pt n} \C^{m_\lambda}}{\C^{2^n}}$ be the unitary-covariant quantum channel introduced above. Then
\begin{equation}
  \Gamma' \circ \Gamma \circ \Theta = \Theta
\end{equation}
where $\Theta$ is the quantum channel that symmetrizes over $\S_n$, see \cref{eq:Theta}.
\end{lemma}

\begin{proof}
Fix an arbitrary input $M \in \L(\C^{2^n})$.
Since $[\Theta(M), P(\pi)] = 0$ for all $\pi \in \S_n$, the output $\Theta(M)$ is in the commutant of the algebra generated by $\set{P(\pi) : \pi \in \S_n}$.
Hence, $\Theta(M) \in \spn \set{U\xp{n} : U \in \U{2}}$ by Schur--Weyl duality (see \cref{sec:SchurWeyl}).
Since $\Usch U\xp{n} \Usch\ct = \bigoplus_{\lambda \pt n} Q_\lambda(U) \x \id_{d_\lambda}$ by \cref{eq:QP}, expressing $\Theta(M)$ in the Schur basis gives us
\begin{equation}
  \Usch \Theta(M) \Usch\ct
  = \bigoplus_{\lambda \pt n} M_\lambda \x \frac{\id_{d_\lambda}}{d_\lambda},
\end{equation}
for some $M_\lambda \in \L(\C^{m_\lambda})$.
Note that Steps~2 and~3 of \cref{alg:Gamma} (namely, measuring $\lambda$ and then discarding the permutation register) are reversible operations when applied on this input.
Since $\Gamma'$ undoes these steps, we recover back the original input $\Theta(M)$.
\end{proof}

The following lemma shows that the implementation of any permutation-invariant quantum channel can start off by first applying $\Gamma$, which justifies the name ``generic pre-processing'' for \cref{alg:Gamma}. Moreover, this decomposition is compatible with unitary-covariance.

\begin{lemma}\label{lem:Generic}
For any permutation-invariant quantum channel
$\Phi \in \CPTP{\C^{2^n}}{\C^2}$
there exists a quantum channel
$\Phi' \in \CPTP{\bigoplus_{\lambda \pt n} \C^{m_\lambda}}{\C^2}$
such that
\begin{equation}
  \Phi = \Phi' \circ \Gamma,
\end{equation}
where $\Gamma \in \CPTP{\C^{2^n}}{\bigoplus_{\lambda \pt n} \C^{m_\lambda}}$ is the generic pre-processing implemented by \cref{alg:Gamma}.
If $\Phi$ is unitary-covariant then so is $\Phi'$.
\end{lemma}

\begin{proof}
By combining permutation-invariance of $\Phi$ with \cref{lem:InverseGeneral}, we get
\begin{align}
  \Phi
  = \Phi \circ \Theta
  = \Phi \circ \Gamma' \circ \Gamma \circ \Theta
  = (\Phi \circ \Gamma') \circ (\Gamma \circ \Theta)
  = \Phi' \circ \Gamma
\end{align}
where the last equality follows by defining
$\Phi' := \Phi \circ \Gamma'$
and noting that
$\Gamma \circ \Theta = \Gamma$
thanks to the permutation-invariance of $\Gamma$ from \cref{lem:Gamma}.
If $\Phi$ is unitary-covariant then so is $\Phi'$ since it is a composition of two unitary-covariant channels ($\Gamma'$ is unitary-covariant by \cref{lem:Inverse}).
\end{proof}

We can now conclude that, without loss of generality, an optimal quantum algorithm for computing a symmetric and equivariant Boolean function in a basis-independent way starts off by applying the generic pre-processing $\Gamma$.

\begin{cor}\label{cor:Generic}
Every symmetric and equivariant Boolean function $f: \set{0,1}^n \to \set{0,1}$ has an optimal quantum algorithm $\Phi \in \UCPTP{\C^{2^n}}{\C^2}$ such that $\Phi = \Phi' \circ \Gamma$ where $\Phi'$ is unitary-covariant and $\Gamma$ is the generic pre-processing implemented by \cref{alg:Gamma}.
\end{cor}

\begin{proof}
Thanks to \cref{cor:PICU}, we can assume that $\Phi$ is permutation-invariant and unitary-equivariant.
The claim then follows from \cref{lem:Generic}.
\end{proof}

\subsection{Example: $3$-qubit quantum majority}\label{sec:Example}

Suppose we are interested in computing the 3-bit majority function in a basis-independent way. As input we receive a state $U\xp{3} \ket{x}$ for some unknown unitary $U \in \U{2}$ and a 3-bit string $x$.
Without loss of generality we first apply the generic pre-processing $\Gamma$.
Recall from \cref{lem:Gamma output} that the output of $\Gamma$ is a partition $\lambda = (\lambda_1,\lambda_2) \pt 3$, occurring with probability $p_\lambda(x)$ defined in \cref{eq:plx}, and a state $Q_\lambda(U) \ket{\wlx}$ where $\wlx = \abs{x} - \lambda_2$ and $Q_\lambda(U)$ is a representation of $U$ of dimension $m_\lambda = \lambda_1 - \lambda_2 + 1$.
For the sake of this example, we will apply an additional isometry $\ket{w} \mapsto \ket{s_\ell(w)}$ that maps the output state into the symmetric subspace of $\ell = \lambda_1 - \lambda_2$ qubits.
Then the output state is
\begin{equation}
  U\xp{\ell} \ket{s_\ell(\wlx)},
\end{equation}
where $\ket{s_\ell(\wlx)}$ is the $\ell$-qubit symmetric state of Hamming weight $\wlx$, see \cref{eq:slw}.

There are only two three-element partitions relevant to us: $\lambda=(2,1)$ and $\lambda=(3,0)$.
Moreover, since the generic pre-processing is permutation-invariant and unitary-covariant (see \cref{lem:Gamma}), \cref{lem:Fidelity} tells us that it is enough to consider only two inputs: $x=000$ and $x=001$.
For example, we can think of the input $x=101$ as $X\xp{3} P(\pi) \ket{001}$ where $\pi$ is the permutation $(132)$.
For each input $x$ and measurement outcome $\lambda$, the table below describes the corresponding Hamming weight $\wlx$, the number of output qubits $\ell$, the probability $p_\lambda(x)$, and the corresponding output state.

\begin{center}
\renewcommand{\arraystretch}{1.5}
\begin{tabular}{c|c|c}
	& $x = 000$ & $x = 001$ \\ \hline
	& $\wlx = -1, \ell = 1$
	& $\wlx =  0, \ell = 1$ \\
		$\lambda = (2,1)$
	& $p_\lambda(x) = 0$
	& $p_\lambda(x) = \frac{2}{3}$ \\
	& -----
	& $U \ket{0}$ \\[3pt] \hline
	& $\wlx = 0, \ell = 3$
	& $\wlx = 1, \ell = 3$ \\
		$\lambda = (3,0)$
	& $p_\lambda(x) = 1$
	& $p_\lambda(x) = \frac{1}{3}$ \\
	& $U\xp{3} \ket{000}$
	& $U\xp{3} \frac{\ket{001}+\ket{010}+\ket{100}}{\sqrt{3}}$
\end{tabular}
\end{center}

Let us now go a step further and think about the actions we could take after the generic pre-processing to finish computing the three-qubit majority function. Note that for both $x=000$ and $x=001$, the correct output is $U \ket{0}$.
\begin{itemize}
\item If the measurement outcome is $\lambda = (2,1)$, we can simply return the remaining 1-qubit state $U \ket{0}$ since it is the desired output.
Since the outcome $\lambda = (2,1)$ cannot occur on inputs $x \in \set{000,111}$, we also have additional classical side-information about the Hamming weight of $x$.
\item If the outcome is $\lambda = (3,0)$, we cannot eliminate any input $x$. However, we can simply return any of the three qubits. If $x = 000$ we end up returning the desired target state $U \ket{0}$ while for $x = 001$ we return a mixed state $U \of[\big]{\frac{2}{3} \proj{0} + \frac{1}{3} \proj{1}} U\ct$ that has fidelity $2/3$ with the desired target state.
\end{itemize}
Overall, this algorithm achieves worst-case fidelity
\begin{equation}
	\min \set*{0 + 1 \cdot 1, \frac{2}{3} \cdot 1 + \frac{1}{3} \cdot \frac{2}{3}} = \frac{8}{9},
\end{equation}
which we later show is optimal (see \cref{sec:Numerics}).
Note that
$8/9 \approx 0.888889$ is better than
$2/3 \approx 0.666667$ and
$5/6 \approx 0.833333$ obtained by the trivial strategy for $n = 3$ in the no-promise and promise cases, respectively (see \cref{fig:Comparison}).

\section{Completing the algorithm}\label{sec:Completing}

The goal of this section is to complete the algorithm by finding an optimal way to proceed after the generic pre-processing $\Gamma$.
We will show in \cref{lem:1param} that the second step of the algorithm can be described by a unitary-covariant channel that has only one degree of freedom.
A complete optimal algorithm can therefore be found by choosing this parameter optimally for each pre-processing measurement outcome $\lambda$.
Our main result, \cref{thm:main}, is a simple linear program that describes the optimal choice.

\subsection{Representation theory refresher}\label{sec:Refresher}

Let us briefly recall some basic facts about the representation theory of $\U{2}$ and $\SU{2}$. For more details, see \cref{apx:2x2}.

Any (polynomial) irreducible representation of $\U{2}$ is given by a map $Q_\lambda: \U{2} \to \U{m_\lambda}$, where $\lambda = (\lambda_1, \lambda_2)$ for some integers $\lambda_1 \geq \lambda_2 \geq 0$, and $m_\lambda := \lambda_1 - \lambda_2 + 1$.
Alternatively, we can use parameters $r := \lambda_2 \geq 0$ and $\ell := \lambda_1 - \lambda_2 \geq 0$, see \cref{fig:lambda}, so that
\begin{equation}
  (\lambda_1, \lambda_2) = (r+\ell,r).
  \label{eq:lambda-lr}
\end{equation}
Note that $\lambda$ is a partition of $n = \lambda_1 + \lambda_2 = 2r + \ell$.
With these conventions,
\begin{equation}
  Q_{\lambda}(M)
  := (\det M)^r T^\ell(M),
  \qquad \forall M \in \L(\C^2),
\end{equation}
where $\det M$ denotes the determinant of $M$ and the map
$T^\ell: \L(\C^2) \to \L(\C^{\ell+1})$
is explicitly described in \cref{apx:Construction}.
The matrix entries of $T^\ell(M)$ are degree-$\ell$ homogeneous polynomials in the entries of $M$ (\eg, $T^0(M) = 1$ and $T^1(M) = M$).
Since $n$ is required to be odd for an $n$-variable Boolean function to be equivariant, see \cref{sec:Boolean functions}, we cannot have $\ell = 0$.
From now on we assume that $\ell \geq 1$.

\begin{figure}
\centering


\begin{tikzpicture}[> = stealth, semithick,
    l/.style = {fill = white, inner sep = 2pt},
  ]
  \def\t{9} 
  \def\s{5} 
  \def\w{0.6cm}
  \draw (\w/2,   \w/2) -- (\t*\w+\w/2,   \w/2);
  \draw (\w/2,  -\w/2) -- (\t*\w+\w/2,  -\w/2);
  \draw (\w/2,-3*\w/2) -- (\s*\w+\w/2,-3*\w/2);
  \foreach \i in {0,...,\s}
    \draw (\w/2+\i*\w,\w/2) -- (\w/2+\i*\w,-3*\w/2);
  \pgfmathsetmacro\sp{\s+1}
  \foreach \i in {\sp,...,\t}
    \draw (\w/2+\i*\w,\w/2) -- (\w/2+\i*\w,-\w/2);
  \draw[<->] (\w/2,\w) -- (\t*\w+\w/2,\w);
  \node[l] at (\t*\w/2+\w/2,\w) {$\lambda_1$};
  \draw[<->] (\w/2,-2*\w) -- (\s*\w+\w/2,-2*\w);
  \node[l] at (\s*\w/2+\w/2,-2*\w) {$\lambda_2 = r$};
  \pgfmathsetmacro\st{\t-\s}
  \draw[<->] (\s*\w+\w/2,-\w) -- (\t*\w+\w/2,-\w);
  \node[l] at (\st*\w/2+\s*\w+\w/2,-\w) {$\ell$};
\end{tikzpicture}
\caption{\label{fig:lambda}Correspondence between the parameters $r,\ell$ and partitions $\lambda = (\lambda_1,\lambda_2)$ in \cref{eq:lambda-lr}. The total number of boxes is $n = \lambda_1 + \lambda_2 = 2r + \ell$.}
\end{figure}
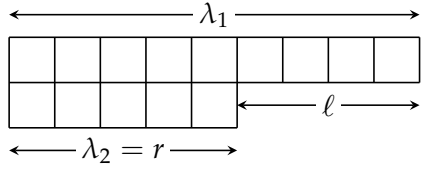

If $M \in \L(\C^2)$ is invertible, we denote its \emph{dual}\footnote{Not to be confused with $\overline{M}$, the entry-wise complex conjugate of $M$. Note that $M^* = \overline{M}$ only when $M$ is unitary.} by $M^* := (M^{-1})\tp$.
Note that $M \mapsto M^*$ is a homomorphism, also known as the \emph{dual of the defining representation}.
Its tensor product $M^* \x T^\ell(M)$ with any of the representations $T^\ell$ is also a representation, albeit not irreducible.
It is decomposed into irreducibles by the
\emph{dual Clebsch--Gordan transform}
$\DCG_\ell \in \U{\C^{2(\ell+1)}}$:
\begin{equation}
  \DCG_\ell \of[\big]{M^* \x T^\ell(M)} \DCGct_\ell
  = T^{\ell-1}(M) \+ \frac{1}{\det M} T^{\ell+1}(M).
  \label{eq:Dual CG on any M}
\end{equation}
Note that both sides have the same dimension since $2 \cdot (\ell+1) = \ell + (\ell+2)$.
Moreover, when multiplied by $\det M$, the matrix entries on both sides are polynomials of degree $\ell+1$ in the matrix entries of $M$.
While \cref{eq:Dual CG on any M} can be taken as an implicit definition\footnote{Multiplying both sides of \cref{eq:Dual CG on any M} from the right by $\DCG_\ell$ and comparing the coefficients at all monomials produces a linear system of equations that determines $\DCG_\ell$ almost uniquely (up to two signs).} of $\DCG_\ell$, we also provide an explicit formula in \cref{eq:DualCl}.
For more details on $\DCG_\ell$, see \cref{apx:Dual CG}.

\subsection{Two extremal channels}\label{sec:Extremal}

\newcommand{\UCl}{\UCPTP{\C^{\ell+1}}{\C^2}}

Recall from \cref{lem:Gamma output} that the generic pre-processing $\Gamma$ given by \cref{alg:Gamma} produces a measurement outcome $\lambda = (\lambda_1, \lambda_2) \pt n$ and a corresponding post-measurement state of dimension $m_\lambda = \ell + 1$ where $\ell := \lambda_1 - \lambda_2 \geq 1$. From \cref{cor:Opt} we know that an optimal algorithm for computing a function in a basis-independent way can be assumed to be unitary-equivariant. Therefore, to complete the algorithm it suffices to consider only unitary-covariant channels from $\C^{\ell+1}$ to a single qubit. Following \cref{def:unitary covariance}, we denote the set of such channels by $\UCl$, where the $\U{2}$-representations on the input and output spaces are $T^\ell(U)$ and $T^1(U) = U$, respectively. It turns out that, irrespective of the value of $\ell \geq 1$, any such channel can be described by a single parameter. The following lemma proves this and derives the Choi matrices of the two extremal channels.

\begin{lemma}\label{lem:1param}
For any integer $\ell \geq 1$, the set $\UCl$ of unitary-covariant channels is convex with two extremal points. The Choi matrices of the two extremal channels are
\begin{align}
	\J{\Phi^\ell_1}
 &= \DCGct_\ell
    \sof*{
      \frac{\ell+1}{\ell} \id_\ell \+ 0_{\ell+2}
    }
    \DCG_\ell, &
  \J{\Phi^\ell_2}
 &= \DCGct_\ell
    \sof*{
      0_\ell \+ \frac{\ell+1}{\ell+2} \id_{\ell+2}
    }
    \DCG_\ell,
  \label{eq:Choi12}
\end{align}
where $\DCG_\ell \in \U{\C^{2(\ell+1)}}$ denotes the dual Clebsch--Gordan transform (see \cref{eq:Dual CG on any M} or \cref{apx:Dual CG}) and $0_d$ denotes the $d \times d$ all-zeroes matrix.
In particular, their Choi ranks are $\rank \J{\Phi^\ell_1} = \ell$ and $\rank \J{\Phi^\ell_2} = \ell+2$, and they satisfy
\begin{equation}
	\ell \J{\Phi^\ell_1} + (\ell+2) \J{\Phi^\ell_2}
  = (\ell+1) \id_{2(\ell+1)}.
  \label{eq:ChoiRel}
\end{equation}
\end{lemma}

\begin{proof}
The set $\UCl$ is convex since it is the intersection of two convex sets -- the set of all channels and the set of all unitary-covariant maps from $\L(\C^{\ell+1})$ to $\L(\C^2)$.

To argue that there are only two extremal points, consider an arbitrary channel $\Phi \in \UCl$.
From \cref{cor:Commuting} we know that its Choi matrix $\J{\Phi}$ commutes with $\overline{U} \x T^\ell(U)$, for all unitaries $U \in \U{2}$.
Note that $\overline{U} = (U\ct)\tp = (U^{-1})\tp = U^*$, so based on \cref{eq:Dual CG on any M} we can use the dual Clebsch--Gordan transform $\DCG_\ell$ to decompose $\overline{U} \x T^\ell(U)$ as a direct sum:
\begin{equation}
  \DCG_\ell
  \of[\big]{\overline{U} \x T^\ell(U)}
  \DCGct_\ell
= T^{\ell-1}(U) \+
  \tfrac{1}{\det U} T^{\ell+1}(U),
\end{equation}
where the two blocks are of dimension $\ell$ and $\ell+2$, respectively.
Since the representations $T^{\ell \pm 1}$ are irreducible (see \cref{lem:irreducibility}), we can use Schur's lemma (\cref{lem:Schur}) to conclude that in the dual Clebsch--Gordan basis $\J{\Phi}$ takes the form
\begin{equation}
  \DCG_\ell \J{\Phi} \DCGct_\ell
  = c_1 \id_\ell \+ c_2 \id_{\ell+2},
  \label{eq:Choi'}
\end{equation}
for some $c_1, c_2 \in \C$.

Let us consider the CPTP constraints \eqref{eq:ChoiCPTP} for the quantum channel $\Phi$.
Since $\Phi$ is completely positive, $c_1, c_2 \geq 0$.
Since $\Phi$ is trace-preserving, on the one hand $\Tr \J{\Phi} = \Tr(\id_{\ell+1}) = \ell + 1$, while on the other we see from \cref{eq:Choi'} that $\Tr \J{\Phi} = \ell c_1 + (\ell+2) c_2$, so
\begin{equation}
  \ell + 1 = \ell c_1 + (\ell+2) c_2.
\end{equation}
Because of this constraint there is only one free parameter.
Furthermore, since $c_1, c_2 \geq 0$, there are two extremal points:
\begin{align}
  (c_1,c_2) &= (\tfrac{\ell+1}{\ell},0), &
  (c_1,c_2) &= (0,\tfrac{\ell+1}{\ell+2}).
\end{align}
Substituting them in \cref{eq:Choi'}, the corresponding two extremal channels $\Phi^\ell_1$ and $\Phi^\ell_2$ have Choi matrices
\begin{align}
	\DCG_\ell \J{\Phi^\ell_1} \DCGct_\ell
  &= \frac{\ell+1}{\ell} \id_\ell \+ 0_{\ell+2}, &
  \DCG_\ell \J{\Phi^\ell_2} \DCGct_\ell
  &= 0_\ell \+ \frac{\ell+1}{\ell+2} \id_{\ell+2}
  \label{eq:extremal Choi matrices}
\end{align}
in the dual Clebsch--Gordan basis.
Clearly, $\rank \J{\Phi^\ell_1} = \ell$ and $\rank \J{\Phi^\ell_2} = \ell+2$, and
\begin{equation}
  \ell \J{\Phi^\ell_1} + (\ell+2) \J{\Phi^\ell_2}
  = (\ell+1) \id_{\ell+(\ell+2)}
  \label{eq:ChoiRel again}
\end{equation}
since $\DCG_\ell$ is unitary.
\end{proof}

\begin{example}[Extremal channels for $\ell = 1$]
Choi matrices of the two $\ell = 1$ extremal channels are:
\begin{align}
  \J{\Phi^1_1} &= \mx{
    1 & 0 & 0 & 1 \\
    0 & 0 & 0 & 0 \\
    0 & 0 & 0 & 0 \\
    1 & 0 & 0 & 1
  }, &
  \J{\Phi^1_2} &= \frac{1}{3} \mx{
    1 & 0 & 0 & -1 \\
    0 & 2 & 0 & 0 \\
    0 & 0 & 2 & 0 \\
   -1 & 0 & 0 & 1
  }.
\end{align}
Note that $\J{\Phi^1_1} + 3 \J{\Phi^1_2} = 2 \id_4$, which agrees with \cref{eq:ChoiRel again} when $\ell = 1$.
The corresponding dual Clebsch--Gordan transform (see \cref{apx:Dual CG}) is given by\footnote{In the context of mixed Schur--Weyl duality, this transform also appears in Example~3.9 of \cite{GO22}. The matrix $U$ in eq.~(58) there slightly differs from our $D_1$ due to different conventions.}
\begin{equation}
  \DCG_1 = \frac{1}{\sqrt{2}}
  \mx{
    1 & 0 & 0 & 1 \\
    0 & 0 & \sqrt{2} & 0 \\
   -1 & 0 & 0 & 1 \\
    0 & -\sqrt{2} & 0 & 0
  }.
\end{equation}
In particular,
$\DCG_1 \J{\Phi^1_1} \DCGct_1 = \diag(2,0,0,0)$ and
$\DCG_1 \J{\Phi^1_2} \DCGct_1 = \diag(0,\frac{2}{3},\frac{2}{3},\frac{2}{3})$
as claimed.
Since the Choi matrix \eqref{eq:J UNOT} of the ideal universal $\NOT$ operation can be expressed as
$\J{\UNOT} = -\frac{1}{2} \J{\Phi^1_1} + \frac{3}{2} \J{\Phi^1_2}$,
the ideal operation is not completely positive.
\end{example}

\subsection{Understanding the extremal channels}\label{sec:Understanding}

This section provides an intuitive way to think about the above two extremal channels and shows that they are related to the partial trace and the universal $\NOT$.

Let us denote the \emph{symmetric subspace} of $\ell \geq 1$ qubits by
\begin{equation}
  \Sym{\ell} := \spn \set[\big]{\ket{s_\ell(w)} : w \in \set{0,\dotsc,\ell}}
  \subseteq \of{\C^2}\xp{\ell},
\end{equation}
where $\ket{s_\ell(w)}$ is the symmetric $\ell$-qubit state of Hamming weight $w$, see \cref{eq:slw}.
Since $\braket{s_\ell(w)}{s_\ell(w')} = \delta_{ww'}$, the projector onto the $\ell$-qubit symmetric subspace can be written as
\begin{align}
  \Pi_{\Sym{\ell}}
:= \sum_{w=0}^\ell \proj{s_\ell(w)}
 = \frac{1}{\ell!} \sum_{\pi \in \S_\ell} P(\pi)
 = (\ell+1) \int \of[\big]{\proj{\psi}}\xp{\ell} \, d\psi,
  \label{eq:PiSym3}
\end{align}
where $P(\pi) \in \U{2^\ell}$ are the $\ell$-qubit permutations defined in \cref{eq:P} and $d\psi$ denotes the uniform measure on the Bloch sphere, see \cref{eq:dpsi}.
The last two identities are proved in \cref{lem:PiSym3}.
For more details on the symmetric subspace, see \cref{apx:Sym}.

Let $V^\ell \in \U{\C^{\ell+1},\of{\C^2}\xp{\ell}}$ denote the isometry that maps the $k$-th standard basis state of $\C^{\ell+1}$ to the symmetric state $\ket{s_\ell(k)}$ of Hamming weight $k$ in the $\ell$-qubit symmetric subspace:
\begin{align}
  V^\ell := \sum_{k=0}^\ell \ketbra{s_\ell(k)}{k},
  \label{eq:Vl!}
\end{align}
and let $\Sigma^\ell \in \CPTP{\C^{\ell+1}}{\of{\C^2}\xp{\ell}}$ denote the quantum channel corresponding to the isometry $V^\ell$:
\begin{equation}
  \Sigma^\ell(\rho)
  := V^\ell \rho V^{\ell\dagger},
  \label{eq:Sigmal!}
\end{equation}
for all $\rho \in \L(\C^{\ell+1})$.

The following two maps from $\L(\C^{\ell+1})$ to $\L(\C^2)$ will play an important role.
They first embed the input into the $\ell$-qubit symmetric subspace $\Sym{\ell}$.
The first map then traces out all qubits but the first, while the second performs the uniform measurement on all tensor power states and outputs for each outcome the corresponding orthogonal state.
Because of this, we denote them by $\Phi^\ell_{\Tr}$ and $\Phi^\ell_{\UNOT}$, respectively.

\begin{definition}[$\Phi^\ell_{\Tr}$ and $\Phi^\ell_{\UNOT}$]\label{def:TrUNOT}
For any integer $\ell \geq 1$, let
$\Phi^\ell_{\Tr}, \Phi^\ell_{\UNOT}: \L(\C^{\ell+1}) \to \L(\C^2)$
act on $\rho \in \L(\C^{\ell+1})$ as follows:
\begin{align}
  \Phi^\ell_{\Tr}(\rho)
&:= \Tr_{2,\dotsc,\ell} \sof*{\Sigma^\ell(\rho)}, \label{eq:PhiTr} \\
  \Phi^\ell_{\UNOT}(\rho)
&:= (\ell+1) \int
    \bra{\psi}\xp{\ell}
    \sof*{\Sigma^\ell(\rho)}
    \ket{\psi}\xp{\ell}
    \of[\big]{\id_2 - \proj{\psi}} \, d\psi, \label{eq:PhiUNOT}
\end{align}
where $\ket{\psi} \in \C^2$ and $d\psi$ denotes the uniform measure on the Bloch sphere, see \cref{eq:dpsi}.
Sometimes we will use the notation
$\Phi^\lambda_{\Tr} := \Phi^\ell_{\Tr}$ and
$\Phi^\lambda_{\UNOT} := \Phi^\ell_{\UNOT}$
where $\lambda = (\lambda_1, \lambda_2)$ and $\ell := \lambda_1 - \lambda_2$.
\end{definition}

Let us verify that both maps are in fact quantum channels.

\begin{prop}
For any integer $\ell \geq 1$, the maps $\Phi^\ell_{\Tr}$ and $\Phi^\ell_{\UNOT}$ are quantum channels.
\end{prop}

\begin{proof}
Let us verify that $\Phi^\ell_{\Tr}$ and $\Phi^\ell_{\UNOT}$ conform to Stinespring and Kraus representations, respectively (see \cref{apx:Kraus and Stinespring} for more background on quantum channel representations).

It is evident that $\Phi^\ell_{\Tr}$ is a quantum channel since \cref{eq:PhiTr} constitutes its Stinespring dilation.
Similarly, \cref{eq:PhiUNOT} constitutes a Kraus representation of $\Phi^\ell_{\UNOT}$:
\begin{equation}
  \Phi^\ell_{\UNOT}(\rho) = \int K^\ell_\psi \rho K^{\ell\dagger}_\psi \, d\psi,
\end{equation}
with Kraus operators
$K^\ell_\psi := \sqrt{\ell+1} \ket{\psi^\perp} \cdot \bra{\psi}\xp{\ell} V^\ell \in \L(\C^{\ell+1},\C^2)$
where $\ket{\psi^\perp} \in \C^2$ is such that $\proj{\psi^\perp} = \id_2 - \proj{\psi}$.
In particular, $\Phi^\ell_{\UNOT}$ is completely positive.
To verify that $\Phi^\ell_{\UNOT}$ is trace-preserving, we compute
\begin{equation}
  \int K^{\ell\dagger}_\psi K^\ell_\psi \, d\psi
  = (\ell+1) \int V^{\ell\dagger} \of[\big]{\proj{\psi}}\xp{\ell} V^\ell \, d\psi
  = V^{\ell\dagger} \Pi_{\Sym{\ell}} V^\ell
  = \id_{\ell+1},
\end{equation}
where the last two equalities follow from \cref{eq:PiSym3} and the fact that $V^\ell$ defined in \cref{eq:Vl!} is an isometry.
\end{proof}

\Cref{lem:identification} in \cref{apx:Identification} identifies the above two quantum channels with the two extremal channels from \cref{lem:1param}:
\begin{align}
  \Phi^\ell_1 &= \Phi^\ell_{\Tr}, &
  \Phi^\ell_2 &= \Phi^\ell_{\UNOT}.
\end{align}
Hence, we conclude from \cref{lem:1param} that any unitary-covariant quantum channel from $\C^{\ell+1}$ to $\C^2$ must be a convex combination of $\Phi^\ell_{\Tr}$ and $\Phi^\ell_{\UNOT}$.

\begin{cor}\label{cor:covariant channels}
Any unitary-covariant quantum channel $\Phi_\ell \in \UCl$ is of the form
\begin{equation}
  \Phi_\ell = t \Phi^\ell_{\Tr} + (1-t) \Phi^\ell_{\UNOT},
\end{equation}
for some $t \in [0,1]$,
where $\Phi^\ell_{\Tr}, \Phi^\ell_{\UNOT} \in \UCl$ are defined in \cref{def:TrUNOT}.
\end{cor}

\subsection{The optimal algorithm and its fidelity}\label{sec:Optimal algorithm}

In this section, we describe a parametrized \emph{template algorithm} $\mc{A}_\vt$ (see \cref{alg:Template}) that depends on a vector of \emph{interpolation parameters} $\vt = \of{t_\lambda \in [0,1]: \lambda \pt n}$.
First, it applies the generic pre-processing $\Gamma$ (see \cref{alg:Gamma}) and then, depending on the measurement outcome $\lambda$, a general unitary-covariant channel $\Phi_\lambda$.
According to \cref{cor:covariant channels}, $\Phi_\lambda$ is a convex combination $t_\lambda \Phi^\lambda_{\Tr} + (1 - t_\lambda) \Phi^\lambda_{\UNOT}$ of the two extremal maps discussed in \cref{sec:Understanding}, with some interpolation parameter $t_\lambda \in [0,1]$.
The following theorem shows that, for any symmetric and equivariant Boolean function $f$, the optimal fidelity for computing $f$ in a basis-independent way can be achieved by the algorithm $\mc{A}_\vt$ with appropriately chosen interpolation parameters $\vt$.

\begin{theorem}\label{thm:template}
Let $f: \set{0,1}^n \to \set{0,1}$ be a symmetric and equivariant Boolean function.
For some choice of interpolation parameters $\vt$, the quantum channel $\mc{A}_\vt \in \CPTP{\C^{2^n}}{\C^2}$ implemented by \cref{alg:Template} is optimal for computing $f$ in a basis-independent way.
Moreover, $\mc{A}_\vt$ is permutation-invariant and unitary-equivariant.
\end{theorem}

\begin{alg}[label={alg:Template}]{Template algorithm $\mc{A}_\vt$}
\textbf{Input:}
Quantum state $U\xp{n} \ket{x}$ with an unknown $x \in \set{0,1}^n$ and $U \in \U{2}$.\medskip\\
\textbf{Output:}
An approximation of $U \ket{f(x)}$.\medskip\\
\textbf{Parameters:}
A vector of \emph{interpolation parameters} $\vt = \of{t_\lambda \in [0,1]: \lambda \pt n}$.
\begin{enumerate}[\bf Step 1:]
  \item Perform the generic pre-processing $\Gamma$ (see \cref{alg:Gamma}).
  \item If the measurement outcome was $\lambda$, apply the channel (see \cref{def:TrUNOT})
  \begin{equation*}
    \Phi_\lambda
   := t_\lambda \Phi^\lambda_{\Tr}
    + (1 - t_\lambda) \Phi^\lambda_{\UNOT}
  \end{equation*}
\end{enumerate}
\end{alg}

\begin{proof}
Based on \cref{cor:Generic}, an optimal quantum channel for computing $f$ is of the form $\Phi' \circ \Gamma$, where $\Gamma$ is the generic pre-processing (see \cref{alg:Gamma}) and $\Phi'$ is some unitary-covariant channel.
Depending on the measurement outcome $\lambda = (\lambda_1, \lambda_2)$ obtained by $\Gamma$, the second step $\Phi'$ simply consists of applying some unitary-covariant channel $\Phi_\lambda \in \UCPTP{\C^{\ell+1}}{\C^2}$ on the remaining $(\ell+1)$-dimensional state where $\ell := \lambda_1 - \lambda_2$.

According to \cref{cor:covariant channels}, the unitary-covariant channel $\Phi_\lambda$ is a convex combination of the two extremal channels $\Phi^\lambda_{\Tr}$ and $\Phi^\lambda_{\UNOT}$ introduced in \cref{def:TrUNOT}:
\begin{equation}
  \Phi_\lambda
 := t_\lambda \Phi^\lambda_{\Tr}
  + (1 - t_\lambda) \Phi^\lambda_{\UNOT},
\end{equation}
for some parameter $t_\lambda \in [0,1]$.
Since \cref{alg:Template} describes a general quantum channel $\mc{A}_\vt$ that consists of applying $\Gamma$ followed by some unitary-covariant $\Phi_\lambda$, an optimal channel for computing $f$ can be found by optimizing the parameters $\vt$ in $\mc{A}_\vt$.
The template algorithm $\mc{A}_\vt$ is permutation-invariant since so is $\Gamma$ (see \cref{lem:Gamma}), and it is unitary-equivariant since it is a composition of two unitary-covariant channels (see \cref{lem:Gamma,cor:covariant channels}).
\end{proof}

Based on \cref{thm:template}, we can find an optimal quantum channel for any symmetric and equivariant Boolean function $f$ simply by finding for each measurement outcome $\lambda \pt n$ an optimal choice of the corresponding parameter $t_\lambda \in [0,1]$ that determines the probability of applying $\Phi^\lambda_{\Tr}$ versus $\Phi^\lambda_{\UNOT}$ in Step~2 of \cref{alg:Template}.
Our main result, \cref{thm:main} below, shows that the optimal fidelity $F_f$ and an optimal choice of the interpolation parameters $\vt$ in the template algorithm $\mc{A}_\vt$ can be found by solving a linear program of size $O(n)$.
The resulting algorithm $\mc{A}_\vt$ computes $f$ with optimal fidelity.

\newcommand{\FTr}{a_k(\ham)}
\newcommand{\FUNOT}{b_k(\ham)}

\begin{theorem}\label{thm:main}
Let $f: \set{0,1}^n \to \set{0,1}$ be a symmetric and equivariant Boolean function.
The optimal fidelity $F_f$ for computing $f$ in a basis-independent way can be found by solving the following linear program with $\floor{n/2}+2$ variables and $3\floor{n/2}+3$ constraints:
\begin{equation}
  \begin{array}{llll}
  	& F_f = \displaystyle\max_{c,\vt} c \\
  	& \text{such that}
    & \displaystyle
      \sum_{k=0}^{\ham} p_k(\ham)
      \sof[\Big]{
        t_k \FTr
      +	(1-t_k) \FUNOT
      }
      \geq c,
    & \forall \ham \in \set*{0,\dotsc,\floor*{\frac{n}{2}}}, \\
  & & 0 \leq t_k \leq 1,
    & \forall k \in \set*{0,\dotsc,\floor*{\frac{n}{2}}}.
  \end{array}
  \tag{LP}
  \label{eq:LP}
\end{equation}
Here we are using the following shorthands:
\begin{align}
	p_k(\ham)
&:= \frac{\binom{n}{k} - \binom{n}{k-1}}
         {\binom{n}{\ham}}, \label{eq:pkh} \\
  \FTr
&:= \frac{n-\ham-k}{n-2k}
  - \symf(\ham) \frac{n-2\ham}{n-2k}, \label{eq:FTr} \\
  \FUNOT
&:= \frac{\ham-k+1}{n-2k+2}
  + \symf(\ham) \frac{n-2\ham}{n-2k+2}, \label{eq:FUNOT}
\end{align}
where $\symf(\ham) := f(0^{n-\ham}1^{\ham})$ denotes the value of $f$ on strings of Hamming weight $\ham$.
Given an optimal solution $\of*{t^*_k : k = 0,\dotsc,\floor{n/2}}$ of the above linear program, the template algorithm $\mc{A}_{\vt^*}$ (see \cref{alg:Template}) achieves the optimal fidelity $F_f$ if we set $t^*_\lambda := t^*_k$ where $\lambda = (n-k,k)$ and $k \in \set{0,\dotsc,\floor{n/2}}$.
\end{theorem}

\begin{proof}
According to \cref{thm:template}, an optimal quantum channel for computing $f$ can be found by optimizing the interpolation parameters $\vt$ in the template algorithm $\mc{A}_\vt$ (\cref{alg:Template}).
Moreover, we also known from \cref{thm:template} that $\mc{A}_\vt$ is unitary-equivariant, so we can omit the unitaries $U$ from the input state $U\xp{n} \ket{x}$ and consider only standard basis inputs $\ket{x}$.

Let us first derive the output of $\mc{A}_\vt$ on input $x \in \set{0,1}^n$.
The first step of $\mc{A}_\vt$ is the generic pre-processing $\Gamma$.
According to \cref{lem:Gamma output},
\begin{equation}
  \Gamma \of[\big]{\proj{x}}
= \bigoplus_{\lambda \pt n}
  p_\lambda(x)
  \proj{\wlx},
  \label{eq:Gamma of x}
\end{equation}
where $\wlx := \abs{x} - \lambda_2$ and the probability $p_\lambda(x)$ of measurement outcome $\lambda$ is given in \cref{eq:plx}.
For each $\lambda$, the second step of $\mc{A}_\vt$ applies the map
$\Phi_\lambda := t_\lambda \Phi^\lambda_{\Tr} + (1 - t_\lambda) \Phi^\lambda_{\UNOT}$
for some $t_\lambda \in [0,1]$,
so the final state of the overall procedure is
\begin{equation}
  \mc{A}_\vt \of[\big]{\proj{x}}
= \sum_{\lambda \pt n}
  p_\lambda(x)
  \sof*{
    t_\lambda \Phi^\lambda_{\Tr} \of[\big]{\proj{\wlx}}
    + (1 - t_\lambda) \Phi^\lambda_{\UNOT} \of[\big]{\proj{\wlx}}
  },
  \label{eq:At of x}
\end{equation}
where the direct sum has been replaced by a regular sum since the measurement outcome $\lambda$ is effectively discarded.

If we parameterize the partition as $\lambda = (n-k,k)$, we can sum over $k \in \set{0,\dotsc,\floor{n/2}}$ instead of $\lambda \pt n$.
Since $p_\lambda(x)$ and $\wlx$ depend only on the Hamming weight $\ham := \abs{x}$, we can write them in terms of $k$ and $\ham$ as follows:
$\wlx = \ham - k$ and $p_\lambda(x) = p_k(\ham)$ where $p_k(\ham)$ is defined in \cref{eq:pkh}.
Since $\ell := \lambda_1 - \lambda_2 = n - 2 \lambda_2 = n - 2k$,
we get from \cref{lem:identification} that
$\Phi^\lambda_{\Tr} = \Phi^\ell_1 = \Phi^{n-2k}_1$ and
$\Phi^\lambda_{\UNOT} = \Phi^\ell_2 = \Phi^{n-2k}_2$.
Finally, we replace $t_\lambda$ by the corresponding $t_k$.
With this notation, \cref{eq:At of x} becomes
\begin{equation}
  \mc{A}_\vt \of[\big]{\proj{x}}
  = \sum_{k=0}^\ham p_k(\ham)
  \sof*{
    t_k \Phi^{n-2k}_1 \of[\big]{\proj{\ham-k}}
    + (1 - t_k) \Phi^{n-2k}_2 \of[\big]{\proj{\ham-k}}
  },
\end{equation}
where we have truncated the sum at $k = \ham$ since $p_\lambda(x) = 0$ when $k > \abs{x} = \ham$, see \cref{eq:plx}.

To evaluate the fidelity
$F\of[\big]{ \ket{f(x)}, \mc{A}_\vt \of[\big]{\proj{x}} }
= \bra{f(x)} \mc{A}_\vt \of[\big]{\proj{x}} \ket{f(x)}$
with the correct output state $\ket{f(x)}$, we can use linearity and deal with each term separately.
Since $f$ is symmetric, the ideal output state for an input $x \in \set{0,1}^n$ of Hamming weight $\ham := \abs{x}$ is
$\ket{f(x)} = \ket{f(0^{n-\ham}1^{\ham})} =: \ket{\symf(\ham)}$.
If we let
\begin{align}
  \FTr
&:= \bra{\symf(\ham)}
    \Phi^{n-2k}_1
    \of[\big]{\proj{\ham-k}}
    \ket{\symf(\ham)}, \label{eq:FTrdef} \\
  \FUNOT
&:= \bra{\symf(\ham)}
    \Phi^{n-2k}_2
    \of[\big]{\proj{\ham-k}}
    \ket{\symf(\ham)}, \label{eq:FUNOTdef}
\end{align}
we can write the output fidelity on input $x$ as follows:
\begin{equation}
  \bra{\symf(\ham)} \mc{A}_\vt \of[\big]{\proj{x}} \ket{\symf(\ham)}
= \sum_{k=0}^{\ham} p_k(\ham)
  \sof[\Big]{
    t_k \FTr
  +	(1-t_k) \FUNOT
  }.
  \label{eq:average fidelity}
\end{equation}

Since $\mc{A}_\vt$ is permutation-invariant and unitary-equivariant, according to \cref{lem:Fidelity} its fidelity for computing $f$ in a basis-independent way is
\begin{equation}
  F_f(\mc{A}_\vt)
= \min_{0 \leq \ham \leq \floor{n/2}}
  F\of[\big]{
    \ket{f(0^{n-\ham}1^{\ham})},
    \mc{A}_\vt(\proj{0^{n-\ham}1^{\ham}})
  }.
\end{equation}
Hence, it suffices to determine the output fidelity of $\mc{A}_\vt$ only on input strings of the form $x = 0^{n-\ham}1^{\ham}$, for all Hamming weights $\ham \in \set{0,\dotsc,\floor{n/2}}$.
From \cref{eq:average fidelity} we see that
\begin{equation}
  F_f(\mc{A}_\vt)
= \min_{0 \leq \ham \leq \floor{n/2}}
  \sum_{k=0}^{\ham} p_k(\ham)
  \sof[\Big]{
    t_k \FTr
  +	(1-t_k) \FUNOT
  }.
\end{equation}
The optimal worst-case fidelity $F_f$ for computing $f$ can be obtained by maximizing the interpolation parameters $t_k$ and a constant $c \in \R$, subject to $F_f(\mc{A}_\vt) \geq c$ and $0 \leq t_k \leq 1$ for all $k \in \set*{0,\dotsc,\floor*{\frac{n}{2}}}$.
This produces the desired linear program.

To complete the proof, it remains to show that
\cref{eq:FTrdef,eq:FUNOTdef}
defining $\FTr$ and $\FUNOT$ agree with
\cref{eq:FTr,eq:FUNOT} in the statement of the theorem.
From \cref{lem:Phi12} in \cref{apx:Action} we have that for any $\ell \geq 1$ and $w \in \set{0,\dotsc,\ell}$,
\begin{align}
  \Phi^\ell_1 \of[\big]{\proj{w}}
  &= \frac{\ell-w}{\ell} \proj{0}
   + \frac{w}{\ell} \proj{1}, \label{eq:ww1} \\
  \Phi^\ell_2 \of[\big]{\proj{w}}
  &= \frac{w+1}{\ell+2} \proj{0}
   + \frac{\ell+1-w}{\ell+2} \proj{1}. \label{eq:ww2}
\end{align}
Substituting $\ell = n - 2k$ and $w = h - k$ into these formulas, we find that
\begin{align}
  \Phi^{n-2k}_1\of[\big]{\proj{h-k}}
  &= \frac{n-\ham-k}{n-2k} \proj{0}
   + \frac{\ham-k}{n-2k} \proj{1}, \\
  \Phi^{n-2k}_2\of[\big]{\proj{h-k}}
  &= \frac{\ham-k+1}{n-2k+2} \proj{0}
   + \frac{n-\ham-k+1}{n-2k+2} \proj{1}.
\end{align}
Plugging this back into \cref{eq:FTrdef,eq:FUNOTdef} gives us the following explicit formulas:
\begin{align}
  \FTr &=
  \begin{cases}
    \frac{n-\ham-k}{n-2k} & \text{if $\symf(\ham) = 0$}, \\
    \frac{\ham-k}{n-2k}   & \text{if $\symf(\ham) = 1$},
  \end{cases} &
   \FUNOT &=
   \begin{cases}
     \frac{\ham-k+1}{n-2k+2}   & \text{if $\symf(\ham) = 0$}, \\
     \frac{n-\ham-k+1}{n-2k+2} & \text{if $\symf(\ham) = 1$},
   \end{cases}
\end{align}
which agree with \cref{eq:FTr,eq:FUNOT}.
\end{proof}

\begin{theorem}\label{thm:implementation}
For any choice of the interpolation parameters $\vt$, the $n$-qubit template algorithm $\mc{A}_\vt$ (\cref{alg:Template}) can be implemented using $O(n^4 \log n)$ elementary quantum gates.
\end{theorem}

\begin{proof}
Step~1 of \cref{alg:Template} consists of the generic pre-processing $\Gamma$ described in \cref{alg:Gamma}.
The only non-trivial step of $\Gamma$ is the $n$-qubit Schur transform $\Usch$, which can be implemented exactly using $O(n^4 \log n)$ quantum gates if CNOT and all single-qubit gates are available~\cite{KS17}.

Step~2 of \cref{alg:Template} probabilistically applies either $\Phi^\ell_{\Tr}$ or $\Phi^\ell_{\UNOT}$.
\Cref{lem:Phi complexity} in \cref{apx:Isometries} shows that both channels can be implemented using $O(n \log n)$ elementary gates.

Hence, the overall gate complexity of $\mc{A}_\vt$ is $O(n^4 \log n)$ irrespective of the interpolation parameters $\vt$.
While computing these parameters takes additional time, the linear program \eqref{eq:LP} is of size $O(n)$ and solving such program takes significantly less than $O(n^4)$ time, hence this does not affect the asymptotic scaling.
\end{proof}

\section{Applications of the general framework}\label{sec:Applications}

This section contains several applications of the general framework for computing symmetric equivariant Boolean functions derived in \cref{sec:Optimal algorithm}.
First, in \cref{sec:Numerics}, we present exact numerical values of optimal fidelities for all functions up to $n = 7$ arguments, and some empirical observations regarding these values.
We also compute exact optimal fidelities for two infinite families of functions: majority $\MAJ_n$ and parity $\PAR_n$.
In \cref{sec:Majority}, we analyze the asymptotic behavior of the optimal fidelity for the majority function.
Finally, in \cref{sec:U-equivariant functions} we relax complete positivity and introduce the abstract notion of unitary-equivariant Boolean functions.

\subsection{Exact numerical results}\label{sec:Numerics}

For any symmetric and equivariant Boolean function $f: \set{0,1}^n \to \set{0,1}$, \cref{thm:main} shows that the optimal fidelity for computing $f$ in a basis-independent way is given by the solution to a linear program of size $O(n)$, which can easily be solved using a computer. This section lists several exact numerical results obtained in this way.

\subsubsection{Detailed solutions for $n = 1$ and $n = 3$}

For $n = 1$, the linear program \eqref{eq:LP} reads
\begin{equation}
  \max_{c,t_0} c \quad
  \of[\big]{1 - f(0)} t_0 + \frac{1 + f(0)}{3} (1 - t_0) \geq c, \quad
  0 \leq t_0 \leq 1,
\end{equation}
while for $n = 3$ it reads
\begin{align}
  \max_{c,t_0,t_1} c \quad
& (1-f(0)) t_0 + \tfrac{1+3f(0)}{5} (1 - t_0) \geq c, \nonumber \\
& \frac{1}{3} \of*{ \tfrac{2-f(1)}{3} t_0 + \tfrac{2+f(1)}{5} (1-t_0) }
+ \frac{2}{3} \of*{ (1-f(1)) t_1 + \tfrac{1+f(1)}{3} (1-t_1) } \geq c, \\
& 0 \leq t_0 \leq 1, \quad
  0 \leq t_1 \leq 1. \nonumber
\end{align}
Solutions of these linear programs for all symmetric equivariant Boolean functions $f$ with $n = 1$ and $n = 3$ arguments are summarized in \cref{tab:Optimal}, and the Choi matrices of the corresponding quantum channels are listed in \cref{apx:Physical}.

\begin{table}
  \centering
  \begin{tabular}{c|c|c|c|c|c|c}
    Function & $\symf$ & $F_f$ & $t_0$ & $t_1$ & $c_0$ & $c_1$ \\ \hline
    $\ID$  & $0$ & $1$   & $1$ & -- & $1$   & -- \\
    $\NOT$ & $1$ & $2/3$ & $0$ & -- & $2/3$ & -- \\ \hline
    $\MAJ_3$  & $00$ & $8/9$   & $1$   & $1$ & $1$   & $8/9$ \\
    $\PAR_3$  & $01$ & $3/5$   & $1/2$ & $0$ & $3/5$ & $3/5$ \\
    $\NPAR_3$ & $10$ & $4/5$   & $0$   & $1$ & $4/5$ & $4/5$ \\
    $\NMAJ_3$ & $11$ & $29/45$ & $0$   & $0$ & $4/5$ & $29/45$
  \end{tabular}
  \caption{\label{tab:Optimal}Optimal fidelities $F_f$ and the corresponding interpolation parameters $t_0,t_1$ for all symmetric equivariant Boolean functions with $n = 1$ and $n = 3$ arguments from \cref{tab:1-3bit}.
  Here $\symf$ stands for the truth table of $\symf(\abs{x}) := f(x)$
  and $c_\ham$ denotes the fidelity achieved on inputs of Hamming weight $\ham \in \set{0,\dotsc,\floor{n/2}}$.
  The Choi matrices corresponding to these channels are listed in \cref{apx:Physical}.}
\end{table}

\subsubsection{Optimal fidelities for all functions up to $n = 7$}\label{sec:Optimal fidelities}

The optimal fidelities for all symmetric equivariant Boolean functions on up to $7$ bits are as follows (for $n = 1,3,5,7$):
\begin{equation*}
\def\arraystretch{1.2}
\left\{
\begin{array}{cc}
  0 & 1 \\
  1 & \frac{2}{3}
\end{array}
\right.
\qquad
\left\{
\begin{array}{cc}
  00 & \frac{8}{9} \\
  01 & \frac{3}{5} \\
  10 & \frac{4}{5} \\
  11 & \frac{29}{45}
\end{array}
\right.
\qquad
\left\{
\begin{array}{cc}
  000 & \frac{62}{75} \\
  001 & \frac{4}{7} \\
  010 & \frac{5}{7} \\
  011 & \frac{95}{153} \\
  100 & \frac{124}{153} \\
  101 & \frac{4}{7} \\
  110 & \frac{5}{7} \\
  111 & \frac{331}{525}
\end{array}
\right.
\qquad
\left\{
\begin{array}{cc}
  0000 & \frac{2888}{3675} \\
  0001 & \frac{5}{9} \\
  0010 & \frac{2}{3} \\
  0011 & \frac{47}{78} \\
  0100 & \frac{59}{78} \\
  0101 & \frac{5}{9} \\
  0110 & \frac{2}{3} \\
  0111 & \frac{1141}{1845} \\
  1000 & \frac{1444}{1845} \\
  1001 & \frac{5}{9} \\
  1010 & \frac{2}{3} \\
  1011 & \frac{47}{78} \\
  1100 & \frac{59}{78} \\
  1101 & \frac{5}{9} \\
  1110 & \frac{2}{3} \\
  1111 & \frac{6841}{11025}
\end{array}
\right.
\end{equation*}
Each block corresponds to a choice of $n$, while each row within a block corresponds to a different function $f$. The function is represented by a bit string that indicates its values on inputs of Hamming weight $0,\dotsc,\floor{n/2}$.\footnote{This determines the remaining values of $f$, including the values on inputs of Hamming weight larger than $\floor{n/2}$, thanks to symmetry and equivariance of $f$.}
Equivalently, the bit string representing $f$ is the \emph{truth table} $\symf(0)\symf(1)\dots\symf(\floor{n/2})$ of the function $\symf(\abs{x}) := f(x)$.
For example, $\MAJ_7$ corresponds to the sequence $0000$ while $\PAR_7$ corresponds to $0101$.

Careful inspection of these values reveals a certain pattern.
For example, for $n=7$ all functions whose truth table ends with $10$ have fidelity $2/3$, while those that end with $01$ have fidelity $5/9$.
Similarly, suffixes $100$ and $011$ lead to fidelities $59/78$ and $47/78$, respectively.
More generally, the optimal fidelity $F_f$ of $f$ depends only on the final digit $\symf(\floor{n/2})$ and the size of the gap around $n/2$ in the truth table of $\symf$.
The truth table of $\symf$ ends either with a string of $k$ zeroes or $k$ ones:
\begin{align}
  *\dots*1\underbrace{0\dots0}_k &&
  *\dots*0\underbrace{1\dots1}_k.
\end{align}
This number $k$, together with the last digit $\symf(\floor{n/2})$, seems to completely determine the optimal fidelity $F_f$.\footnote{This is an empirical observation. We leave it for future work to establish this formally and to determine the exact dependence on $k$.}
This observation is useful for displaying the optimal fidelities in a more compact way (see \cref{fig:Tree}).
Moreover, it reminds of the behavior of the regular quantum query complexity $\mathsf{Q}(f)$ of symmetric Boolean functions established in \cite[Theorem~3.3]{Beals}.
It is an interesting open problem to determine whether there is a closer relationship between the optimal fidelity $F_f$ and the quantum query complexity $\mathsf{Q}(f)$ of all equivariant Boolean functions $f$.

\begin{figure}
\centering


\begin{tikzpicture}[thick,
    pt/.style = {circle, draw = black, fill = black, inner sep = 0.8pt}
  ]
  \def\H{1}
  \def\t{0.6}
  \node[pt] (O) at (0,0) {};
  \draw (O)  -- node[above]{$0$} ++(-2.0,-\H) node[pt] (L1) {};
  \draw (O)  -- node[above]{$1$} ++( 2.0,-\H) node[pt] (R1) {};
  \draw (L1) -- node[above]{$0$} ++(-1.2,-\H) node[pt] (L2) {};
  \draw (L1) -- node[above]{$1$} ++( 1.2,-\H) node[pt] {} +(0,-\t) node {$\displaystyle\frac{2}{3}$};
  \draw (R1) -- node[above]{$0$} ++(-1.2,-\H) node[pt] {} +(0,-\t) node {$\displaystyle\frac{5}{9}$};
  \draw (R1) -- node[above]{$1$} ++( 1.2,-\H) node[pt] (R2) {};
  \draw (L2) -- node[above]{$0$} ++(-0.9,-\H) node[pt] (L3) {};
  \draw (L2) -- node[above]{$1$} ++( 0.9,-\H) node[pt] {} +(0,-\t) node {$\displaystyle\frac{59}{78}$};
  \draw (R2) -- node[above]{$0$} ++(-0.9,-\H) node[pt] {} +(0,-\t) node {$\displaystyle\frac{47}{78}$};
  \draw (R2) -- node[above]{$1$} ++( 0.9,-\H) node[pt] (R3) {};
  \draw (L3) -- node[left ]{$0$} ++(-0.6,-\H) node[pt] {} +(0,-\t) node {$\displaystyle\frac{2888}{3675}$};
  \draw (L3) -- node[right]{$1$} ++( 0.6,-\H) node[pt] {} +(0,-\t) node {$\displaystyle\frac{1444}{1845}$};
  \draw (R3) -- node[left ]{$0$} ++(-0.6,-\H) node[pt] {} +(0,-\t) node {$\displaystyle\frac{1141}{1845}$};
  \draw (R3) -- node[right]{$1$} ++( 0.6,-\H) node[pt] {} +(0,-\t) node {$\displaystyle\frac{6841}{11025}$};
\end{tikzpicture}
\caption{\label{fig:Tree}Optimal fidelities of all symmetric and equivariant $7$-bit Boolean functions, arranged in a tree. To determine $F_f$ for a given function $f$, walk down the tree while reading the truth table of $\symf$ backwards. For example, if the truth table ends with $\dots10$ then the fidelity is $2/3$, while the suffix $\dots011$ leads to fidelity $47/78$.}
\end{figure}
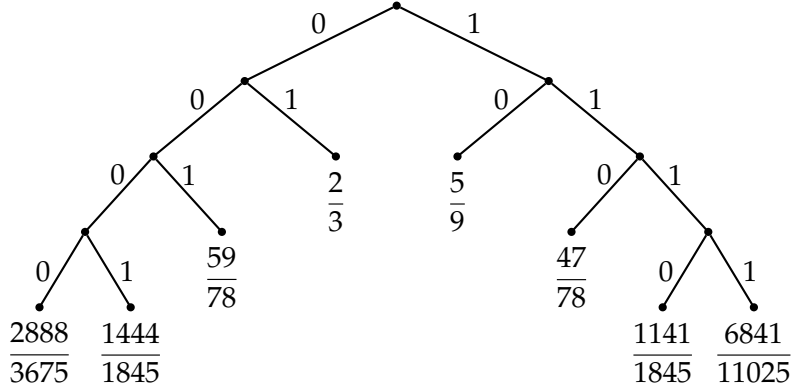

\subsubsection{Optimal fidelities for majority}\label{sec:NumericMajority}

We have also calculated fidelities for $\MAJ_n$ for larger values of $n = 1,3,5,\dotsc$:
\begin{equation*}
  1,\frac{8}{9},\frac{62}{75},\frac{2888}{3675},\frac{15014}{19845},\frac{117548}{160083},\frac{13848922}{19324305},\frac{5816048}{8281845},\frac{183562382}{265939245},\frac{30465827276}{44801898141},\frac{6378478534}{9503432939},\dotsc 
\end{equation*}
which, in decimal, are
\begin{equation*}
  1.,0.888889,0.826667,0.785850,0.756563,0.734294,0.716658,0.702265,0.690242,0.680012,\dotsc 
\end{equation*}
These and larger values can also be obtained by computing $g(\frac{n+1}{2})$ where $g$ is given by the following recursive formula (see proof in \cref{apx:Monotonicity}):

\begin{restatable}{lemma}{recursion}\label{lem:RecursionG}
Let $n \geq 1$ be odd.
The optimal fidelity for the $n$-qubit quantum majority vote is
$F_{\MAJ}(n) = g(\frac{n+1}{2})$
where
\begin{align*}
  g(1) &:= 1, \quad
  g(2)  := 8/9, \\
  g(m) &:= \frac{2m}{(2m-1)^2 (2m+1)}
  \sof[\Big]{ \of[\big]{2m(4m-7)+5} g(m-1) - 4(m-1)(m-2) g(m-2) + 1 }. \nonumber
\end{align*}
\end{restatable}

\subsubsection{Optimal fidelities for parity}

Finally, we have also computed fidelities for the parity function $\PAR_n$, $n=1,3,5,\dotsc$:
\begin{equation*}
  1,\frac{3}{5},\frac{5}{7},\frac{5}{9},\frac{7}{11},\frac{7}{13},\frac{3}{5},\frac{9}{17},\frac{11}{19},\frac{11}{21},\frac{13}{23},\frac{13}{25},\frac{5}{9},\frac{15}{29},\frac{17}{31},\frac{17}{33},\frac{19}{35},\frac{19}{37},\frac{7}{13},\frac{21}{41},\dotsc
\end{equation*}
Based on these numerical results, we conjecture that
\begin{equation}
  F_{\PAR_n} = \frac{2 \lceil \frac{n+1}{4} \rceil + 1}{n+2}
  = 1/2 + O(1/n).
\end{equation}

\subsection{Asymptotic fidelity of quantum majority}\label{sec:Majority}

Armed with our general result, \cref{thm:main}, we are able to calculate the optimal fidelity that can be achieved for the majority function, hence proving \cref{thm:majority}.
Recall that for odd $n$, the majority function evaluates to 0 on all $x$ such that $|x| \leq n/2$, and to 1 on all $x$ such that $|x| > n/2$. By \cref{thm:main}, the maximal fidelity that can be achieved is
\begin{equation}
  \min_{0 \leq \ham \leq \frac{n-1}{2}}
  \sum_{k=0}^h	p_k(\ham)
  \of*{
    t_k \frac{n-\ham-k}{n-2k}
  + (1-t_k) \frac{\ham-k+1}{n-2k+2}
  }.
  \label{eq:overallmax}
\end{equation}
Observe that as $n-\ham-k \geq \ham-k+1$ for all $\ham \leq (n-1)/2$, the term weighted by $t_k$ in \cref{eq:overallmax} is larger than the term weighted by $1-t_k$, for all $k$. So we can assume that $t_k = 1$ for all $k$.
Then the expression for the optimal worst-case fidelity is
$F_{\MAJ}(n) := \min_{0 \leq \ham \leq \frac{n-1}{2}} F_n(\ham)$
where
\begin{equation}
  F_n(\ham) := \sum_{k=0}^{\ham}
  \frac{\binom{n}{k} - \binom{n}{k-1}}{\binom{n}{\ham}} \cdot
  \frac{n-\ham-k}{n-2k}.
  \label{eq:Fnx}
\end{equation}

While this will not be directly used in the proofs below, we show in \cref{lem:Monotonicity} of \cref{apx:Monotonicity} that $F_n(\ham)$ is strictly decreasing in $\ham \in \set{0,\dotsc,\frac{n-1}{2}}$.
This implies that the minimum in the definition of $F_{\MAJ}(n)$ is achieved at $\ham = \frac{n-1}{2}$ and proves the intuitive claim that, as the number of minority states increases, it becomes harder to output the correct majority state.

\majority*

\begin{proof}
For simplicity, let us write $F(\ham)$ instead of $F_n(\ham)$.
We now rewrite \cref{eq:Fnx} as
\begin{equation}
  F(\ham)
  = \sum_{k=0}^{\ham}
    \frac{\binom{n}{k} - \binom{n}{k-1}}{\binom{n}{\ham}}
    \of*{1 - \frac{\ham-k}{n-2k}}
  = 1 - \sum_{k=0}^{\ham}
    \frac{\binom{n}{k} - \binom{n}{k-1}}{\binom{n}{\ham}} \cdot
    \frac{\ham-k}{n-2k}.
  \label{eq:fexp}
\end{equation}
We will consider two regimes. In the first, $\ham$ is arbitrary and we seek a lower bound on $F(\ham)$. Here it will be helpful to parametrize $\ham = (n-1)/2 - \delta$, giving
\begin{equation}
  F(\ham)
  = 1 - \sum_{k=0}^{\ham}
    \frac{\binom{n}{k} - \binom{n}{k-1}}{\binom{n}{\ham}} \cdot
    \frac{1}{2} \of*{1 - \frac{2\delta+1}{n-2k}}
  = \frac{1}{2} + \frac{1}{2} \sum_{k=0}^{\ham}
    \frac{\binom{n}{k} - \binom{n}{k-1}}{\binom{n}{\ham}} \cdot
    \frac{2\delta+1}{n-2k}.
\end{equation}
Since $\binom{n}{k} - \binom{n}{k-1} = \binom{n}{k} \of[\big]{1 - \frac{k}{n-k+1}}$,
\begin{equation}
  F(\ham)
  = \frac{1}{2} + \frac{2\delta + 1}{2}
    \sum_{k=0}^{\ham} \frac{\binom{n}{k}}{\binom{n}{\ham}}
    \frac{n-2k+1}{(n-k+1)(n-2k)}.
\end{equation}
Since the last fraction is between $1/n$ and $4/n$ for all $k$ between 0 and $\ham \leq (n-1)/2$,
\begin{equation}
  \frac{1}{2} + \frac{2\delta + 1}{2n\binom{n}{\ham}} \sum_{k=0}^{\ham} \binom{n}{k} \leq F(\ham) \leq \frac{1}{2} + \frac{2(2\delta + 1)}{n\binom{n}{\ham}} \sum_{k=0}^{\ham} \binom{n}{k}.
  \label{eq:fbound}
\end{equation}

\textbf{Fidelity with no promise on $|x|$.} To complete the bound in this case, we consider $\delta \leq \sqrt{n}$ and $\delta > \sqrt{n}$ separately. In the former case, as $\delta = o(n^{2/3})$ we have~\cite{spencer14}
\begin{equation}
  \binom{n}{\ham} = \binom{n}{n/2} e^{-(n-2\ham)^2/n} (1+o(1)) = \binom{n}{n/2} e^{-(2\delta + 1)^2/n} (1+o(1)) \leq C \frac{2^n}{\sqrt{n}},
\end{equation}
while $\sum_{k=0}^{\ham} \binom{n}{k} \geq C' 2^n$ (for example, via the normal approximation to the binomial distribution), for some universal constants $C$, $C'$. So
\begin{equation}
  F(\ham) \geq \frac{1}{2} + \frac{\sqrt{n}}{2Cn 2^n} C' 2^n = \frac{1}{2} + \frac{C''}{\sqrt{n}}
\end{equation}
for some new universal constant $C''$. If $\delta > \sqrt{n}$, as the sum in \cref{eq:fbound} is clearly larger than $\binom{n}{\ham}$, $F(\ham) > 1/2 + (2\sqrt{n} + 1)/(2n) > 1/2 + 1/\sqrt{n}$.

To see that this is close to tight, consider the case $\ham = (n-1)/2$, for which an upper bound $F(\ham) \leq 1/2 + O(1/\sqrt{n})$ holds by \cref{eq:fbound}.

\textbf{Fidelity with a promise on $|x|$.}
If we assume that $\ham := |x| \leq n/3$, from \cref{eq:fexp} we have that in this regime
\begin{align}
  F(\ham) &
  \geq 1 - \frac{3}{2n} \sum_{k=0}^{\ham}  \frac{\binom{n}{k} - \binom{n}{k-1}}{\binom{n}{\ham}} (\ham-k) = 1 - \frac{3}{2n}\left(\ham - \frac{1}{\binom{n}{\ham}} \sum_{k=0}^{\ham}  \left(\binom{n}{k} - \binom{n}{k-1}\right)k \right) \\
  &= 1 - \frac{3}{2n}\left(\ham - \frac{1}{\binom{n}{\ham}} \sum_{k=0}^{\ham} k \binom{n}{k} + \frac{1}{\binom{n}{\ham}} \sum_{k=0}^{\ham} k \binom{n}{k-1} \right) \\
  &= 1 - \frac{3}{2n}\left(\ham - \ham - \frac{1}{\binom{n}{\ham}} \sum_{k=0}^{\ham-1} k \binom{n}{k} + \frac{1}{\binom{n}{\ham}} \sum_{k=0}^{\ham-1} (k+1) \binom{n}{k} \right) \\
  &= 1 - \frac{3}{2n \binom{n}{\ham}} \sum_{k=0}^{\ham-1} \binom{n}{k}.
\end{align}
As $\binom{n}{k} \geq 2 \binom{n}{k-1}$ for $k \leq n/3$ (via $\binom{n}{k} / \binom{n}{k-1} = (n-k+1)/k$), we have $\sum_{k=0}^{\ham-1} \binom{n}{k} \leq 2\binom{n}{\ham}$ and hence
\begin{equation}
  F(\ham) \geq 1 - \frac{3}{2n}.
\end{equation}
Running the same argument through for an upper bound, and modifying appropriately to start with
\begin{equation}
  F(\ham) \leq 1 - \frac{1}{n} \sum_{k=0}^{\ham}  \frac{\binom{n}{k} - \binom{n}{k-1}}{\binom{n}{\ham}} (\ham-k),
\end{equation}
we obtain
\begin{equation}
  F(\ham) \leq 1 - \frac{1}{n \binom{n}{\ham}} \sum_{k=0}^{\ham-1} \binom{n}{k} \leq 1 - \frac{\binom{n}{\ham-1}}{n \binom{n}{\ham}} = 1 - \frac{\ham}{n(n-\ham+1)}.
\end{equation}
Fixing $\ham = \lfloor n/3 \rfloor$, we get
\begin{equation}
  F(\ham) \leq 1 - \frac{\lfloor n/3 \rfloor}{n(n-\lfloor n/3 \rfloor+1)} \leq 1 - \frac{n/3 - 1}{n(2n/3 + 2)}  = 1 - \Theta(1/n),
\end{equation}
as desired.
\end{proof}

\subsection{Unitary-equivariant Boolean functions}\label{sec:U-equivariant functions}

Can a Boolean function have a continuous symmetry?
And not just any continuous symmetry, but a unitary one?
At first glance these questions may seem absurd since Boolean functions are inherently discrete, however they are implicit in the premise of our work.
Indeed, the main question posed and answered in this work can be rephrased as follows:
to what extent can a Boolean function be extended to a quantum channel with continuous symmetries?
In this section we briefly elaborate on this perspective, leaving a more in-depth exploration for future work.

\Cref{thm:template} shows that any symmetric and equivariant Boolean function can be computed with optimal fidelity by the template algorithm $\mc{A}_\vt$, for some choice of the interpolation parameters $\vt$ that can be found using the linear program stated in \cref{thm:main}.
While no function (except for the trivial $1$-argument identity function $\ID$) can be computed with fidelity one, we can relax the requirement that the linear map $\Phi: \L(\C^{2^n}) \to \L(\C^2)$ for computing it is a quantum channel.
In particular, we can drop the requirement that $\Phi$ is completely positive and ask whether now perfect fidelity can be achieved.
This is equivalent to relaxing the requirement that the Choi matrix $J(\Phi)$ is positive semidefinite.
In the context of \cref{thm:main}, this corresponds to dropping the constraints $0 \leq t_k \leq 1$ in the linear program \eqref{eq:LP}.
To achieve perfect fidelity, we also replace the ``$\geq c$'' constraint by ``$= 1$'', which reduces the problem to a linear system of equations.
Since this system involves a lower-triangular matrix, it has a unique solution that can be found by eliminating one variable at a time.
Based on numerical investigations, we believe this solution to be
\begin{equation}
  t_k :=
    \frac{n-2k}{2(n-2k+1)}
  - \frac{k(n-2k)}{2(n-2k+1)^2} (-1)^{\symf(k-1)}
  + \frac{(n-k+1)(n-2k+2)}{2(n-2k+1)^2} (-1)^{\symf(k)},
\end{equation}
for all $k \in \set{0,\dotsc,\floor{n/2}}$.
We leave it for future work to verify this.

This suggests that for any symmetric and equivariant Boolean function $f: \set{0,1}^n \to \set{0,1}$ there is a unique linear map $\Phi_f: \L(\C^{2^n}) \to \L(\C^2)$ that is permutation-invariant and unitary-equivariant, and agrees with $f$ on all inputs.
This map $\Phi_f$ can be thought of as a \emph{unitary-equivariant extension} of $f$ or, alternatively, as a ``\emph{unitary-equivariant Boolean function}'' in its own right, \ie, a Boolean function with continuous unitary symmetries.

\begin{definition}\label{def:Phif}
Let $f: \set{0,1}^n \to \set{0,1}$ be a symmetric and equivariant Boolean function.
Its \emph{unitary-equivariant extension} is the linear map $\Phi_f: \L(\C^{2^n}) \to \L(\C^2)$ such that
\begin{align}
  \Phi_f(U\xp{n} \proj{x} U\ctxp{n}) &= U \proj{f(x)} U\ct, &
  \forall x &\in \set{0,1}^n, \quad
  \forall U \in \U{2}, \label{eq:ideal map} \\
  \Phi_f(P(\pi) \rho P(\pi)\ct) &= \Phi(\rho), &
  \forall \rho &\in \L(\C^{2^n}),
  \quad \forall \pi \in \S_n.
\end{align}
\end{definition}

It is not obvious at all that such map $\Phi_f$ exists and is unique.
In fact, it is not even obvious that a map satisfying \cref{eq:ideal map} is linear!
To illustrate that this definition is not vacuous, we provide a simple procedure for computing the Choi matrix of the map $\Phi_f$.

First, let $\rho(\vec{r})$ denote the state in \cref{eq:rho xyz}.
For any $s \in \set{0,1}^n$, let
\begin{equation}
  \rho(s,\vec{r}) := \bigotimes_{i=1}^n \rho \of[\big]{(-1)^{s_i} \vec{r}}
\end{equation}
denote the $n$-qubit state corresponding to the input string $s$ in an arbitrary basis.
Then the unitary-equivariant extension of $f$ has Choi matrix $J \in \L(\C^{2^{n+1}})$ such that
\begin{align}
  \Tr_2 \sof[\big]{J \cdot (\id_2 \x \rho(s,\vec{r})\tp)}
  &= \rho\of[\big]{(-1)^{f(s)}\vec{r}}, &
  \forall s &\in \set{0,1}^n, \quad
  \forall \vec{r} \in \Sphere^2, \label{eq:ideal map on all r} \\
  \of[\big]{\id_2 \x P(\pi)} J
  \of[\big]{\id_2 \x P(\pi)}\ct &= J, &
  \forall \pi &\in \S_n, \label{eq:perm-invariant J}
\end{align}
where $\Sphere^2 := \set{(x,y,z) \in \R^3 : x^2 + y^2 + z^2 = 1}$ denotes the Bloch sphere.
Reducing \cref{eq:ideal map on all r} modulo the polynomial $x^2 + y^2 + z^2 - 1$
and then comparing the coefficients at all monomials $x^i y^j z^k$ produces a system of linear equations in the matrix entries of $J$.
Combining this with the linear system \eqref{eq:perm-invariant J} produces a system that uniquely determines the Choi matrix $J$.\footnote{This is an empirical observation.}

Note that for the sake of efficiency in \cref{eq:ideal map on all r} it is enough to consider only one $s \in \set{0,1}^n$ for each Hamming weight $|s| \in \set{0,\dotsc,\floor{n/2}}$.
Similarly, it is enough to impose \cref{eq:perm-invariant J} only for transpositions $\pi = (i,i+1)$ with $i \in \set{1,\dotsc,n-1}$ since they generate $\S_n$.

We provide explicit Choi matrices for all symmetric and unitary-equivariant Boolean functions with $n = 1$ and $n = 3$ arguments in \cref{apx:Ideal}.

\section*{Acknowledgements}

We acknowledge support from the QuantERA ERA-NET Cofund in Quantum Technologies implemented within the European Union's Horizon 2020 Programme (QuantAlgo project), EPSRC grants EP/R043957/1 and EP/T001062/1, and EPSRC Early Career Fellowship EP/L021005/1. This project has received funding from the European Research Council (ERC) under the European Union's Horizon 2020 research and innovation programme (grant agreement No.\ 817581). No new data were created during this study.
MO was supported by an NWO Vidi grant (Project No. VI.Vidi.192.109).


\bibliographystyle{alphaurl}
\bibliography{References}

\appendix

\section{Graphical notation}\label{apx:Graphical}

In this appendix, we introduce graphical notation and prove some basic identities that will be useful in \cref{apx:Nice}. For a much more thorough introduction to tensor networks and graphical notation, see \cite{Biamonte}.

\subsection{Vectorization}

Throughout this section, let us fix some dimension $d \geq 1$.

\begin{definition}[Vectorization]\label{def:vec}
The \emph{vectorization} of a matrix
\begin{equation}
  A = \sum_{\substack{i\in[\dout]\\j\in[\din]}} A_{ij} \ketbra{i}{j}
  \in \L(\C^\din,\C^\dout)
\end{equation}
is defined\footnote{We define $\bra{A}$ as $\ket{A}\tp$ instead of $\ket{A}\ct$ since in graphical notation transposition is more natural than conjugate transposition. We use this convention only in the context of matrix vectorization, not for general ``ket'' vectors.} as follows:
\begin{align}
  \ket{A}
   &:= \sum_{\substack{i\in[\dout]\\j\in[\din]}} A_{ij} \ket{i} \ket{j}, &
  \bra{A}
    := \ket{A}\tp
    &= \sum_{\substack{i\in[\dout]\\j\in[\din]}} A_{ij} \bra{i} \bra{j}.
  \label{eq:vectorization}
\end{align}
\end{definition}

\begin{prop}\label{prop:mini snake}
For any $A,B \in \L(\C^d)$,
\begin{enumerate}
  \item $\braket{B}{A} = \Tr \of[\big]{A B\tp}$, \label{i:Tr}
  \item $\SWAP \ket{A} = \ket{A\tp}$, \label{i:SWAP}
  \item $(B \x \id_d) \ket{A} = \ket{BA}$, \label{i:B}
  \item $\of[\big]{\id_d \x \bra{B}} \of[\big]{\ket{A} \x \id_d} = AB$, \label{i:snake}
\end{enumerate}
where $\SWAP$ denotes the two-qudit swap operation:
$\SWAP \ket{i} \ket{j} = \ket{j} \ket{i}$,
for all $i,j \in [d]$.
\end{prop}

\begin{proof}
The first three identities are quite straightforward:
\begin{align}
  \braket{B}{A}
 &= \sum_{i,j \in [d]} B_{ij} A_{ij}
  = \sum_{i,j \in [d]} A_{ij} (B\tp)_{ji}
  = \Tr \of[\big]{A B\tp}, \\
  \SWAP \ket{A}
 &= \sum_{i,j \in [d]} A_{ij} \SWAP \of[\big]{\ket{i} \ket{j}}
  = \sum_{i,j \in [d]} A_{ij} \ket{j} \ket{i}
  = \ket{A\tp}, \\
  \of[\big]{B \x \id_d} \ket{A}
 &= \sum_{i,j \in [d]} A_{ij}
    \of*{\sum_{k,l \in [d]} B_{kl} \ket{k} \braket{l}{i}}
    \ket{j}
  = \sum_{k,j \in [d]}
    \of*{\sum_{i \in [d]} B_{ki} A_{ij}}
    \ket{k} \ket{j}
  = \ket{BA}.
\end{align}
The third identity can be obtained as follows:
\begin{align}
  \of[\big]{\id_d \x \bra{B}}
  \of[\big]{\ket{A} \x \id_d}
 &= \sum_{i,j \in [d]} B_{ij}
    \of[\big]{\id_d \x \bra{i} \x \bra{j}}
    \sum_{k,l \in [d]} A_{kl}
    \of[\big]{\ket{k} \x \ket{l} \x \id_d} \\
 &= \sum_{i,j,k,l \in [d]}
    A_{kl} B_{ij}
    \ket{k} \braket{i}{l} \bra{j} \\
 &= \of[\Bigg]{
      \sum_{k,l \in [d]}
      A_{kl} \ketbra{k}{l}
    }
    \of[\Bigg]{
      \sum_{i,j \in [d]}
      B_{ij} \ketbra{i}{j}
    } \\
 &= AB,
\end{align}
which completes the proof.
\end{proof}



\newcommand{\snakefig}[1]{%
\begin{tikzpicture}[thick, 
    l/.style = {anchor = east},
    b/.style = {fill = white},
    round/.style = {rounded corners = 9pt}
  ]
  \def\H{0.8cm}
  \def\W{2.4cm}
  \def\h{6pt}
  \def\w{10pt}

#1 

\end{tikzpicture}%
}



\newcommand{\gate}[2]{
  \draw[b] (#1)+(-\w,-\w) rectangle +(\w,\w);
  \node[text height = 1.5ex] at (#1) {$#2$};
}


\newcommand{\SWAPgate}[1]{
  \draw (#1)++(-2.5*\w, 1.5*\h) -- ++(\w,0) to [out = 0, in = 180]
            ++( 3.0*\w,-3.0*\h) -- ++(\w,0);
  \draw (#1)++(-2.5*\w,-1.5*\h) -- ++(\w,0) to [out = 0, in = 180]
            ++( 3.0*\w, 3.0*\h) -- ++(\w,0);
}


\newcommand{\EPR}[2]{
  \draw[round, xscale = #2]
   (#1)++(-2.5*\w,-1.5*\h)
    -- ++( 5.0*\w,0)
    -- ++( 0,3*\h)
    -- ++(-5.0*\w,0);
}


\newcommand{\Notation}{\snakefig{

  \foreach \s/\r/\a in
     {1/5/A1,1/3/B1,1/1/C1,
      3/5/A2,3/3/B2,3/1/C2,
      5/5/A3,5/3/B3,5/1/C3} {
    \path (\s*\W,\r*\H) coordinate (L\a);
    \path (L\a)+(2.8*\w,0) coordinate (\a);
  }

  \node[l] at (LA1) {$A \;= $};
  \draw (A1)+(-2.5*\w,0) -- +(2.5*\w,0);
  \gate{A1}{A}

  \node[l] at (LB1) {$A\tp \;= $};
  \draw[round]
  (B1) ++( 2.5*\w,-3*\h)
    -- ++(-5.0*\w,0)
    -- ++( 0,3*\h)
    -- ++( 5.0*\w,0)
    -- ++( 0,3*\h)
    -- ++(-5*\w,0);
  \gate{$(B1)+(0,0)$}{A}

  \node[l] at (LA2) {$\ket{A} \;= $};
  \node[l] at (LB2) {$\bra{B} \;= $};
  \EPR{A2}{ 1}
  \EPR{B2}{-1}
  \gate{$(A2)+(0, 1.5*\h)$}{A}
  \gate{$(B2)+(0,-1.5*\h)$}{B}

  \node[l] at (LA3) {$\Tr(A) \;= $};
  \draw[round]
   (A3)++(-2.5*\w,-1.5*\h)
    -- ++( 5.0*\w,0)
    -- ++( 0,3*\h)
    -- ++(-5.0*\w,0) -- cycle;
  \gate{$(A3)+(0,1.5*\h)$}{A}

  \node[l] at (LB3) {$A \x B \;= $};
  \draw (B3)++(-2.5*\w, 2*\h) -- +(5.0*\w,0);
  \draw (B3)++(-2.5*\w,-2*\h) -- +(5.0*\w,0);
  \gate{$(B3)+(0, 2*\h)$}{A}
  \gate{$(B3)+(0,-2*\h)$}{B}

  \node[l] at (LC1) {$A \cdot B \;= $};
  \path (C1)+(\w,0) coordinate (C1);
  \draw (C1)+(-3.5*\w,0) -- +(3.5*\w,0);
  \gate{$(C1)+(-1.5*\w,0)$}{A}
  \gate{$(C1)+( 1.5*\w,0)$}{B}

  \node[l] at (LC2) {$I \;= $};
  \draw (C2)+(-2.5*\w,0) -- +(2.5*\w,0);

  \node[l] at (LC3) {$\SWAP \;= $};
  \SWAPgate{C3}

}}


\newcommand{\Snake}{\snakefig{

  \path (-1.65*\W, 1.3*\H) coordinate (I1);
  \path ( 1.65*\W, 1.3*\H) coordinate (I2);
  \path (-1.65*\W,-1.3*\H) coordinate (I3);
  \path ( 1.65*\W,-1.3*\H) coordinate (I4);

  \foreach \i in {1,...,4} {
    \path (I\i)+(-11*\w,0) node {\textit{\i.}};
  }


  \node at (I1) {$=$};
  \path (I1)+(-5.5*\w,0) coordinate (L1);
  \path (I1)+( 5.5*\w,0) coordinate (R1);

  \draw[round]
   (L1)++(-4.0*\w,-1.5*\h)
    -- ++( 8.0*\w,0)
    -- ++( 0,3*\h)
    -- ++(-8.0*\w,0) -- cycle;
  \gate{$(L1)+( 1.5*\w, 1.5*\h)$}{A}
  \gate{$(L1)+(-1.5*\w,-1.5*\h)$}{B}

  \draw[round]
   (R1)++(-4.0*\w,-1.5*\h)
    -- ++( 8.0*\w,0)
    -- ++( 0,3*\h)
    -- ++(-8.0*\w,0) -- cycle;
  \gate{$(R1)+(-1.5*\w,1.5*\h)$}{A}
  \gate{$(R1)+( 1.5*\w,1.5*\h)$}{B\tp}


  \node at (I2) {$=$};
  \path (I2)+(-6.0*\w,0) coordinate (L2);
  \path (I2)+( 4.0*\w,0) coordinate (R2);

  \SWAPgate{$(L2)+(-\w,0)$}
  \draw[round] (L2)++(1.5*\w,1.5*\h) -- ++(3*\w,0) -- ++(0,-3*\h) -- ++(-3*\w,0);
  \gate{$(L2)+(2*\w,1.5*\h)$}{A}

  \EPR{R2}{1}
  \gate{$(R2)+(0,1.5*\h)$}{A\tp}


  \node at (I3) {$=$};
  \path (I3)+(-5.5*\w,0) coordinate (L3);
  \path (I3)+( 4.0*\w,0) coordinate (R3);

  \EPR{L3}{1.6}
  \gate{$(L3)+(-1.5*\w,1.5*\h)$}{B}
  \gate{$(L3)+( 1.5*\w,1.5*\h)$}{A}

  \EPR{R3}{1}
  \gate{$(R3)+(0,1.5*\h)$}{BA}


  \node at (I4) {$=$};
  \path (I4)+(-5.5*\w,0) coordinate (L4);
  \path (I4)+( 5.5*\w,0) coordinate (R4);

  \draw[round]
  (L4) ++( 4.0*\w,-3*\h)
    -- ++(-8.0*\w,0)
    -- ++( 0,3*\h)
    -- ++( 8.0*\w,0)
    -- ++( 0,3*\h)
    -- ++(-8*\w,0);
  \gate{$(L4)+( 1.5*\w, 3*\h)$}{A}
  \gate{$(L4)+(-1.5*\w,-3*\h)$}{B}

  \draw (R4)+(-4.0*\w,0) -- +(4.0*\w,0);
  \gate{$(R4)+(-1.5*\w,0)$,0}{A}
  \gate{$(R4)+( 1.5*\w,0)$,0}{B}


}}


\newcommand{\BigSnake}{\snakefig{

  \path (-6.0*\w,0) coordinate (L);
  \path (11.0*\w,0) coordinate (R);
  \path (L)+(0, 7.0*\h) coordinate (L1);
  \path (L)+(0,-7.0*\h) coordinate (L2);

  \draw[round]
  (L1) ++(-4*\w,3*\h)
    -- ++( 8*\w,0)
    -- ++( 0,-3*\h)
    -- ++(-8*\w,0)
    -- ++( 0,-3*\h)
    -- ++( 4.5*\w,0);
  \gate{$(L1)+( 1.5*\w, 3*\h)$}{A_1}
  \gate{$(L1)+(-1.5*\w,-3*\h)$}{B_1}

  \path (L1)++( 1.5*\w,-3*\h) node {$\ldots$};
  \path (L2)++(-1.5*\w, 3*\h) node {$\ldots$};
  \path (L)+(0,0.5ex) node {$\vdots$};

  \draw[round]
  (L2) ++(-0.5*\w,3*\h)
    -- ++( 4.5*\w,0)
    -- ++( 0,-3*\h)
    -- ++(-8*\w,0)
    -- ++( 0,-3*\h)
    -- ++( 8*\w,0);
  \gate{$(L2)+( 1.5*\w, 3*\h)$}{A_k}
  \gate{$(L2)+(-1.5*\w,-3*\h)$}{B_k}

  \node at (0,0) {$=$};
  \node at (R) {$\ldots$};

  \draw (R)+(-9*\w,0) -- +(-\w,0);
  \draw (R)+( 9*\w,0) -- +( \w,0);
  \gate{$(R)+(-6.5*\w,0)$,0}{A_1}
  \gate{$(R)+(-3.5*\w,0)$,0}{B_1}
  \gate{$(R)+( 3.5*\w,0)$,0}{A_k}
  \gate{$(R)+( 6.5*\w,0)$,0}{B_k}

}}


The intuition behind these identities is best captured by \emph{graphical notation}.
Let us convert basic linear algebra operations involving matrices $A,B \in \L(\C^d)$ into the following diagrams:
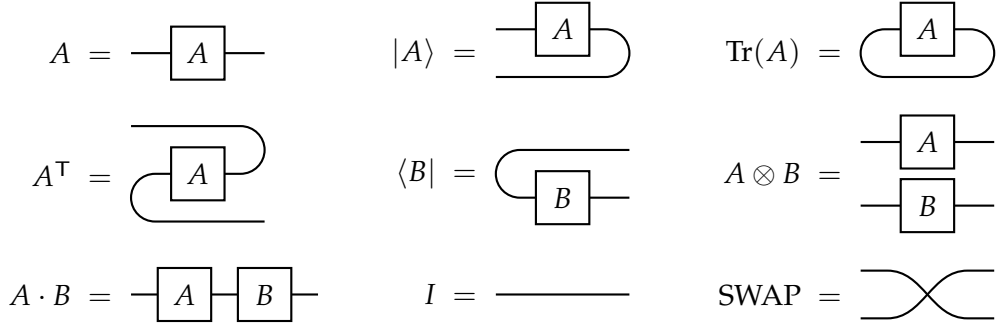
\begin{center}
  \Notation
\end{center}
Then the identities in \cref{prop:mini snake} have the following straightforward proofs:
\begin{center}
  \Snake
\end{center}

\Cref{i:snake} of \cref{prop:mini snake} can be generalized by chaining together any number of vectorized operators.
The resulting chain has one of three possible types, depending on where it starts and ends:
it may start on one side but end on the other (see the diagram for \cref{i:snake}),
it may start and end on the same side (see the diagram for $\ket{A}$),
or it may be a closed loop (see the diagram for $\Tr(A)$).
These chains describe different types of objects:
a matrix, a vector, and a scalar, respectively.
The following proposition expresses this object as the matrix product of the operators involved.

\begin{prop}\label{prop:snakes}
For any $k \geq 0$ and $A_1,\dotsc,A_{k+1},B_1,\dotsc,B_{k+1} \in \L(\C^d)$,
\begin{align}
  \of[\big]{\id_d \x \bra{B_1} \x \dotsb \x \bra{B_k}}
  \of[\big]{\ket{A_1} \x \dotsb \x \ket{A_k} \x \id_d}
  &= A_1 B_1 \dotsb A_k B_k, \\
  \of[\big]{\id_d \x \bra{B_1} \x \dotsb \x \bra{B_k} \x \id_d}
  \of[\big]{\ket{A_1} \x \dotsb \x \ket{A_{k+1}}}
  &= \ket{A_1 B_1 \dotsb A_k B_k A_{k+1}}, \\
  \bra{B_{k+1}} \SWAP
  \of[\big]{\id_d \x \bra{B_1} \x \dotsb \x \bra{B_k} \x \id_d}
  \of[\big]{\ket{A_1} \x \dotsb \x \ket{A_{k+1}}}
  &= \Tr \of[\big]{A_1 B_1 \dotsb A_{k+1} B_{k+1}}.
\end{align}
\end{prop}

\begin{proof}
The $k = 0$ case of the first identity is trivial since $\id_d \id_d = \id_d$. Assuming the identity holds for $k-1$, we get
\begin{align}
  & \of[\big]{\id_d \x \bra{B_1} \x \dotsb \x \bra{B_{k-1}} \x \bra{B_k}}
    \of[\big]{\ket{A_1} \x \dotsb \x \ket{A_{k-1}} \x \ket{A_k} \x \id_d} \\
 &= \of[\big]{\id_d \x \bra{B_1} \x \dotsb \x \bra{B_{k-1}}}
    \of[\big]{\ket{A_1} \x \dotsb \x \ket{A_{k-1}} \x \id_d}
    \of[\big]{\id_d \x \bra{B_k}}
    \of[\big]{\ket{A_k} \x \id_d} \\
 &= A_1 B_1 \dotsb A_{k-1} B_{k-1} \cdot A_k B_k,
\end{align}
where we used the inductive assumption and \cref{i:snake} from \cref{prop:mini snake}.
To prove the second identity, we combine the previous identity with \cref{i:B} from \cref{prop:mini snake}:
\begin{align}
  & \of[\big]{\id_d \x \bra{B_1} \x \dotsb \x \bra{B_k} \x \id_d}
    \of[\big]{\ket{A_1} \x \dotsb \x \ket{A_{k+1}}} \\
 &= \of[\big]{\id_d \x \bra{B_1} \x \dotsb \x \bra{B_k} \x \id_d}
    \of[\big]{\ket{A_1} \x \dotsb \x \ket{A_k} \x \id_d \x \id_d}
    \ket{A_{k+1}} \\
 &= \of[\big]{A_1 B_1 \dotsb A_k B_k \x \id_d} \ket{A_{k+1}} \\
 &= \ket{A_1 B_1 \dotsb A_k B_k A_{k+1}}.
\end{align}
For the third identity, note from \cref{eq:vectorization} and \cref{i:SWAP} of \cref{prop:mini snake} that
\begin{equation}
  \bra{B_{k+1}} \SWAP
  = \of[\big]{\SWAP\tp \bra{B_{k+1}}\tp}\tp
  = \of[\big]{\SWAP \ket{B_{k+1}}}\tp
  = \of[\big]{\ket{B_{k+1}\tp}}\tp
  = \bra{B_{k+1}\tp}.
\end{equation}
Combining this with the previous identity,
\begin{align}
  & \bra{B_{k+1}} \SWAP
    \of[\big]{\id_d \x \bra{B_1} \x \dotsb \x \bra{B_k} \x \id_d}
    \of[\big]{\ket{A_1} \x \dotsb \x \ket{A_{k+1}}} \\
 &= \braket{B_{k+1}\tp}{A_1 B_1 \dotsb A_k B_k A_{k+1}} \\
 &= \Tr \of[\big]{A_1 B_1 \dotsb A_{k+1} B_{k+1}},
\end{align}
where the last equality follows from \cref{i:Tr} of \cref{prop:mini snake}.
\end{proof}

The identities in \cref{prop:snakes} also have straightforward graphical proofs.
For example, the first identity amounts to
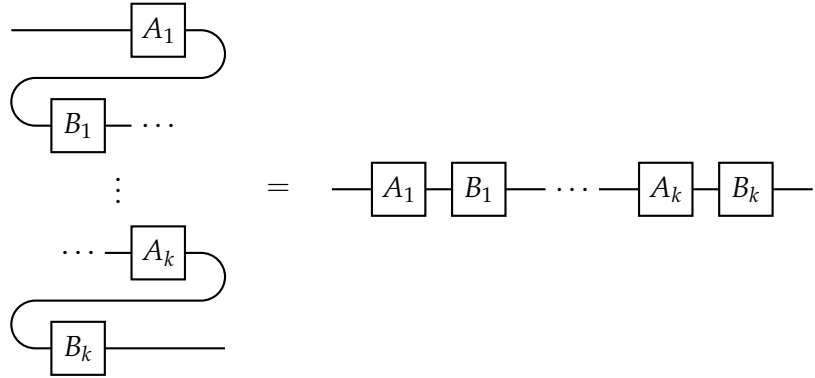
\begin{center}
  \BigSnake
\end{center}

\subsection{Singlet chains}

In \cref{apx:Nice}, we will be particularly interested in the special case when all operators involved describe singlets. For this it is useful to introduce the matrix
\begin{equation}
  S := \ketbra{0}{1}
     - \ketbra{1}{0}
     = \mx{0&1\\-1&0}.
  \label{eq:S}
\end{equation}
The \emph{singlet state} $\ket{\Psi^-}$ is proportional to the vectorization of $S$:
\begin{equation}
  \ket{\Psi^-}
  = \frac{1}{\sqrt{2}} \of[\big]{\ket{0}\ket{1} - \ket{1}\ket{0}}
  = \frac{1}{\sqrt{2}} \ket{S}
  = \frac{1}{\sqrt{2}} \;
    \snakefig{
      \tikzset{baseline = -2pt}
      \EPR{0,0}{1}
      \gate{0,1.5*\h}{S}
    }
\end{equation}
Note that $S$ and $\ket{\Psi^-}$ are anti-symmetric: $S\tp = -S$ and $\SWAP \ket{\Psi^-} = - \ket{\Psi^-}$ (the two identities are related via \cref{i:SWAP} of \cref{prop:mini snake}). Since $\ket{cA} = c \ket{A}$ for any $c \in \C$ and $A \in \L(\C^d)$,
\begin{equation}
  \ket{S\tp} = - \ket{S}.
\end{equation}
The matrix $S$ behaves like the complex number $i$ since
\begin{align}
  S^2  &= -\id, &
  S^3  &= -S, &
  S^4  &=  \id, &
  S\tp &= -S,
  \label{eq:small S powers}
\end{align}
where the transposition is analogous to complex conjugation. Consequently, for any integer $k$,
\begin{align}
  S^{2k  } &= (-1)^k \id, &
  S^{2k+1} &= (-1)^k S.
  \label{eq:S powers}
\end{align}

The following corollary is a direct consequence of \cref{prop:snakes,eq:S powers}.

\begin{cor}\label{cor:S chains}
For any $k \geq 0$,
\begin{align}
  \of[\big]{\id \x \bra{S}\xp{k}}
  \of[\big]{\ket{S}\xp{k} \x \id}
  &= S^{2k}
   = (-1)^k \id, \\
  \of[\big]{\id \x \bra{S}\xp{k} \x \id}
  \of[\big]{\ket{S}\xp{k+1}}
  &= \ket{S^{2k+1}}
   = (-1)^k \ket{S}, \\
  \bra{S} \SWAP
  \of[\big]{\id \x \bra{S}\xp{k} \x \id}
  \of[\big]{\ket{S}\xp{k+1}}
  &= \Tr \of[\big]{S^{2k+2}}
   = 2 (-1)^{k+1}.
\end{align}
Moreover, if $j$ instances of $\ket{S}$ are replaced by $\ket{S\tp} = - \ket{S}$, the final value of each expression picks up an extra factor of $(-1)^j$.
\end{cor}

\section{Properties of the vectors \texorpdfstring{$\ket{(\lambda,w,i)}$}{|(lambda,w,i)>}}\label{apx:Nice}

In this appendix, we prove several results concerning the vectors $\ket{(\lambda,w,i)}$ defined in \cref{sec:BasisQubits}. These vectors constitute a particularly nice choice of \emph{Schur basis} for qubits, and define the Schur transform $\Usch$ via \cref{eq:NiceSchur}, which is used in the generic pre-processing procedure described in \cref{sec:Preproc}.
This appendix is concerned with verifying that $\ket{(\lambda,w,i)}$ is indeed a Schur basis.
The main results proved here are as follows:
\begin{itemize}
  \item \cref{lem:P blocks}:
    the permutation operators $P(\pi)$ are block-diagonal in the Schur basis,
  \item \cref{lem:FullOrthonormal}:
    the vectors $\ket{(\lambda,w,i)}$ form an orthonormal basis,
  \item \cref{lem:QInvar}:
    the subspaces $\mc{Q}_{\lambda,i} := \spn \set[\big]{\ket{(\lambda,w,i)} : w \in [m_\lambda]}$ are invariant under matrices of the form $M\xp{n}$, for any $M \in \L(\C^2)$,
  \item \cref{lem:Unitary blocks}:
    the action of $M\xp{n}$ on $\mc{Q}_{\lambda,i}$ is given by the corresponding $\lambda$-irrep $Q_\lambda(M)$ defined in \cref{eq:Q-lambda};
    this is the same as saying that we have the so-called Gelfand--Tsetlin basis on the unitary register in \cref{eq:QP} (see \cref{apx:Construction} for more details).
\end{itemize}
Along the way we also establish several technical results that we state as propositions.

Recall the following notation from \cref{sec:SchurWeyl,sec:BasisQubits}.
For any partition $\lambda = (\lambda_1,\lambda_2) \pt n$ and Hamming weight $w \in [m_\lambda] := \set{0,\dotsc,\lambda_1-\lambda_2}$, we define an $n$-qubit state
\begin{equation}
  \ket{(\lambda,w,0)}
  := \ket{s_{\lambda_1-\lambda_2}(w)} \x \ket{\Psi^-}\xp{\lambda_2},
  \label{eq:0}
\end{equation}
where $\ket{s_\ell(w)}$ is the symmetric state on $\ell$ qubits with Hamming weight $w$,
\begin{equation}
  \ket{s_\ell(w)}
 := \binom{\ell}{w}^{-1/2}
    \sum_{\substack{x\in\set{0,1}^\ell\\\abs{x}=w}} \ket{x},
  \label{eq:s}
\end{equation}
and $\ket{\Psi^-} := (\ket{01} - \ket{10}) / \sqrt{2}$ is the singlet state.
More generally, recall that the dimension $d_\lambda$ of the permutation register is given in \cref{eq:dims}.
Then, for any $i \in [d_\lambda]$, we define
\begin{equation}
  \ket{(\lambda,w,i)} := \sum_{\pi \in \S_n} \alpha_\pi^{\lambda,i} P(\pi) \ket{(\lambda,w,0)},
  \label{eq:i}
\end{equation}
for some coefficients $\alpha_\pi^{\lambda,i} \in \R$, where $P(\pi) \in \U{2^n}$ permutes $n$ qubits according to the permutation $\pi \in \S_n$, see \cref{eq:P}.

\subsection{Block-diagonalization of permutations}

The goal of this section is to establish that the permutation operators $P(\pi)$ are block-diagonal in the Schur basis $\ket{(\lambda,w,i)}$, with blocks labeled by partitions $\lambda$ and Hamming weights $w$.
More precisely, we will establish \cref{eq:QP} which says that
\begin{equation}
  \hat{P}(\pi)
 := \Usch P(\pi) \Usch\ct
  = \bigoplus_{\lambda \pt n} \id_{m_\lambda} \x P_\lambda(\pi),
  \label{eq:P in Schur basis}
\end{equation}
where $\Usch$ denotes the Schur transform and $P_\lambda(\pi) \in \L(\C^{d_\lambda})$ are some matrices.
Note that we will only determine the block structure of $\hat{P}(\pi)$ but not the actual matrix entries of $P_\lambda(\pi)$.\footnote{We cannot determine the matrix entries unambiguously since \cref{eq:i} does not fully specify the states $\ket{(\lambda,w,i)}$ with $i \neq 0$.}
Luckily, for the purpose of this paper, the matrix entries of $P_\lambda(\pi)$ will not be necessary.

The following lemma establishes the block structure of $P(\pi)$ in the Schur basis
(it relies on \cref{prop:lambda mu,prop:overlap} that are proved subsequently).

\begin{lemma}\label{lem:P blocks}
Let $n \geq 1$.
For any permutation $\pi \in \S_n$,
partitions $\lambda,\mu \pt n$,
Hamming weights $w \in [m_\lambda]$ and $v \in [m_\mu]$,
and indices $i,j \in [d_\lambda]$,
\begin{equation}
  \bra{(\lambda,w,i)} P(\pi) \ket{(\mu,v,j)}
  = \delta_{\lambda\mu} \delta_{wv} \bra{i} P_\lambda(\pi) \ket{j},
  \label{eq:Pij}
\end{equation}
for some operator $P_\lambda(\pi) \in \L(\C^{d_\lambda})$ that does not depend on the Hamming weights $w$ and $v$.
\end{lemma}

\begin{proof}
Substituting $\bra{(\lambda,w,i)}$ and $\ket{(\mu,v,j)}$ from \cref{eq:i} and expanding both sides,
\begin{align}
  \bra{(\lambda,w,i)} P(\pi) \ket{(\mu,v,j)}
 &= \bra{(\lambda,w,0)}
    \of*{
      \sum_{\sigma \in \S_n}
      \alpha_\sigma^{\lambda,i}
      P(\sigma)\ct
    }
    P(\pi)
    \of*{
      \sum_{\tau \in \S_n}
      \alpha_\tau^{\mu,j}
      P(\tau)
    }
    \ket{(\mu,v,0)} \\
 &= \sum_{\sigma,\tau \in \S_n}
    \alpha_\sigma^{\lambda,i}
    \alpha_\tau^{\mu,j}
    \bra{(\lambda,w,0)}
    P(\sigma^{-1}\pi\tau)
    \ket{(\mu,v,0)}. \label{eq:just zero}
\end{align}
\begin{itemize}
  \item If $\lambda \neq \mu$, all terms vanish thanks to \cref{prop:lambda mu}, giving us the first delta function $\delta_{\lambda\mu}$.
  \item If $\lambda = \mu$ but $w \neq v$, all terms vanish since $\bra{(\lambda,w,0)}$ and $P(\sigma^{-1}\pi\tau) \ket{(\lambda,v,0)}$ expanded in the standard basis consist of terms with different Hamming weights, namely $w+\lambda_2$ and $v+\lambda_2$.
  This gives us the second delta function $\delta_{wv}$.
  \item Finally, if $\lambda = \mu$ and $w = v$, each term in \cref{eq:just zero} is of the form
  \begin{equation}
    \alpha_\sigma^{\lambda,i}
    \alpha_\tau^{\lambda,j}
    \bra{(\lambda,w,0)} P(\sigma^{-1}\pi\tau) \ket{(\lambda,w,0)},
  \end{equation}
  which does not depend on $w$ according to \cref{prop:overlap}. Hence the sum in \cref{eq:just zero} depends only on $\lambda,\pi,i,j$, so we can write it as $\bra{i} P_\lambda(\pi) \ket{j}$ for some $P_\lambda(\pi) \in \L(\C^{d_\lambda})$.
\end{itemize}
Since this covers all cases, the proof is complete.
\end{proof}

The following proposition establishes that the value of
$\bra{(\lambda, w, 0)} P(\pi) \ket{(\lambda, w, 0)}$
does not depend on $w$ and hence can be written as $\bra{0} P_\lambda(\pi) \ket{0}$ for some matrix $P_\lambda(\pi)$.

\begin{prop}\label{prop:overlap}
Let $n \geq 1$ and fix any partition $\lambda = (\lambda_1, \lambda_2) \pt n$ and a permutation $\pi \in \S_n$. Then the value of
\begin{equation}
  \bra{(\lambda, w, 0)} P(\pi) \ket{(\lambda, w, 0)}
  \label{eq:overlap}
\end{equation}
does not depend on $w \in [m_\lambda]$.
\end{prop}

\begin{proof}
Since $(U \x U) \ket{\Psi^-} = \ket{\Psi^-}$ for any $U \in \SU{2}$,
we see from \cref{eq:0} that
\begin{equation}
  U\xp{n} \ket{(\lambda, w, 0)}
  = U\xp{\ell} \ket{s_\ell(w)} \x \ket{\Psi^-}\xp{r}
\end{equation}
where $\ell := \lambda_1 - \lambda_2$ and $r := \lambda_2$.
By combining \cref{lem:irreducibility,lem:tjk from symmetric subspace}, we know that $U\xp{\ell}$ acts irreducibly in the $(\ell+1)$-dimensional symmetric subspace $\Sym{\ell} := \spn \set{\ket{s_\ell(w)} : w \in \set{0,\dotsc,\ell}}$. According to Burnside's theorem (see \cite{LR04} for a simple proof), the matrix algebra corresponding to this action coincides with the set of all $(\ell+1) \times (\ell+1)$ complex matrices. Hence any linear transformation on $\Sym{\ell}$ can be expressed as a linear combination of $U\xp{\ell}$, for some choice of $U \in \SU{2}$.

\newcommand{\Shift}{O}

Consider a unitary $\Shift$ that cyclically permutes the values $w \in \set{0, \dotsc, \ell}$:
\begin{equation}
  \Shift \ket{s_\ell(w)} := \Shift \ket{s_\ell(w\+1)},
  \label{eq:C}
\end{equation}
where addition is modulo $\ell+1$.
Note that this specifies $\Shift$ only on the $(\ell+1)$-dimensional subspace $\Sym{\ell}$ of the ambient space $(\C^2)\xp{\ell}$.
Thanks to Burnside's theorem, we can write
\begin{equation}
  \Shift = \sum_i \alpha_i U_i\xp{\ell},
\end{equation}
for some $\alpha_i \in \C$ and $U_i \in \SU{2}$,
which extends $\Shift$ to the rest of $(\C^2)\xp{\ell}$ while still agreeing with \cref{eq:C}.
We can extend this further to the whole of $(\C^2)\xp{n}$ by letting
\begin{equation}
  \tilde{\Shift} := \sum_i \alpha_i U_i\xp{n}.
\end{equation}

Note that $\tilde{\Shift}$ acts on the states $\ket{(\lambda,w,0)}$ as follows:
\begin{align}
  \tilde{\Shift} \ket{(\lambda,w,0)}
  &= \sum_i (\alpha_i U_i\xp{\ell} \x U_i\xp{2r})
     \of*{\ket{s_\ell(w)} \x \ket{\Psi^-}\xp{r}} \\
  &= \sum_i \alpha_i U_i\xp{\ell} \ket{s_\ell(w)} \x
     U_i\xp{2r} \ket{\Psi^-}\xp{r} \\
  &= \Shift \ket{s_\ell(w)} \x \ket{\Psi^-}\xp{r} \\
  &= \ket{s_\ell(w\+1)} \x \ket{\Psi^-}\xp{r} \\
  &= \ket{(\lambda,w\+1,0)},
\end{align}
for any $w \in \set{0,\dotsc,\ell}$.
Furthermore, since $\Shift\ct \ket{s_\ell(w\+1)} = \ket{s_\ell(w)}$, we can run the same calculation backwards and conclude that $\tilde{\Shift}\ct \ket{(\lambda,w\+1,0)} = \ket{(\lambda,w,0)}$ for all $w \in \set{0,\dotsc,\ell}$. Since $\tilde{\Shift}\ct$ is a linear combination of $(U_i\ct)\xp{n}$, we see that $\tilde{\Shift}\ct P(\pi) = P(\pi) \tilde{\Shift}\ct$ for all $\pi \in \S_n$. Hence,
\begin{align}
  \bra{(\lambda,w\+1,0)} P(\pi)
  \ket{(\lambda,w\+1,0)}
&= \bra{(\lambda,w,0)} \tilde{\Shift}\ct P(\pi) \tilde{\Shift}
   \ket{(\lambda,w,0)} \\
&= \bra{(\lambda,w,0)} P(\pi) \tilde{\Shift}\ct \tilde{\Shift}
   \ket{(\lambda,w,0)} \\
&= \bra{(\lambda,w,0)} P(\pi)
   \ket{(\lambda,w,0)}.
\end{align}
Since the above holds for any $w$, the value of $\bra{(\lambda,w,0)} P(\pi) \ket{(\lambda,w,0)}$ is independent of $w$.
\end{proof}

The second ingredient that was used in \cref{lem:P blocks} to reveal the block structure of $P(\pi)$ in the Schur basis was \cref{prop:lambda mu}, which shows that the quantity $\bra{(\lambda, w, 0)} P(\pi) \ket{(\mu, v, 0)}$ vanishes when $\lambda \neq \mu$.
To establish this, we first need to introduce some additional terminology that will allow us to make use of the graphical notation introduced in \cref{apx:Graphical}.



\newcommand{\networkfig}[1]{%
\begin{tikzpicture}[> = stealth, thick, baseline,
  dot/.style = {radius = 1.3pt, fill = black},
  round/.style = {rounded corners = 4pt},
  qubit/.style = {radius = 0.48*\W, thick},
  qubit1/.style = {qubit, fill = lightgray, draw = gray},
  qubit2/.style = {qubit, draw = lightgray, densely dotted},
  singlet/.style = {thick, round, postaction = {decorate}, decoration = {
    markings, mark = at position 0.5 with {
      \node [draw, sharp corners, inner sep = 4.5pt, fill = lightgray] {};
      \draw [->] +(-0.1,0) -- +(0.1,0);
    }}},
  p/.style = {thick, out = 0, in = -180},
  fat/.style = {p, pink, line width = 4pt, line cap = round}
]

\pgfdeclarelayer{base}
\pgfdeclarelayer{background}
\pgfsetlayers{base,background,main}

\def\H{11}   
\def\w{0.30} 
\def\W{2*\w} 
\def\v{0.7}  
\def\d{0.8}  
\def\D{4cm}  

#1 

\end{tikzpicture}%
}



\newcommand{\qubits}[2]{
  \def\name{\if#1-x\else y\fi}
  \pgfmathsetmacro{\singlets}{div(\H-#2,2)}
  \pgfmathsetmacro{\start}{#2+1}
  \path (#1\d,0) coordinate (\name);
  \foreach \i in {1,...,\H} {
    \path (\name)+(0,\W-\i*\W) coordinate (\name\i);
  }
  \ifx0#2\else
    \foreach \i in {1,...,#2} {
      \path[qubit1] (\name\i) circle;
      \node at (\name\i) {$\name_\i$};
    }
  \fi
  \if0\singlets\else
    \begin{pgfonlayer}{base}
      \foreach \i in {\start,...,\H} {
        \path[qubit2] (\name\i) circle;
      }
    \end{pgfonlayer}
    \foreach \s in {1,...,\singlets} {
      \pgfmathsetmacro{\orig}{\H-2*\s+1}
      \pgfmathsetmacro{\targ}{\H-2*\s+2}
      \fill (\name\orig) circle [dot];
      \fill (\name\targ) circle [dot];
      \draw[singlet] (\name\orig) -- ++(#1\v,0) -- ++(0,-\W) -- (\name\targ);
    }
  \fi
}



\newcommand{\diagram}[5]{
  \begin{scope}[xshift = #1*\D]
    \ifx&#2&\else
      \path (0,-\W*\H-\W) node {(#2)};
    \fi
    \qubits{-}{#3}
    \qubits{+}{#4}
    \begin{scope}[on background layer]
      \foreach \i/\j in {#5} {
        \draw[p] (x\i) to (y\j);
      }
    \end{scope}
  \end{scope}
}



\newcommand{\FourDiagrams}{\networkfig{
  \diagram{0}{a}{3}{3}{1/4,2/1,3/2,4/5,5/3,6/7,7/6,8/11,9/8,10/9,11/10}
  \diagram{1}{b}{3}{3}{1/4,2/1,3/2,4/3,5/5,6/6,7/7,8/8,9/11,10/9,11/10}
  \diagram{2}{c}{3}{3}{1/2,2/4,3/5,4/1,5/7,6/8,7/11,8/3,9/6,10/9,11/10}
  \diagram{3}{d}{5}{3}{1/3,2/1,3/2,4/6,5/4,6/8,7/10,8/5,9/7,10/11,11/9}
}}


\newcommand{\PairedBits}{\networkfig{
  \def\H{2}
  \tikzset{baseline = -11pt}
  \diagram{0}{}{2}{0}{1/1,2/2}
}}


\newcommand{\SinglePath}{\networkfig{
  \def\H{4}
  \tikzset{baseline = -28pt}
  \diagram{0}{}{2}{0}{1/1,2/3,3/2,4/4}
}}



\newcommand{\minisinglet}[3]{
  \draw[singlet] (#1) -- ++(#3 0.4,0) -- ++(0,-\W) -- (#2);
}


\newcommand{\singletchain}[4]{
\begin{scope}[yshift = -#1*1.8cm]
  \foreach \i in {1,2,3} {
    \path (-0.5,-\i*\W) coordinate (x\i);
    \path ( 0.5,-\i*\W) coordinate (y\i);
    \fill (x\i) circle [dot];
    \fill (y\i) circle [dot];
  }
  \minisinglet{y2}{y3}{+}
  \minisinglet{x2}{x3}{-}
  \foreach \i/\j in {#2} {
    \draw[p] (x\i) to (y\j);
  }
  \path (y2)+(3.0,0) coordinate (c);
  \def\s{0.6}
  \path (c)+(-\s,0) coordinate (c1);
  \path (c)+(+\s,0) coordinate (c2);
  \fill (c) circle [dot];
  \fill (c)+( 2*\s,0) circle [dot];
  \fill (c)+(-2*\s,0) circle [dot];
  \draw[singlet] (c1)+(-#3\s,0) -- +(+#3\s,0);
  \draw[singlet] (c2)+(-#4\s,0) -- +(+#4\s,0);
  \path (y2)+(1.1,-0.05) node {$=$};
  \path (c) +(1.8,-0.05) node {$=$};
  \path (c) +(3.9,-0.05) node {$=$};
  \path (c)+(2.1,0) node [anchor = west] {$\displaystyle \frac{1}{2} S\tps{#3} S\tps{#4}$};
  \path (c)+(4.7,0) node {$\displaystyle \same{#3}{#4} \frac{1}{2} I$};
\end{scope}
}

\newcommand{\tps}[1]{\if#1+\else\tp\fi}

\newcommand{\same}[2]{\ifx#1#2-\else\phantom{+}\fi}

\newcommand{\SingletChains}{\networkfig{
  \singletchain{0}{1/2,2/3,3/1}{+}{+}
  \singletchain{1}{1/2,2/1,3/3}{+}{-}
  \singletchain{2}{1/3,2/2,3/1}{-}{+}
  \singletchain{3}{1/3,2/1,3/2}{-}{-}
}}


Fix some $n \geq 1$ and consider integers $\ell_1, \ell_2, r_1, r_2 \geq 0$ such that $\ell_1 + 2 r_1 = \ell_2 + 2 r_2 = n$, and let $x \in \set{0,1}^{\ell_1}$ and $y \in \set{0,1}^{\ell_2}$ be arbitrary bit strings.
Fix any permutation $\pi \in \S_n$ and consider the expression
\begin{equation}
  \bra{x} \x \bra{\Psi^-}\xp{r_1}
  \cdot P(\pi) \cdot
  \ket{y} \x \ket{\Psi^-}\xp{r_2}.
  \label{eq:xPy}
\end{equation}
where $\ket{\Psi^-}$ denotes the singlet state.
We associate to such expression a \emph{diagram} (see examples in \cref{fig:Diagrams}) that depicts the qubits of both states by two columns of circles and the permutation $\pi$ by lines connecting these circles. The bits of strings $x$ and $y$ are depicted by gray circles labeled by $x_i$ and $y_j$, respectively, while the singlets are depicted by arcs that connect the corresponding two qubits. Since the singlet state $\ket{\Psi^-} = (\ket{01} - \ket{10}) / \sqrt{2}$ is not symmetric under exchanging the two qubits, we include an arrow pointing from its first qubit to its second qubit:
\begin{equation}
  \ket{\Psi^-}
  \; = \;
  \networkfig{
    \tikzset{baseline = -11pt}
    \def\H{2}
    \qubits{+}{0}
  }
  \; = \; - \;
  \networkfig{
    \tikzset{baseline = 6pt}
    \tikzset{yscale = -1}
    \def\H{2}
    \qubits{+}{0}
  } \;.
  \label{eq:singlet}
\end{equation}
Flipping the arrow introduces an overall minus sign since $\ket{\Psi^-}$ is anti-symmetric.

Singlets in \cref{fig:Diagrams} are chained together by permutation $\pi$, forming paths between the gray circles (ignore the arrow directions for now). These paths effectively pair up the elements of the set $\set{x_1, \dotsc, x_{\ell_1}, y_1, \dotsc, y_{\ell_2}}$ representing the bits of $x$ and $y$. Since each bit is paired up with exactly one other bit, we call this a \emph{matching}. For example, the four diagrams in \cref{fig:Diagrams} produce the following matchings:
\begin{align*}
  \text{(a), (b)\quad} & \set{(x_1,y_3), (x_2,y_1), (x_3,y_2)}, \\
  \text{(c)\quad} & \set{(x_1,y_2), (x_2,x_3), (y_1,y_3)}, \\
  \text{(d)\quad} & \set{(x_1,y_3), (x_2,y_1), (x_3,y_2), (x_4,x_5)}.
\end{align*}

\begin{figure}
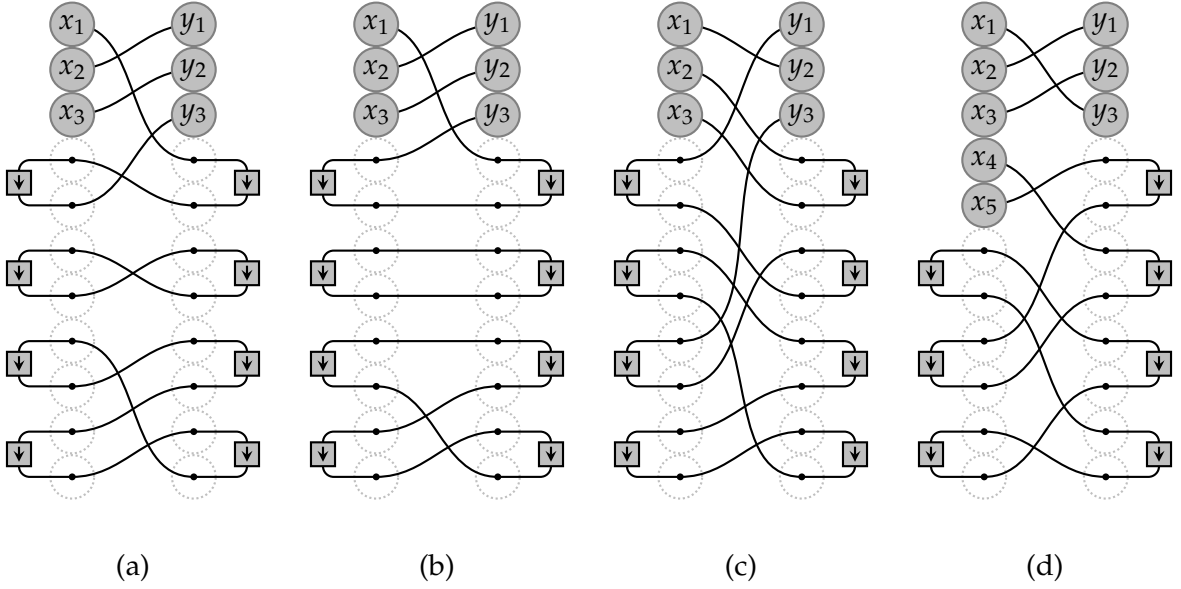

\centering\FourDiagrams
\caption{\label{fig:Diagrams}Four different instances of \cref{eq:xPy} represented in graphical notation. Different permutations $\pi$ result in different matchings among the bits of $x$ and $y$ induced by singlet chains. Each singlet $\ket{\Psi^-} = (\ket{01} - \ket{10}) / \sqrt{2}$ is denoted by an arrow pointing from its first qubit to the second, see \cref{eq:singlet}. The left two matchings (a) and (b) are proper since the bits of $x$ are bijectively matched with the bits of $y$, while the right two matchings (c) and (d) are improper, since two bits of $x$ or $y$ are matched.}
\end{figure}

We call a matching \emph{proper} if each $x_i$ is matched to some $y_j$ (\eg, matchings~(a) and~(b)); otherwise we call it \emph{improper} (\eg, matchings~(c) and~(d)). Note that an improper matching always pairs up two bits of $x$ or two bits of $y$. In particular, if the strings $x$ and $y$ have different lengths, such as in case~(d), the matching will always be improper.

To evaluate \cref{eq:xPy}, we will use the following restatement of \cref{cor:S chains} in terms of the normalized singlet state $\ket{\Psi^-} = \ket{S} / \sqrt{2}$.

\begin{cor}\label{cor:singlet chains}
For any $k \geq 0$,
\begin{align}
  \of[\big]{\id \x \bra{\Psi^-}\xp{k}}
  \of[\big]{\ket{\Psi^-}\xp{k} \x \id}
  &= \frac{1}{2^k} (-1)^k \id, \label{eq:different sides} \\
  \of[\big]{\id \x \bra{\Psi^-}\xp{k} \x \id}
  \of[\big]{\ket{\Psi^-}\xp{k+1}}
  &= \frac{1}{2^k} (-1)^k \ket{\Psi^-}, \label{eq:same side} \\
  \bra{\Psi^-} \SWAP
  \of[\big]{\id \x \bra{\Psi^-}\xp{k} \x \id}
  \of[\big]{\ket{\Psi^-}\xp{k+1}}
  &= \frac{1}{2^k} (-1)^{k+1}. \label{eq:loop}
\end{align}
Moreover, if $j$ singlets are in the opposite direction, \ie, $\ket{\Psi^-}$ is replaced by $\SWAP \ket{\Psi^-} = - \ket{\Psi^-}$, the final value of each expression picks up an extra factor of $(-1)^j$.
\end{cor}

The following result shows that certain matrix entries of $P(\pi)$ in the Schur basis vanish whenever the matching resulting from the corresponding two basis vectors is improper.

\begin{prop}\label{prop:lambda mu}
Let $n \geq 1$ and fix any partitions $\lambda = (\lambda_1, \lambda_2) \pt n$ and $\mu = (\mu_1, \mu_2) \pt n$, and a permutation $\pi \in \S_n$. If the diagram corresponding to
\begin{equation}
  \bra{(\lambda, w, 0)} P(\pi) \ket{(\mu, v, 0)}
  \label{eq:vanishing}
\end{equation}
yields an improper matching, the expression vanishes irrespective of the Hamming weights $w \in [m_\lambda]$ and $v \in [m_\mu]$. In particular, it vanishes whenever $\lambda \neq \mu$.
\end{prop}

\begin{proof}
If $\lambda \neq \mu$ then $\lambda_1 - \lambda_2 \neq \mu_1 - \mu_2$, so there will be more bits on one side of the diagram than the other (\eg, case~(d) in \cref{fig:Diagrams}), so a proper matching between them is impossible.

To establish the main claim, we first write
$\lambda = (r_1 + \ell_1, r_1)$ and
$\mu = (r_2 + \ell_2, r_2)$
for some integers $\ell_1, \ell_2, r_1, r_2 \geq 0$ such that
$\ell_1 + 2 r_1 = \ell_2 + 2 r_2 = n$
(see \cref{fig:lambda}).
Next, recall from \cref{eq:0,eq:s} that
\begin{equation}
  \bra{(\lambda,w,0)}
= \binom{\ell_1}{w}^{-1/2}
  \sum_{\substack{x\in\set{0,1}^{\ell_1}\\\abs{x}=w}}
  \bra{x} \x \bra{\Psi^-}\xp{r_1},
\end{equation}
and similarly for $\ket{(\mu,v,0)}$.
Hence, we can expand \eqref{eq:vanishing} as follows:
\begin{equation}
  \binom{\ell_1}{w}^{-1/2}
  \binom{\ell_2}{v}^{-1/2}
  \sum_{\substack{x \in \set{0,1}^{\ell_1} \\ \abs{x} = w}}
  \sum_{\substack{y \in \set{0,1}^{\ell_2} \\ \abs{y} = v}}
  \bra{x} \x \bra{\Psi^-}\xp{r_1} \cdot P(\pi) \cdot \ket{y} \x \ket{\Psi^-}\xp{r_2}.
  \label{eq:xy}
\end{equation}
Each term in this expansion can be represented by one of the diagrams discussed before.
Since the wiring of the diagram depends only on $\pi$ and the number of singlets $r_1,r_2$, it is in fact the same diagram for all terms, except with different values of the bit strings $x$ and $y$ substituted in the gray circles (see examples in \cref{fig:Diagrams}).
Note that this expression can be factorized over all paths and cycles in the diagram.
Hence, to show that it vanishes, it suffices to find a single path with a vanishing contribution.

By assumption, the diagram corresponding to \cref{eq:xy} yields an improper matching, meaning that two bits on the same side, say $x_1$ and $x_2$, are matched with each other. Hence, the number of singlets (or arrows) on the path between them is odd. According to \cref{eq:same side} in \cref{cor:singlet chains}, such path is equivalent to a single singlet. More specifically, \cref{eq:same side} allows us to replace a path of $2k+1$ singlets by a contracted path consisting only of $(-1)^{k+j} \ket{\Psi^-} / 2^k$, where $j$ is the number of singlets on the original path that point in the direction opposite to the singlet on the contracted path.\footnote{This can also be seen by iteratively contracting two singlets at a time using the identities displayed in \cref{fig:Singlet chains}.}
For example,
\begin{equation}
  \SinglePath
  \;\; = \;\;
  - \frac{1}{2} \;
  \PairedBits \;.
\end{equation}

Since our goal is to show that the expression \eqref{eq:xy} vanishes, we can ignore the sign and the normalization, and simply treat the contracted path as $\ket{\Psi^-}$.
Hence, it suffices to analyze only the case when the path from $x_1$ to $x_2$ consists of a single singlet.
In this case, the relevant part of the diagram looks as follows:
\begin{equation}
  \braket{x_1,x_2}{\Psi^-}
  \; = \;
  \PairedBits \;,
  \label{eq:single singlet}
\end{equation}
which means that we can pull out from \eqref{eq:xy} a factor of
\begin{equation}
  \sum_{h \in H}
  \sum_{\substack{x \in \set{0,1}^2 \\ |x| = h}}
  \braket{x_1,x_2}{\Psi^-},
  \label{eq:vanishing factor}
\end{equation}
where the range of $h$ is some\footnote{More specifically,
$H := \set{\max \set{2-(\ell_1-w),0}, \dotsc, \min \set{w,2}}$.} subset $H \subseteq \set{0,1,2}$ that depends on the length $\ell_1$ of $x$ and its Hamming weight $w$ (or the analogous parameters $\ell_2$ and $v$ if the permutation $\pi$ matched up two bits of $y$ instead of $x$).
Note that \eqref{eq:vanishing factor} vanishes for any
$H \subseteq \set{0,1,2}$ because
\begin{align}
  \braket{0,0}{\Psi^-} &= 0, &
  \braket{0,1}{\Psi^-} + \braket{1,0}{\Psi^-} &= 0, &
  \braket{1,1}{\Psi^-} &= 0.
\end{align}
Since the pre-factor \eqref{eq:vanishing factor} vanishes, then so does the whole expression \eqref{eq:xy}.
Moreover, this phenomenon is independent of the values of $w$ and $v$.
\end{proof}

\newcommand{\loops}{\mathrm{loops}}
\newcommand{\sign}{\mathrm{sign}}

While we will not need this, a more elaborate version of the above argument can be used to show a more general result that subsumes \cref{prop:overlap,prop:lambda mu}, and determines the top left entry $\bra{0} P_\lambda(\pi) \ket{0}$ of the matrix $P_\lambda(\pi)$ implicitly defined in \cref{eq:P in Schur basis}.
Note that the remaining entries of $P_\lambda(\pi)$ cannot be determined in principle unless one makes a specific choice of the coefficients $\alpha_\pi^{\lambda,i}$ in \cref{eq:i} that orthogonalize the states $\ket{(\lambda,w,i)}$.

\begin{lemma}\label{lem:overlap value}
Let $n \geq 1$ and fix any partitions $\lambda,\mu \pt n$, Hamming weights $w \in [m_\lambda]$ and $v \in [m_\mu]$, and a permutation $\pi \in \S_n$. Then
\begin{equation}
  \bra{(\lambda, w, 0)} P(\pi) \ket{(\mu, v, 0)}
= \delta_{\lambda\mu} \delta_{wv}
  \bra{0} P_\lambda(\pi) \ket{0},
  \label{eq:00}
\end{equation}
where
\begin{equation}
  \bra{0} P_\lambda(\pi) \ket{0} =
  \begin{cases}
    \sign(\gamma) 2^{\loops(\gamma) - \lambda_2}
      & \text{if $\gamma$ induces a proper matching}, \\
    0 & \text{otherwise}.
  \end{cases}
\end{equation}
Here $\gamma$ denotes the diagram depicting $\bra{(\lambda, w, 0)} P(\pi) \ket{(\lambda, w, 0)}$ which depends only on $\pi$ and $\lambda$ (see \cref{fig:Diagrams} for some examples and the description below \cref{eq:xPy} for definition),
$\sign(\gamma) \in \set{+1,-1}$ is the sign of $\gamma$,
and $\loops(\gamma)$ is the number of loops in $\gamma$
(see proof for more details).
\end{lemma}

\begin{proof}
We can proceed the same way as in the proof of \cref{prop:lambda mu}, except that now we want to obtain an explicit formula for \eqref{eq:xy}.
As before, we observe that the matching induced by \eqref{eq:xy} is not proper unless $\lambda = \mu$, which gives us the first Kronecker delta function $\delta_{\lambda\mu}$.
If $\lambda = \mu$, we can use the same argument as in \cref{lem:P blocks} to obtain the second Kronecker delta function $\delta_{wv}$.
Namely, we notice from \cref{eq:0} that the standard basis expansions of $\bra{(\lambda,w,0)}$ and $\ket{(\mu,v,0)}$ contain only states of Hamming weights $w + \lambda_2$ and $v + \mu_2$, respectively.
Since the permutation $P(\pi)$ preserves Hamming weight, \cref{eq:xy} vanishes unless $w = v$.

From now on we assume that $\lambda = \mu$ and $w = v$, and our goal is to derive a formula for $\bra{(\lambda,w,0)} P(\pi) \ket{(\lambda,w, 0)}$.
Note that in this case \cref{eq:xy} simplifies to
\begin{equation}
  \bra{(\lambda,w,0)} P(\pi) \ket{(\lambda,w, 0)}
= \binom{\ell}{w}^{-1}
  \sum_{\substack{x,y \in \set{0,1}^{\ell} \\ \abs{x} = \abs{y} = w}}
  \bra{x} \x \bra{\Psi^-}\xp{\lambda_2} \cdot P(\pi) \cdot \ket{y} \x \ket{\Psi^-}\xp{\lambda_2}
  \label{eq:wxy}
\end{equation}
where $\ell := \lambda_1 - \lambda_2$.
There are two cases, depending on the diagram $\gamma(\pi,\lambda)$ that depicts the terms on the right-hand side of \cref{eq:wxy}. If the matching induced by $\gamma(\pi,\lambda)$ is improper, \cref{eq:wxy} vanishes independently of $w$ by the same argument as in \cref{prop:lambda mu}. Otherwise, the bits of $x \in \set{0,1}^\ell$ and $y \in \set{0,1}^\ell$ are bijectively matched up by $\ell$ chains of singlets, where each chain contains an even number of singlets (see examples (a) and (b) in \cref{fig:Diagrams}).
In addition, $\gamma(\pi,\lambda)$ may also contain some loops with an even number of singlets in each of them.

We can completely contract all singlet chains and loops in $\gamma(\pi,\lambda)$ either by removing two singlets at a time by using graphical notation (see \cref{fig:Singlet chains}) or by applying \cref{cor:singlet chains} (use \cref{eq:different sides} to contract singlet chains between bits $x_i$ and $y_j$, and \cref{eq:loop} to contract singlet loops).
Note that we can uniformly use the same contraction procedure for all terms in \cref{eq:wxy} since they correspond to the same diagram $\gamma(\pi,\lambda)$.

The contracted diagram has no more singlets but consists only of the bits of $x$ and $y$ with direct connections between them.
Since the original diagram yielded a proper matching, the contracted diagram encodes a Kronecker delta function between $x$ and a permutation of $y$, which establishes a bijection between the values of $x$ and $y$.
Since non-matching pairs of $x$ and $y$ drop out, we can discard one of the sums in \cref{eq:wxy}.
Moreover, the remaining terms are all equal since they depend only on the diagram $\gamma(\pi,\lambda)$, so we can forget about the second sum as well and cancel out the binomial coefficient $\binom{\ell}{w}$.
This means that the right-hand side of \cref{eq:wxy} does not depend on $w$, and it can be computed by evaluating a single non-vanishing term whose diagram is $\gamma(\pi,\lambda)$.
In particular, it is appropriate to denote the resulting value by $\bra{0} P_\lambda(\pi) \ket{0}$, as in \cref{eq:00}, since it depends only on $\pi$ and $\lambda$.
We will now separately determine the absolute value and the sign of $\bra{0} P_\lambda(\pi) \ket{0}$ as a function of the diagram $\gamma = \gamma(\pi,\lambda)$.

\begin{figure}
\centering\SingletChains
\caption{\label{fig:Singlet chains}Contraction of two consecutive singlets $\ket{\Psi^-} = \ket{S} / \sqrt{2}$ via \cref{eq:different sides}. The result has \emph{negative} sign if both singlets point in the same direction and positive sign otherwise.}
\end{figure}

Since each term of \cref{eq:wxy} contains $2 \lambda_2$ singlets, with a normalization of $2^{-1/2}$ each, we need to introduce a factor of $2^{-\lambda_2}$ to compensate for their removal during contraction.
Moreover, each contracted loop contributes an additional factor of $\Tr \id_2 = 2$, so the final result has an absolute value of $\abs{\bra{0} P_\lambda(\pi) \ket{0}} = 2^{\loops(\gamma) - \lambda_2}$.
The sign of $\bra{0} P_\lambda(\pi) \ket{0}$ is determined by accounting for the singlet arrows in $\gamma$.
According to \cref{cor:singlet chains}, contracting a path or a loop of $2k$ singlets that all point in the same direction produces the sign $(-1)^k$.
The overall sign of $\gamma$, which we denoted by $\sign(\gamma)$, is then determined by taking the product of individual signs over all paths and loops.
Since $\SWAP \ket{\Psi^-} = - \ket{\Psi^-}$, the sign must be inverted whenever the direction of any singlet is inverted (this rule is consistent since we are dealing with the case when each path and loop has an even number of singlets).
\end{proof}

\subsection{Orthonormality}

In this section, we show that the vectors $\ket{(\lambda,w,i)}$ form an orthonormal basis of the $n$-qubit space $\C^{2^n}$.
First, in \cref{prop:Orthonormal} we establish orthonormality only for subsets indexed by $i$, and then in \cref{lem:FullOrthonormal} we establish it for the full set.

\begin{prop}\label{prop:Orthonormal}
Let $n \geq 1$ and fix any partition $\lambda \pt n$ and Hamming weight $w \in [m_\lambda]$. Then, for any $i,j \in [d_\lambda]$,
\begin{equation}
  \braket{(\lambda,w,i)}{(\lambda,w,j)} = \delta_{ij}.
  \label{eq:orthonormal}
\end{equation}
\end{prop}

\begin{proof}
Recall from \cref{eq:i} that
\begin{equation}
  \ket{(\lambda,w,i)} = \sum_{\pi \in \S_n} \alpha_\pi^{\lambda,i} P(\pi) \ket{(\lambda,w,0)}
  \label{eq:lwi2}
\end{equation}
where the coefficients $\alpha_\pi^{\lambda,i} \in \R$ are chosen so that $\braket{(\lambda,0,i)}{(\lambda,0,j)} = \delta_{ij}$, see \cref{eq:loi}.
Hence,
\begin{align}
  \braket{(\lambda,w,i)}{(\lambda,w,j)}
  &= \bra{(\lambda,w,0)}
     \sum_{\pi \in \S_n}
     \alpha_\pi^{\lambda,i} P(\pi)\ct
     \sum_{\sigma \in \S_n}
     \alpha_\sigma^{\lambda,j} P(\sigma)
     \ket{(\lambda,w,0)} \\
  &= \sum_{\pi,\sigma \in \S_n}
     \alpha_\pi^{\lambda,i}
     \alpha_\sigma^{\lambda,j}
     \bra{(\lambda,w,0)} P(\pi^{-1}\sigma) \ket{(\lambda,w,0)}.
\end{align}
Using \cref{prop:overlap}, we can replace $w$ by $0$ and perform the same calculation in reverse. This produces $\braket{(\lambda,0,i)}{(\lambda,0,j)}$, which is equal to the desired value $\delta_{ij}$ by \cref{eq:loi}.
\end{proof}

Orthonormality of the full set of $\ket{(\lambda,w,i)}$ follows by combining \cref{lem:P blocks,prop:Orthonormal}.

\begin{lemma}\label{lem:FullOrthonormal}
For any $n \geq 1$, partitions $\lambda,\mu \pt n$, Hamming weights $w \in [m_\lambda]$ and $v \in [m_\mu]$, and indices $i \in [d_\lambda]$ and $j \in [d_\mu]$,
\begin{equation}
  \braket{(\lambda,w,i)}{(\mu,v,j)}
  = \delta_{\lambda\mu} \delta_{wv} \delta_{ij}.
\end{equation}
\end{lemma}

\begin{proof}
If $\pi \in \S_n$ is the trivial permutation then $P(\pi) = \id_{2^n}$. We already know from \cref{lem:P blocks} that, for any $\pi \in \S_n$,
\begin{equation}
  \bra{(\lambda,w,i)} P(\pi) \ket{(\mu,v,j)}
  = \delta_{\lambda\mu} \delta_{wv} \bra{i} P_\lambda(\pi) \ket{j},
\end{equation}
for some operator $P_\lambda(\pi) \in \L(\C^{d_\lambda})$. This expression vanishes when $\lambda \neq \mu$ or $w \neq v$, so it only remains to understand the case when $\lambda = \mu$ and $w = v$. In this case,
\begin{equation}
  \braket{(\lambda,w,i)}{(\lambda,w,j)}
  = \delta_{ij}
\end{equation}
according to \cref{prop:Orthonormal}.
\end{proof}

\subsection{Unitary action}

This section considers the action of $M\xp{n}$ on the vectors $\ket{(\lambda,w,i)}$ for any $M \in \L(\C^2)$.
As a special case, this also includes the action of $U\xp{n}$ where $U \in \U{2}$ is unitary.

\begin{lemma}\label{lem:QInvar}
For any $n \geq 1$, partition $\lambda \pt n$, and index $i\in[d_\lambda]$,
the corresponding subspace
$\mc{Q}_{\lambda,i} := \spn \set[\big]{\ket{(\lambda,w,i)} : w \in [m_\lambda]}$
is invariant under $M\xp{n}$ for all $M \in \L(\C^2)$.
\end{lemma}

\begin{proof}
For any $\ell \geq 0$, $\pi \in \S_\ell$, and $w \in \set{0,\dotsc,\ell}$,
\begin{equation}
  P(\pi) M\xp{\ell} \ket{s_\ell(w)}
  = M\xp{\ell} P(\pi) \ket{s_\ell(w)}
  = M\xp{\ell} \ket{s_\ell(w)},
\end{equation}
so the vector $M\xp{\ell} \ket{s_\ell(w)}$ is fixed by all qubit permutations. Hence, the $\ell$-qubit symmetric space, $\spn\set{\ket{s_\ell(w)} : w\in\set{0,\dotsc,\ell}}$, is invariant under $M\xp{\ell}$ and we can write
\begin{equation}
  M\xp{\ell} \ket{s_\ell(w)}
  = \sum_{v\in[m_\lambda]} \beta'_v \ket{s_\ell(v)},
\end{equation}
for some $\beta'_v \in \C$.
Combining the above and the fact that
$(M \x M) \ket{\Psi^-} = \det(M) \ket{\Psi^-}$
together with the definition of $\ket{(\lambda,w,0)}$ from \cref{eq:lw0}, we get
\begin{align}
  M\xp{n} \ket{(\lambda,w,0)}
  = M\xp{(\lambda_1-\lambda_2)} \ket{s_{\lambda_1-\lambda_2}(w)} \x \of*{\det(M) \ket{\Psi^-}}\xp{\lambda_2}
  = \sum_{v\in[m_\lambda]} \beta_v \ket{(\lambda,v,0)},
  \label{eq:Closed}
\end{align}
where $\beta_v := \beta'_v \of{\det(M)}^{\lambda_2}$.
Using the definition of $\ket{(\lambda,w,i)}$ from \cref{eq:i},
\begin{align}
  M\xp{n} \ket{(\lambda,w,i)} &=
  \sum_{\pi \in \S_n} \alpha_{\pi}^{\lambda,i} P(\pi) \Big( M\xp{n} \ket{(\lambda,w,0)} \Big)\\
 & = \sum_{\pi \in \S_n} \alpha_{\pi}^{\lambda, i} P(\pi) \sum_{v\in [m_\lambda]} \beta_v \ket{(\lambda,v,0)} \\
 & = \sum_{v\in [m_\lambda]} \beta_v \left( \sum_{\pi \in \S_n} \alpha_{\pi}^{\lambda,i} P(\pi) \ket{(\lambda,v,0)} \right)\\
  &= \sum_{v\in[m_\lambda]} \beta_v \ket{(\lambda,v,i)},
\end{align}
where the final state lies in the subspace $\mc{Q}_{\lambda,i}$ as desired.
\end{proof}

Having established that the subspaces spanned by $\ket{(\lambda,w,i)}$ with $w \in [m_\lambda]$ are invariant under the action of $M\xp{n}$, we would like to derive explicit matrix entries of this action.
We will show that they coincide with $Q_\lambda(M)$, where $Q_\lambda$ is the representation of $\U{2}$ defined in \cref{apx:U2}.
This stems from the fact that the vectors $\ket{(\lambda,w,i)}$ are related to $\ket{(\lambda,w,0)}$, which can in turn be expressed in terms of the symmetric states $\ket{s_{\ell}(w)}$ and singlets $\ket{\Psi^-}$.
The result the follows from \cref{lem:tjk from symmetric subspace} in \cref{apx:Sym} which is an analogous statement for the symmetric states $\ket{s_{\ell}(w)}$.

\begin{lemma}\label{lem:Unitary blocks}
For any $n \geq 1$, partitions $\lambda,\mu \pt n$, Hamming weights $w \in [m_\lambda]$ and $v \in [m_\mu]$, indices $i \in [d_\lambda]$ and $j \in [d_\mu]$, and matrix $M \in \L(\C^2)$,
\begin{align}
  M\xp{n} \ket{(\mu,v,j)}
 &= \sum_{w=0}^{m_\mu}
    Q_\mu(M)_{wv}
    \ket{(\mu,w,j)}, &
  \bra{(\lambda,w,i)} M\xp{n} \ket{(\mu,v,j)}
 &= \delta_{\lambda\mu} \delta_{ij}
    Q_\mu(M)_{wv},
\end{align}
where
$m_\mu := \mu_1 - \mu_2 + 1$ and $Q_\mu: \L(\C^2) \to \L(\C^{m_\mu})$ is the representation defined in \cref{eq:Q-lambda}.
\end{lemma}

\begin{proof}
By \cref{lem:FullOrthonormal}, the vectors $\ket{(\lambda,w,i)}$ form an orthonormal basis of $\C^{2^n}$, so the second identity follows by multiplying both sides of the first identity by $\bra{(\lambda,w,i)}$.

To prove the first identity, note from \cref{eq:i} that
\begin{equation}
  M\xp{n} \ket{(\mu,v,j)}
= \sum_{\pi \in \S_n} \alpha_\pi^{\mu,j}
  P(\pi) M\xp{n} \ket{(\mu,v,0)}.
  \label{eq:Mn action}
\end{equation}
Recall from \cref{eq:0} that
$\ket{(\mu,v,0)} := \ket{s_{\mu_1-\mu_2}(v)} \x \ket{\Psi^-}\xp{\mu_2}$.
Since the singlet state satisfies
$(M \x M) \ket{\Psi^-} = \det(M) \ket{\Psi^-}$,
\cref{lem:tjk from symmetric subspace} implies that
\begin{align}
  M\xp{n} \ket{(\mu,v,0)}
 &= M\xp{\mu_1-\mu_2} \ket{s_{\mu_1-\mu_2}(v)} \x
    M\xp{2\mu_2} \ket{\Psi^-}\xp{\mu_2} \\
 &= \sum_{w=0}^{m_\mu}
    t^{\mu_1-\mu_2}_{wv}(M) \,
    \ket{s_{\mu_1-\mu_2}(w)} \x
    (\det M)^{\mu_2}
    \ket{\Psi^-}\xp{\mu_2} \\
 &= \sum_{w=0}^{m_\mu}
    Q_\mu(M)_{wv}
    \ket{(\mu,w,0)},
\end{align}
where we used the definition of $Q_\mu$ from \cref{eq:Q-lambda}.
The result follows by substituting this in \cref{eq:Mn action} and exchanging the two sums.
\end{proof}

\section{Representations of $2 \times 2$ matrices}\label{apx:2x2}

\renewcommand{\H}{\mathcal{H}}
\newcommand{\abcd}{\smx{a&b\\c&d}}
\newcommand{\ABCD}{\mx{a&b\\c&d}}

In this appendix, we briefly summarize the representation theory of $\SU{2}$ and $\U{2}$.
While representations of $\U{2}$ are defined as homomorphisms of the form $R: \U{2} \to \U{d}$ for some $d \geq 1$, the entries of the represented matrix $R(U)$ are actually polynomials (or more generally rational functions) in the entries of the original matrix $U$.
We take the perspective of \cite{Louck} and treat the matrix $U$ as if having symbolic entries $\abcd$ so that we can think of $R(U)$ as a matrix consisting of polynomials.
This has the advantage that we can represent not just unitary matrices but any $2 \times 2$ matrix, whether numeric or symbolic.
In particular, we can represent the set $\L(\C^2)$ of all $2 \times 2$ complex matrices, viewed as a monoid under matrix multiplication~\cite{YoungDiagrams}.

\subsection{Construction}\label{apx:Construction}

Let us recap an explicit construction of all irreducible representations of $\SU{2}$.
A concise derivation of this construction is given in \cite{Koornwinder}.
For more details, see \cite[p.~108]{Vilenkin}, \cite[p.~276]{VK1}, \cite[p.~279]{Wawrzynczyk}, or \cite[p.~90]{Zelobenko}.

For any function $\psi: \C^2 \to \C$ and matrix $M \in \L(\C^2)$, define the corresponding transformed function $\psi_M: \C^2 \to \C$ via
\begin{equation}
  \psi_M(v) := \psi(vM), \quad \forall v \in \C^2,
  \label{eq:psiM}
\end{equation}
where $v$ is treated as a row vector. More specifically, if $v = \mx{x&y}$ and $M = \ABCD$ then
\begin{equation}
  \psi_M(x,y) = \psi(ax+cy,bx+dy).
  \label{eq:psiMxy}
\end{equation}

Let us fix an integer\footnote{In physics, it is common to use the half-integer $\ell/2$ instead of $\ell$.} $\ell \geq 0$ and construct an $(\ell+1)$-dimensional representation of $\L(\C^2)$ that acts on $\H_\ell := \spn \set{x^\ell, x^{\ell-1} y, \dotsc, x y^{\ell-1}, y^\ell}$, the complex linear space of homogeneous degree-$\ell$ bivariate polynomials. We can treat any $\psi \in \H_\ell$ as a function $\psi: \C^2 \to \C$ by evaluating it at a given point $x,y \in \C$. For each $M \in \L(\C^2)$, define a map $T^\ell(M): \H_\ell \to \H_\ell$ that applies $M$ to the arguments of $\psi$:
\begin{equation}
  T^\ell(M): \psi \mapsto \psi_M,
  \quad \forall \psi \in \H_\ell.
  \label{eq:TsM}
\end{equation}
Note from \cref{eq:psiMxy} that $\psi_M$ is also homogeneous and of degree $\ell$.
Moreover, note from \cref{eq:psiM} that $(\psi_{M_1})_{M_2} = \psi_{M_1 M_2}$, since $(v M_1) M_2 = v (M_1 M_2)$ for all $v \in \C^2$ and $M_1, M_2 \in \L(\C^2)$.
This implies that $T^\ell(M_1 M_2) = T^\ell(M_1) T^\ell(M_2)$, meaning that $T^\ell$ is a representation of $\L(\C^2)$ on the $(\ell+1)$-dimensional space $\H_\ell$, where we view $\L(\C^2)$ as a monoid\footnote{Note that the map $T^\ell$ is \emph{not} linear, so it does \emph{not} constitute a representation of $\L(\C^2)$ as an algebra.} under matrix multiplication~\cite{YoungDiagrams}.

To compute the matrix entries of $T^\ell(M)$ in $\H_\ell$, we introduce the following basis consisting of normalized degree-$\ell$ monomials:
\begin{equation}
  \psi^\ell_k(x,y) := \binom{\ell}{k}^{1/2} x^{\ell-k} y^k
  \label{eq:psk}
\end{equation}
where $k \in \set{0,\dotsc,\ell}$.
In this basis,
\begin{equation}
  T^\ell(M) \psi^\ell_k = \sum_{j=0}^\ell t^\ell_{jk}(M) \psi^\ell_j,
\end{equation}
where $t^\ell_{jk}(M)$ are the matrix entries of $T^\ell(M)$.
Substituting the definition of $T^\ell(M)$ from \cref{eq:TsM,eq:psiMxy}, and $\psi^\ell_k$ from \cref{eq:psk},
\begin{equation}
  \binom{\ell}{k}^{1/2} (ax+cy)^{\ell-k} (bx+dy)^k
  = \sum_{j=0}^\ell t^\ell_{jk}\sof*{\ABCD} \binom{\ell}{j}^{1/2} x^{\ell-j} y^j.
  \label{eq:expansion}
\end{equation}
To obtain a formula for $t^\ell_{jk}$, we need to compare the coefficients at $x^{\ell-j} y^j$ on both sides.

Expanding the left-hand side of \cref{eq:expansion} yields
\begin{equation}
  \binom{\ell}{k}^{1/2}
  \sum_{r=0}^{\ell-k} \binom{\ell-k}{r} (ax)^r (cy)^{\ell-k-r}
  \sum_{s=0}^k \binom{k}{s} (bx)^s (dy)^{k-s}.
\end{equation}
To compare this with the right-hand side of \cref{eq:expansion}, we need to obtain the coefficient at $x^{\ell-j} y^j$ for $j \in \set{0,\dotsc,\ell}$.
If we pick $s := \ell-j-r$ in the second sum, we can guarantee that $s \geq 0$ and $s \leq k$ by restricting $r$ so that $\ell-j \geq r$ and $\ell-j-k \leq r$, respectively.
This produces
\begin{equation}
  \binom{\ell}{k}^{1/2}
  \sum_{r=\max\of{\ell-j-k,0}}^{\min\of{\ell-j,\ell-k}}
  \binom{\ell-k}{r} a^r c^{\ell-k-r}
  \binom{k}{\ell-j-r} b^{\ell-j-r} d^{j+k+r-\ell}.
\end{equation}
Comparing this with the coefficient at $x^{\ell-j} y^j$ on the right-hand side of \cref{eq:expansion}, we get the following formula for the matrix entries of $T^\ell(M)$:
\begin{align}
  t^\ell_{jk}\sof*{\ABCD}
  :=
  \sqrt{j!(\ell-j)!k!(\ell-k)!}
  \sum_{r=\max\of{\ell-j-k,0}}^{\min\of{\ell-j,\ell-k}}
  \frac{a^r b^{\ell-j-r} c^{\ell-k-r} d^{j+k+r-\ell}}{r!(\ell-j-r)!(\ell-k-r)!(j+k+r-\ell)!},
  \label{eq:tjk}
\end{align}
for any integer $\ell \geq 0$ and $j,k \in \set{0,\dotsc,\ell}$.
In physics literature, $T^\ell(M)$ is known as \emph{Wigner $D$-matrix} since it was first derived by Wigner (see eq.~(15.21) on p.~164 of \cite{Wigner}).
The basis in which these representations are written is known as the \emph{Gelfand--Tsetlin basis}, see \cite{GelfandII}, \cite[Section~11.3]{Louck}, \cite[Chapter~18]{VK3}, or \cite[Section~67]{Zelobenko}.

As advocated in \cite{Louck}, it is useful to think of $M$ as a symbolic matrix and of $T^\ell(M)$ as a collection of homogeneous polynomials in the entries of $M$. Since $T^\ell(M_1 M_2) = T^\ell(M_1) T^\ell(M_2)$ even for \emph{symbolic} $2 \times 2$ matrices $M_1$ and $M_2$, we can view $T^\ell$ as a representation of any $2 \times 2$ matrix group, such as $\SU{2}$, $\U{2}$, $\SL{2}$, $\GL{2}$, or even the full set of all $2 \times 2$ complex matrices $\L(\C^2)$, viewed as a monoid under matrix multiplication~\cite{YoungDiagrams}.

The first few values of $\ell$ yield the following expressions for $T^\ell(M)$:
\begin{align*}
  T^0\sof*{\ABCD} &= \mx{1}, \quad
  T^1\sof*{\ABCD} = \ABCD, \quad
  T^2\sof*{\ABCD} =
  \mx{
    a^2 & \sqrt{2} a b & b^2 \\
    \sqrt{2} a c & a d + b c & \sqrt{2} b d \\
    c^2 & \sqrt{2} c d & d^2
  }, \\
  T^3\sof*{\ABCD} &=
  \mx{
    a^3 & \sqrt{3} a^2 b & \sqrt{3} a b^2 & b^3 \\
    \sqrt{3} a^2 c & a (ad+2bc) & b (2ad+bc) & \sqrt{3} b^2 d \\
    \sqrt{3} a c^2 & c (2ad+bc) & d (ad+2bc) & \sqrt{3} b d^2 \\
    c^3 & \sqrt{3} c^2 d & \sqrt{3} c d^2 & d^3
  }, \\
  T^4\sof*{\ABCD} &=
  \mx{
    a^4 & 2 a^3 b & \sqrt{6} a^2 b^2 & 2 a b^3 & b^4 \\
    2 a^3 c & a^2 (ad+3bc) & \sqrt{6} ab (ad+bc) & b^2 (3ad+bc) & 2 b^3 d \\
    \sqrt{6} a^2 c^2 & \sqrt{6} ac (ad+bc) & a^2 d^2 + 4 abcd + b^2 c^2 & \sqrt{6} bd (ad+bc) & \sqrt{6} b^2 d^2 \\
    2 a c^3 & c^2 (3ad+bc) & \sqrt{6} cd (ad+bc) & d^2 (ad+3bc) & 2 b d^3 \\
    c^4 & 2 c^3 d & \sqrt{6} c^2 d^2 & 2 c d^3 & d^4 \\
  }.
\end{align*}
Since $T^1(M) = M$, this is also known as the \emph{defining representation}.

\subsection{Properties}\label{apx:Properties}

By construction, the map $T^\ell$ defined by \cref{eq:tjk} is a representation of $\L(\C^2)$, meaning that $T^\ell(M_1 M_2) = T^\ell(M_1) T^\ell(M_2)$ for all $M_1, M_2 \in \L(\C^2)$.
Let us denote its dimension by
\begin{equation}
  m_\ell := \ell+1
\end{equation}
and briefly discuss several of its properties.

First, notice from \cref{eq:tjk} that all matrix entries of $T^\ell(M)$ are homogeneous degree-$\ell$ polynomials in the entries of $M$.
Hence, $T^\ell(tM) = t^\ell T(M)$, for any $t \in \C$.

Next, if we assume that $M$ is diagonal, say $M = \smx{a&0\\0&d}$, then $\psi_M(x,y) = \psi(ax,dy)$ by \cref{eq:psiMxy}.
According to \cref{eq:psk}, $\psi^\ell_k(ax,dy) = a^{\ell-k} d^k \psi^\ell_k(x,y)$, so
\begin{equation}
  T^\ell\sof*{\mx{a&0\\0&d}} = \sum_{k=0}^\ell a^{\ell-k} d^k \proj{k}
  \label{eq:Tdiagonal}
\end{equation}
is diagonal as well.
In the special cases when $a = d = 1$ or when $a = 1$ and $d = 0$ we get
\begin{align}
  T^\ell(\id_2) &= \id_{m_\ell}, &
  T^\ell\of[\big]{\proj{0}_2} &= \proj{0}_{m_\ell},
  \label{eq:T(0)}
\end{align}
respectively, where the subscripts indicate dimensions and we use the convention that $0^0 = 1$.

Next, note that \cref{eq:tjk} is invariant under simultaneously exchanging $j \leftrightarrow k$ and $b \leftrightarrow c$:
\begin{equation}
  t^\ell_{jk}\sof*{\mx{a&b\\c&d}} =
  t^\ell_{kj}\sof*{\mx{a&c\\b&d}},
\end{equation}
which is ultimately thanks to our choice of normalization in \cref{eq:psk}.
This property is equivalent to $T^\ell(M\tp) = T^\ell(M)\tp$.
Since $T^\ell(\overline{M}) = \overline{T^\ell(M)}$, we also get $T^\ell(M\ct) = T^\ell(M)\ct$.
By combining these observations we conclude that if $M$ is unitary, \ie, $M M\ct = \id_2$, then $T^\ell(M) T^\ell(M)\ct = T^\ell(M M\ct) = T^\ell(\id_2) = \id_{m_\ell}$, meaning that $T^\ell(M)$ is also unitary.
In other words, not only is $T^\ell$ a representation of $\U{2}$ and $\SU{2}$, it is also a \emph{unitary} representation.

Combining unitarity with \cref{eq:Tdiagonal} we conclude that if $M$ has singular values $\set{a,d}$, \ie, $M = U \diag(a,d) V\ct$ for some unitaries $U,V \in \U{2}$, then
\begin{equation}
  T^\ell(M) = T^\ell(U) \diag \of{a^{\ell-k}d^k : k \in [m_\ell]} T^\ell(V\ct)
\end{equation}
has singular values $\set{a^{\ell-k}d^k : k \in [m_\ell]}$.
In particular, $T^\ell(\ketbra{\psi}{\varphi})$ is a rank-$1$ matrix for any $\ket{\psi}, \ket{\varphi} \in \C^2$.

The following proposition provides a simple formula for the image of a pure state $\proj{\psi}$ under $T^\ell$.
In physics, the resulting state $\proj{\psi^\ell}$ is known as \emph{coherent spin state} \cite{Radcliffe,Arecchi}.

\begin{prop}[Image of a pure state]\label{prop:Tl of psi}
Let $\ket{\psi} = \smx{a\\c} \in \C^2$.
Then, for any $\ell \geq 0$,
\begin{equation}
  T^\ell\of[\big]{\proj{\psi}}
  = \proj{\psi^\ell},
  \qquad \text{where} \qquad
  \ket{\psi^\ell}
  := \sum_{k=0}^\ell \psi^\ell_k(a,c) \ket{k} \in \C^{m_\ell}
  \label{eq:psil}
\end{equation}
and $\psi^\ell_k$ is the function from \cref{eq:psk}.
Moreover, $\ket{\psi^\ell}$ is a unit vector if and only if $\ket{\psi}$ is.
\end{prop}

\begin{proof}
Let $M = \smx{a&b\\c&d}$ so that $M \ket{0} = \smx{a\\c} = \ket{\psi}$. Then using \cref{eq:T(0)} we get
\begin{equation}
  T^\ell\of[\big]{\proj{\psi}}
  = T^\ell\of[\big]{M \proj{0} M\ct}
  = T^\ell(M) \proj{0}_{m_\ell} T^\ell(M)\ct,
\end{equation}
so it only remains to show that $T^\ell(M) \ket{0}_{m_\ell} = \ket{\psi^\ell}_{m_\ell}$.
According to the definition of $T^\ell$ from \cref{eq:psiMxy,eq:TsM}, $T^\ell(M)$ acts on the monomial $\psi^\ell_0(x,y)$ as follows:
\begin{align}
  T^\ell(M) \psi^\ell_0(x,y)
 &= \psi^\ell_0(ax+cy,bx+dy) \\
 &= (ax+cy)^\ell \\
 &= \sum_{k=0}^\ell
    \binom{\ell}{k}
    (ax)^{\ell-k} (cy)^k \\
 &= \sum_{k=0}^\ell
    \psi^\ell_k(a,c)
    \psi^\ell_k(x,y),
\end{align}
where we used the definition of $\psi^\ell_k(x,y)$ from \cref{eq:psk} several times.
This identity is equivalent to $T^\ell(M) \ket{0}_{m_\ell} = \sum_{k=0}^\ell \psi^\ell_k(a,c) \ket{k}_{m_\ell} = \ket{\psi^\ell}$, which is what we had to show.
Finally, note that
\begin{equation}
  \braket{\psi^\ell}{\psi^\ell}
  = \sum_{k=0}^\ell
    \binom{\ell}{k}
    \of{\abs{a}^2}^{\ell-k}
    \of{\abs{c}^2}^k
    = \of*{\abs{a}^2 + \abs{c}^2}^\ell
    = \of{\braket{\psi}{\psi}}^\ell,
\end{equation}
meaning that $\ket{\psi^\ell}$ is normalized if and only if $\ket{\psi}$ is.
\end{proof}

\subsection{Irreducibility}

Since the representation $T^\ell$ has dimension $m_\ell = \ell+1$, it is clear that the $T^\ell$ are distinct for different values of $\ell$.
However, it is not obvious that they are irreducible, \ie, cannot be expressed as a direct sum.
For the sake of completeness, we provide a proof of this standard fact below.
For more details, see Proposition~8.2.3 on p.~285 of \cite{Wawrzynczyk} or Proposition~5.1 on p.~85 of \cite{BrockerDieck}.

\begin{lemma}\label{lem:irreducibility}
For any integer $\ell \geq 0$, the $\SU{2}$ representation $T^\ell$ defined by \cref{eq:tjk} is irreducible.
\end{lemma}

\newcommand{\A}{\mathscr{A}}

\begin{proof}
Fix the value of $\ell$ and consider the matrix algebra $\A^\ell$ generated by $T^\ell(U)$ as $U$ ranges over $\SU{2}$, \ie,
\begin{equation}
  \A^\ell := \langle T^\ell(U) : U \in \SU{2} \rangle.
\end{equation}
In particular, $\A^\ell$ also includes $T^\ell(U)$ for infinitesimally small rotations $U$.
To conclude that $T^\ell$ is irreducible, it suffices to show that $\A^\ell$ is the full matrix algebra $\L(\C^{m_\ell})$.
We will show that $\A^\ell = \L(\C^{m_\ell})$ by explicitly expressing each elementary matrix $\ketbra{j}{k}$ where $j,k \in \set{0,\dotsc,\ell}$ using the elements of $\A^\ell$.

\newcommand{\sx}{\sin \frac{t}{2}}
\newcommand{\cx}{\cos \frac{t}{2}}

Letting $\omega := \exp(\frac{2 \pi i}{m_\ell})$, note from \cref{eq:Tdiagonal} that,
for any $j \in \set{0,\dotsc,\ell}$,
\begin{align}
  \frac{1}{m_\ell} \sum_{m=0}^\ell \omega^{-m j} T^\ell\sof*{\mx{1&0\\0&\omega^m}}
  &= \frac{1}{m_\ell} \sum_{m=0}^\ell \omega^{-m j} \sum_{k=0}^\ell \omega^{mk} \proj{k} \\
  &= \sum_{k=0}^\ell \of*{\frac{1}{m_\ell} \sum_{m=0}^\ell \omega^{m(k-j)}} \proj{k} \\
  &= \proj{j},
\end{align}
which constitutes all diagonal elements of $\L(\C^{m_\ell})$.

An explicit construction of off-diagonal elements of $\L(\C^{m_\ell})$ is more involved (indeed, diagonal elements are a special case of this).
Let $X := \smx{0&1\\1&0}$ and $Y := \smx{0&-i\\i&0}$ denote the Pauli matrices and let
\begin{align}
  R_X(t)
 &:= \exp \of*{-i \tfrac{t}{2} X}
   = \mx{\cx & -i\sx \\ -i\sx & \cx}, \\
  R_Y(t)
 &:= \exp \of*{-i \tfrac{t}{2} Y}
   = \mx{\cx & -\sx \\ \sx & \cx}
\end{align}
denote the $\SU{2}$ elements that rotate the Bloch sphere by angle $t$ around the $x$ and $y$ axis, respectively.
Let us denote the infinitesimal versions of their $T^\ell$-representations by
\begin{align}
  X^\ell &:= \lim_{t \to 0} \frac{d}{dt} T^\ell(R_X(t)), &
  Y^\ell &:= \lim_{t \to 0} \frac{d}{dt} T^\ell(R_Y(t)).
\end{align}
Since $T^\ell(R_X(0)) = T^\ell(\id_2) = \id_{m_\ell}$, from the definition of a derivative we get that
\begin{equation}
  X^\ell = \lim_{\varepsilon \to 0} \frac{T^\ell(R_X(\varepsilon))-\id_{m_\ell}}{\varepsilon},
\end{equation}
and similarly for $Y^\ell$. This shows that $X^\ell$ and $Y^\ell$ are limits of elements that belong to the algebra $\A^\ell$ generated by $T^\ell(U)$.
Since $\A^\ell$ is a linear subspace of $\L(\C^{m_\ell})$, it is closed and thus $X^\ell$ and $Y^\ell$ themselves belong to $\A^\ell$.

We can use $X^\ell$ and $Y^\ell$ to generate more elements that belong to $\A^\ell$.
In particular, let
\begin{equation}
  E_{\pm}^\ell := i X^\ell \pm Y^\ell.
\end{equation}
Using \cref{eq:tjk}, one can obtain the following explicit formula for their matrix entries:
\begin{equation}
  E_+^\ell
  = (E_-^\ell)\tp
  = \sum_{k=0}^{\ell-1} \sqrt{(\ell-k)(k+1)} \; \ketbra{k+1}{k},
\end{equation}
\cf, eq.~(8.4.6) on p.~296 of \cite{Wawrzynczyk}.
In other words, $E_\pm^\ell \ket{k} \propto \ket{k \pm 1}$ whenever $k \pm 1 \in \set{0,\dotsc,\ell}$, and $E_\pm^\ell \ket{k} = 0$ otherwise.
Hence, up to a constant factor, $E_\pm^\ell$ corresponds to shifting $\ket{k}$ by one in the corresponding direction.

For any $j \in \set{1,\dotsc,\ell}$, let us denote the $j$-th matrix power of $E_\pm^\ell$ (with an appropriate normalization) by
\begin{equation}
  S_\pm^\ell(j) := \sqrt{\frac{(\ell-j)!}{\ell!j!}} (E_\pm^\ell)^j.
  \label{eq:Sdef}
\end{equation}
Since $E_\pm^\ell$ are singular, we also set $S_\pm^\ell(0) := \id_{m_\ell}$ by hand.
Then $S_\pm^\ell(j) \ket{k} \propto \ket{k \pm j}$ whenever $k \pm j \in \set{0,\dotsc,\ell}$, and $S_\pm^\ell(j) \ket{k} = 0$ otherwise. In other words, $S_\pm^\ell(j)$ performs a shift by $\pm j$, up to an overall constant factor.
The normalization in \cref{eq:Sdef} is chosen so that, for any $j \in \set{0,\dotsc,\ell}$,
\begin{align}
  S_+^\ell(j) \ket{0   } &= \ket{j}, &
  S_-^\ell(j) \ket{\ell} &= \ket{\ell-j},
  \label{eq:Sj}
\end{align}
\ie, no additional factors are incurred when shifting the first and the last vector, respectively.
By transposing the first formula and noting that $\smash[b]{S_+^\ell(j)}\tp = S_-^\ell(j)$, we also get
\begin{equation}
  \bra{0} S_-^\ell(j) = \bra{j}.
  \label{eq:S0}
\end{equation}
Finally, note that the shifts by $\ell$ have the following particularly simple matrix form:
\begin{equation}
  S_+^\ell(\ell)
  = \smash[b]{S_-^\ell(\ell)}\tp
  = \ketbra{\ell}{0}.
  \label{eq:Sl}
\end{equation}
By combining \cref{eq:Sj,eq:S0,eq:Sl} we can express any elementary matrix as follows:
\begin{equation}
  S_-^\ell(\ell-i)
  S_+^\ell(\ell)
  S_-^\ell(j)
  = \ketbra{i}{j},
\end{equation}
where $i,j \in \set{0,\dotsc,\ell}$ are arbitrary.
This implies that $\A^\ell = \L(\C^{m_\ell})$, \ie, $\A^\ell$ is the whole matrix algebra, and hence $T^\ell$ is irreducible.
\end{proof}

\subsection{Completeness}

It turns out that every irreducible representation of $\SU{2}$ is equivalent to one of the representations $T^\ell$ constructed in \cref{apx:Construction}.

\begin{theorem}
Any irreducible representation of $\SU{2}$ is equivalent to $T^\ell$, for some integer $\ell \geq 0$.
\end{theorem}
\noindent
For a proof, see Theorem~3 on p.~93 of \cite{Zelobenko},
Corollary~8.6.2 on p.~301 of \cite{Wawrzynczyk},
or Proposition~5.3 on p.~86 of \cite{BrockerDieck}.
This constitutes a complete characterization of all irreducible representations of $\SU{2}$.

\subsection{Representations of $\U{2}$}\label{apx:U2}

Let us start with $1$-dimensional representations of $\U{2}$.
Any matrix in $\U{2}$ can be written as $e^{i\varphi} U$, for some $\varphi \in \R$ and $U \in \SU{2}$. Since $e^{i\varphi} \in \U{1}$, it is useful to first understand the irreducible representations of $\U{1}$. They correspond to non-zero rational functions $R: \C \to \C$ such that $R(a) R(b) = R(ab)$, for all $a,b \in \C^\times$. Any such function is of the form $R(a) := a^r$,
for some $r \in \Z$, which exhausts all irreducible representations of $\U{1}$.
If $D: \L(\C^2) \to \C$ is a $1$-dimensional representation of $\U{2}$ then so is $R \circ D$.
In particular, the determinant is a $1$-dimensional representation since $\det(MN) = \det(M) \det(N)$. Composing it with $R$ we get
$D^r(M) := (\det M)^r$
where $r \in \Z$. This exhausts all $1$-dimensional (and thus irreducible) representations of $\U{2}$.

Any representation can be entry-wise multiplied by a $1$-dimensional representation to obtain a new representation. In particular, we can multiply any of the representations $T^\ell$ constructed in \cref{apx:Construction} by $D^r$ to obtain
\begin{equation}
  T^{r,\ell}(M) := (\det M)^r T^\ell(M).
  \label{eq:Trl}
\end{equation}
For each $\ell \geq 0$, this exhausts all $(\ell+1)$-dimensional irreducible representations of $\U{2}$ as we range $r$ over $\Z$.
Consequently, any irreducible representation of $\U{2}$ is of the form $T^{r,\ell}$, for some integers $r \in \Z$ and $\ell \geq 0$ \cite{Stembridge,Stanley}.
Note from \cref{eq:tjk} that all entries of $T^{r,\ell}(M)$ are homogeneous polynomials (or rational functions when $r < 0$) of degree $2r + \ell$ in the entries of $M$.

It is useful to establish a correspondence between the parameters $r,\ell$ of the representation $T^{r,\ell}$ and partitions $\lambda = (\lambda_1,\lambda_2)$ defined in \cref{eq:pt}. When $r \geq 0$, the two are related as follows:
\begin{equation}
  (\lambda_1, \lambda_2)
  = (r+\ell,r).
  \label{eq:lambda vs lr}
\end{equation}
Note that $\lambda$ is a partition of $n = \lambda_1 + \lambda_2 = 2r + \ell$, which is also the total number of boxes in the corresponding Young diagram in \cref{fig:lambda}.

Given $\lambda \pt n$, we can recover $r = \lambda_2$ and $\ell = \lambda_1 - \lambda_2$ via \cref{eq:lambda vs lr}, so it is natural to associate to $\lambda$ the representation $T^{\lambda_2,\lambda_1-\lambda_2}$. We denote this representation by
\begin{equation}
  Q_{\lambda}(M)
  := (\det M)^{\lambda_2}
     T^{\lambda_1-\lambda_2}(M),
  \label{eq:Q-lambda}
\end{equation}
where the entries of $T^{\lambda_1-\lambda_2}(M)$ are given in \cref{eq:tjk}. These are precisely the representations that appear in the decomposition of $M\xp{n}$ in Schur basis, see \cref{eq:QP} in \cref{sec:SchurWeyl}.
Note that representations $T^{r,\ell}$ with $r < 0$ do not appear in this decomposition since the corresponding $\lambda$ is not a partition.

The matrix entries of $Q_\lambda(M)$ are homogeneous polynomials of degree $n = \lambda_1 + \lambda_2$ in the entries of $M$. In particular, they can be evaluated for any $M \in \L(\C^2)$, so $Q_\lambda$ is not only an irreducible representation of $\U{2}$ but also of $\L(\C^2)$ as a monoid under matrix multiplication~\cite{YoungDiagrams}.
Note from \cref{eq:Q-lambda} that adding or removing a full column of boxes from the diagram of $\lambda$ corresponds to multiplying or dividing the corresponding representation by $\det(M)$.
This observation extends also to $\U{d}$, for any $d \geq 1$ \cite{Stembridge,Stanley}.

\subsection{Clebsch--Gordan transform}

Recall from \cref{apx:Properties} that the matrix entries of $T^\ell(M)$ are homogeneous degree-$\ell$ polynomials in the entries of $M$.
Hence, the entries of $M \x T^\ell(M)$ are also homogeneous but of degree $\ell+1$.
It turns out that they can be expressed as linear combinations of the entries of $T^{1,\ell-1}(M)$ and $T^{0,\ell+1}(M)$, which are also polynomials of degree $\ell+1$.
The relationship between the two is expressed via Clebsch--Gordan transform, see \cite[p.~165]{Bohm}, \cite[p.~174]{Vilenkin}, \cite[p.~498]{VK1}, \cite[p.~304]{Wawrzynczyk}.

For any integer $\ell \geq 0$, the \emph{Clebsch--Gordan transform} $\CG_\ell \in \U{2\ell+2}$ is a map whose input and output spaces are interpreted as follows:
\begin{equation}
  \CG_\ell: \C^2 \x \C^{\ell+1} \to \C^\ell \+ \C^{\ell+2}.
  \label{eq:Cdims}
\end{equation}
Its defining property is that, for all $M \in \L(\C^2)$,
\begin{equation}
  \CG_\ell \of[\big]{M \x T^\ell(M)} \CG_\ell\ct
  = T^{1,\ell-1}(M) \+ T^{0,\ell+1}(M).
  \label{eq:CG decomposition}
\end{equation}
Equivalently, in terms of $\lambda$ and $Q_\lambda$,
\begin{equation}
  \CG_{\lambda_1-\lambda_2} \of[\big]{M \x Q_{(\lambda_1,\lambda_2)}(M)} \CG_{\lambda_1-\lambda_2}\ct
  = Q_{(\lambda_1,\lambda_2+1)}(M)
 \+ Q_{(\lambda_1+1,\lambda_2)}(M),
  \label{eq:CG decomposition for Q}
\end{equation}
which holds for any $\lambda \pt n$.

Up to certain phases, such map $\CG_\ell$ is essentially unique.
If we let
\begin{equation}
  R^\ell(k) := \frac{1}{\sqrt{\ell+1}}
  \mx{
    \sqrt{k} & \sqrt{\ell+1-k} \\
   -\sqrt{\ell+1-k} & \sqrt{k}
  }
  \label{eq:Rlk}
\end{equation}
and denote its matrix entries by $R^\ell_{ab}(k)$ where $a,b \in \set{0,1}$, we can express $\CG_\ell$ as follows:
\begin{equation}
  \CG_\ell :=
  \sum_{i=0}^1
  \sum_{j=0}^\ell
  \of[\Big]{
    R^\ell_{i0}(j+i) \ket{j+i-1}_\ell \+
    R^\ell_{i1}(j+i) \ket{j+i}_{\ell+2}
  }
  \bra{i}_2 \x \bra{j}_{\ell+1},
  \label{eq:Cl}
\end{equation}
where the subscripts indicate the dimension of each state,
\cf~\cref{eq:Cdims}.
Since $R^\ell_{00}(0) = 0$, the term with $\ket{i+j-1}_\ell$ vanishes when $i = j = 0$, so we do not need to worry about interpreting the resulting state $\ket{-1}_\ell$.
Similarly, $R^\ell_{01}(\ell+1) = 0$, so we can also ignore $\ket{\ell}_\ell$.
It follows from $R^\ell(k) R^\ell(k)\tp = \id_2$ that $\CG_\ell \CG_\ell\tp = \id_{2\ell+2}$, meaning that $\CG_\ell$ is orthogonal and hence unitary.

This formula can be obtained from the known expressions of Clebsch--Gordan coefficients in the special case when one of the representations is the defining (\ie, $\ell = 1$) representation,
see \cite[p.~173]{Bohm}, \cite[p.~187]{Vilenkin}, \cite[p.~512]{VK1}.
We have adapted them to the case when the defining representation is in the first register instead of the second, and for using $\set{0,\dotsc,\ell}$ instead of $\set{-\ell/2,\dotsc,\ell/2}$ to label the standard basis vectors of $\C^{\ell+1}$.
Finally, we are using a slightly different phase convention in \cref{eq:Rlk} to match with our derivations in \cref{apx:SymCG}.

For completeness, in \cref{apx:SymCG} we provide a self-contained derivation of Clebsch--Gordan transform based on a decomposition of $\C^2 \x \Sym{\ell}$.

\subsection{Dual Clebsch--Gordan transform}\label{apx:Dual CG}

\newcommand{\iY}{\mx{0&1\\-1&0}}
\newcommand{\miY}{\mx{0&-1\\1&0}}

Let us denote the \emph{dual} of the defining representation $T^1(M) = M$ by
\begin{equation}
  M^* := (M^{-1})\tp.
\end{equation}
Note that if $U$ is unitary then $U^{-1} = U\ct$ so $U^* = (U\ct)\tp = \overline{U}$.
The \emph{dual} Clebsch--Gordan transformation $\DCG_\ell$ provides a decomposition of $M^* \x T^\ell(M)$ analogous to \cref{eq:CG decomposition}.

Let us derive a formula for $\DCG_\ell$ in terms of $\CG_\ell$. First, we can relate $M^*$ and $M$ by noting that
\begin{equation}
  \sof[\Bigg]{\ABCD^*}\tp
  = \ABCD^{-1}
  = \frac{1}{ad-bc} \mx{d&-b\\-c&a}.
\end{equation}
Hence, for any invertible $M \in \L(\C^2)$,
\begin{equation}
  \det(M) \, M^*
  = \miY M \miY\ct.
\end{equation}
This allows expressing $M^* \x T^\ell(M)$ in terms of $M \x T^\ell(M)$:
\begin{equation}
  \det(M)
  \of[\big]{M^* \x T^\ell(M)}
= \sof*{\miY \x \id_{\ell+1}} \cdot
  \of[\big]{M \x T^\ell(M)} \cdot
  \sof*{\miY \x \id_{\ell+1}}\ct.
\end{equation}

To obtain the desired decomposition of $M^* \x T^\ell(M)$, we can first undo the undesired basis change and then apply the regular Clebsch--Gordan transform $\CG_\ell$.
In other words, let
\begin{equation}
  \DCG_\ell := \CG_\ell \cdot \sof*{\iY \x \id_{\ell+1}},
  \label{eq:Cl dual and Cl}
\end{equation}
where the $2 \times 2$ matrix is the same as $S$ in \cref{eq:S} and hence satisfies $S (-S) = \id_2$.
Then, according to \cref{eq:CG decomposition},
\begin{equation}
  \DCG_\ell \det(M) \of[\big]{M^* \x T^\ell(M)} \DCGct_\ell
  = T^{1,\ell-1}(M) \+ T^{0,\ell+1}(M).
\end{equation}
After cancelling $\det(M)$ from both sides, we get the following dual of \cref{eq:CG decomposition}:
\begin{equation}
  \DCG_\ell \of[\big]{M^* \x T^\ell(M)} \DCGct_\ell
  = T^{0,\ell-1}(M) \+ T^{-1,\ell+1}(M),
  \label{eq:dual CG decomposition}
\end{equation}
which justifies calling $\DCG_\ell$ the \emph{dual Clebsch--Gordan transform}. Equivalently, similar to \cref{eq:CG decomposition for Q}, we can also express this in terms of $\lambda$ and $Q_\lambda$:
\begin{equation}
  \DCG_{\lambda_1-\lambda_2} \of[\big]{M^* \x Q_{(\lambda_1,\lambda_2)}(M)} \DCGct_{\lambda_1-\lambda_2}
  = Q_{(\lambda_1-1,\lambda_2)}(M)
 \+ Q_{(\lambda_1,\lambda_2-1)}(M),
  \label{eq:dual CG decomposition for Q}
\end{equation}
which holds for any $\lambda \pt n$.
Noting that, for any $i \in \set{0,1}$,
\begin{equation}
  \bra{i} \iY = (-1)^i \bra{1-i},
\end{equation}
we can use \cref{eq:Cl dual and Cl} to modify \cref{eq:Cl} and obtain the following explicit formula for $\DCG_\ell$:
\begin{equation}
  \DCG_\ell =
  \sum_{i=0}^1
  \sum_{j=0}^\ell
  \of[\Big]{
    R^\ell_{i0}(j+i) \ket{j+i-1}_\ell \+
    R^\ell_{i1}(j+i) \ket{j+i}_{\ell+2}
  }
  (-1)^i \bra{1-i}_2 \x \bra{j}_{\ell+1}.
  \label{eq:DualCl}
\end{equation}

\section{Symmetric subspace}\label{apx:Sym}

In this appendix, we introduce the $\ell$-qubit symmetric subspace and discuss its connections with irreducible representations of $\SU{2}$, Clebsch--Gordan transform, and the two extremal unitary-covariant channels $\Phi^\ell_1$ and $\Phi^\ell_2$ introduced in \cref{sec:Extremal} which correspond to partial trace and universal $\NOT$.

For any $\ell \geq 0$, we denote the $\ell$-qubit \emph{symmetric subspace} by
\begin{equation}
  \Sym{\ell} := \spn \set[\big]{\ket{s_\ell(w)} : w \in \set{0,\dotsc,\ell}}
  \subseteq \of{\C^2}\xp{\ell},
\end{equation}
where $\ket{s_\ell(w)}$ denotes the $\ell$-qubit symmetric state of Hamming weight $w$, see \cref{eq:s}.

\subsection{Projector onto the symmetric subspace}\label{apx:Sym projector}

Since $\braket{s_\ell(w)}{s_\ell(w')} = \delta_{ww'}$, the set
$\set[\big]{\ket{s_\ell(w)} : w \in \set{0,\dotsc,\ell}}$
is in fact an orthonormal basis for the $\ell$-qubit symmetric subspace, hence the projector onto this subspace is given by
\begin{equation}
  \Pi_{\Sym{\ell}}
  := \sum_{w=0}^\ell \proj{s_\ell(w)}.
  \label{eq:PiSym}
\end{equation}
There are several equivalent ways of writing this projector.

\begin{lemma}\label{lem:PiSym3}
For any $\ell \geq 0$, the following three expressions for the projector $\Pi_{\Sym{\ell}}$ onto the $\ell$-qubit symmetric subspace are equivalent:
\begin{align}
  \Pi_{\Sym{\ell}}
:= \sum_{w=0}^\ell \proj{s_\ell(w)}
 = \frac{1}{\ell!} \sum_{\pi \in \S_\ell} P(\pi)
 = (\ell+1) \int \of[\big]{\proj{\psi}}\xp{\ell} \, d\psi,
\end{align}
where $\ket{s_\ell(w)}$ are the symmetric states from \cref{eq:s},
$P(\pi) \in \U{2^\ell}$ are the $\ell$-qubit permutations defined in \cref{eq:P},
and $d\psi$ is the uniform measure on the Bloch sphere, see \cref{eq:dpsi}.
\end{lemma}

\begin{proof}
Let us denote the above expressions by $\Pi_{\Sym{\ell}}$, $\Pi_{\Sym{\ell}}'$, and $\Pi_{\Sym{\ell}}''$, respectively.
We will evaluate their matrix entries and show that they are equal.

Let us start with $\Pi_{\Sym{\ell}}$.
Note from \cref{eq:s} that, for any $\ham,w \in \set{0,\dotsc,\ell}$,
\begin{equation}
  \braket{0^{\ell-\ham}1^{\ham}}{s_\ell(w)}
  = \delta_{\ham w} \binom{\ell}{w}^{-1/2}.
\end{equation}
More generally, for any $x,y \in \set{0,1}^\ell$,
\begin{equation}
  \braket{x}{s_\ell(w)}
  \braket{s_\ell(w)}{y}
  = \delta_{|x| w} \delta_{|y| w} \binom{\ell}{w}^{-1},
\end{equation}
so we conclude that
\begin{equation}
  \bra{x}
  \Pi_{\Sym{\ell}}
  \ket{y}
= \sum_{w=0}^\ell
  \braket{x}{s_\ell(w)}
  \braket{s_\ell(w)}{y}
= \delta_{|x||y|} \binom{\ell}{|x|}^{-1}.
  \label{eq:Symxy1}
\end{equation}

To evaluate the matrix entries of the second expression, note that
\begin{equation}
  \bra{x}
  \Pi_{\Sym{\ell}}'
  \ket{y}
= \bra{x}
  \frac{1}{\ell!} \sum_{\pi \in \S_\ell} P(\pi)
  \ket{y}
= \delta_{|x||y|} \frac{(\ell-|x|)!|x|!}{\ell!}
= \delta_{|x||y|} \binom{\ell}{|x|}^{-1}.
  \label{eq:Symxy2}
\end{equation}
The second equality holds since permutations preserve Hamming weight, and there are $|x|!$ ways to map the ones of $x$ to those of $y$, and $(\ell-|x|)!$ ways to map the zeroes of $x$ to those of $y$.

For the third expression, let $\ket{\psi} = \smx{\alpha \\ \beta} \in \C^2$ and note that
\begin{equation}
  \bra{x} \cdot \ket{\psi}\xp{\ell}
  = \alpha^{\ell-|x|} \beta^{|x|}
\end{equation}
for any $x \in \set{0,1}^\ell$.
Hence,
\begin{equation}
  \bra{x}
  \Pi_{\Sym{\ell}}''
  \ket{y}
= (\ell+1) \int
  \bra{x} \cdot \ket{\psi}\xp{\ell} \,
  \bra{\psi}\xp{\ell} \cdot \ket{y}
  \, d\psi
= (\ell+1) \int
  \alpha^{\ell-|x|} \beta^{|x|}
  \bar\alpha^{\ell-|y|} \bar\beta^{|y|}
  \, d\psi.
\end{equation}
Recall from \cref{eq:psi angles} that we can take
$\alpha := \cos \frac{\theta}{2}$ and
$\beta := e^{i \varphi} \sin \frac{\theta}{2}$,
and from \cref{eq:dpsi} that the uniform measure on the Bloch sphere is
$d\psi := \frac{1}{4\pi} \sin \theta \, d\theta \, d\varphi$.
Putting everything together,
\begin{equation}
  \bra{x}
  \Pi_{\Sym{\ell}}''
  \ket{y}
= \frac{\ell+1}{4\pi}
  \int_{\theta=0}^{\pi}
  \of*{\cos \frac{\theta}{2}}^{2\ell-|x|-|y|}
  \of*{\sin \frac{\theta}{2}}^{|x|+|y|}
  \sin \theta \, d\theta
  \int_{\varphi=0}^{2\pi}
  e^{i\varphi(|x|-|y|)} \, d\varphi.
\end{equation}
The second integral evaluates to $2 \pi \delta_{|x||y|}$, so we are left with
\begin{equation}
  \bra{x}
  \Pi_{\Sym{\ell}}''
  \ket{y}
= \delta_{|x||y|}
  \frac{\ell+1}{2}
  \int_{\theta=0}^{\pi}
  \of*{\cos \frac{\theta}{2}}^{2\ell-2|x|}
  \of*{\sin \frac{\theta}{2}}^{2|x|}
  \sin \theta \, d\theta
= \delta_{|x||y|}
  \binom{\ell}{|x|}^{-1}.
  \label{eq:Symxy3}
\end{equation}
We are done since the matrix entries in \cref{eq:Symxy1,eq:Symxy2,eq:Symxy3} agree for all $x,y \in \set{0,1}^\ell$.
\end{proof}

\subsection{Symmetric subspace and irreducible representations of $\SU{2}$}

Here we provide an alternative way of deriving the irreducible representations $T^\ell$ of $\SU{2}$ described in \cref{apx:Construction}. The following lemma shows that the matrix entries of $T^\ell(M)$, for any $M \in \L(\C^2)$, can be obtained by restricting the action of $M\xp{\ell}$ to the $\ell$-qubit symmetric subspace.

\begin{lemma}\label{lem:tjk from symmetric subspace}
For any integer $\ell \geq 0$, indices $j,k \in \set{0,\dotsc,\ell}$, and matrix $M \in \L(\C^2)$,
\begin{align}
  M\xp{\ell} \ket{s_\ell(k)}
  &= \sum_{j=0}^{\ell} t^\ell_{jk}(M) \ket{s_\ell(j)}, &
  \bra{s_\ell(j)} M\xp{\ell} \ket{s_\ell(k)} &= t^\ell_{jk}(M),
\end{align}
where $t^\ell_{jk}(M)$ are the matrix entries of $T^\ell(M)$, see \cref{eq:tjk}, and $\ket{s_\ell(w)}$ is the symmetric state of Hamming weight $w$, see \cref{eq:s}.
\end{lemma}

\begin{proof}
Since $\braket{s_\ell(j)}{s_\ell(k)} = \delta_{jk}$, we can obtain the second identity from the first one by multiplying both sides with $\bra{s_\ell(j)}$.

To prove the first identity, first note that we can express the sum of all $\ell$-qubit standard basis states of Hamming weight $k \in \set{0,\dotsc,\ell}$ as follows:
\begin{align}
  \sum_{\substack{x\in\set{0,1}^\ell\\\abs{x}=k}} \ket{x}
 &= \frac{1}{(\ell-k)!k!}
    \sum_{\pi \in \S_\ell} P(\pi)
    \ket{0}\xp{\ell-k}
    \ket{1}\xp{k}
  = \binom{\ell}{k} \;
    \Pi_{\Sym{\ell}} \of*{
      \ket{0}\xp{\ell-k}
      \ket{1}\xp{k}
    },
\end{align}
where $\Pi_{\Sym{\ell}}$ denotes the projector onto the symmetric subspace, see \cref{eq:PiSym}.
Adjusting the normalization, we can write the symmetric state $\ket{s_\ell(k)}$ of Hamming weight $k$ as follows:
\begin{equation}
  \ket{s_\ell(k)}
= \Pi_{\Sym{\ell}} \sof*{
    \binom{\ell}{k}^{1/2}
    \ket{0}\xp{\ell-k}
    \ket{1}\xp{k}
  }.
  \label{eq:slk in Sym}
\end{equation}

Note how the above equation is similar to \cref{eq:psk}. Here the standard basis vectors~$\ket{0}$ and~$\ket{1}$ play a similar role to the variables $x$ and $y$ of $\psi^\ell_k(x,y)$. To establish a more formal correspondence, consider the following vector, which is a function of variables $x,y \in \C$:
\begin{equation}
  \ket{v(x,y)}
 := \of[\big]{x \ket{0} + y \ket{1}}\xp{\ell}
  = \sum_{a \in \set{0,1}^\ell}
    x^{\ell-\abs{a}} y^{\abs{a}} \ket{a}.
\end{equation}
Since $\ket{v(x,y)} \in \Sym{\ell}$, we have
$\bra{v(x,y)} \Pi_{\Sym{\ell}} = \bra{v(x,y)}$.
Using this and \cref{eq:slk in Sym},
\begin{equation}
  \braket{v(x,y)}{s_\ell(k)}
  = \bra{v(x,y)}
    \binom{\ell}{k}^{1/2}
    \ket{0}\xp{\ell-k}
    \ket{1}\xp{k}
  = \binom{\ell}{k}^{1/2}
    x^{\ell-k}
    y^k
  = \psi^\ell_k(x,y),
\end{equation}
establishing a formal correspondence between the symmetric states $\ket{s_\ell(k)}$ and the functions $\psi^\ell_k$ defined in \cref{eq:psk}.

Let us now consider the action of $M\xp\ell$ on $\ket{s_\ell(k)}$. For concreteness, let $M = \abcd$. Then
$M \ket{0} = a \ket{0} + c \ket{1}$ and
$M \ket{1} = b \ket{0} + d \ket{1}$,
which is analogous to how $M$ acts on the arguments of functions $\psi$ in \cref{eq:psiMxy}.
Since $\Pi_{\Sym{\ell}}$ is a linear combination of permutations, see \cref{lem:PiSym3}, it commutes with $M\xp{\ell}$ and thus
\begin{equation}
  M\xp{\ell} \ket{s_\ell(k)}
= \Pi_{\Sym{\ell}} \sof*{
    \binom{\ell}{k}^{1/2}
    \of[\big]{a \ket{0} + c \ket{1}}\xp{\ell-k}
    \of[\big]{b \ket{0} + d \ket{1}}\xp{k}
  },
  \label{eq:Ml on slk}
\end{equation}
which is analogous to the left-hand side of \cref{eq:expansion}.

To obtain an analogue of the right-hand side, note that $M\xp\ell \ket{s_\ell(k)}$ is symmetric under permutations and thus $M\xp\ell \ket{s_\ell(k)} \in \Sym{\ell} = \spn \set[\big]{\ket{s_\ell(w)} : w \in \set{0,\dotsc,\ell}}$. Hence,
\begin{equation}
  M\xp\ell \ket{s_\ell(k)}
  = \sum_{j=0}^\ell
    \tilde{t}^\ell_{jk}(M)
    \ket{s_\ell(j)},
  \label{eq:Ml on slk again}
\end{equation}
for some coefficients $\tilde{t}^\ell_{jk}(M) \in \C$. Combining this with \cref{eq:slk in Sym,eq:Ml on slk},
\begin{align*}
  \Pi_{\Sym{\ell}} \sof*{
    \binom{\ell}{k}^{1/2}
    \of[\big]{a \ket{0} + c \ket{1}}\xp{\ell-k}
    \of[\big]{b \ket{0} + d \ket{1}}\xp{k}
  }
= \sum_{j=0}^\ell
  \tilde{t}^\ell_{jk}(M)
  \Pi_{\Sym{\ell}} \sof*{
    \binom{\ell}{j}^{1/2}
    \ket{0}\xp{\ell-j}
    \ket{1}\xp{j}
  }.
\end{align*}
Multiplying both sides by $\bra{v(x,y)}$, we recover \cref{eq:expansion}. Comparing the coefficients at $x^{\ell-j} y^j$ we conclude that $t^\ell_{jk}(M) = \tilde{t}^\ell_{jk}(M)$.
\end{proof}

\subsection{Symmetric subspace and tensor power states}\label{apx:Power states}

Fix any $\ell \geq 0$ and let $V^\ell \in \U{\C^{\ell+1},\of{\C^2}\xp{\ell}}$ denote the isometry from $\ell+1$ dimensions to the $\ell$-qubit symmetric subspace $\Sym{\ell}$:
\begin{align}
  V^\ell := \sum_{k=0}^\ell \ketbra{s_\ell(k)}{k}.
  \label{eq:Vl}
\end{align}
The following proposition shows that this isometry relates the states $\ket{\psi^\ell} \in \C^{\ell+1}$ defined in \cref{eq:psil} to the tensor power states $\ket{\psi}\xp{\ell} \in \of{\C^2}\xp{\ell}$.

\begin{prop}\label{prop:Vl}
For any integer $\ell \geq 0$ and any vector $\ket{\psi} \in \C^2$,
\begin{equation}
  V^\ell \ket{\psi^\ell}
  = \ket{\psi}\xp{\ell},
\end{equation}
where $\ket{\psi^\ell} \in \C^{\ell+1}$ is the vector defined in \cref{eq:psil}.
\end{prop}

\begin{proof}
For concreteness, let $\ket{\psi} = \smx{a\\c}$.
Recall from \cref{eq:psk,eq:psil} that
\begin{equation}
  \ket{\psi^\ell}
 := \sum_{k=0}^\ell
    \binom{\ell}{k}^{1/2}
    a^{\ell-k} c^k
    \ket{k}.
\end{equation}
By applying $V^\ell$ from \cref{eq:Vl} and substituting the definition of the symmetric state $\ket{s_\ell(k)}$ from \cref{eq:s}, we find that
\begin{align}
  V^\ell \ket{\psi^\ell}
 &= \sum_{k=0}^\ell
    \binom{\ell}{k}^{1/2}
    a^{\ell-k} c^k
    \ket{s_\ell(k)} \\
 &= \sum_{k=0}^\ell
    a^{\ell-k} c^k
    \sum_{\substack{x\in\set{0,1}^\ell\\\abs{x}=k}} \ket{x} \\
 &= \sum_{x\in\set{0,1}^\ell}
    a^{\ell-\abs{x}} c^{\abs{x}} \ket{x} \\
 &= \ket{\psi}\xp{\ell},
\end{align}
as desired.
\end{proof}

Let $\Sigma^\ell \in \CPTP{\C^{\ell+1}}{\of{\C^2}\xp{\ell}}$ denote the quantum channel corresponding to the isometry $V^\ell$:
\begin{equation}
  \Sigma^\ell(\rho)
  := V^\ell \rho V^{\ell\dagger},
  \label{eq:Sigmal}
\end{equation}
for all $\rho \in \L(\C^{\ell+1})$.
Then the states $\ket{\psi}$, $\ket{\psi^\ell}$, and $\ket{\psi}\xp{\ell}$ are related as follows.

\begin{cor}\label{cor:three states}
For any integer $\ell \geq 0$ and any vector $\ket{\psi} \in \C^2$,
\begin{equation}
  \proj{\psi}
  \overset{T^\ell}{\longmapsto}
  \proj{\psi^\ell}
  \overset{\Sigma^\ell}{\longmapsto}
  \of{\proj{\psi}}\xp{\ell},
\end{equation}
where $T^\ell$ and $\Sigma^\ell$ are the maps defined by \cref{eq:tjk,eq:Sigmal}, respectively.
\end{cor}

\begin{proof}
This follows by combining \cref{prop:Vl,prop:Tl of psi}.
\end{proof}

\subsection{Invariant subspaces of \texorpdfstring{$\C^2 \x \Sym{\ell}$}{C2 x Sym(l)}}\label{apx:Invariant subspaces}

Fix an integer $\ell \geq 1$ and let
\begin{equation}
  \mc{V} := \C^2 \x \Sym{\ell}
  = \C^2 \x \spn \set{\ket{s_\ell(w)} : w \in \set{0,\dotsc,\ell}}
  \label{eq:V}
\end{equation}
be the subspace of $\ell+1$ qubits where the first qubit is arbitrary but the last $\ell$ are symmetric. We would like to decompose this as a direct sum of the subspace where \emph{all} $\ell+1$ qubits are symmetric and its orthogonal complement.

Let us denote the symmetric subspace of $\ell+1$ qubits by
\begin{equation}
  \mc{V}_{\UNOT} := \Sym{\ell+1}
  = \spn \set{\ket{s_{\ell+1}(w)} : w \in \set{0,\dotsc,\ell+1}}.
  \label{eq:V-UNOT}
\end{equation}
Using the definition in \cref{eq:s}, we can recursively relate the symmetric states on $\ell+1$ qubits to those on $\ell$ qubits:
\begin{align}
  \ket{s_{\ell+1}(w+1)}
 &= \binom{\ell+1}{w+1}^{-1/2}
    \sof*{
      \binom{\ell}{w+1}^{1/2} \ket{0} \ket{s_\ell(w+1)}
    + \binom{\ell}{w  }^{1/2} \ket{1} \ket{s_\ell(w  )}
    } \\
 &= \sqrt{\frac{\ell-w}{\ell+1}} \ket{0} \ket{s_\ell(w+1)}
  + \sqrt{\frac{w   +1}{\ell+1}} \ket{1} \ket{s_\ell(w  )}.
  \label{eq:s-recursion}
\end{align}
This holds for $w \in \set{0,\dotsc,\ell-1}$, but the identity also makes sense when $w = -1$ or $w = \ell$ since the coefficients of the undefined terms $\ket{s_\ell(-1)}$ and $\ket{s_\ell(\ell+1)}$ vanish.

Let us define another set of unit vectors by simply exchanging the two coefficients in \cref{eq:s-recursion} and adding a minus sign to the second term:
\begin{equation}
  \ket{s'_{\ell-1}(w)}
 := \sqrt{\frac{w+1}{\ell+1}} \ket{0} \ket{s_\ell(w+1)}
  - \sqrt{\frac{\ell-w}{\ell+1}} \ket{1} \ket{s_\ell(w  )}
  \label{eq:s'}
\end{equation}
where $w \in \set{0,\dotsc,\ell-1}$ (we will justify the notation $\ket{s'_{\ell-1}(w)}$ in \cref{lem:s'}).
We denote the span of these vectors by
\begin{equation}
  \mc{V}_{\Tr} := \spn \set{\ket{s'_{\ell-1}(w)} : w \in \set{0,\dotsc,\ell-1}}.
  \label{eq:V-Tr}
\end{equation}

\begin{lemma}
The subspaces defined in \cref{eq:V,eq:V-UNOT,eq:V-Tr} satisfy
\begin{align}
  \dim \mc{V} &= 2(\ell+1), &
  \dim \mc{V}_{\Tr} &= \ell, &
  \dim \mc{V}_{\UNOT} &= \ell+2, &
  \mc{V} &= \mc{V}_{\Tr} \+ \mc{V}_{\UNOT}.
\end{align}
\end{lemma}

\begin{proof}
Since $\braket{s_\ell(w)}{s_\ell(v)} = \delta_{wv}$, it is evident from \cref{eq:V,eq:V-UNOT} that $\mc{V}$ and $\mc{V}_{\UNOT}$ have the claimed dimensions. Since all terms in $\ket{s'_{\ell-1}(w)}$ have Hamming weight $w+1$, we conclude that $\braket{s'_\ell(w)}{s'_\ell(v)} = \delta_{wv}$ and hence $\mc{V}_{\Tr}$ also has the claimed dimension.

It is evident from \cref{eq:s-recursion,eq:s'} that $\ket{s_{\ell+1}(w)}$ and $\ket{s'_{\ell-1}(w)}$ are in $\mc{V}$, so $\mc{V}_{\Tr}$ and $\mc{V}_{\UNOT}$ are subspaces of $\mc{V}$. By construction of $\ket{s'_{\ell-1}(w)}$,
\begin{equation}
  \braket{s_{\ell+1}(w+1)}{s'_{\ell-1}(w)} = 0,
\end{equation}
for any $w \in \set{0,\dotsc,\ell-1}$.
More generally, for any $w \in \set{0,\dotsc,\ell+1}$ and $v \in \set{0,\dotsc,\ell-1}$,
\begin{equation}
  \braket{s_{\ell+1}(w)}{s'_{\ell-1}(v)} = 0
\end{equation}
since the two states have different Hamming weights when $w \neq v+1$. This means the subspaces $\mc{V}_{\Tr}$ and $\mc{V}_{\UNOT}$ are orthogonal to each other. Since $\dim \mc{V}_{\Tr} + \dim \mc{V}_{\UNOT} = \ell + (\ell + 2) = \dim \mc{V}$, the union of $\mc{V}_{\Tr}$ and $\mc{V}_{\UNOT}$ constitutes the whole of $\mc{V}$ and we indeed get a direct sum.
\end{proof}

The following lemma provides an alternative formula for the vectors $\ket{s'_{\ell-1}(w)}$ defined in \cref{eq:s'}. They can be obtained by symmetrizing the $(\ell-1)$-qubit symmetric state $\ket{s_{\ell-1}(w)}$ and one qubit of the singlet state $\ket{\Psi^-}$. This also explains why $\ket{s'_{\ell-1}(w)}$ has the same parameters as a symmetric state, even though it is an $(\ell+1)$-qubit state of Hamming weight $w+1$.

\begin{lemma}\label{lem:s'}
For any integer $\ell \geq 1$ and $w \in \set{0,\dotsc,\ell-1}$,
\begin{equation}
  \ket{s'_{\ell-1}(w)}
= \sqrt{\frac{2\ell}{\ell+1}}
  \of[\big]{\id_2 \x \Pi_{\Sym{\ell}}}
  \ket{\Psi^-}
  \ket{s_{\ell-1}(w)},
  \label{eq:s'-alternative}
\end{equation}
where $\Pi_{\Sym{\ell}} := \frac{1}{\ell!} \sum_{\pi \in S_\ell} P(\pi)$ is the projector onto the $\ell$-qubit symmetric subspace, see \cref{eq:PiSym}.
\end{lemma}

\begin{proof}
Substituting the definition of the symmetric state $\ket{s_{\ell-1}(w)}$ from \cref{eq:s}, we can rewrite the right-hand side of \cref{eq:s'-alternative} as follows:
\begin{align}
  & \sqrt{\frac{2\ell}{\ell+1}} \cdot
    \frac{1}{\ell!}
    \sum_{\pi \in \S_\ell}
    (\id_2 \x P(\pi))
    \frac{\ket{01}-\ket{10}}{\sqrt{2}}
    \ket{s_{\ell-1}(w)} \\
 &= \sqrt{\frac{\ell}{\ell+1}} \cdot
    \frac{1}{\ell!}
    \sof*{
      \ket{0}
      \sum_{\pi \in \S_\ell}
      P(\pi) \ket{1} \ket{s_{\ell-1}(w)}
    - \ket{1}
      \sum_{\pi \in \S_\ell}
      P(\pi) \ket{0} \ket{s_{\ell-1}(w)}
    } \\
 &= \sqrt{\frac{\ell}{\ell+1}} \cdot
    \frac{1}{\ell!}
    \binom{\ell-1}{w}^{-1/2}
    \hspace{-1em}
    \sum_{\substack{x\in\set{0,1}^{\ell-1}\\\abs{x}=w}}
    \sof[\Bigg]{
      \ket{0}
      \underbrace{
        \sum_{\pi \in \S_\ell}
        P(\pi) \ket{1} \ket{x}
      }_{\ket{h_1}}
    - \ket{1}
      \underbrace{
        \sum_{\pi \in \S_\ell}
        P(\pi) \ket{0} \ket{x}
      }_{\ket{h_0}}
    }.
  \label{eq:s'-symmetrization}
\end{align}
For each $i \in \set{0,1}$, the expression for $\ket{h_i}$ is symmetric and of Hamming weight $w+i$. This means that $\ket{h_i} \propto \ket{s_\ell(w+i)}$. To determine the proportionality constant, we need to count how many times each string appears in the symmetrization. Since permuting $0$'s and $1$'s among themselves does not change the string, we simply need to insert appropriate factorials to account for this. For any $x \in \set{0,1}^{\ell-1}$ of Hamming weight $\abs{x} = w$,
\begin{align}
  \ket{h_1}
 &= \sum_{\pi \in \S_\ell}
    P(\pi) \ket{1} \ket{x}
  = \underbrace{(w+1)!}_{1\text{'s}}
    \underbrace{(\ell-w-1)!}_{0\text{'s}}
    \binom{\ell}{w+1}^{1/2}
    \ket{s_\ell(w+1)}, \\
  \ket{h_0}
 &= \sum_{\pi \in \S_\ell}
    P(\pi) \ket{0} \ket{x}
  = \underbrace{\phantom{(\!\!}w!}_{1\text{'s}}
    \underbrace{(\ell-w)!}_{0\text{'s}}
    \binom{\ell}{w}^{1/2}
    \ket{s_\ell(w)}.
\end{align}
Note that neither of $\ket{h_i}$ depends on $x$, so substituting them back into \cref{eq:s'-symmetrization}
we get
\begin{equation*}
  \sqrt{\frac{\ell}{\ell+1}} \cdot
  \frac{1}{\ell!}
  \binom{\ell-1}{w}^{1/2}
  \sof[\Big]{
    \ket{0} \ket{h_1}
  - \ket{1} \ket{h_0}
  }
= \sqrt{\frac{w+1}{\ell+1}}
  \ket{0}
  \ket{s_\ell(w+1)}
- \sqrt{\frac{\ell-w}{\ell+1}}
  \ket{1}
  \ket{s_\ell(w)},
\end{equation*}
which agrees with \cref{eq:s'}.
\end{proof}

The following lemma derives the matrix entries of $M\xp{\ell+1}$ in the subspaces $\mc{V}_{\UNOT}$ and $\mc{V}_{\Tr}$ and relates them to the representations $T^{\ell+1}$ and $T^{\ell-1}$, respectively. In particular, this justifies the notation $\ket{s'_{\ell-1}(w)}$ for the basis of $\mc{V}_{\Tr}$.

\begin{lemma}
For any $M \in \L(\C^2)$ and $\ell \geq 1$, the entries of the matrix $M\xp{\ell+1}$ restricted to the subspaces $\mc{V}_{\UNOT}$ and $\mc{V}_{\Tr}$ are as follows:
\begin{align}
  \bra{s_{\ell+1}(j)} M\xp{\ell+1} \ket{s_{\ell+1}(k)}
  &= t^{\ell+1}_{jk}(M),
  & \forall j,k \in \set{0,\dotsc,\ell+1}, \\
  \bra{s'_{\ell-1}(j)} M\xp{\ell+1} \ket{s'_{\ell-1}(k)}
  &= \det(M) t^{\ell-1}_{jk}(M),
  & \forall j,k \in \set{0,\dotsc,\ell-1},
\end{align}
where $t^\ell_{jk}(M)$ are the entries of the irreducible representation $T^\ell$ defined in \cref{eq:tjk}.
\end{lemma}

\begin{proof}
The first identity follows immediately from \cref{eq:V-UNOT,lem:tjk from symmetric subspace}. For the second identity we use \cref{lem:s'}:
\begin{equation}
  M\xp{\ell+1} \ket{s'_{\ell-1}(k)}
  = \sqrt{\frac{2\ell}{\ell+1}}
    \of[\big]{\id_2 \x \Pi_{\Sym{\ell}}}
    \of[\Big]{
      M\xp{2} \ket{\Psi^-} \x
      M\xp{\ell-1} \ket{s_{\ell-1}(k)}
    }.
\end{equation}
Note that $(M \x M) \ket{\Psi^-} = \det(M) \ket{\Psi^-}$ and
$M\xp{\ell-1} \ket{s_{\ell-1}(k)}
= \sum_{j=0}^{\ell-1}
  t^{\ell-1}_{jk}(M)
  \ket{s_\ell(j)}$
according to \cref{eq:Ml on slk again}. Thus
\begin{equation}
  M\xp{\ell+1} \ket{s'_{\ell-1}(k)}
  = \det(M)
    \sum_{j=0}^{\ell-1}
    t^{\ell-1}_{jk}(M)
    \ket{s'_\ell(j)}.
\end{equation}
Recalling that $\braket{s'_{\ell-1}(j)}{s'_{\ell-1}(j')} = \delta_{jj'}$, we get the desired formula by multiplying both sides by $\bra{s'_{\ell-1}(j)}$.
\end{proof}

\subsection{Symmetric subspace and (dual) Clebsch--Gordan transform}\label{apx:SymCG}

The decomposition of the space $\mc{V} = \C^2 \x \Sym{\ell}$ as $\mc{V}_{\Tr} \+ \mc{V}_{\UNOT}$ is closely related to Clebsch--Gordan transform. Similar to the map $\CG_\ell$ from \cref{eq:Cdims}, consider a map
\begin{equation}
  \CG_{\Sym{\ell}}: \C^2 \x \Sym{\ell} \to \C^\ell \+ \C^{\ell+2}
\end{equation}
defined by stacking the vectors $\bra{s'_{\ell-1}(v)}$ and $\bra{s_{\ell+1}(w)}$ as rows of a matrix:
\begin{equation}
  \CG_{\Sym{\ell}}
  = \sum_{v=0}^{\ell-1}
    \of[\big]{\ket{v}_\ell \+ \0_{\ell+2}}
    \bra{s'_{\ell-1}(v)}
  + \sum_{w=0}^{\ell+1}
    \of[\big]{\0_\ell \+ \ket{w}_{\ell+2}}
    \bra{s_{\ell+1}(w)},
\end{equation}
where $\0_\ell$ denotes the $\ell$-dimensional all zeroes column vector.
Substituting the formulas for $\bra{s'_{\ell-1}(v)}$ and $\bra{s_{\ell+1}(w)}$ from \cref{eq:s',eq:s-recursion}, we get
\begin{align}
  \CG_{\Sym{\ell}}
 &= \sum_{v=0}^{\ell-1}
    \of[\big]{\ket{v}_\ell \+ \0_{\ell+2}}
    \of*{
      \sqrt{\frac{v+1   }{\ell+1}} \bra{0} \bra{s_\ell(v+1)}
    - \sqrt{\frac{\ell-v}{\ell+1}} \bra{1} \bra{s_\ell(v  )}
    } \label{eq:stacking} \\
 &+ \sum_{w=0}^{\ell+1}
    \of[\big]{\0_\ell \+ \ket{w}_{\ell+2}}
    \of*{
      \sqrt{\frac{\ell+1-w}{\ell+1}} \bra{0} \bra{s_\ell(w  )}
    + \sqrt{\frac{w       }{\ell+1}} \bra{1} \bra{s_\ell(w-1)}
    }. \nonumber
\end{align}
We can rearrange the terms in each sum to pull out $\bra{s_\ell(v)}$ and $\bra{s_\ell(w)}$, respectively:
\begin{align}
  \CG_{\Sym{\ell}}
 &= \sum_{v=0}^\ell
    \of*{
      \sqrt{\frac{v     }{\ell+1}}
      \of[\big]{\ket{v-1}_\ell \+ \0_{\ell+2}} \bra{0}
    - \sqrt{\frac{\ell-v}{\ell+1}}
      \of[\big]{\ket{v  }_\ell \+ \0_{\ell+2}} \bra{1}
    } \bra{s_\ell(v)} \\
 &+ \sum_{w=0}^\ell
    \of*{
      \sqrt{\frac{\ell+1-w}{\ell+1}}
      \of[\big]{\0_\ell \+ \ket{w  }_{\ell+2}} \bra{0}
    + \sqrt{\frac{w+1}{\ell+1}}
      \of[\big]{\0_\ell \+ \ket{w+1}_{\ell+2}} \bra{1}
    } \bra{s_\ell(w)}. \nonumber
\end{align}
We can now combine the two sums into a single sum:
\begin{align}
  \CG_{\Sym{\ell}} =
  \sum_{w=0}^\ell
  \Biggl[
  & \of*{
      \sqrt{\frac{w       }{\ell+1}} \ket{w-1}_\ell \+
      \sqrt{\frac{\ell+1-w}{\ell+1}} \ket{w  }_{\ell+2}
    } \bra{0} \\
+ & \of*{
    - \sqrt{\frac{\ell-w}{\ell+1}} \ket{w  }_\ell \+
      \sqrt{\frac{w+1   }{\ell+1}} \ket{w+1}_{\ell+2}
    } \bra{1}
  \Biggr]
  \bra{s_\ell(w)}. \nonumber
\end{align}
We can simplify this even further by introducing another sum to group the two terms together:
\begin{align}
  \CG_{\Sym{\ell}} =
  \sum_{i=0}^1
  \sum_{w=0}^\ell
  \of*{
    R^\ell_{i0}(w+i) \ket{w+i-1}_\ell \+
    R^\ell_{i1}(w+i) \ket{w+i}_{\ell+2}
  } \bra{i}
  \bra{s_\ell(w)}
  \label{eq:CSyml}
\end{align}
where
\begin{equation}
  R^\ell(k) := \frac{1}{\sqrt{\ell+1}}
  \mx{
    \sqrt{k       } &
    \sqrt{\ell+1-k} \\
  - \sqrt{\ell+1-k} &
    \sqrt{k       }
  }
  \label{eq:Rlk again}
\end{equation}
is the same matrix as in \cref{eq:Rlk}.
Note that \cref{eq:CSyml} is identical to \cref{eq:Cl} for $\CG_\ell$, except for the input space being $\Sym{\ell}$ instead of $\C^{\ell+1}$.
In other words, the two maps are related as follows:
\begin{equation}
  \CG_\ell
  = \CG_{\Sym{\ell}} \cdot
    \sof*{\id_2 \x V^\ell}
  \label{eq:Cl and CSyml}
\end{equation}
where $V^\ell$ is the isometry from \cref{eq:Vl}.

Similar to \cref{eq:Cl dual and Cl}, the \emph{dual} of \cref{eq:CSyml} is given by
\begin{align}
  \DCG_{\Sym{\ell}}
  :={}& \CG_{\Sym{\ell}} \cdot \sof*{\iY \x \id_2\xp{\ell}} \label{eq:relation to dual} \\
   ={}&
  \sum_{i=0}^1
  \sum_{w=0}^\ell
  \of*{
    R^\ell_{i0}(w+i) \ket{w+i-1}_\ell \+
    R^\ell_{i1}(w+i) \ket{w+i}_{\ell+2}
  } (-1)^i \bra{1-i}
  \bra{s_\ell(w)},
  \label{eq:DualCSyml}
\end{align}
which is identical to \cref{eq:DualCl}, except for the input space being $\Sym{\ell}$ instead of $\C^{\ell+1}$.
This constitutes a self-contained derivation of Clebsch--Gordan transform and its dual.

Let $S := \smx{0&1\\-1&0}$ denote the $2 \times 2$ matrix in \cref{eq:relation to dual}.
Noting that $\bra{0} S = \bra{1}$ and $\bra{1} S = -\bra{0}$, we can derive the following dual version of \cref{eq:stacking} that will be useful later:
\begin{align}
  \DCG_{\Sym{\ell}}
 &= \sum_{v=0}^{\ell-1}
    \of[\big]{\ket{v}_\ell \+ \0_{\ell+2}}
    \of*{
      \sqrt{\frac{\ell-v}{\ell+1}} \bra{0} \bra{s_\ell(v  )}
    + \sqrt{\frac{v+1   }{\ell+1}} \bra{1} \bra{s_\ell(v+1)}
    } \label{eq:dual Sym stacking} \\
 &+ \sum_{w=0}^{\ell+1}
    \of[\big]{\0_\ell \+ \ket{w}_{\ell+2}}
    \of*{
    - \sqrt{\frac{w       }{\ell+1}} \bra{0} \bra{s_\ell(w-1)}
    + \sqrt{\frac{\ell+1-w}{\ell+1}} \bra{1} \bra{s_\ell(w  )}
    }. \nonumber
\end{align}

\section{Implementation of the extremal covariant channels $\Phi^\ell_{\Tr}$ and $\Phi^\ell_{\UNOT}$}\label{apx:Extremal}

The goal of this appendix is to derive quantum circuits for implementing the channels $\Phi^\ell_{\Tr}$ and $\Phi^\ell_{\UNOT}$ from \cref{sec:Understanding} which are used in \cref{alg:Template}.
Recall from \cref{lem:identification} that we have identified them with the two extremal $\U{2}$-covariant channels $\Phi^\ell_1$ and $\Phi^\ell_2$, respectively, whose Choi matrices are given by \cref{lem:1param}.

Our strategy for deriving quantum circuits for these channels consists of the following steps:
(i)~use dual Clebsch--Gordan transform to extract a set of Kraus operators from the Choi matrices of the channel,
(ii)~combine the Kraus operators into a Stinespring isomtery, and
(iii)~derive a quantum circuit implementing this isometry.

Note that extremal $\SU{2}$-covariant quantum channels have also been studied in \cite{Nuwairan,SU2CovariantLowRank,NoethersPrinciple}.
Some results concerning their structure that are proved here can also be found in these papers.

\subsection{Representations of quantum channels}\label{apx:Kraus and Stinespring}

Let us briefly summarize different ways of representing a quantum channel (see~\cite{Watrous} for more background).
By interconverting between these representations, we will find isometries that correspond to the two extremal channels $\Phi^\ell_1$ and $\Phi^\ell_2$.

Recall from \cref{sec:Math preliminaries} that a quantum channel $\Phi \in \CPTP{\C^\din}{\C^\dout}$ is fully described by its Choi matrix $\J{\Phi} \in \L(\C^\dout \x \C^\din)$ defined in \cref{eq:Choi}.
Since $\J{\Phi}$ is positive semidefinite, see \cref{eq:ChoiCPTP}, we can find $r = \rank \J{\Phi}$ vectors $\ket{k_1}, \dotsc, \ket{k_r} \in \C^\dout \x \C^\din$ such that
\begin{equation}
  \J{\Phi} = \sum_{i=1}^r \proj{k_i}.
  \label{eq:Choi and Kraus}
\end{equation}
Let $K_i^\Phi \in \L(\C^\din,\C^\dout)$ be such that $\ket{K_i^\Phi} = \ket{k_i}$, \ie, the vectorization of $K_i^\Phi$ (see \cref{def:vec}) coincides with $\ket{k_i}$.
Then
\begin{equation}
  \Phi(\rho) = \sum_{i=1}^r K_i^\Phi \rho K_i^{\Phi\dagger},
  \label{eq:Kraus output}
\end{equation}
for any $\rho \in \L(\C^\din)$.
This is known as \emph{Kraus representation} of $\Phi$ and the $K_i^\Phi$ are known as \emph{Kraus operators}.
The map $\Phi$ is CPTP if and only if
\begin{equation}
  \sum_{i=1}^r K_i^{\Phi\dagger} K_i^\Phi = \id_\din.
  \label{eq:KKI}
\end{equation}

Another useful way of representing a quantum channel $\Phi \in \CPTP{\C^\din}{\C^\dout}$ is by a map $U_\Phi \in \L(\C^\din, \C^\dout \x \C^r)$ defined in terms of Kraus operators of $\Phi$:
\begin{equation}
  U_\Phi := \sum_{i=1}^r K_i^\Phi \x \ket{i}.
  \label{eq:UPhi}
\end{equation}
Note that \cref{eq:KKI} is equivalent to $U_\Phi\ct U_\Phi = \id_\din$, meaning that $U_\Phi$ is an isometry.
The action of $\Phi$ on any $\rho \in \L(\C^\din)$ can be expressed as follows:
\begin{equation}
  \Phi(\rho)
  = \Tr_2 \sof[\big]{U_\Phi \rho U_\Phi\ct},
\end{equation}
which produces the same output as \cref{eq:Kraus output}.
This is known as \emph{Stinespring representation} of $\Phi$ and the map $U_\Phi$ is called \emph{Stinespring isometry}.

\subsection{Kraus operators of $\Phi^\ell_1$ and $\Phi^\ell_2$}\label{apx:Kraus}

Let us find Kraus representations of the two extremal channels $\Phi^\ell_1, \Phi^\ell_2 \in \CPTP{\C^{\ell+1}}{\C^2}$ from \cref{sec:Extremal}.
Recall from \cref{eq:Choi12} in \cref{lem:1param} that their Choi matrices are
\begin{align}
	\J{\Phi^\ell_1}
 &= \DCGct_\ell
    \sof*{
      \frac{\ell+1}{\ell} \id_\ell \+ 0_{\ell+2}
    }
    \DCG_\ell, &
  \J{\Phi^\ell_2}
 &= \DCGct_\ell
    \sof*{
      0_\ell \+ \frac{\ell+1}{\ell+2} \id_{\ell+2}
    }
    \DCG_\ell,
  \label{eq:Chois}
\end{align}
where $\DCG_\ell$ denotes dual Clebsch--Gordan transform, see \cref{apx:Dual CG}.
Similar to \cref{eq:Cl and CSyml}, we can express $\DCG_\ell$ as follows:
\begin{equation}
  \DCG_\ell = \DCG_{\Sym{\ell}} \cdot \sof*{\id_2 \x V^\ell},
\end{equation}
where $V^\ell$ is the isometry defined in \cref{eq:Vl} that acts as $\bra{s_\ell(w)} V^\ell = \bra{w}$.
Using this identity, we can convert \cref{eq:dual Sym stacking} from $\DCG_{\Sym{\ell}}$ to $\DCG_\ell$:
\begin{align}
  \DCG_\ell
 &= \sum_{v=0}^{\ell-1}
    \of[\big]{\ket{v}_\ell \+ \0_{\ell+2}}
    \of*{
      \sqrt{\frac{\ell-v}{\ell+1}} \bra{0} \bra{v  }
    + \sqrt{\frac{v+1   }{\ell+1}} \bra{1} \bra{v+1}
    } \label{eq:dual stacking} \\
 &+ \sum_{w=0}^{\ell+1}
    \of[\big]{\0_\ell \+ \ket{w}_{\ell+2}}
    \of*{
    - \sqrt{\frac{w       }{\ell+1}} \bra{0} \bra{w-1}
    + \sqrt{\frac{\ell+1-w}{\ell+1}} \bra{1} \bra{w  }
    }. \nonumber
\end{align}

Following the procedure in \cref{eq:Choi and Kraus}, we can find Kraus operators of the channels $\Phi^\ell_1$ and $\Phi^\ell_2$ from the spectral decomposition of their Choi matrices in \cref{eq:Chois}:
\begin{align}
  \J{\Phi^\ell_1}
 &= \sum_{v \in [\ell]}
    \ketbra{K_v^{\Tr}}{\overline{K_v^{\Tr}}}, &
  \J{\Phi^\ell_2}
 &= \sum_{w \in [\ell+2]}
    \ketbra{K_w^{\UNOT}}{\overline{K_w^{\UNOT}}},
  \label{eq:ChoiKraus}
\end{align}
where the Kraus operators $K_v^{\Tr}$ and $K_w^{\UNOT}$ are derived below.
The reason we denote them like this is because eventually in \cref{lem:identification} we will identify $\Phi^\ell_1$ and $\Phi^\ell_2$ with the partial trace and universal NOT channels $\Phi^\ell_{\Tr}$ and $\Phi^\ell_{\UNOT}$, which were introduced in \cref{def:TrUNOT}.

\begin{lemma}[Kraus operators for $\Phi^\ell_1$ and $\Phi^\ell_2$]\label{lem:Kraus}
For any integer $\ell \geq 1$, the two extremal unitary-covariant channels $\Phi^\ell_1, \Phi^\ell_2 \in \CPTP{\C^{\ell+1}}{\C^2}$ from \cref{lem:1param} have the following Kraus operators:
\begin{align}
  K_v^{\Tr} &:=
    \sqrt{\frac{\ell-v}{\ell}} \ketbra{0}{v  }
  + \sqrt{\frac{v+1   }{\ell}} \ketbra{1}{v+1}, &
  v &\in [\ell], \label{eq:KTr} \\
  K_w^{\UNOT} &:=
  - \sqrt{\frac{w       }{\ell+2}} \ketbra{0}{w-1}
  + \sqrt{\frac{\ell+1-w}{\ell+2}} \ketbra{1}{w  }, &
  w &\in [\ell+2]. \label{eq:KUNOT}
\end{align}
In particular, the extremal cases $w = 0$ and $w = \ell+1$ of $K_w^{\UNOT}$ correspond to
\begin{align}
  K_0^{\UNOT}
    &= \sqrt{\frac{\ell+1}{\ell+2}} \ketbra{1}{0}, &
  K_{\ell+1}^{\UNOT}
    &=-\sqrt{\frac{\ell+1}{\ell+2}} \ketbra{0}{\ell}.
\end{align}
\end{lemma}

Note that similar formulas can also be found in Sections~2.2.1 and 2.2.2 of \cite{SU2CovariantLowRank}.
That paper discusses the duals of these channels and uses slightly different sign conventions for Clebsch--Gordan coefficients.

\begin{proof}
We can immediately read off the vectorized Kraus operators from \cref{eq:Chois,eq:dual stacking}:
\begin{align}
  \ket{K_v^{\Tr}} &=
    \sqrt{\frac{\ell+1}{\ell}}
    \of*{
      \sqrt{\frac{\ell-v}{\ell+1}} \bra{0} \bra{v  }
    + \sqrt{\frac{v+1   }{\ell+1}} \bra{1} \bra{v+1}
    }, &
  v &\in [\ell], \\
  \ket{K_w^{\UNOT}} &=
  \sqrt{\frac{\ell+1}{\ell+2}}
    \of*{
    - \sqrt{\frac{w       }{\ell+1}} \bra{0} \bra{w-1}
    + \sqrt{\frac{\ell+1-w}{\ell+1}} \bra{1} \bra{w  }
    }, &
  w &\in [\ell+2].
\end{align}
The normalization factors at the front account for the eigenvalues of the Choi operators in \cref{eq:Chois}.
Since in \cref{eq:ChoiKraus} we have absorbed a square root of the eigenvalue in each of the vectorized Kraus operators, we must take a square root of the normalization factor.
We get the desired formulas by flipping the first ``bra'' into a ``ket'' to unvectorize, see \cref{def:vec}.
\end{proof}

\begin{example}[UNOT for $\ell=1$]
The $\ell = 1$ case of \cref{eq:KUNOT} provides Kraus operators for the single-qubit $\UNOT$ operation, which are as follows:
\begin{align}
  K_0 &=-\sqrt{\frac{2}{3}} \ketbra{1}{0}, &
  K_1 &= \sqrt{\frac{1}{3}} \of[\big]{\proj{0}-\proj{1}}, &
  K_2 &= \sqrt{\frac{2}{3}} \ketbra{0}{1}.
\end{align}
Under a unitary basis change in their span, they are equivalent to $1/\sqrt{3}$ times a Pauli matrix:
\begin{align}
  \frac{1}{\sqrt{3}} X &= \frac{-K_0+K_2}{\sqrt{2}}, &
  \frac{1}{\sqrt{3}} Z &= K_1, &
  \frac{1}{\sqrt{3}} Y &= -i \frac{K_0+K_2}{\sqrt{2}}.
\end{align}
Thus the $\ell=1$ $\UNOT$ map can be implemented by applying the Pauli matrices $X,Y,Z$, each with probability $1/3$.
\end{example}

\subsection{The action of $\Phi^\ell_1$ and $\Phi^\ell_2$}\label{apx:Action}

The following lemma uses the Kraus operators of $\Phi^\ell_1$ and $\Phi^\ell_2$ from \cref{lem:Kraus} to derive their action on any input.

\begin{lemma}\label{lem:Phi12}
For any integer $\ell \geq 1$ and any $w,w' \in \set{0,\dotsc,\ell}$,
\begin{align}
  \Phi^\ell_1 \of[\big]{\ketbra{w}{w'}} &=
	\begin{cases}
		\frac{\ell-w}{\ell} \proj{0} + \frac{w}{\ell} \proj{1} \qquad & \text{if $w = w'$}, \\[3pt]
		\frac{\sqrt{(\ell-w)w'}}{\ell} \ketbra{0}{1} & \text{if $w' = w+1$}, \\[3pt]
		\frac{\sqrt{(\ell-w')w}}{\ell} \ketbra{1}{0} & \text{if $w = w'+1$}, \\
		0 & \text{otherwise},
	\end{cases} \label{eq:output1} \\
  \Phi^\ell_2 \of[\big]{\ketbra{w}{w'}} &=
	\begin{cases}
		\frac{w+1}{\ell+2} \proj{0} + \frac{\ell+1-w}{\ell+2} \proj{1} & \text{if $w = w'$}, \\[3pt]
	- \frac{\sqrt{(\ell+1-w')(w+1)}}{\ell+2} \ketbra{0}{1} & \text{if $w' = w+1$}, \\[3pt]
	- \frac{\sqrt{(\ell+1-w)(w'+1)}}{\ell+2} \ketbra{1}{0} & \text{if $w = w'+1$}, \\
		0 & \text{otherwise}.
	\end{cases} \label{eq:output2}
\end{align}
\end{lemma}

Note that similar-looking formulas for the duals of these channels are provided in Section~2.2.1 of \cite{SU2CovariantLowRank}, so one can recover \cref{lem:Phi12} also from Section~2.2.2 of that paper.

\begin{proof}
Note from \cref{eq:KTr} that
\begin{equation}
  K_v^{\Tr} \ket{w}
  = \sqrt{\frac{\ell-w}{\ell}} \ket{0} \delta_{v,w}
  + \sqrt{\frac{     w}{\ell}} \ket{1} \delta_{v+1,w}.
\end{equation}
The output of $\Phi^\ell_1$ on input $\ketbra{w}{w'}$ is then given by
\begin{align}
  \Phi^\ell_1 \of[\big]{\ketbra{w}{w'}}
&=
  \sum_{v=0}^{\ell-1}
  K_v^{\Tr}
  \ketbra{w}{w'}
  K_v^{\Tr\dagger} \\
&=
  \sum_{v=0}^{\ell-1}
  \of*{
    \sqrt{\frac{\ell-w}{\ell}} \ket{0} \delta_{v,w}
  + \sqrt{\frac{     w}{\ell}} \ket{1} \delta_{v+1,w}
  }
  \of*{
    \sqrt{\frac{\ell-w'}{\ell}} \bra{0} \delta_{v,w'}
  + \sqrt{\frac{     w'}{\ell}} \bra{1} \delta_{v+1,w'}
  }.
\end{align}
After expanding and simplifying this, we get
\begin{align}
  \Phi^\ell_1 \of[\big]{\ketbra{w}{w'}}
  = \frac{\ell-w}{\ell} \proj{0} \delta_{w,w'}
 &+ \sqrt{\frac{\ell-w}{\ell}} \sqrt{\frac{w'}{\ell}} \ketbra{0}{1} \delta_{w,w'-1} \\
{}+ \frac{     w}{\ell} \proj{1} \delta_{w,w'}
 &+ \sqrt{\frac{\ell-w'}{\ell}} \sqrt{\frac{w}{\ell}} \ketbra{1}{0} \delta_{w-1,w'},
\end{align}
which agrees with \cref{eq:output1}.

The proof of \cref{eq:output2} is similar.
Note from \cref{eq:KUNOT} that
\begin{equation}
  K_v^{\UNOT} \ket{w}
= - \sqrt{\frac{     w+1}{\ell+2}} \ket{0} \delta_{v-1,w}
    \sqrt{\frac{\ell+1-w}{\ell+2}} \ket{1} \delta_{v,w}.
\end{equation}
The output of $\Phi^\ell_2$ is then given by
\begin{align}
  \Phi^\ell_2 \of[\big]{\ketbra{w}{w'}}
&=
  \sum_{v=0}^{\ell+1}
  K_v^{\UNOT}
  \ketbra{w}{w'}
  K_v^{\UNOT\dagger} \\
&=
  \sum_{v=0}^{\ell+1}
  \of*{
  - \sqrt{\frac{     w+1}{\ell+2}} \ket{0} \delta_{v-1,w}
  + \sqrt{\frac{\ell+1-w}{\ell+2}} \ket{1} \delta_{v,w}
  } \\
  & \quad\quad
  \of*{
  - \sqrt{\frac{     w'+1}{\ell+2}} \bra{0} \delta_{v-1,w'}
  + \sqrt{\frac{\ell+1-w'}{\ell+2}} \bra{1} \delta_{v,w'}
  }. \nonumber
\end{align}
After expanding and simplifying this, we get
\begin{align}
  \Phi^\ell_2 \of[\big]{\ketbra{w}{w'}}
  = \frac{w+1}{\ell+2} \proj{0} \delta_{w',w}
 &- \sqrt{\frac{\ell+1-w'}{\ell+2}}
    \sqrt{\frac{w+1}{\ell+2}}
    \ketbra{0}{1} \delta_{w+1,w'} \\
{}+ \frac{\ell+1-w}{\ell+2} \proj{1} \delta_{w',w}
 &- \sqrt{\frac{\ell+1-w}{\ell+2}}
    \sqrt{\frac{w'+1}{\ell+2}}
    \ketbra{1}{0} \delta_{w,w'+1},
\end{align}
which agrees with \cref{eq:output2}.
\end{proof}

\subsection{Proof that $\Phi^\ell_1 = \Phi^\ell_{\Tr}$ and $\Phi^\ell_2 = \Phi^\ell_{\UNOT}$}\label{apx:Identification}

\begin{lemma}\label{lem:identification}
For any integer $\ell \geq 1$, the maps from \cref{def:TrUNOT} coincide with the ones from \cref{lem:1param}:
\begin{align}
  \Phi^\ell_1 &= \Phi^\ell_{\Tr}, &
  \Phi^\ell_2 &= \Phi^\ell_{\UNOT}.
\end{align}
\end{lemma}

\begin{proof}
Our strategy will be to show that, for all
$w,w' \in \set{0,\dotsc,\ell}$,
\begin{align}
  \Phi^\ell_1 \of[\big]{\ketbra{w}{w'}}
  &= \Phi^\ell_{\Tr} \of[\big]{\ketbra{w}{w'}}, &
  \Phi^\ell_2 \of[\big]{\ketbra{w}{w'}}
  &= \Phi^\ell_{\UNOT} \of[\big]{\ketbra{w}{w'}}.
\end{align}
Recall from \cref{eq:Vl!,eq:Sigmal!} that
\begin{equation}
  \Sigma^\ell\of[\big]{\ketbra{w}{w'}}
  = V^\ell \ketbra{w}{w'} V^{\ell\dagger}
  = \ketbra{s_\ell(w)}{s_\ell(w')}.
\end{equation}
We can use this to analyze the output of $\Phi^\ell_{\Tr}$:
\begin{equation}
  \Phi^\ell_{\Tr} \of[\big]{\ketbra{w}{w'}}
= \Tr_{2,\dotsc,\ell}
  \sof[\Big]{
    \ketbra{s_\ell(w)}{s_\ell(w')}
  }.
  \label{eq:partial trace}
\end{equation}
Recall from \cref{eq:s-recursion} in \cref{apx:Invariant subspaces} that
\begin{equation}
  \ket{s_\ell(w)}
  = \sqrt{\frac{\ell-w}{\ell}} \ket{0} \ket{s_{\ell-1}(w  )}
  + \sqrt{\frac{     w}{\ell}} \ket{1} \ket{s_{\ell-1}(w-1)}.
  \label{eq:slw recursion}
\end{equation}
After substituting this, we can easily compute the partial trace in \cref{eq:partial trace}:
\begin{align}
  \Phi^\ell_{\Tr} \of[\big]{\ketbra{w}{w'}}
= \Tr_{2,\dotsc,\ell}
  \Bigg[
  & \of*{
      \sqrt{\frac{\ell-w}{\ell}} \ket{0} \ket{s_{\ell-1}(w  )}
    + \sqrt{\frac{     w}{\ell}} \ket{1} \ket{s_{\ell-1}(w-1)}
    } \\
  & \of*{
      \sqrt{\frac{\ell-w'}{\ell}} \bra{0} \bra{s_{\ell-1}(w'  )}
    + \sqrt{\frac{     w'}{\ell}} \bra{1} \bra{s_{\ell-1}(w'-1)}
    }
  \Bigg]. \nonumber
\end{align}
Since $\Tr \sof[\Big]{\ketbra{s_{\ell-1}(w)}{s_{\ell-1}(v)}} = \braket{s_{\ell-1}(v)}{s_{\ell-1}(w)} = \delta_{v,w}$, we get
\begin{align}
  \Phi^\ell_{\Tr} \of[\big]{\ketbra{w}{w'}}
  = \frac{\ell-w}{\ell} \proj{0} \delta_{w,w'}
 &+ \sqrt{\frac{\ell-w}{\ell}} \sqrt{\frac{w'}{\ell}} \ketbra{0}{1} \delta_{w,w'-1} \\
{}+ \frac{     w}{\ell} \proj{1} \delta_{w,w'}
 &+ \sqrt{\frac{\ell-w'}{\ell}} \sqrt{\frac{w}{\ell}} \ketbra{1}{0} \delta_{w',w-1}, \nonumber
\end{align}
which agrees with \cref{eq:output1}.

\newcommand{\Int}{\mathrm{Int}}

The output of the second map is equal to
\begin{align}
  \Phi^\ell_{\UNOT} \of[\big]{\ketbra{w}{w'}}
 &= (\ell+1) \int
    \bra{\psi}\xp{\ell}
    \sof[\Big]{\ketbra{s_\ell(w)}{s_\ell(w')}}
    \ket{\psi}\xp{\ell}
    \of[\big]{\id_2 - \proj{\psi}} \, d\psi \\
 &= (\ell+1) \int
    \bra{s_\ell(w')}
    \of[\big]{\proj{\psi}}\xp{\ell}
    \ket{s_\ell(w)}
    \of[\big]{\id_2 - \proj{\psi}} \, d\psi \\
 &= c \id_2 - A, \label{eq:cIA}
\end{align}
for some constant $c \in \C$ and matrix $A \in \L(\C^2)$ that can be found by splitting the integral into two terms and integrating them separately.
Using \cref{lem:PiSym3}, the first term results in
\begin{equation}
  c
= \bra{s_\ell(w')}
  \sof*{
    (\ell+1) \int
    \of[\big]{\proj{\psi}}\xp{\ell}
    \, d\psi
  }
  \ket{s_\ell(w)}
= \bra{s_\ell(w')}
  \Pi_{\Sym{\ell}}
  \ket{s_\ell(w)}
= \delta_{w',w}
  \label{eq:c}
\end{equation}
since $\ket{s_\ell(w)} \in \Sym{\ell}$.
The second term produces
\begin{align}
  A
&=
  \of[\big]{\id_2 \x \bra{s_\ell(w')}}
  \sof*{
    (\ell+1) \int
    \of{\proj{\psi}}\xp{\ell+1}
    \, d\psi
  }
  \of[\big]{\id_2 \x \ket{s_\ell(w)}} \\
&=
  \of[\big]{\id_2 \x \bra{s_\ell(w')}}
  \sof*{
    \frac{\ell+1}{\ell+2}
    \Pi_{\Sym{\ell+1}}
  }
  \of[\big]{\id_2 \x \ket{s_\ell(w)}} \\
&=
  \frac{\ell+1}{\ell+2}
  \of[\big]{\id_2 \x \bra{s_\ell(w')}}
  \sof*{
    \sum_{v=0}^{\ell+1}
    \proj{s_{\ell+1}(v)}
  }
  \of[\big]{\id_2 \x \ket{s_\ell(w)}},
  \label{eq:A}
\end{align}
where we substituted $\Pi_{\Sym{\ell+1}}$ using \cref{lem:PiSym3}.
Applying \cref{eq:slw recursion} to $\ket{s_{\ell+1}(v)}$, we get
\begin{align}
  \of[\big]{\id_2 \x \bra{s_\ell(w')}}
  \ket{s_{\ell+1}(v)}
&=
  \of[\big]{\id_2 \x \bra{s_\ell(w')}}
  \of*{
    \sqrt{\tfrac{\ell+1-v}{\ell+1}} \ket{0} \ket{s_\ell(v  )}
  + \sqrt{\tfrac{       v}{\ell+1}} \ket{1} \ket{s_\ell(v-1)}
  } \\
&=
  \sqrt{\frac{\ell+1-w'}{\ell+1}} \ket{0} \delta_{w',v}
+ \sqrt{\frac{     w'+1}{\ell+1}} \ket{1} \delta_{w',v-1}.
\end{align}
Substituting this back into \cref{eq:A}, we get
\begin{align}
  A
&=
  \frac{\ell+1}{\ell+2}
  \sum_{v=0}^{\ell+1}
  \of[\big]{\id_2 \x \bra{s_\ell(w')}}
  \ket{s_{\ell+1}(v)}
  \bra{s_{\ell+1}(v)}
  \of[\big]{\id_2 \x \ket{s_\ell(w)}} \\
&=
  \sum_{v=0}^{\ell+1}
  \of*{
    \sqrt{\tfrac{\ell+1-w'}{\ell+2}} \ket{0} \delta_{w',v}
  + \sqrt{\tfrac{     w'+1}{\ell+2}} \ket{1} \delta_{w',v-1}
  }
  \of*{
    \sqrt{\tfrac{\ell+1-w}{\ell+2}} \bra{0} \delta_{w,v}
  + \sqrt{\tfrac{     w+1}{\ell+2}} \bra{1} \delta_{w,v-1}
  }.
\end{align}
Expanding the sum, we are left with the following:
\begin{align}
  A
  = \frac{\ell+1-w}{\ell+2} \proj{0} \delta_{w',w}
 &+ \sqrt{\frac{\ell+1-w'}{\ell+2}}
    \sqrt{\frac{w+1}{\ell+2}}
    \ketbra{0}{1} \delta_{w',w+1} \\
{}+ \frac{w+1}{\ell+2} \proj{1} \delta_{w',w}
 &+ \sqrt{\frac{\ell+1-w}{\ell+2}}
    \sqrt{\frac{w'+1}{\ell+2}}
    \ketbra{1}{0} \delta_{w'+1,w}.
\end{align}
Substituting this and $c = \delta_{w',w}$ from \cref{eq:c} back into $c \id_2 - A$ of \cref{eq:cIA}, we get
\begin{align}
  \Phi^\ell_{\UNOT} \of[\big]{\ketbra{w}{w'}}
  = \frac{w+1}{\ell+2} \proj{0} \delta_{w',w}
 &- \sqrt{\frac{\ell+1-w'}{\ell+2}}
    \sqrt{\frac{w+1}{\ell+2}}
    \ketbra{0}{1} \delta_{w',w+1} \\
{}+ \frac{\ell+1-w}{\ell+2} \proj{1} \delta_{w',w}
 &- \sqrt{\frac{\ell+1-w}{\ell+2}}
    \sqrt{\frac{w'+1}{\ell+2}}
    \ketbra{1}{0} \delta_{w'+1,w},
\end{align}
which agrees with \cref{eq:output2}.
\end{proof}

\subsection{Isometries and quantum circuits for $\Phi^\ell_{\Tr}$ and $\Phi^\ell_{\UNOT}$}\label{apx:Isometries}

\begin{lemma}\label{lem:Us}
For any integer $\ell \geq 1$, the channels
$\Phi^\ell_{\Tr}, \Phi^\ell_{\UNOT} \in \CPTP{\C^{\ell+1}}{\C^2}$
from \cref{def:TrUNOT} are described by isometries
$U^\ell_{\Tr  } \in \U{\C^{\ell+1},\C^2 \x \C^{\ell  }}$ and
$U^\ell_{\UNOT} \in \U{\C^{\ell+1},\C^2 \x \C^{\ell+2}}$
given by
\begin{align}
  U^\ell_{\Tr} &:=
  \sum_{v=0}^\ell
  \of*{
    \sqrt{\frac{\ell-v}{\ell}} \ket{0} \ket{v  }
  + \sqrt{\frac{v     }{\ell}} \ket{1} \ket{v-1}
  } \bra{v}, \label{eq:UTr} \\
  U^\ell_{\UNOT} &:=
  \sum_{w=0}^\ell
  \of*{
  - \sqrt{\frac{w+1     }{\ell+2}} \ket{0} \ket{w+1}
  + \sqrt{\frac{\ell+1-w}{\ell+2}} \ket{1} \ket{w  }
  } \bra{w}. \label{eq:UUNOT}
\end{align}
\end{lemma}

\begin{proof}
Recall from \cref{lem:identification} that
$\Phi^\ell_{\Tr  } = \Phi^\ell_1$ and
$\Phi^\ell_{\UNOT} = \Phi^\ell_2$.
Using the Kraus operators from \cref{lem:Kraus} and the recipe in \cref{eq:UPhi}, we get
\begin{align}
  U^\ell_{\Tr}
&=
  \sum_{v=0}^{\ell-1}
  \of*{
    \sqrt{\frac{\ell-v}{\ell}} \ket{0} \ketbra{v}{v  }
  + \sqrt{\frac{v+1   }{\ell}} \ket{1} \ketbra{v}{v+1}
  } \\
&=
  \sum_{v=0}^\ell
  \of*{
    \sqrt{\frac{\ell-v}{\ell}} \ket{0} \ketbra{v}{v  }
  + \sqrt{\frac{v     }{\ell}} \ket{1} \ketbra{v-1}{v}
  },
\end{align}
where we decreased $v$ by one in the second term and made the sum run up to $v = \ell$ (this does not affect the first term since it vanishes at $v = \ell$).
Similarly, the isometry for $\Phi^\ell_{\UNOT}$ is given by
\begin{align}
  U^\ell_{\UNOT}
&=
  \sum_{w=0}^{\ell+1}
  \of*{
  - \sqrt{\frac{w       }{\ell+2}} \ket{0} \ketbra{w}{w-1}
  + \sqrt{\frac{\ell+1-w}{\ell+2}} \ket{1} \ketbra{w}{w  }
  } \\
&=
  \sum_{w=0}^\ell
  \of*{
  - \sqrt{\frac{     w+1}{\ell+2}} \ket{0} \ketbra{w+1}{w}
  + \sqrt{\frac{\ell+1-w}{\ell+2}} \ket{1} \ketbra{w  }{w}
  },
\end{align}
where we increased $w$ in the first term by one and truncated the sum at $w = \ell$ (this does not affect the second term since it vanishes at $w = \ell+1$).
\end{proof}

\begin{lemma}\label{lem:Phi complexity}
For any integer $\ell \geq 1$, the quantum channels $\Phi^\ell_{\Tr}$ and $\Phi^\ell_{\UNOT}$ can be implemented using $O(\ell \log \ell)$ elementary quantum gates.
\end{lemma}

\begin{proof}
Note from \cref{eq:UTr} that, for any $v \in \set{0,\dotsc,\ell}$,
\begin{equation}
  U^\ell_{\Tr} \ket{v}
= \sqrt{\frac{\ell-v}{\ell}} \ket{0} \ket{v  }
+ \sqrt{\frac{v     }{\ell}} \ket{1} \ket{v-1}.
  \label{eq:UTrv}
\end{equation}
Let $R^\ell(v)$ denote the following rotation:
\begin{equation}
  R^\ell(v) := \frac{1}{\sqrt{\ell}}
  \mx{
    \sqrt{\ell-v} & -\sqrt{v} \\
    \sqrt{v} & \sqrt{\ell-v}
  }.
  \label{eq:R}
\end{equation}
Note that the state in \cref{eq:UTrv} is very similar to
\begin{equation}
  R^\ell(v) \ket{0} \x \ket{v}
  = \sqrt{\frac{\ell-v}{\ell}} \ket{0} \ket{v}
  + \sqrt{\frac{v     }{\ell}} \ket{1} \ket{v},
  \label{eq:Rotated v}
\end{equation}
except for having $\ket{v-1}$ instead of $\ket{v}$ in the second term.
We can fix this by applying an additional subtraction gate that decrements the second register if the first register is in state $\ket{1}$.
The overall circuit for implementing $U^\ell_{\Tr}$ is shown in \cref{fig:Circuits}, left.
The first gate
\begin{equation}
  \sum_{v=0}^\ell
  R^\ell(v) \x \proj{v}
  \label{eq:cR}
\end{equation}
is a conditional rotation that applies $R^\ell(v)$ on the first register if the second register is in state $\ket{v}$, while the second gate ``$-$'' denotes the modulo $\ell+1$ \emph{decrementer}
\begin{equation}
  \sum_{i=0}^1 \proj{i} \x
  \sum_{v=0}^\ell \ketbra{v-i}{v}
  \label{eq:decrementer}
\end{equation}
where we treat $\ket{-1}$ as $\ket{\ell}$ to make the gate reversible.
Since the term $\ket{1} \ket{v}$ drops out from \cref{eq:Rotated v} when $v = 0$, this convention does not affect the output of the left circuit in \cref{fig:Circuits}.
Indeed, since we never apply the decrementer onto $\ket{1} \ket{0}$, the second output register cannot ever be in state $\ket{\ell}$. Thus we can effectively treat this output as $\ell$-dimensional (\ie, as belonging to $\spn \set{0,\dotsc,\ell-1}$), which agrees with the output dimension of $U^\ell_{\Tr}$ in \cref{lem:Us}.



\newcommand{\circuit}[7]{%
\begin{tikzpicture}[thick,
    gate/.style = {draw, fill = white, minimum width = 0.7cm, minimum height = 0.7cm},
    blob/.style = {draw = none, fill = lightgray, rounded corners = 6pt},
    dot/.style = {radius = 2pt, fill = black},
    arc/.style = {start angle = 90, end angle = 0, radius = 0.3cm},
    ell/.style = {x radius = 0.2cm, y radius = 0.1cm}
  ]
  \def\W{0.8cm}
  \def\H{0.6cm}
  \draw[blob] (-2.8*\W+#7,-2.0*\H) rectangle (1.8*\W,2.0*\H);
  \draw (-3.5*\W+#7, \H) -- (2.5*\W, \H);
  \draw (-3.5*\W+#7,-\H) -- (2.0*\W,-\H);
  \path (-\W, \H) coordinate (R);
  \path (-\W,-\H) coordinate (Rc);
  \path (Rc)+(0,-10pt) node {$#2$};
  \draw (Rc) -- (R);
  \draw[dot] (Rc) circle;
  \node[gate] at (R) {#3};
  \path ( \W,-\H) coordinate (P);
  \path ( \W, \H) coordinate (Pc);
  \draw (Pc) -- (P);
  \draw[dot, fill = #4] (Pc) circle;
  \node[gate] at (P) {$#5$};
  \node at (-0.4*\W+0.5*#7,2.7*\H) {$U^\ell_{#6}$};
  \node[anchor = east] at (-3.6*\W+#7, \H) {$\ket{#1}$};
  \node[anchor = east] at (-3.6*\W+#7,-\H) {$\C^{\ell+1}$};
  \node at ( 3.1*\W, \H) {$\C^2$};
  \def\r{0.06cm}
  \draw (2*\W,-\H) arc[arc] circle[ell] ++(0.2cm,0) -- ++(-\r,-0.35cm) arc [start angle = 0, end angle = -180, x radius = 0.2cm-\r, y radius = 0.05cm] -- ++(-\r,0.35cm);
\end{tikzpicture}%
}

\begin{figure}
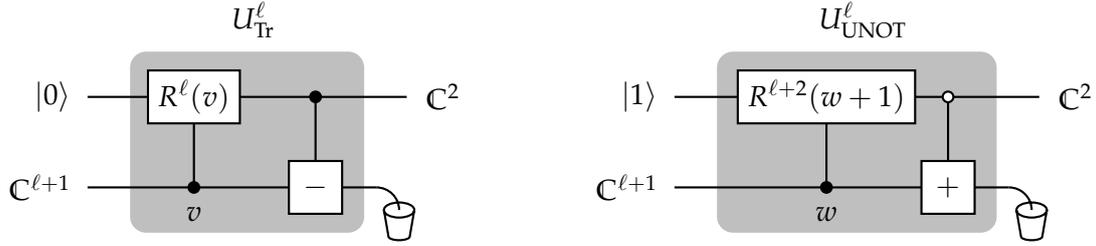

\centering
\circuit{0}{v}{$R^{\ell  }(v  )$}{black}{-}{\Tr}{0.6cm}\qquad\qquad
\circuit{1}{w}{$R^{\ell+2}(w+1)$}{white}{+}{\UNOT}{0}
\caption{\label{fig:Circuits}Quantum circuits for implementing the channels $\Phi^\ell_{\Tr}$ (left) and $\Phi^\ell_{\UNOT}$ (right) based on the isometries $U^\ell_{\Tr}$ and $U^\ell_{\UNOT}$, denoted by gray boxes. The ``$-$'' and ``$+$'' gates are defined in \cref{eq:decrementer,eq:incrementer}, respectively.}
\end{figure}

We can derive a quantum circuit for $U^\ell_{\UNOT}$ in a similar manner.
Note from \cref{eq:UUNOT} that, for any $w \in \set{0,\dotsc,\ell+1}$,
\begin{align}
  U^\ell_{\UNOT} \ket{w}
= - \sqrt{\frac{w+1     }{\ell+2}} \ket{0} \ket{w+1}
  + \sqrt{\frac{\ell+1-w}{\ell+2}} \ket{1} \ket{w  }.
\end{align}
This is almost identical to
\begin{equation}
  R^{\ell+2}(w+1) \ket{1} \x \ket{w}
= - \sqrt{\frac{w+1     }{\ell+2}} \ket{0} \ket{w}
  + \sqrt{\frac{\ell+1-w}{\ell+2}} \ket{1} \ket{w},
\end{equation}
except for having $\ket{w+1}$ instead of $\ket{w}$ in the first term.
We can fix this by applying an additional \emph{incrementer} gate ``$+$'' that increments the second register modulo $\ell+2$ if the first register is in state $\ket{0}$:
\begin{equation}
  \sum_{i=0}^1 \proj{i} \x
  \sum_{v=0}^{\ell+1} \ketbra{v+1-i}{v}.
  \label{eq:incrementer}
\end{equation}
Note that the second register of the incrementer has dimension $\ell+2$. Thus, to match dimensions with the previous gate, before applying the incrementer we need to bump the second register's dimension up from $\ell+1$ to $\ell+2$ by trivially embedding $\C^{\ell+1}$ into $\C^{\ell+2}$.
This guarantees that the second output register of $U^\ell_{\UNOT}$ indeed has dimension $\ell+2$ as stated in \cref{lem:Us}.
The overall circuit for implementing $\Phi^\ell_{\UNOT}$ is shown in \cref{fig:Circuits}, right.

To complete the implementation of the channels $\Phi^\ell_{\Tr}$ and $\Phi^\ell_{\UNOT}$, we simply need to initialize the first register in $\ket{0}$ and $\ket{1}$, respectively, and discard the second register after the isometry has been applied.
The conditional rotation in \cref{eq:cR} can be implemented using $O(\ell \log \ell)$ elementary gates, while the incrementer and decrementer require only $O(\log \ell)$ gates \cite{Adder}.
\end{proof}

\section{Monotonicity of fidelity in the Hamming weight}\label{apx:Monotonicity}

The main technical result of this appendix is \cref{lem:Monotonicity} which shows that the fidelity $F_n(h)$ for computing quantum majority on~$n$ qubits is strictly decreasing in the Hamming weight $h \in \set{0,\dotsc,\frac{n-1}{2}}$ of the input string.
That is, as the number of non-majority states increases, it becomes harder to output the correct majority state.
We use this to derive a recursive formula for the optimal fidelity of the $n$-qubit quantum majority vote in \cref{lem:RecursionG} at the end of this appendix.

Recall from \cref{eq:Fnx} that
\begin{equation}\label{eq:minF}
  F_{\MAJ}(n) = \min_{0 \leq h \leq \frac{n-1}{2}} F_n(h)
\end{equation}
where
\begin{equation}\label{eq:Fnh}
  F_n(h)
  := \sum_{k=0}^h  \frac{\binom{n}{k} - \binom{n}{k-1}}{\binom{n}{h}}
     \left( \frac{n-k-h}{n-2k} \right).
\end{equation}
Our goal is to show that $F_n(h)$ is monotonically decreasing in the Hamming weight $h$, and hence the minimum in \cref{eq:minF} is achieved at $h = \frac{n-1}{2}$. We will do this by first deriving a recurrence relation for $F_n(h)$ in $h$.

\begin{lemma}\label{lem:Recurrence}
$F_n(h)$, as defined in \cref{eq:Fnh}, satisfies the following linear recurrence relation:
\begin{align}
  F_n(0) &= 1, \\
  F_n(h) &= a(n,h) F_n(h-1) + b(n,h), \label{eq:Frec}
\end{align}
where the coefficients $a(n,h)$ and $b(n,h)$ are given by the following rational functions:
\begin{align}\label{eq:ab}
  a(n,h) &:= \frac{(n-2h)h}{(n-2h+2)(n-h+1)}, &
  b(n,h) &:= \frac{4 h^2 - (4n+5) h + (n+1)(n+2)}{(n-2h+2)(n-h+1)}.
\end{align}
\end{lemma}

\begin{proof}
Note that $F_n(0) = 1$ because $\binom{n}{0} = 1$ and $\binom{n}{-1} = 0$, so we only need to show \cref{eq:Frec}. Define
\begin{align}\label{eq:PnT}
  P(n,h,k) &:= \frac{\binom{n}{k}-\binom{n}{k-1}}{\binom{n}{h}}, &
  T(n,h,k) &:= P(n,h,k) \frac{n-k-h}{n-2k},
\end{align}
and notice that
\begin{align}
  \sum_{k=0}^h P(n,h,k) &= 1, &
  \sum_{k=0}^h T(n,h,k) &= F_n(h),
  \label{eq:PTsums}
\end{align}
where the first identity follows by canceling all intermediate terms in the telescoping sum, and the second identity follows directly from \cref{eq:Fnh}. Let us use these identities to write both sides of \cref{eq:Frec} as sums of the form $\sum_{k=0}^{h-1}$:
\begin{align}
  F_n(h)
   = \sum_{k=0}^{h-1} T(n,h,k) + T(n,h,h)
  &= \sum_{k=0}^{h-1} \sof[\Big]{ T(n,h,k) + T(n,h,h) P(n,h-1,k) }, \\
  a(n,h) F_n(h-1) + b(n,h)
  &= \sum_{k=0}^{h-1} \sof[\Big]{ a(n,h) T(n,h-1,k) + b(n,h) P(n,h-1,k) }.
\end{align}
The reason why both sides are equal is because the two sums agree term-wise. That is,
\begin{equation}
  T(n,h,k) + T(n,h,h) P(n,h-1,k) = a(n,h) T(n,h-1,k) + b(n,h) P(n,h-1,k).
\end{equation}
Verifying this is a matter of substituting the definitions of all quantities and observing that once the identity $\binom{n}{x-1} = \frac{x}{n-x+1} \binom{n}{x}$ is used, everything cancels out after a tedious simplification.
\end{proof}

We will now use the recurrence to show that $F_n(h)$ is decreasing in $h$. We assume that $n \geq 3$ since $n$ has to be odd and there is nothing to show when $n = 1$ as only one value of $h$ is possible.

\begin{lemma}\label{lem:Monotonicity}
For any odd integer $n \geq 3$, the fidelity $F_n(h)$ of computing the $n$-bit quantum majority on a string of Hamming weight $h$, see \cref{eq:Fnh}, is strictly decreasing in $h \in \set{0,\dotsc,\frac{n-1}{2}}$.
\end{lemma}

\begin{proof}
Let us first show that $a(n,h)$ defined in \cref{eq:ab} is bounded as follows:
\begin{equation}
  0 < a(n,h) < 1, \qquad h \in \sof{1, \tfrac{n-1}{2}},
\end{equation}
which we will be using throughout the rest of the proof. Since $h \leq \frac{n-1}{2}$, we see from \cref{eq:ab} that both the numerator and the denominator of $a(n,t)$ are strictly positive, so $a(n,h) > 0$. The inequality $a(n,h) < 1$ is equivalent to the denominator of $a(n,h)$ being strictly larger than the numerator. This in turn is equivalent to
\begin{equation}\label{eq:p}
  p(n,h) := 4 h^2 - 4 (n+1) h + (n+1) (n+2) > 0, \qquad h \in \sof{1, \tfrac{n-1}{2}}.
\end{equation}
Note that $p(n,1) = n(n-1)+2 > 0$, $\frac{d}{dh} p(n,h) = -4(n+1-2h) < 0$, and $p(n,\frac{n-1}{2}) = n+5 > 0$, confirming that $p(n,h) > 0$ when $h$ is in the specified interval, and hence that $a(n,h) < 1$.

The following quantity will play an important role in the rest of the proof:
\begin{equation}\label{eq:fnt}
  f_n(h) := \frac{b(n,h)}{1-a(n,h)}.
\end{equation}
We will need two properties of $f_n(h)$:
\begin{align}
  &a(n,h) f_n(h) + b(n,h) = f_n(h), \label{eq:fprop1}\\
  &f_n(h) > f_n(h+1), \qquad h \in \set{1, \dotsc, \tfrac{n-1}{2}-1}. \label{eq:fprop2}
\end{align}
The first property follows immediately from the definition in \cref{eq:fnt} and it says that $f_n(h)$ is fixed by the map $x \mapsto a(n,h) x + b(n,h)$ (this is in contrast to $F_n(h-1)$ which gets mapped to $F_n(h)$ by the same map because of the recurrence in \cref{eq:Frec}). We will show the second property later.

Recall that our goal is to show that $F_n(h-1) > F_n(h)$ when $h \in \set{1,\dotsc,\frac{n-1}{2}}$. If we substitute $F_n(h-1) = \of{F_n(h)-b(n,h)}/a(n,h)$ from the recurrence in \cref{eq:Frec}, this is equivalent to showing that
\begin{equation}\label{eq:cnt}
  F_n(h) > f_n(h), \qquad h \in \set{1,\dotsc,\tfrac{n-1}{2}}.
\end{equation}
We can prove \cref{eq:cnt} by induction. The basis case\footnote{Notice that for $h = 0$ we actually have an equality since $F_n(0) = 1 = f_n(0)$.} $h=1$ holds since $n \geq 3$ and thus
\begin{equation}
  F_n(1) = 1 - \frac{1}{n^2} > 1 - \frac{1}{n^2-n+2} = f_n(1).
\end{equation}
Assuming $F_n(h) > f_n(h)$ for some $h \in \set{1, \dotsc, \frac{n-1}{2}-1}$, let us show that $F_n(h+1) > f_n(h+1)$. Using the recurrence in \cref{eq:Frec} and the inductive assumption,
\begin{align}
  F_n(h+1)
  &= a(n,h+1) F_n(h) + b(n,h+1) \\
  &> a(n,h+1) f_n(h) + b(n,h+1) \\
  &> a(n,h+1) f_n(h+1) + b(n,h+1) \\
  &= f_n(h+1),
\end{align}
where we used the two properties of $f_n(h)$ stated in \cref{eq:fprop1,eq:fprop2}.

To conclude the proof, it remains to show the second property of $f_n(h)$ in \cref{eq:fprop2}. More specifically, we will show that $\frac{d}{dh} f_n(h) < 0$ when $h \in \sof{1, \frac{n-1}{2}}$. We can get an explicit formula for $f_n(h)$ using the definition in \cref{eq:fnt} and the formulas for $a(n,h)$ and $b(n,h)$ in \cref{eq:ab}:
\begin{equation}
  f_n(h)
  = \frac{b(n,h)}{1-a(n,h)}
  = 1 - \frac{h}{p(n,h)},
\end{equation}
where $p(n,h)$ is the same polynomial as defined earlier in \cref{eq:p}. Note that
\begin{equation}
  \frac{d}{dh} f_n(h)
  = \frac{h \frac{d}{dh} p(n,h) - p(n,h)}{p(n,h)^2}
  = \frac{4 h^2 - (n+1) (n+2)}{p(n,h)^2} < 0
\end{equation}
because $4 h^2 \leq (n-1)^2 < (n+1)(n+2)$ when $1 \leq h \leq \frac{n-1}{2}$. This completes the proof of \cref{eq:fprop2} and hence also the proof of the lemma.
\end{proof}

We can now combine \cref{lem:Recurrence,lem:Monotonicity} to prove \cref{lem:RecursionG}, which was originally stated in \cref{sec:NumericMajority}.

\recursion*

\begin{proof}
Recall from \cref{eq:minF} that $F_{\MAJ}(n) := \min_{0 \leq \ham \leq \frac{n-1}{2}} F_n(\ham)$
where $F_n(\ham)$ is defined in \cref{eq:Fnh}.
Since we just showed in \cref{lem:Monotonicity} that the minimum is achieved when $\ham = \frac{n-1}{2}$, we conclude that $F_{\MAJ}(n) = F_n(\frac{n-1}{2})$.
\Cref{lem:Recurrence} gives us a linear recursion for $F_n(\ham)$ in terms of $F_n(\ham-1)$.
By unwinding this recursion all the way to $F_n(0) = 1$, we get that
\begin{equation}
  F_n(\ham) = \prod_{k=1}^\ham a(n,k) + \sum_{j=1}^\ham \of*{\prod_{k=j+1}^\ham a(n,k)} b(n,j).
\end{equation}
We can evaluate this with a computer algebra system and get an explicit (but unwieldy) formula for $F_n(\ham)$.
We can then substitute $\ham = \frac{n-1}{2}$ and equate the result with $g(\frac{n+1}{2})$, and use a computer algebra system to verify that $g$ satisfies the stated recursion.
\end{proof}

\section{Numerical Choi matrices}\label{apx:Choi}

In this appendix, we provide exact numerical Choi matrices for computing all symmetric and equivariant Boolean functions with $n = 1$ and $n = 3$ arguments (see \cref{tab:1-3bit}).
In \cref{apx:Physical}, we list Choi matrices for quantum channels with optimal fidelity,
while in \cref{apx:Ideal} we list them for superoperators that implement the ideal functionality, \ie, achieve perfect fidelity on all inputs.
These matrices where obtained using the method outlined in \cref{sec:U-equivariant functions}.

\subsection{Choi matrices for optimal quantum channels}\label{apx:Physical}

A slight extension of the method outlined in \cref{sec:U-equivariant functions} can be used to obtain an explicit Choi matrix of a quantum channel that computes a given symmetric and equivariant Boolean function with optimal fidelity.

Let $J$ denote the Choi matrix of the channel in question.
To achieve fidelities $c_s \in [0,1]$ on inputs $s \in \set{0,1}^n$, we modify \cref{eq:ideal map on all r} as follows:
\begin{align}
  \Tr_2 \sof[\big]{J \cdot (\id_2 \x \rho(s,\vec{r})\tp)}
  &= \rho\of[\big]{(2 c_s - 1)(-1)^{f(s)}\vec{r}}, &
  \forall s &\in \set{0,1}^n, \quad
  \forall \vec{r} \in \Sphere^2.
  \label{eq:fidelity cs}
\end{align}
The coefficient $2 c_s - 1$ indeed guarantees fidelity $c_s$ since one can verify by an explicit calculation that
\begin{equation}
  \Tr \sof[\Big]{ \rho(\vec{r}) \rho\of[\big]{(2c_s-1)\vec{r}} }
  = c_s (x^2 + y^2 + z^2) - \frac{x^2 + y^2 + z^2 - 1}{2},
\end{equation}
which modulo $x^2 + y^2 + z^2 - 1$ is equal to $c_s$.
As before, it suffices to consider only one $s \in \set{0,1}^n$ for each Hamming weight $|s| \in \set{0,\dotsc,\floor{n/2}}$ in \cref{eq:fidelity cs}.
Optimal values for fidelities $c_s$ can be easily obtained from the linear program \eqref{eq:LP}.
In particular, for $n = 1$ and $n = 3$ they are listed in the last two columns of \cref{tab:Optimal}.

The following are Choi matrices for quantum channels that compute all symmetric equivariant Boolean functions with $n = 1$ arguments with optimal fidelity:
\begin{align*}
  J_\ID &=
  \rmx{
    1 & 0 & 0 & 1 \\
    0 & 0 & 0 & 0 \\
    0 & 0 & 0 & 0 \\
    1 & 0 & 0 & 1
  }, &
  J_\NOT &= \frac{1}{3}
  \rmx{
    1 & 0 & 0 &-1 \\
    0 & 2 & 0 & 0 \\
    0 & 0 & 2 & 0 \\
   -1 & 0 & 0 & 1
  }.
\end{align*}
Same, but for functions with $n = 3$ arguments:
\begin{equation*}
  J_{\MAJ_3} = \frac{1}{9}
  \srmx{
    9 & 0 & 0 & 0 & 0 & 0 & 0 & 0 & 0 & 3 & 3 & 0 & 3 & 0 & 0 & 0 \\
    0 & 8 & -1 & 0 & -1 & 0 & 0 & 0 & 0 & 0 & 0 & 5 & 0 & 5 & -4 & 0 \\
    0 & -1 & 8 & 0 & -1 & 0 & 0 & 0 & 0 & 0 & 0 & 5 & 0 & -4 & 5 & 0 \\
    0 & 0 & 0 & 1 & 0 & 1 & 1 & 0 & 0 & 0 & 0 & 0 & 0 & 0 & 0 & 3 \\
    0 & -1 & -1 & 0 & 8 & 0 & 0 & 0 & 0 & 0 & 0 & -4 & 0 & 5 & 5 & 0 \\
    0 & 0 & 0 & 1 & 0 & 1 & 1 & 0 & 0 & 0 & 0 & 0 & 0 & 0 & 0 & 3 \\
    0 & 0 & 0 & 1 & 0 & 1 & 1 & 0 & 0 & 0 & 0 & 0 & 0 & 0 & 0 & 3 \\
    0 & 0 & 0 & 0 & 0 & 0 & 0 & 0 & 0 & 0 & 0 & 0 & 0 & 0 & 0 & 0 \\
    0 & 0 & 0 & 0 & 0 & 0 & 0 & 0 & 0 & 0 & 0 & 0 & 0 & 0 & 0 & 0 \\
    3 & 0 & 0 & 0 & 0 & 0 & 0 & 0 & 0 & 1 & 1 & 0 & 1 & 0 & 0 & 0 \\
    3 & 0 & 0 & 0 & 0 & 0 & 0 & 0 & 0 & 1 & 1 & 0 & 1 & 0 & 0 & 0 \\
    0 & 5 & 5 & 0 & -4 & 0 & 0 & 0 & 0 & 0 & 0 & 8 & 0 & -1 & -1 & 0 \\
    3 & 0 & 0 & 0 & 0 & 0 & 0 & 0 & 0 & 1 & 1 & 0 & 1 & 0 & 0 & 0 \\
    0 & 5 & -4 & 0 & 5 & 0 & 0 & 0 & 0 & 0 & 0 & -1 & 0 & 8 & -1 & 0 \\
    0 & -4 & 5 & 0 & 5 & 0 & 0 & 0 & 0 & 0 & 0 & -1 & 0 & -1 & 8 & 0 \\
    0 & 0 & 0 & 3 & 0 & 3 & 3 & 0 & 0 & 0 & 0 & 0 & 0 & 0 & 0 & 9
  },
\end{equation*}
\begin{equation*}
  J_{\PAR_3} = \frac{1}{15}
  \srmx{
    9 & 0 & 0 & 0 & 0 & 0 & 0 & 0 & 0 & 1 & 1 & 0 & 1 & 0 & 0 & 0 \\
    0 & 6 & 1 & 0 & 1 & 0 & 0 & 0 & 0 & 0 & 0 & -1 & 0 & -1 & 4 & 0 \\
    0 & 1 & 6 & 0 & 1 & 0 & 0 & 0 & 0 & 0 & 0 & -1 & 0 & 4 & -1 & 0 \\
    0 & 0 & 0 & 9 & 0 & -1 & -1 & 0 & 0 & 0 & 0 & 0 & 0 & 0 & 0 & 1 \\
    0 & 1 & 1 & 0 & 6 & 0 & 0 & 0 & 0 & 0 & 0 & 4 & 0 & -1 & -1 & 0 \\
    0 & 0 & 0 & -1 & 0 & 9 & -1 & 0 & 0 & 0 & 0 & 0 & 0 & 0 & 0 & 1 \\
    0 & 0 & 0 & -1 & 0 & -1 & 9 & 0 & 0 & 0 & 0 & 0 & 0 & 0 & 0 & 1 \\
    0 & 0 & 0 & 0 & 0 & 0 & 0 & 6 & 0 & 0 & 0 & 0 & 0 & 0 & 0 & 0 \\
    0 & 0 & 0 & 0 & 0 & 0 & 0 & 0 & 6 & 0 & 0 & 0 & 0 & 0 & 0 & 0 \\
    1 & 0 & 0 & 0 & 0 & 0 & 0 & 0 & 0 & 9 & -1 & 0 & -1 & 0 & 0 & 0 \\
    1 & 0 & 0 & 0 & 0 & 0 & 0 & 0 & 0 & -1 & 9 & 0 & -1 & 0 & 0 & 0 \\
    0 & -1 & -1 & 0 & 4 & 0 & 0 & 0 & 0 & 0 & 0 & 6 & 0 & 1 & 1 & 0 \\
    1 & 0 & 0 & 0 & 0 & 0 & 0 & 0 & 0 & -1 & -1 & 0 & 9 & 0 & 0 & 0 \\
    0 & -1 & 4 & 0 & -1 & 0 & 0 & 0 & 0 & 0 & 0 & 1 & 0 & 6 & 1 & 0 \\
    0 & 4 & -1 & 0 & -1 & 0 & 0 & 0 & 0 & 0 & 0 & 1 & 0 & 1 & 6 & 0 \\
    0 & 0 & 0 & 1 & 0 & 1 & 1 & 0 & 0 & 0 & 0 & 0 & 0 & 0 & 0 & 9
  },
\end{equation*}
\begin{equation*}
  J_{\NPAR_3} = \frac{1}{5}
  \srmx{
    1 & 0 & 0 & 0 & 0 & 0 & 0 & 0 & 0 & -1 & -1 & 0 & -1 & 0 & 0 & 0 \\
    0 & 4 & -1 & 0 & -1 & 0 & 0 & 0 & 0 & 0 & 0 & 1 & 0 & 1 & -4 & 0 \\
    0 & -1 & 4 & 0 & -1 & 0 & 0 & 0 & 0 & 0 & 0 & 1 & 0 & -4 & 1 & 0 \\
    0 & 0 & 0 & 1 & 0 & 1 & 1 & 0 & 0 & 0 & 0 & 0 & 0 & 0 & 0 & -1 \\
    0 & -1 & -1 & 0 & 4 & 0 & 0 & 0 & 0 & 0 & 0 & -4 & 0 & 1 & 1 & 0 \\
    0 & 0 & 0 & 1 & 0 & 1 & 1 & 0 & 0 & 0 & 0 & 0 & 0 & 0 & 0 & -1 \\
    0 & 0 & 0 & 1 & 0 & 1 & 1 & 0 & 0 & 0 & 0 & 0 & 0 & 0 & 0 & -1 \\
    0 & 0 & 0 & 0 & 0 & 0 & 0 & 4 & 0 & 0 & 0 & 0 & 0 & 0 & 0 & 0 \\
    0 & 0 & 0 & 0 & 0 & 0 & 0 & 0 & 4 & 0 & 0 & 0 & 0 & 0 & 0 & 0 \\
    -1 & 0 & 0 & 0 & 0 & 0 & 0 & 0 & 0 & 1 & 1 & 0 & 1 & 0 & 0 & 0 \\
    -1 & 0 & 0 & 0 & 0 & 0 & 0 & 0 & 0 & 1 & 1 & 0 & 1 & 0 & 0 & 0 \\
    0 & 1 & 1 & 0 & -4 & 0 & 0 & 0 & 0 & 0 & 0 & 4 & 0 & -1 & -1 & 0 \\
    -1 & 0 & 0 & 0 & 0 & 0 & 0 & 0 & 0 & 1 & 1 & 0 & 1 & 0 & 0 & 0 \\
    0 & 1 & -4 & 0 & 1 & 0 & 0 & 0 & 0 & 0 & 0 & -1 & 0 & 4 & -1 & 0 \\
    0 & -4 & 1 & 0 & 1 & 0 & 0 & 0 & 0 & 0 & 0 & -1 & 0 & -1 & 4 & 0 \\
    0 & 0 & 0 & -1 & 0 & -1 & -1 & 0 & 0 & 0 & 0 & 0 & 0 & 0 & 0 & 1
  },
\end{equation*}
\begin{equation*}
  J_{\NMAJ_3} = \frac{1}{45}
  \srmx{
    9 & 0 & 0 & 0 & 0 & 0 & 0 & 0 & 0 & -9 & -9 & 0 & -9 & 0 & 0 & 0 \\
    0 & 16 & 1 & 0 & 1 & 0 & 0 & 0 & 0 & 0 & 0 & -11 & 0 & -11 & 4 & 0 \\
    0 & 1 & 16 & 0 & 1 & 0 & 0 & 0 & 0 & 0 & 0 & -11 & 0 & 4 & -11 & 0 \\
    0 & 0 & 0 & 29 & 0 & -1 & -1 & 0 & 0 & 0 & 0 & 0 & 0 & 0 & 0 & -9 \\
    0 & 1 & 1 & 0 & 16 & 0 & 0 & 0 & 0 & 0 & 0 & 4 & 0 & -11 & -11 & 0 \\
    0 & 0 & 0 & -1 & 0 & 29 & -1 & 0 & 0 & 0 & 0 & 0 & 0 & 0 & 0 & -9 \\
    0 & 0 & 0 & -1 & 0 & -1 & 29 & 0 & 0 & 0 & 0 & 0 & 0 & 0 & 0 & -9 \\
    0 & 0 & 0 & 0 & 0 & 0 & 0 & 36 & 0 & 0 & 0 & 0 & 0 & 0 & 0 & 0 \\
    0 & 0 & 0 & 0 & 0 & 0 & 0 & 0 & 36 & 0 & 0 & 0 & 0 & 0 & 0 & 0 \\
    -9 & 0 & 0 & 0 & 0 & 0 & 0 & 0 & 0 & 29 & -1 & 0 & -1 & 0 & 0 & 0 \\
    -9 & 0 & 0 & 0 & 0 & 0 & 0 & 0 & 0 & -1 & 29 & 0 & -1 & 0 & 0 & 0 \\
    0 & -11 & -11 & 0 & 4 & 0 & 0 & 0 & 0 & 0 & 0 & 16 & 0 & 1 & 1 & 0 \\
    -9 & 0 & 0 & 0 & 0 & 0 & 0 & 0 & 0 & -1 & -1 & 0 & 29 & 0 & 0 & 0 \\
    0 & -11 & 4 & 0 & -11 & 0 & 0 & 0 & 0 & 0 & 0 & 1 & 0 & 16 & 1 & 0 \\
    0 & 4 & -11 & 0 & -11 & 0 & 0 & 0 & 0 & 0 & 0 & 1 & 0 & 1 & 16 & 0 \\
    0 & 0 & 0 & -9 & 0 & -9 & -9 & 0 & 0 & 0 & 0 & 0 & 0 & 0 & 0 & 9
  }.
\end{equation*}

\subsection{Choi matrices for unitary-equivariant Boolean functions}\label{apx:Ideal}

This appendix lists exact numerical Choi matrices for superoperators $\L(\C^{2^n}) \to \L(\C^2)$ that perfectly implement the ideal functionality, \ie, they compute an equivariant Boolean function $f: \set{0,1}^n \to \set{0,1}$ in a basis-independent way with fidelity one.
These superoperators can be considered as unitary-equivariant extensions of Boolean functions (see \cref{sec:U-equivariant functions}).

To distinguish the Choi matrices in this appendix from those in \cref{apx:Physical}, we use $\JJ$ instead of $J$.
An important difference between the two is that the $\JJ$ are \emph{not} positive semidefinite, and hence the corresponding superoperators are \emph{not} completely positive.\footnote{Except for the identity function $\ID$ for which $J_\ID = \JJ_\ID \geq 0$.}
Similar to \cref{apx:Physical}, the $\JJ$ listed below are obtained using the method described in \cref{sec:U-equivariant functions}, with the output fidelity set to one for all inputs.

The following Choi matrices are for superoperators that exactly compute all symmetric equivariant Boolean functions with $n = 1$ arguments:
\begin{align*}
  \JJ_\ID &=
  \rmx{
    1 & 0 & 0 & 1 \\
    0 & 0 & 0 & 0 \\
    0 & 0 & 0 & 0 \\
    1 & 0 & 0 & 1
  }, &
  \JJ_\NOT &=
  \rmx{
    0 & 0 & 0 &-1 \\
    0 & 1 & 0 & 0 \\
    0 & 0 & 1 & 0 \\
   -1 & 0 & 0 & 0
  }.
\end{align*}
Same, but for functions with $n = 3$ arguments (see \cref{tab:1-3bit}):
\begin{equation*}
  \JJ_{\MAJ_3} = \frac{1}{6}
  \srmx{
    6 & 0 & 0 & 0 & 0 & 0 & 0 & 0 & 0 & 2 & 2 & 0 & 2 & 0 & 0 & 0 \\
    0 & 6 & -1 & 0 & -1 & 0 & 0 & 0 & 0 & 0 & 0 & 4 & 0 & 4 & -4 & 0 \\
    0 & -1 & 6 & 0 & -1 & 0 & 0 & 0 & 0 & 0 & 0 & 4 & 0 & -4 & 4 & 0 \\
    0 & 0 & 0 & 0 & 0 & 1 & 1 & 0 & 0 & 0 & 0 & 0 & 0 & 0 & 0 & 2 \\
    0 & -1 & -1 & 0 & 6 & 0 & 0 & 0 & 0 & 0 & 0 & -4 & 0 & 4 & 4 & 0 \\
    0 & 0 & 0 & 1 & 0 & 0 & 1 & 0 & 0 & 0 & 0 & 0 & 0 & 0 & 0 & 2 \\
    0 & 0 & 0 & 1 & 0 & 1 & 0 & 0 & 0 & 0 & 0 & 0 & 0 & 0 & 0 & 2 \\
    0 & 0 & 0 & 0 & 0 & 0 & 0 & 0 & 0 & 0 & 0 & 0 & 0 & 0 & 0 & 0 \\
    0 & 0 & 0 & 0 & 0 & 0 & 0 & 0 & 0 & 0 & 0 & 0 & 0 & 0 & 0 & 0 \\
    2 & 0 & 0 & 0 & 0 & 0 & 0 & 0 & 0 & 0 & 1 & 0 & 1 & 0 & 0 & 0 \\
    2 & 0 & 0 & 0 & 0 & 0 & 0 & 0 & 0 & 1 & 0 & 0 & 1 & 0 & 0 & 0 \\
    0 & 4 & 4 & 0 & -4 & 0 & 0 & 0 & 0 & 0 & 0 & 6 & 0 & -1 & -1 & 0 \\
    2 & 0 & 0 & 0 & 0 & 0 & 0 & 0 & 0 & 1 & 1 & 0 & 0 & 0 & 0 & 0 \\
    0 & 4 & -4 & 0 & 4 & 0 & 0 & 0 & 0 & 0 & 0 & -1 & 0 & 6 & -1 & 0 \\
    0 & -4 & 4 & 0 & 4 & 0 & 0 & 0 & 0 & 0 & 0 & -1 & 0 & -1 & 6 & 0 \\
    0 & 0 & 0 & 2 & 0 & 2 & 2 & 0 & 0 & 0 & 0 & 0 & 0 & 0 & 0 & 6
  },
\end{equation*}
\begin{equation*}
  \JJ_{\PAR_3} = \frac{1}{3}
  \srmx{
    3 & 0 & 0 & 0 & 0 & 0 & 0 & 0 & 0 & 1 & 1 & 0 & 1 & 0 & 0 & 0 \\
    0 & 0 & 1 & 0 & 1 & 0 & 0 & 0 & 0 & 0 & 0 & -1 & 0 & -1 & 4 & 0 \\
    0 & 1 & 0 & 0 & 1 & 0 & 0 & 0 & 0 & 0 & 0 & -1 & 0 & 4 & -1 & 0 \\
    0 & 0 & 0 & 3 & 0 & -1 & -1 & 0 & 0 & 0 & 0 & 0 & 0 & 0 & 0 & 1 \\
    0 & 1 & 1 & 0 & 0 & 0 & 0 & 0 & 0 & 0 & 0 & 4 & 0 & -1 & -1 & 0 \\
    0 & 0 & 0 & -1 & 0 & 3 & -1 & 0 & 0 & 0 & 0 & 0 & 0 & 0 & 0 & 1 \\
    0 & 0 & 0 & -1 & 0 & -1 & 3 & 0 & 0 & 0 & 0 & 0 & 0 & 0 & 0 & 1 \\
    0 & 0 & 0 & 0 & 0 & 0 & 0 & 0 & 0 & 0 & 0 & 0 & 0 & 0 & 0 & 0 \\
    0 & 0 & 0 & 0 & 0 & 0 & 0 & 0 & 0 & 0 & 0 & 0 & 0 & 0 & 0 & 0 \\
    1 & 0 & 0 & 0 & 0 & 0 & 0 & 0 & 0 & 3 & -1 & 0 & -1 & 0 & 0 & 0 \\
    1 & 0 & 0 & 0 & 0 & 0 & 0 & 0 & 0 & -1 & 3 & 0 & -1 & 0 & 0 & 0 \\
    0 & -1 & -1 & 0 & 4 & 0 & 0 & 0 & 0 & 0 & 0 & 0 & 0 & 1 & 1 & 0 \\
    1 & 0 & 0 & 0 & 0 & 0 & 0 & 0 & 0 & -1 & -1 & 0 & 3 & 0 & 0 & 0 \\
    0 & -1 & 4 & 0 & -1 & 0 & 0 & 0 & 0 & 0 & 0 & 1 & 0 & 0 & 1 & 0 \\
    0 & 4 & -1 & 0 & -1 & 0 & 0 & 0 & 0 & 0 & 0 & 1 & 0 & 1 & 0 & 0 \\
    0 & 0 & 0 & 1 & 0 & 1 & 1 & 0 & 0 & 0 & 0 & 0 & 0 & 0 & 0 & 3
  },
\end{equation*}
\begin{equation*}
  \JJ_{\NPAR_3} = \frac{1}{3}
  \srmx{
    0 & 0 & 0 & 0 & 0 & 0 & 0 & 0 & 0 & -1 & -1 & 0 & -1 & 0 & 0 & 0 \\
    0 & 3 & -1 & 0 & -1 & 0 & 0 & 0 & 0 & 0 & 0 & 1 & 0 & 1 & -4 & 0 \\
    0 & -1 & 3 & 0 & -1 & 0 & 0 & 0 & 0 & 0 & 0 & 1 & 0 & -4 & 1 & 0 \\
    0 & 0 & 0 & 0 & 0 & 1 & 1 & 0 & 0 & 0 & 0 & 0 & 0 & 0 & 0 & -1 \\
    0 & -1 & -1 & 0 & 3 & 0 & 0 & 0 & 0 & 0 & 0 & -4 & 0 & 1 & 1 & 0 \\
    0 & 0 & 0 & 1 & 0 & 0 & 1 & 0 & 0 & 0 & 0 & 0 & 0 & 0 & 0 & -1 \\
    0 & 0 & 0 & 1 & 0 & 1 & 0 & 0 & 0 & 0 & 0 & 0 & 0 & 0 & 0 & -1 \\
    0 & 0 & 0 & 0 & 0 & 0 & 0 & 3 & 0 & 0 & 0 & 0 & 0 & 0 & 0 & 0 \\
    0 & 0 & 0 & 0 & 0 & 0 & 0 & 0 & 3 & 0 & 0 & 0 & 0 & 0 & 0 & 0 \\
    -1 & 0 & 0 & 0 & 0 & 0 & 0 & 0 & 0 & 0 & 1 & 0 & 1 & 0 & 0 & 0 \\
    -1 & 0 & 0 & 0 & 0 & 0 & 0 & 0 & 0 & 1 & 0 & 0 & 1 & 0 & 0 & 0 \\
    0 & 1 & 1 & 0 & -4 & 0 & 0 & 0 & 0 & 0 & 0 & 3 & 0 & -1 & -1 & 0 \\
    -1 & 0 & 0 & 0 & 0 & 0 & 0 & 0 & 0 & 1 & 1 & 0 & 0 & 0 & 0 & 0 \\
    0 & 1 & -4 & 0 & 1 & 0 & 0 & 0 & 0 & 0 & 0 & -1 & 0 & 3 & -1 & 0 \\
    0 & -4 & 1 & 0 & 1 & 0 & 0 & 0 & 0 & 0 & 0 & -1 & 0 & -1 & 3 & 0 \\
    0 & 0 & 0 & -1 & 0 & -1 & -1 & 0 & 0 & 0 & 0 & 0 & 0 & 0 & 0 & 0
  },
\end{equation*}
\begin{equation*}
  \JJ_{\NMAJ_3} = \frac{1}{6}
  \srmx{
    0 & 0 & 0 & 0 & 0 & 0 & 0 & 0 & 0 & -2 & -2 & 0 & -2 & 0 & 0 & 0 \\
    0 & 0 & 1 & 0 & 1 & 0 & 0 & 0 & 0 & 0 & 0 & -4 & 0 & -4 & 4 & 0 \\
    0 & 1 & 0 & 0 & 1 & 0 & 0 & 0 & 0 & 0 & 0 & -4 & 0 & 4 & -4 & 0 \\
    0 & 0 & 0 & 6 & 0 & -1 & -1 & 0 & 0 & 0 & 0 & 0 & 0 & 0 & 0 & -2 \\
    0 & 1 & 1 & 0 & 0 & 0 & 0 & 0 & 0 & 0 & 0 & 4 & 0 & -4 & -4 & 0 \\
    0 & 0 & 0 & -1 & 0 & 6 & -1 & 0 & 0 & 0 & 0 & 0 & 0 & 0 & 0 & -2 \\
    0 & 0 & 0 & -1 & 0 & -1 & 6 & 0 & 0 & 0 & 0 & 0 & 0 & 0 & 0 & -2 \\
    0 & 0 & 0 & 0 & 0 & 0 & 0 & 6 & 0 & 0 & 0 & 0 & 0 & 0 & 0 & 0 \\
    0 & 0 & 0 & 0 & 0 & 0 & 0 & 0 & 6 & 0 & 0 & 0 & 0 & 0 & 0 & 0 \\
    -2 & 0 & 0 & 0 & 0 & 0 & 0 & 0 & 0 & 6 & -1 & 0 & -1 & 0 & 0 & 0 \\
    -2 & 0 & 0 & 0 & 0 & 0 & 0 & 0 & 0 & -1 & 6 & 0 & -1 & 0 & 0 & 0 \\
    0 & -4 & -4 & 0 & 4 & 0 & 0 & 0 & 0 & 0 & 0 & 0 & 0 & 1 & 1 & 0 \\
    -2 & 0 & 0 & 0 & 0 & 0 & 0 & 0 & 0 & -1 & -1 & 0 & 6 & 0 & 0 & 0 \\
    0 & -4 & 4 & 0 & -4 & 0 & 0 & 0 & 0 & 0 & 0 & 1 & 0 & 0 & 1 & 0 \\
    0 & 4 & -4 & 0 & -4 & 0 & 0 & 0 & 0 & 0 & 0 & 1 & 0 & 1 & 0 & 0 \\
    0 & 0 & 0 & -2 & 0 & -2 & -2 & 0 & 0 & 0 & 0 & 0 & 0 & 0 & 0 & 0
  }.
\end{equation*}

\end{document}